\documentclass[%
 reprint,
nofootinbib,
 amsmath,amssymb,
 aps,
 pra,
]{revtex4-2}

\usepackage{graphicx}
\usepackage{dcolumn}
\usepackage{bm}
\usepackage{hyperref}

\usepackage{contour,xcolor}
\contourlength{1pt}

\usepackage{subcaption}
\usepackage{amsthm} 
\usepackage{braket}
\usepackage{mathtools}

\usepackage{thmtools, thm-restate} 

\newtheorem{theorem}{Theorem}
\newtheorem{lemma}{Lemma}
\newtheorem{corollary}{Corollary}

\newcommand{\comm}[2]{ { \left[  #1 , #2 \right]_-  } }
\newcommand{\acomm}[2]{ { \{  #1 , #2 \}_+  } }

\DeclareMathOperator{\Tr}{Tr}
\DeclareMathOperator{\anc}{anc}
\DeclareMathOperator{\N}{N}
\DeclareMathOperator{\C}{C}
\DeclareMathOperator{\depth}{depth}
\DeclareMathOperator{\p}{V}

\begin{document}

\title{Optimal logical Bell measurements on stabilizer codes with linear optics}
\author{Simon D. Reiß}
\email{sreiss@uni-mainz.de}
\author{Peter van Loock}
\email{loock@uni-mainz.de}
\affiliation{Johannes-Gutenberg University of Mainz, Institute of Physics, Staudingerweg 7, 55128 Mainz, Germany}

\begin{abstract}
Bell measurements~(BMs) are ubiquitous in quantum information and technology. They are basic elements for quantum commmunication, computation, and error correction. In particular, when performed on logical qubits encoded in physical photonic qubits, they allow for a read-out of stabilizer syndrome information to enhance loss tolerance in qubit-state transmission and fusion. However, even in an ideal setting without photon loss, BMs cannot be done perfectly based on the simplest experimental toolbox of linear optics. Here we demonstrate that any logical BM on stabilizer codes can always be mapped onto a single physical BM perfomed on any qubit pair from the two codes. As a necessary condition for the success of a logical BM, this provides a general upper bound on its success probability, especially ruling out the possibility that the stabilizer information obtainable from  only partially succeeding, physical linear-optics BMs could be combined into the full logical stabilizer information. We formulate sufficient criteria to find schemes for which a single successful BM on the physical level will always allow to obtain the full logical information by suitably adapting the subsequent physical measurements. Our approach based on stabilizer group theory is generally applicable to any stabilizer code, which we demonstrate for quantum parity, five-qubit, standard and rotated planar surface, tree, and seven-qubit Steane codes. Our schemes attain the general upper bound for all these codes, while this bound had previously only been reached for the quantum parity code.
\end{abstract}
 
\maketitle

\section{Introduction}
\label{sec:introduction}
All hardware used in information technology is susceptible to errors, and quantum technology is no exception. In fact, owing to the fragility of quantum states, quantum hardware is especially vulnerable. In the pursuit of fault-tolerant quantum information processing, quantum error correction codes~\cite{nielsen_chuang_2010} play a crucial role. Stabilizer codes~\cite{gottesman1997stabilizer} form an especially relevant subclass of quantum error correction codes that are highly prevalent in the literature, with extensive research committed to the topic. Their structured framework, efficient error correction, scalability, and fault-tolerant gate implementation, combined with strong theoretical foundations, make them particularly suitable for developing reliable and scalable quantum information processing. In this work, we focus on optical platforms, a promising approach because of the intrinsic error robustness of photons, the ease of implementing single-qubit gates at room temperature, and the use of photons as flying qubits for quantum communication and quantum computers.

Bell measurements~(BMs) are a key computational primitive in quantum information processing, with applications in both quantum computation and quantum communication. In measurement-based quantum computation~(MBQC)~\cite{PhysRevLett.86.5188}, BMs provide the fusion operations that combine smaller photonic resource states into the large-scale graph states serving as the universal resource for computation~\cite{PhysRevA.68.022312,PhysRevLett.95.010501,Lobl2025transforminggraph}. Fusion-based quantum computation~(FBQC) builds on this idea by formulating universal quantum computation directly in terms of such fusion operations, enabling scalable and loss-tolerant architectures for photonic quantum computing~\cite{Bartolucci2023}. In the context of all-optical quantum repeaters, BMs enable entanglement swapping at the logical level or general teleportation-based quantum error correction~(QEC),	 which provides resilience against photon loss and thus enhances both the scalability and efficiency of long-distance quantum communication~\cite{Azuma2015,PhysRevA.100.052303,PhysRevLett.117.210501,PhysRevA.95.012327,Hilaire_2021,patil2024improveddesignallphotonicquantum,PRXQuantum.4.040322}.

Unfortunately, in linear optical setups, BMs are only achievable with a non-unit success probability~\cite{PhysRevA.59.3295,Calsamiglia2001,PhysRevA.69.012302}. In this setting, stabilizer codes may serve a dual purpose. Besides providing error correction, they enable logical BMs on encoded qubits with an improved success probability \cite{PhysRevLett.117.210501,PhysRevA.95.012327,PhysRevA.99.062308,PhysRevA.100.052303,PhysRevLett.133.050604,PhysRevLett.133.050605}. Furthermore, logical BMs are necessary for computation on encoded qubits, which are regarded as essential for realizing fault-tolerant architectures~\cite{PhysRevLett.98.190504,Bartolucci2023}. However, for a given encoding, the success probability still depends heavily on the specific logical measurement scheme. The success rate of logical BMs is critical to the scalability of optical implementations, given their prevalent use in quantum computation and communication. Therefore, optimizing this probability without increasing the number of qubits can substantially improve their viability.

Various static and feedforward-based photonic linear-optics logical BMs have been proposed~\cite{PhysRevLett.117.210501,PhysRevA.95.012327,PhysRevA.99.062308,PhysRevA.100.052303,Hilaire_2021,PRXQuantum.4.040322,patil2024improveddesignallphotonicquantum}. Static logical measurement schemes for the quantum parity code~(QPC)~\cite{PhysRevLett.95.100501} were introduced in Refs.~\cite{PhysRevLett.117.210501,PhysRevA.95.012327}. Ref.~\cite{PhysRevA.99.062308} showed that any logical BM on stabilizer codes can be decomposed into transversal physical BMs. It also proposed efficient static logical BMs for three Calderbank–Shor–Steane~(CSS) codes~\cite{PhysRevA.54.1098,10.1098/rspa.1996.0136}, a subclass of stabilizer codes derived from classical codes, namely, the QPC code, the standard planar surface code~\cite{KITAEV20032}, and the seven-qubit Steane code~\cite{doi:10.1098/rspa.1996.0136}. Ref.~\cite{PhysRevA.100.052303} presented a feedforward-based logical BM for the QPC code that achieves a higher success probability than the static schemes in Refs.~\cite{PhysRevLett.117.210501,PhysRevA.95.012327,PhysRevA.99.062308}, at the expense of requiring feedforward. For the tree code~\cite{PhysRevLett.97.120501}, feedforward-based logical BMs were presented in Refs.~\cite{Hilaire_2021} and~\cite{patil2024improveddesignallphotonicquantum}, with a focus on improving loss tolerance and error robustness. Finally, Ref.~\cite{PRXQuantum.4.040322} derived tight, fundamental upper bounds on the loss tolerance thresholds for both static and feedforward-based logical linear-optics BMs.

This paper addresses several critical open questions in the field. Firstly, there is currently no established bound on the success probability for feedforward-based logical BMs. Additionally, while the schemes proposed in Ref.~\cite{PhysRevA.99.062308} have been optimized within static setups, it remained open whether these strategies can be extended to feedforward-based schemes and if such extensions offer further improvements. Furthermore, while Ref.~\cite{PhysRevA.100.052303} provides a feedforward-based scheme with optimal no-loss success probability for the QPC code, the development of optimal feedforward-based schemes for other stabilizer codes remains an open problem.

In this work, we develop a comprehensive mathematical framework for analyzing logical BM schemes with great generality. We rigorously formalize the process of logical measurement schemes from the ground up. Along the way, we derive several minor results that provide valuable insights into the dynamics of physical measurements on entangled quantum states, particularly those encoded with stabilizer codes. Based on these foundations, we eventually address feedforward-based logical BMs on two logical qubits encoded in arbitrary, not necessarily identical, stabilizer codes where the physical BMs have a non-unit success probability. The most prominent scenario for this are BMs on photonic dual-rail qubits using linear optics. These schemes make use of single-qubit Clifford operations, single-qubit Pauli measurements, and BMs. The inclusion of single-qubit gates is well motivated, as they are typically simple to implement on dual-rail qubits. The restriction to single-qubit Clifford gates arises from the limitations of the stabilizer formalism. We focus on a scenario without errors, considering only the probabilistic nature of physical BMs as the sole imperfection. We prove that having at least one successful physical BM performed on a pair of physical qubits, one from each of the two separate codes, where we have no prior knowledge about the outcome, is a necessary condition for a successful logical BM. This requirement imposes an upper bound on the logical success probability for these schemes, given by the probability of having at least one such successful physical BM. This improves upon a proof previously given in Ref.~\cite{PhysRevA.100.052303} by extending it from static linear optics to feedforward-based schemes and by circumventing the restrictive assumption that photon-number-resolving detectors can distinguish only up to two-photon events.

We also propose a conceptual framework for designing logical BM schemes akin to the approach used in Ref.~\cite{PhysRevA.99.062308}. However, while Ref.~\cite{PhysRevA.99.062308} is limited to static linear optics and CSS codes, our approach extends to the full class of stabilizer codes and incorporates feedforward-based schemes. This broader scope is achieved by employing a more group-theoretic approach, rather than the classical vector space methods used in Ref.~\cite{PhysRevA.99.062308}. We further present sufficient conditions for the optimality of a scheme. In this work, we define a scheme as optimal if it reaches the bound that we have determined. Our approach based on stabilizer group theory is generally applicable to any stabilizer code, which we demonstrate for quantum parity, five-qubit~\cite{PhysRevLett.77.198}, standard and rotated planar surface~\cite{Horsman_2012}, tree, and seven-qubit Steane codes. Our schemes attain the general upper bound for all these codes, while this bound had previously only been reached for the QPC code in Ref.~\cite{PhysRevA.100.052303}. Additionally, we present an optimized static scheme for the rotated planar surface code. While this scheme does not achieve the success probability of the feedforward-based bound, it still performs significantly better than a simple static scheme. To the best of the authors' knowledge, no logical BM schemes have been proposed for the five-qubit code or the rotated planar surface code to date. For the tree code, feedforward-based logical BM schemes were presented in Refs.~\cite{Hilaire_2021} and~\cite{patil2024improveddesignallphotonicquantum}, improving loss tolerance and error robustness compared to a simpler approach. However, in the absence of loss, our scheme achieves a significantly higher success probability than these schemes. For the Steane code, Ref.~\cite{PhysRevA.99.062308} presents a logical BM scheme that is sub-optimal, in contrast to ours, but has the advantage of requiring only static linear optics.

Interestingly, the scheme we develop for the five-qubit code does not require feedforward, thus it can be fully implemented using static operations alone. This observation implies that, in general, there is no tighter bound for static schemes than for feedforward-based ones. However, for certain codes, the standard toolbox, which relies on what we call guaranteed partial information BMs, fails to achieve the bound when constrained to static operations~\cite{PhysRevA.99.062308}. Therefore, even though we have disproved the existence of a tighter bound for static schemes for general stabilizer codes, it seems unlikely that the bound can be achieved with static means in many cases.

On a more conceptual level, our work provides deeper insights into the dynamics of physical measurements on entangled quantum states, particularly those encoded using stabilizer codes. Furthermore, we expect that our results are of high relevance to the implementation of fault-tolerant optical quantum technologies, because in a regime of sufficiently small errors a realistic logical BM starts closely resembling the ideal measurement. Although a treatment including photon loss would ultimately yield a more complete picture, a full analysis of the idealized, lossless case is a necessary prerequisite. The present work undertakes this foundational step in detail, providing the basis upon which any treatment of imperfections can be built.

In Sec.~\ref{sec:physical-bell-measurements}, we introduce the mathematical formalism, and discuss general physical BMs and their application to subspaces of entangled quantum states. In Sec.~\ref{sec:logical-bell-measurements-on-stabilizer-codes}, we describe the encoding of a uniform mixture of Bell states in stabilizer codes, examine their dynamics under Pauli measurements, and present the proof of the bound for the success probability of a logical BM. In Sec.~\ref{sec:logical-bell-measurements-and-their-optimality}, we discuss our conceptual framework for designing logical BM schemes, present sufficient conditions for an optimal scheme, and discuss heuristics for finding optimal schemes. In Sec.~\ref{sec:optimal-logical-bell-measurements-for-specific-codes} we present optimal schemes for the stabilizer codes considered in this work. In Sec.~\ref{sec:conclusion}, we summarize and conclude our results, and give a brief outlook on future work. Further technical details and extended derivations are provided in the appendices.

\section{Physical Bell measurements}
\label{sec:physical-bell-measurements}
In this section, we establish the mathematical formalism for physical, destructive BMs on a subset of two qubits from an entangled quantum state. In the context of linear optics, a destructive measurement refers to one in which the photonic qubits are absorbed by the detector, making them unavailable for further processing. We define a complete physical destructive BM as a perfect projective measurement characterized by four projectors onto the Bell states. Furthermore, we assume that the entangled quantum state satisfies the following property: when a complete physical BM is applied to a subset of two qubits, the probabilities of the four possible measurement outcomes are uniformly distributed. Intuitively, the local state of these two qubits, after tracing out the rest of the quantum state, mimics a uniform mixture of Bell states. Consequently, we would expect that a non-complete physical BM on these two qubits, allowing for partial outcomes, would have the same success probability as on a uniform mixture of Bell states. Moreover, if the measurement is successful, it is expected that the physical BM will project the remaining quantum state in the same manner as a complete BM.

However, the local state of the two measured qubits may belong to a broader class of states that also satisfy this property. To illustrate, consider the following example of an entangled quantum state:
\begin{equation}
	\ket{\psi} = \frac{1}{\sqrt{2}} \ket{0} \otimes \left( \ket{00} + \ket{11} \right).
\end{equation}
Now, consider a BM applied to the first two qubits of this state. The local state of these two qubits is given by:
\begin{equation}
	\Tr_3 \{ \ket{\psi} \bra{\psi} \}  = \frac{1}{2} \ket{00} \bra{00} + \frac{1}{2} \ket{01} \bra{01}.
\end{equation}
It is straightforward to see, by rewriting the local state as follows, that this state satisfies the condition that all four outcomes of a complete BM are equally likely:
\begin{equation}
	\begin{aligned}
	\Tr_3 \{ \ket{\psi} \bra{\psi} \}	& = \frac{1}{4} \left( \ket{\Phi^+} + \ket{\Phi^-} \right) \left( \bra{\Phi^+} + \bra{\Phi^-} \right) \\
										& + \frac{1}{4} \left( \ket{\Psi^+} - \ket{\Psi^-} \right) \left( \bra{\Psi^+} - \bra{\Psi^-} \right).	
	\end{aligned}
\end{equation}
However, it is important to note that while these probabilities replicate those of a local uniform mixture of Bell states, the local state of these two qubits is clearly distinct from such an ensemble. Thus, rigorously proving the aforementioned expectations appears to be a useful, non-trivial starting point. As a prerequisite for the main results of this paper, we prove that these expectations hold. We conclude the section with a formal definition of partial BM results.

We start with a brief review of linear-optics BMs in Sec.~\ref{sec:linear-optics-bell-measurements}. In Sec.~\ref{sec:generalized-physical-bell-measurements}, we introduce our definition of a physical BM. We then present our formal results addressing physical BMs on a subset of two qubits from an entangled quantum state, as well as the definition of a partial BM, in Sec.~\ref{sec:physical-bell-measurements-on-entangled-quantum-states}.

\subsection{Linear-optics Bell measurements}
\label{sec:linear-optics-bell-measurements}

The Bell states constitute the four simultaneous eigenstates of the two two-qubit operators $X \otimes X$ and $Z \otimes Z$. For convenience, we will omit the tensor product in the notation. The corresponding eigenvalues of the four Bell states are:
\begin{equation}
\begin{aligned}
\ket{\Phi^+} &= \frac{\ket{00} + \ket{11}}{\sqrt{2}} \quad \corresponds \quad ZZ \rightarrow +1, \quad XX \rightarrow +1, \\
\ket{\Phi^-} &= \frac{\ket{00} - \ket{11}}{\sqrt{2}} \quad \corresponds \quad ZZ \rightarrow +1, \quad XX \rightarrow -1, \\
\ket{\Psi^+} &= \frac{\ket{10} + \ket{01}}{\sqrt{2}} \quad \corresponds \quad ZZ \rightarrow -1, \quad XX \rightarrow +1, \\
\ket{\Psi^-} &= \frac{\ket{10} - \ket{01}}{\sqrt{2}} \quad \corresponds \quad ZZ \rightarrow -1, \quad XX \rightarrow -1. \\
\end{aligned}
\end{equation}
This means that a BM essentially consists of the measurement of the two two-qubit observables $XX$ and $ZZ$. Throughout this paper, we shall denote the collection of Bell states as follows:
\begin{equation}
	\{ \ket{\Phi_j} \}_{j \in \{1,2,3,4\}} = \{ \ket{\Phi^+}, \ket{\Phi^-}, \ket{\Psi^+}, \ket{\Psi^-} \}.
\end{equation}

In the context of discrete-variable quantum information processing, photonic qubits are most commonly dual-rail encoded. In this encoding, a single photon occupying one of two optical modes defines the logical basis states, with $\ket{10}_p \equiv \ket{0}$ and $\ket{01}_p \equiv \ket{1}$. Here, the subscript $p$ is used to indicate a quantum state on the photonic level in a Fock space. In the following discussion on linear-optics BMs, we will assume all qubits to be dual-rail encoded. What makes this encoding particularly appealing is its inherent ability to serve as a loss detection code, where any instance of photon loss instantly removes the qubit from the code space.

Using linear optics it is impossible to unambiguously identify Bell states with unit probability~\cite{PhysRevA.59.3295}. This no-go includes the use of feedforward and ancillary photons. A scenario of particular interest is the unambiguous discrimination of four equiprobable Bell states. It was shown that without the use of feedforward and ancillary photons the success probability in this scenario is upper bounded by $\frac{1}{2}$~\cite{Calsamiglia2001}. With a simple setup which uses two beam splitters it is possible to achieve this bound~\cite{Weinfurter_1994,PhysRevA.51.R1727,PhysRevA.53.R1209}. Note that while the mathematical description of the BM process requires two beam splitter interactions, practical implementations using polarization encoding typically need only a single physical beam splitter to achieve this operation by interacting each polarization separately. The Bell states in dual-rail encoding are:
\begin{equation}
\begin{aligned}
	\ket{\Phi^\pm} & = \frac{1}{\sqrt{2}} \left( \ket{1010}_p \pm \ket{0101}_p \right), \\
	\ket{\Psi^\pm} & = \frac{1}{\sqrt{2}} \left( \ket{1001}_p \pm \ket{0110}_p \right). \\
\end{aligned}
\end{equation}
Two symmetric, 50:50 beam splitters defined by their action on two mode creation operators,
\begin{equation}
	\begin{pmatrix}
		a^\dagger	\\
		b^\dagger	\\
	\end{pmatrix}
	\rightarrow
	\frac{1}{\sqrt{2}}
	\begin{pmatrix}
		1	&	i	\\
		i	&	1	\\
	\end{pmatrix}
	\begin{pmatrix}
		a^\dagger	\\
		b^\dagger	\\
	\end{pmatrix}
\end{equation}
are now used to interfere mode 1 with mode 3 and mode 2 with mode 4. Under this action the Bell states transform as follows:
\begin{eqnarray}
\ket{\Phi^\pm} & \rightarrow & \frac{i}{2} \left( \ket{2000}_p + \ket{0020}_p \pm \ket{0200}_p \pm \ket{0002}_p \right), \nonumber \\
\ket{\Psi^+}   & \rightarrow & \frac{i}{\sqrt{2}} \left( \ket{1100}_p + \ket{0011}_p \right), \nonumber \\
\ket{\Psi^-}   & \rightarrow & \frac{i}{\sqrt{2}} \left( \ket{1001}_p - \ket{0110}_p \right).
\label{eq:std-bm-output}
\end{eqnarray}
With photon-number-resolving-detectors (assuming no loss, even only on-off detectors are sufficient) it is now possible to unambiguously discriminate $\ket{\Psi^+}$ and $\ket{\Psi^-}$, among each other and against $\ket{\Phi^\pm}$. However, if we detect two photons in one of the output modes (assuming no loss, this is simply ``on'' in one mode and ``off'' in any other mode) we cannot discriminate $\ket{\Phi^+}$ and $\ket{\Phi^-}$. Therefore, assuming a uniform mixture of Bell states, we achieve a success probability of $\frac{1}{2}$.

There is another important aspect to this measurement. In the event of an unsuccessful outcome, i.e., a single-mode detection event without loss, we still gain information about the measured quantum state. Since in this case we can still distinguish the two subspaces spanned by $\ket{\Phi^\pm}$ and $\ket{\Psi^\pm}$ we obtain the eigenvalue of $ZZ$ with unit probability. This is referred to as a partial BM. In fact, even more information is obtained in this case. Since the first two terms of the quantum state on the rhs of the first line in Eq.~\eqref{eq:std-bm-output} stem from the transformation of the first term on the lhs in $\ket{\Phi^\pm}$, and the last two terms on the rhs originate from the transformation of the second term in $\ket{\Phi^\pm}$,
\begin{equation}
\begin{aligned}
	\frac{1}{\sqrt{2}} \ket{1010}_p \rightarrow \ket{2000}_p + \ket{0020}_p, \\
	\frac{1}{\sqrt{2}} \ket{0101}_p \rightarrow \ket{0200}_p + \ket{0002}_p,
\end{aligned}
\end{equation}
measuring the two photons in the same mode reveals the eigenvalues of the operators $ZI$ and $IZ$. Furthermore, single-qubit Pauli transformations on the input state offer a simple way to achieve an arbitrary permutation of the Bell states. While this does not change the possible outputs in Eq.~\eqref{eq:std-bm-output}, it does alter which Bell state transforms to which output state. Therefore, we can choose which of the three eigenvalues $XX$, $YY$, and $ZZ$ is obtained with unit probability, or more precisely the eigenvalues of $\{ XI , IX \}$, $\{ YI , IY \}$, and $\{ ZI , IZ \}$. We call this class of BMs guaranteed partial information BMs. For convenience we refer to BMs which obtain the $XX$, $YY$, or $ZZ$ eigenvalues as guaranteed partial information as $XX$-, $YY$-, and $ZZ$-BMs, respectively. Note that this includes but is not restricted to the scheme we just discussed here. The single-qubit Pauli transformations on dual-rail encoded qubits needed for the permutation of Bell states are experimentally rather simple operations and can, in principle, be done with unit probability. As a final note, the linear-optics bound of $\frac{1}{2}$ for physical BMs can be exceeded with the use of ancilla photons~\cite{PhysRevLett.113.140403, PhysRevA.84.042331, doi:10.1126/sciadv.adf4080, Hauser2025, guo2024boostedfusiongatespercolation}. Asymptotically, by allowing highly entangled multi-photon ancillae, a unit success probability can even be approached~\cite{Knill2001,PhysRevLett.113.140403,PhysRevA.84.042331}.

\subsection{Generalized physical Bell measurements}
\label{sec:generalized-physical-bell-measurements}

We now introduce our definition of general physical BMs. We consider two qubits on which a BM is to be performed. These qubits are encoded within a subspace of a larger Hilbert space $\mathcal{H}_B$. The logical computational basis states $\ket{00}$, $\ket{01}$, $\ket{10}$, and $\ket{11}$ are represented by orthonormal vectors in $\mathcal{H}_B$, which together span the two-qubit coding subspace. Further we include an ancillary Hilbert space $\mathcal{H}_A$, which is prepared in some state $\ket{A}_A \in \mathcal{H}_A$. The physical measurement consists of a unitary operator $U$ and a projective measurement described by a Hermitian observable $M = \sum_m m P_m$, both jointly acting on $\mathcal{H}_B \otimes \mathcal{H}_A$. To perform the measurement the unitary $U$ is applied and then, on the resulting state, the observable $M$ is measured with some outcome in $\{ m \}$. Here, $P_m = \ket{m}\bra{m}$ denotes the projectors onto the orthogonal eigenstates of $M$.

Note that we confine ourselves to projective measurements without sacrificing the generality of the measurement, as we have introduced the unspecified ancillary state $\ket{A}_A$ and joint unitary operation $U$~\cite{nielsen_chuang_2010}. Further note that no restrictions were imposed on the encoding of the qubits within the Hilbert space $\mathcal{H}_B$; thus, the encoded qubits do not necessarily span the full Hilbert space $\mathcal{H}_B$. For example, $\mathcal{H}_B$ could be the full 4-mode Fock space with the physical code space being dual-rail encoded and thus being a subspace of $\mathcal{H}_B$. Additionally, the unitary $U$ as well as the measurement operators $P_m$ are kept completely general. This includes cases where the unitary $U$ takes the state out of its coding space, which is a subspace of $\mathcal{H}_B$. This is, for instance, the case for standard linear-optics BMs, as described in the previous section~\ref{sec:linear-optics-bell-measurements}. Furthermore, due to the principle of deferred measurements \cite{nielsen_chuang_2010}, the physical BM we defined here is not restricted to static circuits but can, in principle, include arbitrary feedforward. In conclusion, we have introduced a definition of a very general physical BM. 

We draw some connections between the formal definitions above to the standard linear-optics BM we discussed in the previous section. The dual-rail encoded state on which the measurement is applied lives in $\mathcal{H}_B$, the scheme does not use an ancilla state $\ket{A}$, and the unitary $U$ is defined by the action of the beam splitters. Therefore, Eq.~\eqref{eq:std-bm-output} represents $U\varrho_B U^\dagger$. Lastly, the projectors $\{P_m\}$ are the projectors onto all possible 4-mode Fock states.

Next, we present the definition of an unambiguous BM as introduced in Ref.~\cite{Calsamiglia2001}. First, we shall recall that we denote the four Bell states, here encoded in $\mathcal{H}_B$, as $\{ \ket{\Phi_j}_B \}_{j \in \{1,2,3,4\}}$. Subsequently, we focus on examining the term $P_m U \ket{\Phi_j}_B \ket{A}_A$ for all $j \in \{1,2,3,4\}$, to categorize the measurement results $\{m\}$. The defining properties of a measurement outcome --- probability of occurrence and post-measurement state --- for an arbitrary measured quantum state can be derived solely from this term. For simplicity, we will usually omit the Hilbert space in the subscript of a ket after its introduction to improve readability. The following definition in Eq.~\eqref{eq:s} was essentially introduced in Ref.~\cite{Calsamiglia2001}, but we have adapted its formulation to suit our formalism. We define an unambiguous BM result as a measurement outcome $s \in \{ m \}$ that identifies a Bell state $\ket{\Phi_\sigma} \in \{ \ket{\Phi_j} \}_{j \in \{1,2,3,4\}}$ unambiguously:
\begin{equation}
	\exists \sigma : P_s U \ket{\Phi_\sigma} \ket{A} \neq 0 \quad \wedge \quad \forall i \neq \sigma \quad P_s U \ket{\Phi_i} \ket{A} = 0.
	\label{eq:s}
\end{equation}
For brevity we refer to such a measurement result as a successful BM, or simply success. In the special case where no ancillary state is used, we can simply remove the ket $\ket{A}$ throughout Eq.~\eqref{eq:s}.

We denote by $\mathbb{P}_B$ the probability of having a successful BM on a uniform mixture of Bell states, i.e., a maximally mixed state. A detailed derivation of this quantity is provided in App.~\ref{app:generalized-physical-bell-measurements}. In the case of static linear optics without ancillary photons, $\mathbb{P}_B$ was shown to be upper bounded by $\frac{1}{2}$~\cite{Calsamiglia2001}. In the scheme proposed by Ewert and van Loock~\cite{PhysRevLett.113.140403}, this success probability is increased to $\frac{3}{4}$ by employing four ancilla photons and static linear optics. As in the standard scheme, the dual-rail encoded state on which the measurement is applied lives in $\mathcal{H}_B$. Here, an ancilla state $\ket{A} = \ket{1111}_p$, consisting of four photons in four modes, is introduced. The unitary operation $U$ includes the action of eight beam splitters (in polarization-encoding, counting those acting on distinct polarization nodes separately), while the projectors ${P_m}$ once again correspond to the projectors onto all possible multi-mode Fock states.

\subsection{Physical Bell measurements on entangled quantum states}
\label{sec:physical-bell-measurements-on-entangled-quantum-states}

In the following we investigate destructive physical BMs that are performed on a subset of two qubits of an entangled quantum state. We will approach this in a general fashion without making any assumptions about the quantum state. We start by considering an arbitrary $n$-qubit pure state $\ket{\psi}$. Without loss of generality we split the $n$-qubit Hilbert space of $\ket{\psi}$ such that the encoded two-qubit Hilbert space $\mathcal{H}_B$ on which the physical BM will be performed is separated out:
\begin{equation}
	\ket{\psi}_{BR} \in \mathcal{H}_B \otimes \mathcal{H}_R.
\end{equation}
From now on, throughout the paper we consider destructive physical BMs. Therefore, when we consider the post-measurement state after we performed a BM on the two qubits encoded in $\mathcal{H}_B$ we will trace out $\mathcal{H_B}$. So, in conclusion, there are two facets which completely characterize a measurement outcome $m \in \{m\}$. The probability of the outcome,
\begin{equation}
	p_{\ket{\psi}\bra{\psi}} (m) = \Tr \{ (P_m U \otimes I_R) ( \ket{\psi} \bra{\psi} \otimes \ket{A} \bra{A} ) (U^\dagger \otimes I_R) \},
\end{equation}
where $\Tr \equiv \Tr_{BAR}$, and its post-measurement state,
\begin{equation}
	\frac{1}{p_{\ket{\psi}\bra{\psi}} (m)} \Tr_{BA} \{ (P_m U \otimes I_R) ( \ket{\psi} \bra{\psi} \otimes \ket{A} \bra{A} ) (U^\dagger \otimes I_R) \},
\end{equation}
where we recall that $\ket{A} \in \mathcal{H}_A$ is an ancillary state and the operators $P_m$ and $U$ act on $\mathcal{H}_B \otimes \mathcal{H}_A$. A detailed treatment of this concluding in the proof of the upcoming Lem.~\ref{lem:successful-bell-measurment} can be found in App.~\ref{app:physical-bell-measurements-on-entangled-quantum-states}. We now return to the scenario introduced at the beginning of this section and present Lem.~\ref{lem:successful-bell-measurment}.
\begin{lemma}
\label{lem:successful-bell-measurment}
(Success of a physical BM) We consider a quantum state $\ket{\psi} \in \mathcal{H}_B \otimes \mathcal{H}_R$, where the state in $\mathcal{H}_B$ is entirely within the two-qubit code space. Furthermore, we assume that measuring $ZZ$ and $XX$ on these two qubits in $\mathcal{H}_B$ has uniform probability for all four outcomes. 

Then, a successful physical BM has the same success probability $\mathbb{P}_B$ as measuring a uniform mixture of Bell states, and the post-measurement state on $\mathcal{H}_R$ is identical to the projection of a complete BM which identified the same Bell state.
\end{lemma}
\begin{proof}
Provided in App.~\ref{app:physical-bell-measurements-on-entangled-quantum-states}.
\end{proof}

It should be noted that uniform probability for all four outcomes means that the outcomes of $XX$ and $ZZ$ are independent and thus that the outcomes of $YY$ are also equally likely. Note that this is not equivalent to uniform probability for $XX$ and $ZZ$ separately, which would be an insufficient criteria. This can be easily understood by considering the quantum state $\left( \ket{\Phi^-} + \ket{\Psi^+} \right) / \sqrt{2}$ as an example. For this state, the outcomes of $XX$ and $ZZ$ are uniformly distributed for each observable separately, while the eigenvalue of $YY$ is always one. Thus, in this case the outcomes of $XX$ and $ZZ$ are not independent.

We have concluded our discussion on unambiguous measurement results and will now discuss what is commonly known as a partial BM. We will begin with a simple example. Let us consider an ambiguous measurement outcome $p \in \{ m \}$, with corresponding projector $P_p = \ket{p}\bra{p}$, which satisfies,
\begin{equation}
	P_p U \ket{\Phi^+}_B \ket{A}_A = P_p U \ket{\Phi^-}_B \ket{A}_A = \alpha \ket{p}_{BA},
\end{equation}
where $\alpha$ is a complex number with $|\alpha|^2 < 1$. Using Lem.~\ref{lem:pm-general} from App.~\ref{app:physical-bell-measurements-on-entangled-quantum-states} we obtain the effective projection of this measurement result:
\begin{equation}
	\Pi_p = \frac{1}{2} \left( \ket{\Phi^+} + \ket{\Phi^-} \right) \left( \bra{\Phi^+} + \bra{\Phi^-} \right) = \ket{00}\bra{00}.
\end{equation}
Thus, the measurement result $p$ projects the remaining post-measurement state on $\mathcal{H}_R$ identically to measuring the observables $ZI = +1$ and $IZ = +1$ on $\mathcal{H}_B$. In general, we define a partial BM result as a measurement outcome $p \in \{m\}$, with corresponding projector $P_p = \ket{p}\bra{p}$ which effectively projects onto a simultaneous eigenstate of a pair in $\{ \{XI,IX\} , \{YI,IY\} ,\{ZI,IZ\} \}$. In other words, obtaining a partial BM outcome is equivalent to performing two single-qubit Pauli measurements. Partial BMs thus can be used to leverage stabilizer information which is obtained even when an unambiguous Bell projection has failed~\cite{PhysRevA.99.062308,Hilaire_2021,patil2024improveddesignallphotonicquantum}. Finally we note that even though we are not aware of possibilities to obtain more information from a non-successful BM, to our knowledge, there exists no bound in the literature which limits the amount of information that can be obtained from a measurement which does not unambiguously identify a Bell state.

\section{Logical Bell measurements on stabilizer codes}
\label{sec:logical-bell-measurements-on-stabilizer-codes}
In the previous section we considered two physical qubits which are encoded in a Hilbert space $\mathcal{H}_B$. In this section, we extend our discussion to include a second layer of encoding. Rather than focusing on the low-level encoding of physical qubits we shift our attention to logical qubits which are encoded in multiple physical qubits using quantum error correction codes, specifically stabilizer codes. In Sec.~\ref{sec:encoded-uniform-mixture-of-bell-states} we formally introduce the logical encoding of a uniform mixture of Bell states. Following this, in Sec.~\ref{sec:pauli-measurements-and-logical-operators-on-encoded-uniform-mixtures-of-bell-states} we will examine the dynamics of measurements of elements of the Pauli group and Clifford operations on an encoded Bell state. Finally, in Sec.~\ref{sec:necessary-condition-for-an-optimal-logical-Bell-measurement-with-feedforward-based-linear-optics} we will focus on a more practical scenario, particularly relevant to linear-optics setups. Within this context, we will establish a bound on the success probability of logical BMs with feedforward-based linear optics.

\subsection{Encoded uniform mixture of Bell states}
\label{sec:encoded-uniform-mixture-of-bell-states}
In what follows we will make extensive use of the stabilizer formalism~\cite{nielsen_chuang_2010,gottesman1997stabilizer}. We recall that the Pauli group is defined to consist of all Pauli matrices, together with multiplicative factors $\pm 1, \pm i$:
\begin{equation}
	\mathcal{P}_1 \equiv \{ \pm I, \pm i I, \pm X, \pm i X, \pm Y, \pm i Y, \pm Z, \pm i Z \}.
\end{equation}
The general Pauli group $\mathcal{P}_n$ on $n$ qubits is defined to consist of all $n$-fold tensor products of Pauli matrices, and again we allow multiplicative factors $\pm 1, \pm i$. All elements of the general Pauli group $\mathcal{P}_n$ have exactly two eigenvalues which are $+1$ and $-1$. For convenience, we will omit the tensor products in the notation. We refer to the Hilbert subspace on which an operator acts non-trivially as the support of the operator. For example, the operator $X \otimes I \otimes Z = X I Z \in \mathcal{P}_3$ has support on the first and third qubit. Additionally, for brevity, we define an operator to commute with a set of operators if it commutes with each individual element in that set.

We denote a quantum error correction code that encodes $k$ logical qubits using $n$ physical qubits as an $[n,k]$ error correction code. We now consider the encoding of two logical qubits in separate $[n_i,1]$ stabilizer error correction codes, where $i \in \{ 1, 2 \} $ denotes the two codes. Note that, in our treatment, the two logical qubits are encoded in disjoint sets of physical qubits. Each code is characterized by a set of $(n_i-1)$ independent and commuting stabilizer generators denoted as $G_i = \{ g_{i,s} \}_{s \in \{ 1, \dots, n_i-1 \} }$, where $g_{i,s} \in \mathcal{P}_{n_i}$. To construct a joint code for the two logical qubits, we define the $(n_1+n_2)-2=n_c-2$ independent and commuting stabilizer generators as:
\begin{equation}
	G_c = \{ g_j \}_{j \in \{ 1, \dots, n_c-2\}} = G_1 \cup G_2.
	\label{eq:Gc}
\end{equation}
In Eq.~\eqref{eq:Gc}, we extend all elements in $G_i$ trivially to the other code. A trivial extension of an element $g_{1,s} \in G_1$ is defined as $g_{1,s} \otimes I^{\otimes n_2}$, and analogously for the other code. For convenience, we omit the trivial extension in the notation in Eq.~\eqref{eq:Gc} and throughout the rest of the paper.  We will refer to the elements of $G_c$ as code stabilizers. Generally, we are always free to choose a different set of generators $\tilde{G}_c = \{ \tilde{g}_j \}_{j \in \{ 1, \dots, n_c-2 \} }$ which generates the same group:
\begin{equation}
	\langle G_c \rangle = \langle \tilde{G}_c \rangle.
\end{equation}
The stabilizer group of the joint code is generated by $G_c$:
\begin{equation}
	S_c = \langle G_c \rangle.
\end{equation}
The logical operators of this two-qubit code can be related to the equivalence classes formed by the elements of the quotient group
\begin{equation}
	N(S_c)/S_c = \{ nS_c \mid n \in N(S_c) \},
\end{equation}
where $N(S_c)$ denotes the normalizer of $S_c$. From group theory, we deduce that $N(S_c)/S_c$ is isomorphic to the Pauli group~\cite{gottesman1997stabilizer}. Thus, we define the generators of $N(S_c)/S_c$ as:
\begin{equation}
	\{ [i\overline{II}] , [\overline{XI}] , [\overline{IX}] , [\overline{ZI}] , [\overline{IZ}] \},
	\label{eq:cosets}
\end{equation}
where $ \overline{II}, \overline{XI} , \overline{IX} , \overline{ZI} $ and $ \overline{IZ}$ are arbitrary representatives of the cosets. It is important to note that all operators in a particular coset act identically on the codespace. Consequently, each coset comprises all logical operators that exhibit identical actions on the codespace.

To define a uniform mixture of Bell states, we introduce the random variables $l_x$ and $l_z$, which are independent symmetric Bernoulli random variables with outcomes $\{ -1, +1 \}$. We then define the randomly distributed logical operators:
\begin{equation}
	L_{x,z} = \{ l_x \overline{XX} , l_z \overline{ZZ} \},
	\label{eq:lx-lz}
\end{equation}
where $\overline{XX} \in [\overline{XX}]$ and $\overline{ZZ} \in [\overline{ZZ}]$. The elements of $L_{x,z}$ are independent and commute with the stabilizer group $S_c$, and are not elements of $S_c$. Therefore, we conclude that the stabilizer group of the uniform mixture of encoded Bell states is
\begin{equation}
\begin{aligned}
	S 	& = \langle G_c \cup L_{x,z} \rangle \\
		& \corresponds \overline{\varrho_B},
\end{aligned}
	\label{eq:init-state}
\end{equation}
which defines a uniform mixture of logical Bell states. Since $G_c \cup L_{x,z}$ consists of $n_c = n_1 + n_2$ independent stabilizer generators, the group $S$ stabilizes a stabilizer state. We usually do not specify the logical operators $ \overline{XX}$ and $\overline{ZZ} $, since the stabilizer group $S$ is independent of the chosen representatives.

At the beginning of a logical measurement scheme we know everything about the stabilizers in $S$ except for the values of the random variables $l_x$ and $l_z$. In this work we treat logical measurement schemes which consist of Clifford gates and measurements of Pauli-operators, both of which can be tracked in the stabilizer formalism. Therefore, throughout a logical measurement scheme, we keep track of the quantum state via transformations of the stabilizer generators and all information we obtain is contained in the measurement results. The task of a BM scheme can thus be interpreted as obtaining measurement results which are correlated with $l_x$ and $l_z$ to determine their values. Hence we refer to the values of $l_x$ and $l_z$ as the logical information.

To generalize the logical information, we consider the following. The stabilizer group of each Bell state has order four: $S = \{ II, l_x XX, l_y YY, l_z ZZ \}$. This stabilizer group contains three distinct pairs of commuting, nontrivial Pauli operators that can serve as generators: $\{XX, ZZ\}$, $\{XX, YY\}$, and $\{YY, ZZ\}$. Hence, these pairs comprise all possible Pauli measurements that, when measured jointly, project onto the Bell basis. Therefore, we define the random variable $l_y = -l_x l_z$, which is the logical $\overline{YY}$ information of the encoded state. In conclusion, the complete logical information is contained in any pair chosen from the set $\{l_x, l_y, l_z\}$.

Because global phases of quantum observables are physically irrelevant for measurements~\cite{6778074}, we work with the effective Pauli group for the cosets of measured logical operators. The effective $n$-qubit Pauli group is the set
\begin{equation}
    \Pi^{\otimes n} = \{ I, X, Y, Z \}^{\otimes n}
\end{equation}
endowed with the phase-stripped multiplication rule $(\Pi^{\otimes n},\star)$, where
\begin{equation}
    \sigma_a \star \sigma_b
    = I\,\delta_{a,b}
    + \sum_{c=1}^3 \lvert \epsilon_{abc} \rvert\, \sigma_c,
\end{equation}
with $(\sigma_1,\sigma_2,\sigma_3) = (X,Y,Z)$~\cite{9456887}. For two logical qubits, the effective group is therefore generated by
\begin{equation}
    \{[\overline{XI}], [\overline{IX}], [\overline{ZI}], [\overline{IZ}]\}.
\end{equation}
However, in the definition of the logical Bell information in Eq.~\eqref{eq:lx-lz}, the phase of the logical operators is relevant because it determines the signs of $l_x$ and $l_z$. In such contexts we must therefore work with the full Pauli group, as defined in Eq.~\eqref{eq:cosets}. Throughout the remainder of the paper, we will use the full Pauli group when working with stabilizer groups of encoded states, and the effective Pauli group when considering measurements of logical operators. Which version of the Pauli group is being used will be clear from context, so we do not introduce separate notation for the two.

\subsection{Pauli measurements and logical operators on encoded uniform mixtures of Bell states}
\label{sec:pauli-measurements-and-logical-operators-on-encoded-uniform-mixtures-of-bell-states}

In this section, we examine two topics related to encoded states: measurements of Pauli group elements, and properties of pairs of logical operators that constitute a logical BM. From this point onward we will assume that all observables and operators are elements of the Pauli group. Throughout this section, we will analyze how the generators of the stabilizer group transform under measurements. Additionally, we implicitly account for Clifford gates acting on the quantum state, as they normalize the Pauli group and can therefore always be absorbed into the Pauli measurements. In preparation for the logical measurement schemes we will especially treat multiple successive measurements on a stabilizer state. The reason to consider a sequence of measurements as opposed to combining the set of commuting observables into a single measurement observable lies in the application to measurement schemes using feedforward. In the remainder of this paper it will be necessary to track the full transformations of quantum states through measurement schemes. The only restriction we impose on the measurements in this section is that all successive measurements commute with each other. This is motivated by the fact that, throughout this paper, we will always consider destructive single-qubit measurements and destructive BMs, which naturally commute. We will term the stabilizer group that stabilizes the quantum state at a given time as the current stabilizer state. This means that if we relate an observable to the current stabilizer group we relate the observable to the stabilizer group that defines the quantum state at the time the measurement is performed. Throughout this discussion, we assume that no information regarding the logical variables is known prior to any measurement. While it is evident that this is a reasonable assumption at the beginning of the logical measurement scheme it might seem questionable to be assumed for later measurements. However, the reason for this will become clear in the proof of Thm.~\ref{thm:bound} where we demonstrate that all relevant measurements satisfy this assumption. In Sec.~\ref{sec:pauli-measurements-on-encoded-uniform-mixtures-of-bell-states}, we examine Pauli measurements on encoded uniform mixtures of Bell states. In Sec.~\ref{sec:logical-operators-on-encoded-uniform-mixtures-of-bell-states}, we address properties of pairs of logical operators that constitute a logical BM.

\subsubsection{Pauli measurements on encoded uniform mixtures of Bell states}
\label{sec:pauli-measurements-on-encoded-uniform-mixtures-of-bell-states}

In the following, we discuss all possible cases of Pauli observables measured on an encoded uniform mixture of Bell states. Firstly, we note that multi-qubit Pauli observables exhibit exclusively two behaviors: they either commute or anticommute with elements of the stabilizer group. When the observable anticommutes with at least one stabilizer generator, the two measurement outcomes occur with equal probability, and the measurement necessarily changes the global $n$-qubit state. This implies that the measurement outcome is uncorrelated with the logical variables and obtains no logical information. The new stabilizer generators can be calculated straightforwardly, a process we will detail in the proofs of Lem.~\ref{lem:standard-form} and Lem.~\ref{lem:forbidden-measurements}. In the second case the observable commutes with the stabilizer group. We recall that the stabilizer group is generated by a complete set of $n_c$ elements, corresponding to a stabilizer state. It is well known that consequently, the observable is also an element of the stabilizer group up to a sign. Therefore, the global $n$-qubit state remains unchanged by the measurement, and the outcome is predetermined with unit probability by the stabilizer group. For our logical measurement schemes this implies two possibilities: either the observable is uncorrelated with the logical information or it is correlated. In the case where the observable is uncorrelated with the logical information, we know the outcome beforehand and we obtain no additional knowledge from the measurement; we merely completed measuring a code stabilizer in $\langle G_c \rangle$. In the case where the observable is correlated with the logical information the outcome is predetermined in one-to-one correspondence by one of the logical variables $\{ l_x, l_y, l_z \}$. Therefore, assuming no prior knowledge of the logical information, both outcomes are equally likely, since the logical variables are uniformly distributed. Consequently, the measurement will yield the value of one of the three logical variables. A complete formal treatment, along with an illustrative discussion of the cases where the observables commute with the current stabilizer, is provided in App.~\ref{app:observables-that-commute-with-s}.

In the following we examine the transformation of the stabilizer state under measurements, assuming no information on the logical variables has been acquired. In Lem.~\ref{lem:standard-form}, we present a compact form of the stabilizer generators after successive measurements. Additionally, Lem.~\ref{lem:forbidden-measurements} analyzes a special case of the observable, demonstrating that this measurement invariably removes one logical variable from the quantum state without learning its value. Consequently, we prohibit this measurement in our schemes and in the treatment of Lem.~\ref{lem:standard-form}.

\begin{lemma}
\label{lem:standard-form}
(General stabilizer evolution under measurements) We consider an encoded Bell state as defined in Eq.~\eqref{eq:init-state} and a sequence of mutually commuting observables $\{ M_i \}$ with measurement results $\{ m_i \}$. We define the set of operators:
\begin{equation}
	\mathbb{M} = \{ m_i M_i \}.
\end{equation}
We propose the form $S^{(\mathbb{M})}$ of the current stabilizer group of the global quantum state after the measurements as:
\begin{equation}
	S^{(\mathbb{M})} = \langle G_c^{(\mathbb{M})} \cup \mathbb{M} \cup L^{(\mathbb{M})}_{x,z} \rangle,
	\label{eq:general-state}
\end{equation}
where $\langle G_c^{(\mathbb{M})} \rangle \subseteq \langle G_c \rangle$ and $L^{(\mathbb{M})}_{x,z} = \{ l_x \overline{XX}^{(\mathbb{M})} , l_z \overline{ZZ}^{(\mathbb{M})}  \}$ with $\overline{XX}^{(\mathbb{M})} \in [\overline{XX}]$ and $\overline{ZZ}^{(\mathbb{M})} \in [\overline{ZZ}]$.

With the restriction that we never measure an observable which anticommutes with at least one element of $L^{(\mathbb{M})}_{x,z}$ and commutes with the rest of the generators of the current stabilizer group $G_c^{(\mathbb{M})} \cup \mathbb{M}$, and assuming no information on $l_x$ or $l_z$ has been obtained yet, Eq.~\eqref{eq:general-state} describes the global quantum state at any point in a sequence of multi-qubit Pauli measurements.

\end{lemma}
\begin{proof}
In the beginning of the logical measurement scheme the quantum state is in the stabilizer state defined in Eq.~\eqref{eq:init-state}. Consider an observable $M$ that anticommutes with at least one element $g$ of $G_c$. A standard group-theoretic fact is that a generator may be replaced by its product with an element of the stabilizer subgroup generated by the remaining generators without changing the stabilizer group. Specifically, it is always possible to select the generators of the stabilizer group $S$ in such a manner that $g$ is the sole element that anticommutes with $M$. If there exists another element $g'$ in $G_c$ that anticommutes with $M$, then the operator $gg'$ will commute with $M$. Consequently, we replace each stabilizer generator $g'$ in $(G_c \cup L_{x,z}) \setminus \{g\}$ that anticommutes with $M$ with $gg'$. This ensures that $g$ remains the only element of $G_c\cup L_{x,z}$ that anticommutes with $M$.

We shall illustrate this with an example. Let us assume that $g_1$, $g_3$ and $l_x \overline{XX}$ anticommute with $M$. We rewrite the generators as discussed:

\begin{equation}
\begin{aligned} 
	S & = \langle g_1, g_2 , g_3, \dots , g_{n_c-2} , l_x \overline{XX}, l_z \overline{ZZ} \rangle \\
	& = \langle g_1, g_2 ,g_1 g_3, \dots , g_{n_c-2} , g_1 l_x \overline{XX}, l_z \overline{ZZ} \rangle.
\end{aligned}
\end{equation}
After this transformation, $g_1$ is the sole generator which anticommutes with $M$. To get the post-measurement stabilizer group after measuring $M$ with outcome $m$, we replace the only generator that still anticommutes with $M$ with $m M$. In our example this yields:
\begin{equation}
	S^{(M)} = \langle g_2 ,g_1 g_3, \dots , g_{n_c-2}, m M , g_1 l_x \overline{XX}, l_z \overline{ZZ} \rangle.
\end{equation}
Now we identify all elements of the new generators which are products of elements of $G_c$ as $G_c^{(M)}$. In our example, these elements are:
\begin{equation}
	G_c^{(M)} = \{ g_2, g_1 g_3, \dots g_{n_c-2} \}.
\end{equation}
We also identify the new logical operators $L^{(M)}_{x,z}$ which are the initial logical operators $\{ l_x \overline{XX}, l_z \overline{ZZ} \}$ up to factors of a code stabilizer. In our example, these operators are:
\begin{equation}
\begin{aligned}
	L^{(M)}_{x,z} & = \{ g_1 l_x \overline{XX}, l_z \overline{ZZ} \} \\
	& = \{ l_x \overline{XX}^{(M)} , l_z \overline{ZZ}^{(M)}  \}. \\
\end{aligned}
\end{equation}
We obtain the generators of the stabilizer group after the measurement:
\begin{equation}
	S^{(M)} = \langle G_c^{(M)} \cup m M \cup L^{(M)}_{x,z} \rangle.
\end{equation}
in conclusion, the set $G_c^{(M)}$ is obtained by first choosing a set that generates the same group as $G_c$ and has only one element which anticommutes with the measurement and then removing this element. Therefore it holds that $\langle G_c^{(M)} \rangle \subseteq \langle G_c \rangle$. The logical operators in the set $L^{(M)}_{x,z} = \{ l_x \overline{XX}^{(M)} , l_z \overline{ZZ}^{(M)}  \}$ differ only by factors of a code stabilizer $g \in G_c$, therefore the coset for each element has not changed and it holds that $\overline{XX}^{(M)} \in [\overline{XX}]$ and $\overline{ZZ}^{(M)} \in [\overline{ZZ}]$.

We generalize the above for multiple measurements $\{ M_i \}$ with results $\{ m_i \}$. Recall that we always assume that all $M_i$ commute with each other. We define the set
\begin{equation}
	\mathbb{M} = \{ m_i M_i \},
\end{equation}
and by iteratively applying the result for a single measurement we obtain the stabilizer group of the post-measurement state,
\begin{equation}
	S^{(\mathbb{M})} = \langle G_c^{(\mathbb{M})} \cup \mathbb{M} \cup L^{(\mathbb{M})}_{x,z} \rangle.
\end{equation}
This form also characterizes the stabilizer state at the beginning of the scheme, prior to any measurement, as Eq.~\eqref{eq:init-state} emerges as a special case of Eq.~\eqref{eq:general-state} when ${M}=\emptyset$. We conclude the proof by referring back to the discussion earlier in this section where we argued that any measurements not addressed in this proof either commute with the current stabilizer group and thus leaves the quantum state invariant or are explicitly excluded by the assumptions of this lemma. Therefore, including measurements which commute with the current stabilizer in $\mathbb{M}$ will have no impact on the generated group. However, it is important to note that these operators will be redundant in the generator which will cause the set $G_c^{(\mathbb{M})} \cup \mathbb{M} \cup L^{(\mathbb{M})}_{x,z}$ to not be a minimal generating set.
\end{proof}

Note that Lem.~\ref{lem:standard-form} applies regardless of whether observables that commute with the current stabilizer are included in $\mathbb{M}$. We may exclude all such observables and thereby define the reduced set $\mathbb{M}_r$:
\begin{equation}
\begin{aligned}
	\mathbb{M}_r = \mathbb{M} \setminus \{ 	& m_i M_i \in \mathbb{M} = \{m_i M_i\} \\
											& \mid \forall s \in S^{( \{ m_j M_j \}_{j<i} )} : \comm{m_i M_i}{s} = 0 \}.
\end{aligned}
\end{equation}
With this definition, the set $G_c^{(\mathbb{M})} \cup \mathbb{M}_r \cup L^{(\mathbb{M})}_{x,z}$ forms a minimal generating set. Nevertheless, for the purposes of this work it is often more convenient to keep the current stabilizer in its redundant form. The reduced set $\mathbb{M}_r$ is introduced here to avoid ambiguity and will be used in App.~\ref{app:observables-that-commute-with-s}.

We will now treat the case where an observable anticommutes with at least one element of $L^{(\mathbb{M})}_{x,z}$ and commutes with the rest of the generators of the current stabilizer group $G_c^{(\mathbb{M})} \cup \mathbb{M}$. As we stated earlier in this section this measurement invariably removes one logical variable from the quantum state without learning its value. We will now formally prove this in Lem.~\ref{lem:forbidden-measurements}.

\begin{lemma}
\label{lem:forbidden-measurements}
(Observables that commute with $G_c^{(\mathbb{M})}$ and anticommute with an element of $L_{x,z}$) We consider a current stabilizer state as defined in Eq.~\eqref{eq:general-state}. If an observable $M$ which commutes with all elements of $G_c \cup \mathbb{M}$ and anticommutes with at least one element of $L_{x,z}$ is measured, the new stabilizer of the state does not contain the logical Bell-state information and we obtain no information about $l_x$, $l_y$, and $l_z$.
\end{lemma}
\begin{proof}
Since the operator $M$ anticommutes with a stabilizer generator both outcomes are equally likely and thus uncorrelated with $l_x$ and $l_z$. Thus we learn no information about the value of these two variables from this measurement. Now we consider two cases to discuss the post-measurement stabilizer group. In the first case $M$ anticommutes with exactly one element of $L_{x,z}$ and in the second case with both. In the first case the stabilizer generator which anticommutes with $M$ gets replaced in the stabilizer generators of the state with $m M$, where $m \in \{ -1,+1 \} $ is the measurement outcome. In the second case we replace $l_x \overline{X_1 X}_2$ with $l_x l_z \overline{X_1 X}_2 \ \overline{Z_1 Z}_2 = -l_x l_z \overline{Y_1 Y}_2$ in the stabilizer generators. The operator $M$ commutes with this new element. Now the measurement replaces $l_z \overline{Z_1 Z}_2$ with $m M$ in the stabilizer generators. In both cases only one logical variable is contained in the post-measurement stabilizer generators. Therefore, the necessary information to unambiguously identify the logical Bell state has been irreversibly destroyed.
\end{proof}
Having discussed all possible cases of Pauli measurements, we conclude our discussion on the dynamics of Pauli measurements on entangled quantum states.

\subsubsection{Logical operators on encoded uniform mixtures of Bell states}
\label{sec:logical-operators-on-encoded-uniform-mixtures-of-bell-states}

In this section, we present Lem.~\ref{lem:acomm-logicals}, an observation that will be crucial in the proof of Thm.~\ref{thm:bound} in the subsequent section. We also provide discussion and motivation to frame the problem addressed there. Before we can adequately introduce Lem.~\ref{lem:acomm-logicals} we need to define some terminology. Let us consider two elements of the general Pauli group, denoted as $p_1$ and $p_2$, both belonging to $\mathcal{P}_{n}$. We represent them through their single-qubit decompositions: $p_1 = \bigotimes_{i=1}^n u_i$ and $p_2 = \bigotimes_{i=1}^n w_i$, where $u_i, w_i \in \{ I, X, Y, Z \}$ are single-qubit Pauli operators acting on the $i$-th qubit. We define the number of qubits in which the two operators anticommute as the number of qubits for which the isolated single-qubit operators $u_i$ and $w_i$ anticommute. To illustrate, consider the two operators $IZZX$ and $ZXZZ$ in $\mathcal{P}_4$. These operators anticommute in two qubits, specifically in qubits $2$ and $4$. Formally, this is defined as the cardinality of the set:
\begin{equation}
	| \{ \{ u_i ,w_i \} | \acomm{u_i}{w_i} = 0 \} |.
\end{equation}
In the beginning of this section, we constructed a stabilizer code that encodes two logical qubits by combining two codes, each encoding one logical qubit. For such a joint code, we can still refer to the physical qubits of the initial two codes. Therefore, we can still tell on which of the two initial codes the operators anticommute. This is illustrated with the following example. Consider two logical qubits, each encoded in the QPC($2,2$) code. The stabilizer generators of the QPC($2,2$) code are:
\begin{equation}
	G_i = \{ XXXX, ZZII, IIZZ \},
\end{equation}
where $i \in \{1,2\}$, and the trivial extension is omitted. Therefore, the stabilizer generators of the combined code encoding two logical qubits are:
\begin{equation}
\begin{aligned}
	G_c = 	& \ G_1 \cup G_2 \\
		= 	& \ \{  XXXX \ IIII, ZZII \ IIII, IIZZ \ IIII, \\
			& \ \phantom{\{ } IIII \ XXXX, IIII \ ZZII, IIII \ IIZZ \phantom{,} \}.
\end{aligned}
\end{equation}
This illustrates that the first four qubits of the combined code belong to the first instance of the QPC($2,2$) code, while the next four qubits belong to the second instance of it. Applying this observation, we next consider the operators $XXII \ IIXX$ and $IZZI \ IIZZ$. These operators anticommute in one and two qubits on the first and second code, respectively, thus they anticommute in three qubits in total. We now present Lem.~\ref{lem:acomm-logicals}, which states an observation that serves as a prerequisite and motivation for the proof of Thm.~\ref{thm:bound}.
\begin{restatable}{lemma}{acommlogicalsrestatable}
\label{lem:acomm-logicals}
(Logical operators constituting a BM) We consider a stabilizer code constructed as in Eq.~\eqref{eq:Gc}, which encodes two logical qubits. Then any two logical operators which constitute a logical BM, e.g., $\overline{XX} \in [\overline{XX}]$ and $\overline{ZZ} \in [\overline{ZZ}]$, anticommute in an odd number of qubits in each code.
\end{restatable}
\begin{proof}
Provided in App.~\ref{app:proof-lem-acomm-logicals}.
\end{proof}

To perform a logical BM, we need to measure a pair of observables, specifically two logical operators that constitute a BM. A pair of observables that constitutes a BM does commute. Thus, one might be tempted to assume that, in principle, it could be possible for them to anticommute in zero qubits. If Lem.~\ref{lem:acomm-logicals} did not hold and we knew a pair of observables which constitutes a logical BM while not anticommuting in any qubit, a logical BM scheme with unit success probability would be trivial to find. We could simply decompose the two logical operators into single-qubit Pauli operators and measure them. However, since Lem.~\ref{lem:acomm-logicals} does hold, we deduce that if we tried to measure any two logical operators in this fashion they would conflict on the necessary single-qubit Pauli measurement on at least one qubit in each code.

To illustrate this, let us return to our example of two logical qubits, each encoded in a QPC($2,2$) code. As an example for two logical operators which constitute a logical BM let us consider $XXII \ XXII \in \overline{XX}$ and $ZIZI \ ZIZI \in \overline{ZZ}$. We cannot decompose these two operators into single-qubit measurements, which do not conflict on the measurement on the first qubit of each code, since they anticommute in these qubits. We conclude that a logical BM needs to measure this ``double'' information on the qubits where the two logical operators which constitute the BM anticommute.

One might be tempted to think that we can trivially proof the bound for logical BMs via the following argument. We know that we have to get the double information on two qubits. Seemingly, this is only possible to obtain via a successful physical BM. Thus, one might conclude that the bound for the success probability of a logical BM is simply the probability to have at least one successful physical BM. While this bound is, in fact, correct, as we will see in Thm.~\ref{thm:bound}, the proof is more involved. In stabilizer codes the information on different qubits is correlated. Thus, in principle, we could obtain the information in an indirect fashion, i.e. revealing information on one qubit by measuring another. In Sec.~\ref{sec:pauli-measurements-and-logical-operators-on-encoded-uniform-mixtures-of-bell-states} and App.~\ref{app:observables-that-commute-with-s} we have seen that we can already know the result of an observable with support on exclusively unmeasured qubits, if it completes a stabilizer measurement. In other words, the question is whether it is possible to obtain the logical Bell-state information even without successfully identifying any one of the physical Bell states. In the proof of Thm.~\ref{thm:bound} we will see that indirect measurements cannot be leveraged to overcome the bound.

\subsection{Necessary condition for an optimal logical Bell measurement with feedforward-based linear optics}
\label{sec:necessary-condition-for-an-optimal-logical-Bell-measurement-with-feedforward-based-linear-optics}
Up to this point, our treatment of measurements has been quite general. We only required that the measurements belong to the Pauli group and that all sequential measurements commute with each other. We will now focus on a more practical scenario, particularly relevant to linear-optics setups. In linear optics, implementing single-qubit operations and measurements is typically simple. Therefore, we will allow for arbitrary single-qubit measurements and single-qubit Clifford gates. Recall that all Clifford gates can be absorbed in the measurements, since they normalize the Pauli group. Finally, our set of possible measurements includes physical BMs, which can either succeed and unambiguously identify a Bell state, or yield a partial measurement result. We will show that having at least one successful physical BM performed on a pair of physical qubits, one from each of the two separate codes, where we have no prior knowledge about the outcome is a necessary condition for a successful logical BM. Consequently, the success probability is upper bounded by the probability to have at least one such successful physical BM.

\begin{theorem}
\label{thm:bound}
(Necessary condition for an optimal logical BM with feedforward-based linear optics) The success probability of a destructive logical BM on two logical qubits, each encoded using stabilizer codes of $n_1$ and $n_2$ physical qubits, respectively, using only destructive BMs, destructive single-qubit Pauli measurements, and single-qubit Clifford gates, and allowing for feedforward is upper bounded by:
\begin{equation}
	1- \left( 1 - \mathbb{P}_B \right)^{ \min(n_1,n_2) }.
\end{equation}
\end{theorem}
While in Thm.~\ref{thm:bound} we present the most general bound, it is interesting to note that for the case of two identical error correction codes with $n$ physical qubits each, and assuming standard linear-optics BMs which have a success probability of $P_B=\frac{1}{2}$, the bound simplifies to:
\begin{equation}
	1-2^{-n}.
\end{equation}

\begin{proof}
We proof that at least one successful physical BM performed on a pair of physical qubits, one from each of the two separate codes, where we have no prior knowledge about the outcome is necessary for a successful BM. We term such a success a blind success. The probability to have a blind success for an attempted physical BM was shown to be $\mathbb{P}_B$ in Lem.~\ref{lem:successful-bell-measurment}. In total it is possible to perform $\min(n_1,n_2)$ physical BMs on a pair of physical qubits, one from each of the two separate codes. Thus, the probability to have at least one blind success is $1- \left( 1 - \mathbb{P}_B \right)^{\min(n_1,n_2)}$. We proof this claim via contradiction, by showing that any logical measurement without a blind success will fail to obtain the full logical information.

In this proof we assume without loss of generality, that a logical BM is completed by measuring $\overline{XX} \in [\overline{XX}]$ and $\overline{ZZ} \in [\overline{ZZ}]$. We could have chosen any other pair from $\{ \overline{XX}, \overline{YY}, \overline{ZZ} \}$. For further clarification, the reader can refer to Fig.~\ref{fig:bound} while following the subsequent argument. From Lem.~\ref{lem:acomm-logicals} we know, that $\overline{XX}$ and $\overline{ZZ}$ anticommute in an odd number of qubits in each code. Hence, there exist two indices $\lambda_1$ and $\lambda_2$ where $\overline{XX}$ and $\overline{ZZ}$ anticommute in the first and second code, respectively. It is important to note that this applies to any two representatives of $[\overline{XX}]$ and $[\overline{ZZ}]$. Consequently, if we were to measure any representative of one of the two logical operators destructively without obtaining the double information, it would become impossible to measure any representative of the other logical operator. Hence, we can assume that we have no knowledge about any of the logical variables until the double information is acquired. Therefore, without a blind success we need another way to obtain this double information on $\lambda_1$ and $\lambda_2$. Since we cannot measure the operators which give the double information directly and destructively without a blind success, at least one of them has to be measured in an indirect fashion, leveraging the correlations between physical qubits in a stabilizer code.

Specifically, we need to indirectly measure an operator with support on $\lambda_1$ and $\lambda_2$ or two operators with support on one of the two individually. In the following, we will discuss the first of these two cases, with the discussion of the second proceeding analogously.

From the results derived in Sec.~\ref{sec:pauli-measurements-and-logical-operators-on-encoded-uniform-mixtures-of-bell-states} and App.~\ref{app:observables-that-commute-with-s}, where we covered all possible kinds of measurements, we know that we can only obtain this information indirectly if there exists an observable $M^*$ with support on $\lambda_1$ and $\lambda_2$ which completes a stabilizer measurement:
\begin{equation}
	M^* = \gamma \mu,
\end{equation}
where $\gamma \in \langle G_c^{(\mathbb{M})} \rangle$ and $\mu \in \langle \mathbb{M} \rangle$. Recall that $\langle G_c^{(\mathbb{M})}\rangle \subseteq \langle G_c \rangle$.
We rearrange the terms:
\begin{equation}
	\gamma = M^* \mu.
\end{equation}
Note that either $\gamma$ or $- \gamma$ is a stabilizer of the code. Without loss of generality, assume $\gamma$ is a stabilizer of the code.

Given that $\overline{XX}$ and $\overline{ZZ}$ anticommute in $\lambda_1$, it follows that either one of them anticommutes with $M^*$ in $\lambda_1$. Without loss of generality, assume it is $\overline{XX}$, with the discussion for $\overline{ZZ}$ proceeding analogously. Since the two logical qubits are encoded in independent stabilizer codes on disjoint sets of qubits, we can use the unique factorization:
\begin{equation}
	\overline{XX} = \overline{X}_1 \otimes \overline{X}_2 = \left( \overline{X}_1 \otimes I^{\otimes n_2} \right) \left( I^{\otimes n_1} \otimes \overline{X}_2 \right),
\end{equation}
where only the first factor has support on the first code. Thus, it follows that $\left( \overline{X}_1 \otimes I^{\otimes n_2} \right) \in [\overline{XI}]$ anticommutes with $M^*$ in $\lambda_1$. However, the code stabilizer $\gamma$ commutes with $\overline{XI}$, therefore $\gamma$ has to anticommute with $\overline{X}_1 \otimes I^{\otimes n_2}$ in another qubit $\lambda_1^\prime \neq \lambda_1$ on the first code. An analog argument on the second code leads to another qubit $\lambda_2^\prime \neq \lambda_2$ where $I^{\otimes n_1} \otimes \overline{X}_2 \in [\overline{IX}]$ anticommutes with $\gamma$ on the second code. For the analogous discussion of the second case, we would simply replace $M^*$ with two observables with support exclusively on $\lambda_1$ and $\lambda_2$, respectively.

We have shown that to obtain the double information on $\lambda_1$ and $\lambda_2$ without a blind success we need double information on another pair of qubits $\lambda_1^\prime$ and $\lambda_2^\prime$ which have shared support with $\overline{XX}$ and $\overline{ZZ}$ on the first and second code, respectively. However, the same argument we applied to $\lambda_1$ and $\lambda_2$ also applies to $\lambda_1^\prime$ and $\lambda_2^\prime$. We would need double information on a different pair $\lambda_1^{\prime \prime}$ and $\lambda_2^{\prime \prime}$. Since we have yet to obtain the necessary information on $\lambda_1$ and $\lambda_2$, this pair cannot be identical to the first. This reasoning extends iteratively until we reach the finite size limit of the code. Consequently, it is impossible to acquire the required double information, rendering it impossible to obtain the complete logical Bell information without a blind success.

\end{proof}

\begin{figure}[tb]
	\def\svgwidth{0.45\textwidth}
	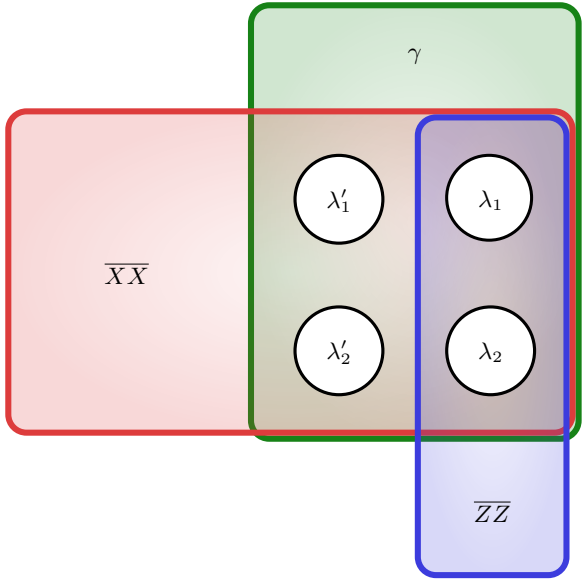
	\caption{\label{fig:bound}This figure illustrates a crucial aspect of the proof for Thm.~\ref{thm:bound}. Any pair of logical operators $\overline{XX}$ and $\overline{ZZ}$ anticommutes in two qubits $\lambda_1$ and $\lambda_2$. In the proof we argue that any code stabilizer $\gamma$ with support on $\lambda_1$ and $\lambda_2$ must also anticommute with one of the logical operators in two other qubits $\lambda_1^\prime$ and $\lambda_2^\prime$.}
\end{figure}

\section{Logical Bell measurements and their optimality}
\label{sec:logical-bell-measurements-and-their-optimality}
In this section, we present the principles of our optimal logical Bell measurement schemes. We characterize a scheme as optimal if it meets the bound defined in Thm.~\ref{thm:bound}. Here, we will consider the scenario where two logical qubits are encoded in identical stabilizer codes, which leads us to define $n_1=n_2=\frac{n_c}{2}\eqqcolon n$.

In Sec.~\ref{sec:exemplary-optimal-measurement-scheme-for-qpc(2,2)}, as an initial simple example, we will present our optimal Bell measurement scheme for the small QPC($2,2$) code in great detail. Subsequently, in Sec.~\ref{sec:strategy-and-single-code-reduction}, we describe the general strategy of our schemes and how our measurement schemes can be reduced to a single-code picture. In Sec.~\ref{sec:sufficient-conditions-for-an-optimal-logical-Bell-measurement-with-feedforward-based-linear-optics}, we present sufficient conditions for an optimal Bell measurement scheme. Finally, in Sec.~\ref{sec:heuristics-for-finding-optimal-logical-bell-measurements}, we will discuss some of the heuristics we used to find optimal schemes.

\subsection{Exemplary optimal Bell measurement for the quantum parity code ($2,2$)}
\label{sec:exemplary-optimal-measurement-scheme-for-qpc(2,2)}
In this section, we describe our logical BM scheme on two logical qubits, each encoded in the QPC($2,2$) code~\cite{PhysRevLett.95.100501}. For additional intuition, App.~\ref{app:observables-that-commute-with-s} includes an illustrative discussion of how a logical $\overline{X}$ measurement is performed on a single logical qubit.

Recall the stabilizer generators of the QPC($2,2$) code:
\begin{equation}
G_i = \{ XXXX, ZZII, IIZZ \},
\end{equation}
where $i \in \{1,2\}$, and the trivial extension is omitted. Furthermore, the relevant logical operators are:
\begin{equation}
\{ XXII, IIXX \} \subset [\overline{X}]
\end{equation}
and
\begin{equation}
\{ ZIZI, ZIIZ, IZZI, IZIZ \} \subset [\overline{Z}].
\end{equation}
For convenience, we designate the first two qubits of the code as the first row and the last two as the second row. In the literature, the rows of this code are often referred to as blocks~\cite{PhysRevLett.95.100501,PhysRevLett.117.210501,PhysRevA.95.012327,PhysRevA.99.062308,PhysRevA.100.052303,PRXQuantum.4.040322}. Covering a row with $X$ operators gives a logical $\overline{X}$ operator and a $Z$ operator on one qubit in each row constitutes a logical $\overline{Z}$ operator. We define the code stabilizers of the two codes, $G_1$ and $G_2$, by extending the single-code stabilizers $G_i$ across the full Hilbert space and reordering the elements:
\begin{equation}
	\begin{aligned}
	G_1 & = \{ ZZII \ IIII, XXXX \ IIII, IIZZ \ IIII \} \\
		& = \{ g_{1,s} \}_{s \in \{ 1, \dots, n-1 \} }
	\end{aligned}
\end{equation}
and
\begin{equation}
	\begin{aligned}
	G_2 & = \{ IIII \ ZZII, IIII \ XXXX, IIII \ IIZZ \} \\
		& = \{ g_{2,s} \}_{s \in \{ 1, \dots, n-1 \} }.
	\end{aligned}
\end{equation}
In the preceding section, we combined the two copies of the code:
\begin{equation}
\begin{aligned}
	G_c = & \ G_1 \cup G_2 \\
	= & \ \{ XXXX \ IIII, ZZII \ IIII, IIZZ \ IIII, \\
	& \ \phantom{\{} IIII \ XXXX, IIII \ ZZII, IIII \ IIZZ \phantom{,} \} \\
	= & \ \{ g_j \}_{j \in \{1, \dots, n-2 \}}.
\end{aligned}
\end{equation}
However, since we are working with identical codes in this section, it is instructive to continue treating the generators separately. In the following, for convenience, we term two qubits that correspond to each other between the two codes as a qubit pair. For example, the second qubit of each code forms a qubit pair, so the operator $IXII \ IXII$ acts on the second qubit pair.

We define an operator that acts identically on both codes as a transversal operator. Similarly, a transversal BM refers to any physical BM performed on two corresponding qubits of the two codes. In our measurement schemes, we aim to measure transversal logical operators, namely, elements of the sets
\begin{equation}
	\{ XXII \ XXII, IIXX \ IIXX \} \subset [\overline{XX}]
	\label{eq:qpc22-logical-xx}
\end{equation}
and
\begin{equation}
	\begin{aligned}
	\{ & ZIZI \ ZIZI, ZIIZ \ ZIIZ, & \\
	& IZZI \ IZZI, IZIZ \ IZIZ \phantom{,} \} & \subset [\overline{ZZ}].
	\end{aligned}
	\label{eq:qpc22-logical-zz}
\end{equation}

We define the initial quantum state as
\begin{equation}
	S = \langle G_1 \cup G_2 \cup L_{x,z} \rangle,
\end{equation}
where
\begin{equation}
	L_{x,z} = \{ l_x XXII \ XXII , l_z ZIZI \ ZIZI \}.
\end{equation}
Recall that the stabilizer group is independent of the choice of the representatives for the logical operators.

In our scheme, the only necessary operations are transversal physical BMs. In this section, we call a partial BM a failed BM. The scheme is illustrated in Fig.~\ref{fig:qpc-2-code-example}. We will start by briefly outlining the strategy of our scheme. To simplify, when measuring a transversal operator on a qubit pair, we refer to this action as measuring the Pauli information of the pair. For instance, when measuring the operator $IXII \ IXII$, we are obtaining the $XX$ information of the second qubit pair. For the measurement scheme to succeed we want to measure an element of each of the sets $[\overline{XX}]$ and $[\overline{ZZ}]$. From Eqs.~\eqref{eq:qpc22-logical-xx} and~\eqref{eq:qpc22-logical-zz} we can deduce how to obtain these operators. To measure an element of $[\overline{XX}]$ we need to obtain the $XX$ information of each qubit pair of one row. To measure an element of $[\overline{ZZ}]$ we need to obtain $ZZ$ information of at least one qubit pair in each row. Thus, the strategy is as follows. In each row we start by measuring the $XX$ information on the first qubit pair using a transversal $XX$-BM. If this measurement succeeds we additionally obtain the $ZZ$ information on the first qubit pair. Therefore, we do not need the $ZZ$ information on the remaining qubit pair in that row which then can be measured with another transversal $XX$-BM. This way, we obtain the logical $\overline{XX}$ information $l_x$ in this row. If, however, the BM on the first qubit pair of the row only obtains a partial outcome the remaining qubit pair is measured with a $ZZ$-BM. This ensures that the $ZZ$ information is obtained on at least one qubit pair of each row. Thus, after measuring any row the logical $ZZ$ information $l_z$ remains obtainable. We will now analyze this scheme in more detail. While the transformation of the stabilizer state could be directly taken from Lem.~\ref{lem:standard-form}, we will instead carry out the derivation explicitly in this example.

\begin{figure}
	\def\svgwidth{0.44\textwidth}
	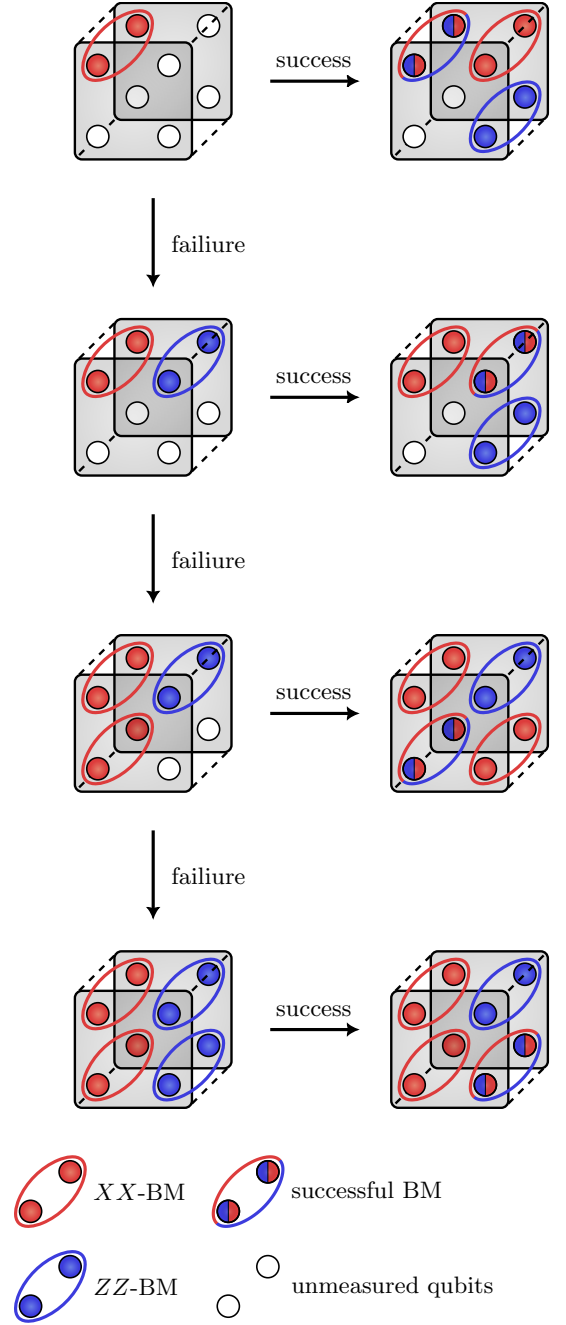
	\caption{\label{fig:qpc-2-code-example}Logical BM scheme for QPC($2,2$). Each gray box represents a single logical qubit encoded in four physical qubits, depicted as circles. Red and blue indicates $XX$-BMs and $ZZ$-BMs, respectively. The qubits filled with both red and blue indicate a successful physical BM.}
\end{figure}

We start our measurement scheme by performing a transversal $XX$-BM on the first qubit pair. The three corresponding transversal operators $XIII \ XIII$, $YIII \ YIII$ and $ZIII \ ZIII$ anticommute with at least one stabilizer generator of the code: e.g., the operators $XIII \ XIII$ and $YIII \ YIII$ anticommute with $g_{1,1}$ and $ZIII \ ZIII$ anticommutes with $g_{1,2}$. We conclude that for each of these operators both outcomes are equally likely. Thus, from Lem.~\ref{lem:successful-bell-measurment} we deduce that the success probability of this measurement is $\mathbb{P}_B$.

From Lem.~\ref{lem:successful-bell-measurment}, we know that a successful physical BM behaves as expected. Thus, upon a successful BM we measure the operators $M_{x,1} = XIII \ XIII$ and $M_{z,1} = ZIII \ ZIII$, yielding results $m_{x,1}$ and $m_{z,1}$, respectively. In this case, we continue to measure the operators $M_{x,2} = IXII \ IXII$ and $M_{z,4} = IIIZ \ IIIZ$ with results $m_{x,2}$ and $m_{z,4}$, respectively. We use the first subscript to denote the Pauli type of the observable and the second subscript to refer to the measured qubit pair. Both of these operators can be measured with unit probability by performing a transversal $XX$-BM and a transversal $ZZ$-BM on the second and fourth qubit pairs, respectively. In this case, since $M_{x,1} M_{x,2} = XXII \ XXII \in [\overline{XX}]$ and $M_{z,1} M_{z,4} = ZIIZ \ ZIIZ \in [\overline{ZZ}]$, we successfully measure the logical Bell information with $l_x = m_{x,1} m_{x,2}$ and $l_z = m_{z,1} m_{z,4}$, thus concluding the scheme.

If, on the other hand, the first physical BM fails, we obtain a partial result. A partial $XX$-BM on the  qubit pair yields results $r_{1,1}$ and $r_{2,1}$ for $XIII \ IIII$ and $IIII \ XIII$, respectively, where we use the first index to denote the code on which the respective observable has support on and the second index to enumerate the measurements. We derive the eigenvalue $m_{x,1} = r_{1,1}r_{2,1}$ of $M_{x,1}$ by multiplying these outcomes. We take note that $g_{1,1}$ and $l_z ZIZI \ ZIZI$ are the only stabilizer generators with which $XIII \ IIII$ anticommutes and $g_{2,1}$ and $l_z ZIZI \ ZIZI$ are the only stabilizer generator with which $IIII \ XIII$ anticommutes. Thus, to capture the change in the quantum state resulting from the measurement, we have to replace the stabilizer generators $g_{1,1}$ and $g_{2,1}$ with the two measured observables and multiply $g_{1,1}$ and $g_{2,1}$ to $l_z ZIZI \ ZIZI$. Therefore, the current stabilizer after the first partial BM reads:
\begin{equation}
	S^{(\mathbb{M}_1)} = \langle G_c^{(\mathbb{M}_1)} \cup \mathbb{M}_1 \cup L_{x,z}^{(\mathbb{M}_1)} \rangle,
\end{equation}
where
\begin{equation}
	\begin{aligned}
	G_c^{(\mathbb{M}_1)} 	= \{ & g_{1,2}, g_{1,3}, g_{2,2}, g_{2,3} \} \\
	 					= \{ & XXXX \ IIII, IIZZ \ IIII, \\
							 & IIII \ XXXX, IIII \ IIZZ \phantom{,} \},
	\end{aligned}	
\end{equation}
\begin{equation}
	\mathbb{M}_1 = \{ r_{1,1} XIII \ IIII , r_{2,1} IIII \ XIII \},
\end{equation}
and
\begin{equation}
	L_{x,z}^{(\mathbb{M}_1)} = \{ l_x XXII \ XXII, l_z IZZI \ IZZI \}.
\end{equation}

We continue by performing a transversal $ZZ$-BM on the second qubit pair. The observables $IYII \ IYII$ and $IZII \ IZII$ each anticommute with a generator of the current stabilizer $S^{(\mathbb{M}_1)}$. Since the eigenvalue $m_{x,1}$ was already obtained, the operator $M_{x,2} = IXII \ IXII$ with outcome $m_{x,2}$ completes a logical $\overline{XX}$ measurement measuring $XXII \ XXII \in [\overline{XX}]$: $m_{x,1} M_{x,1} \times m_{x,2} M_{x,2} = m_{x,1} XIII \ XIII \times m_{x,2} IXII \ IXII = l_x XXII \ XXII \in S^{(\mathbb{M}_1)}$. Since $l_x$ is a symmetric Bernoulli random variable with outcomes $\{ -1,+1 \}$ Lem.~\ref{lem:successful-bell-measurment} applies and the success probability for this second physical BM is $\mathbb{P}_B$ as well.

If this second physical BM is successful we measure $M_{x,2} = IXII \ IXII$ and $M_{z,2} = IZII \ IZII$. Recall, that the partial BM on the first qubit pair obtained the eigenvalue of $M_{x,1} = XIII \ XIII$. We continue the scheme by measuring the observable $M_{z,4} = IIIZ \ IIIZ$ with probability one. We again obtain a successful logical BM since $M_{x,1}M_{x,2} = XIII \ XIII \times IXII \ IXII = XXII \ XXII \in [\overline{XX}]$ and $M_{z,2}M_{z,4} = IZII \ IZII \times IIIZ \ IIIZ = IZIZ \ IZIZ \in [\overline{ZZ}]$.

If the second physical BM fails, we obtain $IZII \ IIII$ and $IIII \ IZII$, with outcomes $r_{1,2}$ and $r_{2,2}$, respectively. Their product again yields the eigenvalue of the transversal operator $M_{z,2}$: $m_{z_2} = r_{1,2} r_{2,2}$. At this point we have obtained the eigenvalues $m_{x,1}$ and $m_{z,2}$ of $M_{x,1} = XIII \ XIII$ and $M_{z,2} = IZII \ IZII$, respectively. The current stabilizer after the second partial BM is
\begin{equation}
	S^{(\mathbb{M}_2)} = \langle G_c^{(\mathbb{M}_2)} \cup \mathbb{M}_2 \cup L_{x,z}^{(\mathbb{M}_2)} \rangle,
\end{equation}
where
\begin{equation}
	\begin{aligned}
	G_c^{(\mathbb{M}_2)} 	= \{ & g_{1,3}, g_{2,3} \} \\ \\
	 					= \{ & IIZZ \ IIII, IIII \ IIZZ \},
	\end{aligned}	
\end{equation}
\begin{equation}
	\begin{aligned}
	\mathbb{M}_2 = \{ 	& r_{1,1} XIII \ IIII, r_{2,1} IIII \ XIII, \\
						& r_{1,2} IZII \ IIII, r_{2,2} IIII \ IZII \phantom{,} \},
	\end{aligned}
\end{equation}
and
\begin{equation}
	L_{x,z}^{(\mathbb{M}_2)} = \{ l_x IIXX \ IIXX, l_z IZZI \ IZZI \}.
\end{equation}

We proceed with the second row of the code identically to the first row. Again, by applying Lem.~\ref{lem:successful-bell-measurment} similarly to the second qubit pair, one can straightforwardly deduce that the success probability of the third BM is $\mathbb{P}_B$. If the $XX$-BM on the third qubit pair succeeds we complete the protocol by performing an $XX$-BM on the last qubit pair. In this case we successfully measure $IIXX \ IIXX \in [\overline{XX}]$ and $IZZI \ IZZI \in [\overline{ZZ}]$. If the BM on the third qubit pair fails we obtain the eigenvalues $m_{x,1}$, $m_{z,2}$ and $m_{x,3}$ of $M_{x,1} = XIII \ XIII$, $M_{z,2} = IZII \ IZII$ and $M_{x,3} = IIXI \ IIXI$, respectively. The current stabilizer after the third partial BM reads:
\begin{equation}
	S^{(\mathbb{M}_3)}	= \langle \mathbb{M}_3 \cup L_{x,z}^{(\mathbb{M}_3)} \rangle,
\end{equation}
where
\begin{equation}
	\begin{aligned}
	\mathbb{M}_3 = \{ 	& r_{1,1} XIII \ IIII, r_{2,1} IIII \ XIII, \\
						& r_{1,2} IZII \ IIII, r_{2,2} IIII \ IZII, \\
						& r_{1,3} IIXI \ IIII, r_{2,3} IIII \ IIXI \phantom{,} \}
	\end{aligned}
\end{equation}
and
\begin{equation}
	L_{x,z}^{(\mathbb{M}_3)} = \{ l_x IIXX \ IIXX, l_z IZIZ \ IZIZ \}.
\end{equation}

In this case, we perform one last physical BM on the final qubit pair. For the last physical BM any transversal physical operator completes a logical operator,
\begin{equation}
	\begin{aligned}
	M_{x,3} \times IIIX \ IIIX 	& = IIXI \ IIXI \times IIIX \ IIIX \\
								& = IIXX \ IIXX \in [\overline{XX}],
	\end{aligned}
\end{equation}
\begin{equation}
	\begin{aligned}
	M_{z,2} \times IIIZ \ IIIZ	& = IZII \ IZII \times IIIZ \ IIIZ \\
								& = IZIZ \ IZIZ \in [\overline{ZZ}],
	\end{aligned}
\end{equation}
\begin{equation}
	\begin{aligned}	
	& M_{z,2} \times M_{x,3} \times IIIY IIIY\\	& = IZII \ IZII \times IIXI \ IIXI \times IIIY \ IIIY \\
												& = IZIY \ IZIY \in [\overline{YY}].
	\end{aligned}
\end{equation}

Therefore, we apply Lem.~\ref{lem:successful-bell-measurment} to deduce that the success probability is $\mathbb{P}_B$. If the last measurement succeeds we successfully measure $IIXX \ IIXX \in [\overline{XX}]$ and $IZIZ \ IZIZ \in [\overline{ZZ}]$. To achieve a successful logical BM, if all previous physical BMs failed, the final physical BM must succeed. This directly follows from Lem.~\ref{lem:forbidden-measurements}, since any single-qubit Pauli observable on one qubit of the last qubit pair commutes with $\mathbb{M}_3$ and anticommutes with at least one element of $\{ l_x IIXX \ IIXX, l_z IZIZ \ IZIZ \}$. Additionally, we showed in Thm.~\ref{thm:bound} that if all physical BMs fail we are unable to obtain both logical eigenvalues $\overline{XX}$ and $\overline{ZZ}$. We can now easily infer the success probability of this scheme. All physical BMs have a success probability of $\mathbb{P}_B$ and if at least one succeeds we are able to measure the logical Bell information with probability one. Thus, we conclude that the success probability of the logical measurement scheme reaches the bound $1- \left( 1 - \mathbb{P}_B \right)^{n} = 1- \left( 1 - \mathbb{P}_B \right)^{4}$. For a standard linear-optics BM with $\mathbb{P}_B = \frac{1}{2}$, we thus reproduce the value of $\frac{15}{16}$ which coincides with the success probability of the scheme in Ref.~\cite{PhysRevA.100.052303} for the QPC($2,2$) code.

\subsection{Strategy and single-code reduction}
\label{sec:strategy-and-single-code-reduction}
Let us summarize some terminology introduced in Sec.~\ref{sec:exemplary-optimal-measurement-scheme-for-qpc(2,2)} we will use throughout this work. We term two qubits that correspond to each other between the two codes as a qubit pair. For example, the second qubit of each code forms the second qubit pair. Furthermore, we call an operator that acts identically on both codes as a transversal operator. Similarly, a transversal BM refers to any physical BM performed on a qubit pair. For brevity, from this point forward, ``transversal BM'' will refer to a transversal guaranteed-partial information physical BM. Furthermore, recall that two operators are said to conflict on a qubit if their decompositions into single-qubit Pauli operators require different Pauli information on that qubit.

In Sec.~\ref{sec:strategy-for-optimal-measurement-schemes} we describe our general strategy to design optimal Bell measurement schemes. In Sec.~\ref{sec:single-code-reduction} we demonstrate how our schemes can be reduced to a single-code picture.

\subsubsection{Strategy for optimal logical Bell measurements}
\label{sec:strategy-for-optimal-measurement-schemes}
Our goal is to present logical BM schemes that reach the bound established in Thm.~\ref{thm:bound} for two logical qubits encoded in the same code. All our measurement schemes follow a similar conceptual framework and consist of two parts. In the first part of the schemes, the scheme remains fixed and qubit pairs are measured one by one in a predetermined order using predefined transversal BMs, until the first successful outcome occurs. This first part of the scheme is designed such that, there always exists a pair of transversal logical operators $\overline{XX}$ and $\overline{ZZ}$ which anticommute exclusively in the qubits where the success occurred and do not conflict with any prior measurement. Therefore, once a transversal BM is successful, the logical Bell information can be obtained with probability one by completing the measurement of these two logical operators, in the second part of scheme. Completing the measurement of the logical operators can for example be performed by transversal BMs. Since the logical operators do not conflict on any unmeasured qubits no other successful BM is necessary in the second part of the scheme. Our strategy may appear restrictive within the broad range of possible strategies, but we have discovered optimal schemes for all codes considered in this work within these simplifications.

We recall from Sec.~\ref{sec:physical-bell-measurements-on-entangled-quantum-states} that an unsuccessful BM, i.e., a partial measurement outcome, is equivalent to two single-qubit Pauli measurements. Therefore, alternatively, one can use single-qubit measurements to complete the logical operators in the second part of the scheme. Using this equivalence, our treatment remains identical, regardless of whether transversal BMs or single-qubit Pauli measurements are used in the second part of the scheme. In a related context, Refs.~\cite{Hilaire_2021,patil2024improveddesignallphotonicquantum} demonstrated that switching to single-qubit Pauli measurements following a successful transversal BM can enhance the loss tolerance of logical BMs on the tree code. However, for our schemes, it is unclear whether single-qubit measurements are always the better choice, because photon loss can make the initially chosen logical operators unobtainable. In this situation, a second successful BM can potentially allow a different logical operator pair to be measured. Thus, for general codes and measurement schemes, determining the loss regimes in which single-qubit Pauli measurements or transversal BMs maximize the success probability remains an open question. Since the use of single-qubit or transversal BMs in the second part of the scheme is equivalent in our loss-less treatment, we restrict our discussion to transversal BMs for simplicity, without loss of generality.

\subsubsection{Single-code reduction}
\label{sec:single-code-reduction}
Up to this point, we have consistently tracked the logical two-qubit code space. In the following, we present an argument demonstrating how our measurement schemes can be reduced to a single-code picture. 

Tracking the stabilizer generators sequentially, instead of combining the commuting measurement operators into a single measurement is necessary whenever feedforward may occur between measurements. However, after the first successful BM, the second part of the scheme is fixed and feedforward is no longer needed, so tracking the stabilizer generators is necessary only for the first part of the scheme. The first part of the scheme consists of partial BMs and one successful BM. At any point in the first part of the scheme, the stabilizer generators can be chosen so that they exhibit the following symmetry between the two codes. Every element of the stabilizer generators that is not transversal has support on only one of the two codes, and there exists another element in the generators which is the same operator on the other code up to a sign. The signs of these two operators only differs for elements of the measurements $\mathbb{M}$, but not for the elements of the code generators $G_c^{\mathbb{M}}$. Furthermore, the logical generators $L_{x,z}^{(\mathbb{M})}$ are transversal. However, in cases where the signs of the two operators in $\mathbb{M}$, which are otherwise identical in their respective codes, differ, we are only interested in the product of the signs, since we aim to obtain eigenvalues of transversal logical operators. Thus, the stabilizer generators can essentially be treated as symmetric under the exchange of the two codes. A detailed derivation of this symmetry is provided in App.~\ref{app:single-code-reduction}.

As a consequence, we can simplify the treatment of our schemes. Since the generators in the code stabilizers $G_c^{(\mathbb{M})}$ are always identical on both codes it is sufficient to track them for one code. This simplification is illustrated in Fig.~\ref{fig:qpc-1-code-example} where we reduce the scheme of Fig.~\ref{fig:qpc-2-code-example} to a single-code picture. For simplicity, we use the same notation for the sets $S_c$, $G_c^{(\mathbb{M})}$, $\mathbb{M}$, and $L_{x,z}^{(\mathbb{M})}$ in both pictures throughout this work. The picture they are defined in will always be made clear in context.

\begin{figure}
	\def\svgwidth{0.44\textwidth}
	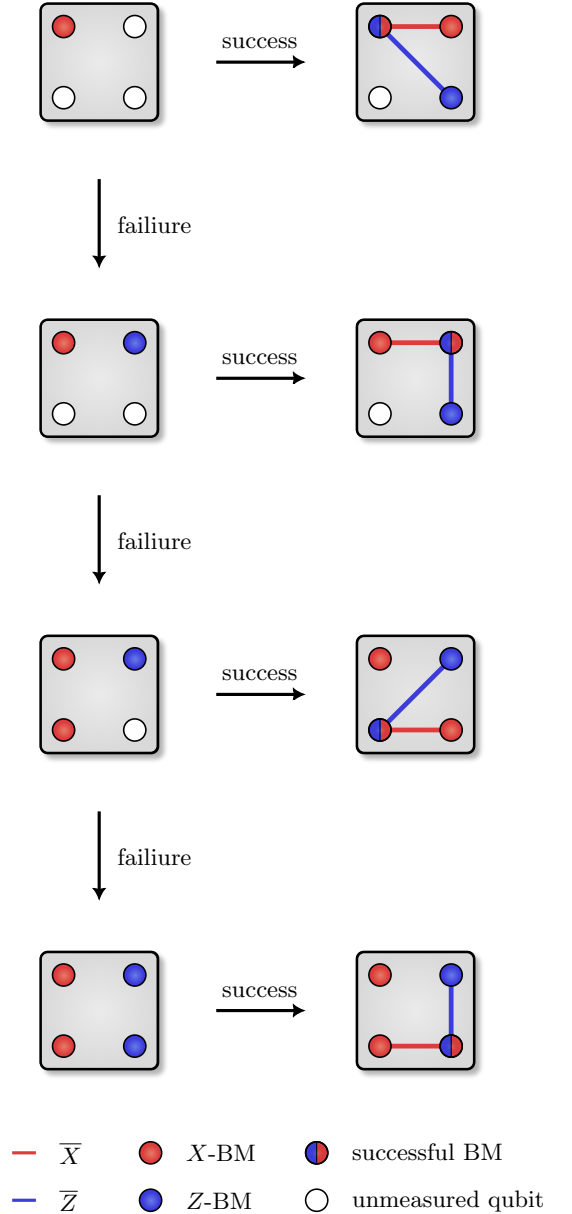	
	\caption{\label{fig:qpc-1-code-example}Logical BM scheme for QPC($2,2$) in the single-code picture. Each gray box represents a single logical qubit encoded in four physical qubits, depicted as circles. Red and blue indicates $X$-BMs and $Z$-BMs, respectively. Qubits filled with both red and blue indicate a successful physical BM. The red and blue lines illustrate the measured logical operators $\overline{X}$ and $\overline{Z}$, respectively.}
\end{figure}

We now discuss how partial transversal BMs transform stabilizer generators in the single-code picture. In App.~\ref{app:single-code-reduction}, we formally demonstrate that the code stabilizers $G_1^{(\mathbb{M})}$ and $G_2^{(\mathbb{M})}$ in each code transform identically under a partial transversal BM. To illustrate how this applies to the single-code picture, let us revisit the example from Sec.~\ref{sec:exemplary-optimal-measurement-scheme-for-qpc(2,2)}. Consider the stabilizer generators of two logical qubits in a uniform mixture of logical Bell states, where each logical qubit is encoded in the QPC($2,2$) code:
\begin{equation}
	\begin{aligned}
	S = \langle & \{ ZZII \ IIII, IIZZ \ IIII, XXXX \ IIII, \\
	 			& \phantom{\{ }IIII \ ZZII, IIII \ IIZZ, IIII \ XXXX \phantom{,} \}	\\
	 			& \cup L_{x,z} \rangle.
	\end{aligned}
\end{equation}
Suppose we perform a partial $XX$-BM on the first qubit pair. The measured observables $XIII \ IIII$ and $IIII \ XIII$ anticommute exclusively with $ZZII \ IIII$ and $IIII \ ZZII$, respectively. As a result, each measured observable replaces the respective stabilizer generator. The resulting stabilizer generators are:
\begin{equation}
	S^{(\mathbb{M}_1)} = \langle G_c^{(\mathbb{M}_1)} \cup \mathbb{M}_1 \cup L_{x,z}^{(\mathbb{M}_1)} \rangle,
\end{equation}
where
\begin{equation}
	\begin{aligned}
	G_c^{(\mathbb{M}_1)} = \{ 	& IIZZ \ IIII, XXXX \ IIII, \\
							 	& IIII \ IIZZ, IIII \ XXXX \phantom{,} \},
	\end{aligned}	
\end{equation}
\begin{equation}
	\mathbb{M}_1 = \{ r_{1,1} XIII \ IIII , r_{2,1} IIII \ XIII \}.
\end{equation}

In the single-code picture, this process simplifies as follows. The single-code stabilizer generators are:
\begin{equation}
	G_c = \{ ZZII, IIZZ, XXXX \}.
\end{equation}
Measuring the first qubit pair with a partial $XX$-BM reduces to measuring the single-qubit observable $XIII$, which anticommutes exclusively with $ZZII$ and thus replaces this stabilizer generator. The updated stabilizer generators are:
\begin{equation}
	G_c^{(\mathbb{M}_1)} = \{ IIZZ, XXXX \},
\end{equation}
and
\begin{equation}
	\mathbb{M}_1 = \{ r_{1,1}r_{2,1} XIII \} = \{ m_1 XIII \},
\end{equation}
where $m_1 = r_{1,1}r_{2,1}$ is the product of the two outcomes of the partial BM in the two-code picture. The observations in the above example are generalized in Lem.~\ref{lem:transversal-bm}.

\begin{restatable}{lemma}{transversalbmrestatable}
\label{lem:transversal-bm}
(Transversal BMs in the single-code picture) Partial transversal BMs in the two-code picture act as single-qubit Pauli measurements in the single-code picture. Specifically, an $XX$-, $YY$-, or $ZZ$-BM on the $i$-th qubit pair with a partial result transforms the single-code stabilizer as an $X$, $Y$, or $Z$ measurement on the $i$-th qubit, respectively.
\end{restatable}
\begin{proof}
Provided in App.~\ref{app:single-code-reduction}.
\end{proof}

For the partial BMs it is only necessary to track the result of the transversal operator, i.e., the product of the two singe-qubit measurement results. Recall that, when measuring a transversal operator on a qubit pair, we refer to this action as measuring the Pauli information of the pair. For example, when dealing with two logical qubits, each encoded in the QPC($2,2$) code, measuring the operator $IXII \ IXII$ yields the $XX$ information of the second qubit pair. In our reduced single-code picture, this measurement corresponds to obtaining the Pauli information of a single qubit. Specifically, in our example, it means measuring the $X$ information of the second qubit. To clarify the context, in the single-code picture, we will refer to transversal BMs as $X$-, $Y$-, and $Z$-BMs from now on. For brevity, we will often associate a Pauli operator acting on a qubit with the corresponding BM. For instance, the BM $X_i$ denotes the $X$-BM on qubit $i$. Similarly, transversal logical operators transform to single-code logical operators in the single-code picture, for instance $\overline{XX}$ becomes $\overline{X}$.

In this section, we have demonstrated that our schemes exhibit a symmetry which can be leveraged to reduce them to a single-code picture. We will employ this concept in the rest of the paper to simplify the discussions of our schemes, providing a more elegant treatment over a two-code treatment.

\subsection{Sufficient conditions for an optimal logical Bell measurement with feedforward-based linear optics}
\label{sec:sufficient-conditions-for-an-optimal-logical-Bell-measurement-with-feedforward-based-linear-optics}
In this section, we define our measurement schemes algebraically and present sufficient conditions for an optimal logical BM with feedforward-based linear optics.

Our measurement schemes can be characterized by two sequences. To define these sequences, we enumerate the qubits of the code in the order they are measured. Furthermore, let us define the notation for a Pauli operator acting on exclusively the $j$-th qubit of an $n$-qubit code space: $W_j \equiv I^{\otimes (j-1)} \otimes W \otimes I^{\otimes (n-j)}$. For instance, the operator $X_2$ on a four-qubit code is $X_2 = IXII$. The first sequence $\mathbb{B}$ consists of $n-1$ Pauli operators,

\begin{equation}
	\mathbb{B} = \left( b_j \right)_{j \in \{ 1, \dots, n-1 \}}, \quad \text{where} \quad b_j \in \{ X_j, Y_j, Z_j \},
	\label{eq:B-def}
\end{equation}
where the index $j$ in $b_j$ indicates the position in the sequence, and the subscript $j$ in $\{ X_j, Y_j, Z_j \}$ specifies the qubit the Pauli operator acts on. They are equal by definition, since the operator $b_j$ acts on qubit $j$. It defines the BMs that are performed until the first successful one, i.e., the $j$-th qubit pair is measured with a transvesal BM corresponding to $b_j$. For instance, if $b_3 = X_3$, then the third qubit pair is measured with an $X$-BM as long as no successful BM has occurred on the first two qubit pairs. The second sequence $\mathbb{L}$ consists of $n$ ordered pairs. Each ordered pair consists of two logical operators of the respective code. The $j$-th pair defines the two logical operators which are measured in the event that the $j$-th BM is successful,
\begin{equation}
	\mathbb{L} = \left( \left( \overline{X}_j , \overline{Z}_j \right) \right)_{j \in \{ 1, \dots, n \}},
\end{equation}
where $\overline{X}_j \in \left[ \overline{X} \right]$ and $\overline{Z}_j \in \left[ \overline{Z} \right]$, and the index $j$ denotes the position in the sequence.

The sequence $\mathbb{B}$ is indexed up to $n - 1$ because the chosen measurement of the last qubit is irrelevant. If no successful BM has occurred up to that point, a successful one is required at the final qubit regardless. The sequence $\mathbb{L}$, however, includes the final index $n$, since a successful physical BM on the last qubit completes the logical BM.

In Thm.~\ref{thm:sufficient} we present sufficient conditions for an optimal Bell measurement scheme. We define a scheme as optimal if it reaches the bound established in Thm.~\ref{thm:bound}. After presenting Thm.~\ref{thm:sufficient} we will offer a discussion to facilitate a more intuitive understanding of these conditions and our schemes.

\begin{restatable}{theorem}{sufficientrestatable}
\label{thm:sufficient}
(Sufficient conditions for an optimal logical BM with feedforward-based linear optics) We consider two logical qubits, each encoded in the same single-qubit stabilizer code defined by the stabilizer group $S_c$. Let us further assume, that there exists a minimal generating set $G_c$ of $S_c$ and a sequence $\mathbb{C} = \left( c_j \right)_{j \in \{1, \dots, n-1 \} }$ in which each element of $G_c$ appears exactly once. Then, the sequences $\mathbb{B}$ and $\mathbb{L}$ characterize an optimal Bell measurement scheme if the following five conditions are met. Due to the sequential structure of the scheme, we refer to an operator as later than another if it appears at a higher position in its sequence, and as prior if it appears at a lower position.

Condition~1: Each operator $b_j$ anticommutes with $c_j$:
\begin{equation}
	\forall j \in \{ 1 , \dots, n-1 \} : \quad \acomm{b_j}{c_j} = 0.
	\label{eq:con-1}
\end{equation}

Condition~2: Each operator $b_j$ commutes with every later stabilizer generator:
\begin{equation}
	\forall k>j: \quad \comm{b_j}{c_k} = 0.
	\label{eq:con-2}
\end{equation}

Condition~3: For all $j \in \{ 1, \dots, n-1 \}$ each operator in the set $\tilde{b}_j \in \{ X_j, Y_j, Z_j \} \setminus \{ b_j \}$ either anticommutes with at least one non-prior stabilizer generator,
\begin{equation}
	\exists k \geq j: \quad \acomm{\tilde{b}_j}{c_k} = 0,
	\label{eq:con-3-a}
\end{equation}
or completes a logical measurement,
\begin{equation}
	\exists \mu \in \langle b_1, \dots, b_{j-1} \rangle: \quad \mu \tilde{b}_j \in \left[ \overline{X} \right] \cup \left[ \overline{Y} \right] \cup \left[ \overline{Z} \right].
	\label{eq:con-3-b}
\end{equation}

Condition~4: The logical operators $\overline{X}_j$ and $\overline{Z}_j$ commute with every prior element of $\mathbb{B}$ for all $j \in \{ 1, \dots, n \}$:
\begin{equation}
	\forall k<j : \quad \comm{\overline{X}_j}{b_k} = 0,
	\label{eq:con-4-a}
\end{equation}
\begin{equation}
	\forall k<j : \quad \comm{\overline{Z}_j}{b_k} = 0.
	\label{eq:con-4-b}
\end{equation}

Condition~5: We decompose the logical operators into single-qubit Pauli operators to formulate the last condition:
\begin{equation}
	\overline{X}_j = \bigotimes_{t=1}^{n} u_{j,t}, \quad \text{where } u_{j,t} \in \{ I, X, Y, Z \},
	\label{eq:logical-decomp-x}
\end{equation}
\begin{equation}
	\overline{Z}_j = \bigotimes_{t=1}^{n} v_{j,t}, \quad \text{where } v_{j,t} \in \{ I, X, Y, Z \}.
	\label{eq:logical-decomp-z}
\end{equation}
The logical operators $\overline{X}_j$ and $\overline{Z}_j$ anticommute only in $j$:
\begin{equation}
	\forall j \in \{ 1, \dots, n \} : \acomm{u_{j,j}}{v_{j,j}} = 0,
	\label{eq:double-info}
\end{equation}
\begin{equation}
	\forall k,j \in \{1, \dots, n \} \wedge k \neq j : \comm{u_{j,k}}{v_{j,k}} = 0.
	\label{eq:double-info-b}
\end{equation}

\end{restatable}

\begin{proof}
Provided in App.~\ref{app:proof-thm:sufficient}.
\end{proof}

We will now provide a more illustrative discussion of the conditions. The first two conditions~1 and~2 ensure that the observable of a partial BM outcome, $b_j$ anticommutes with exactly one element $c_j \in G^{(\mathbb{M}_j)}$ of the current code stabilizer generators. Therefore, each measurement $b_j$ replaces the corresponding code stabilizer $c_j$ with the product of its measured eigenvalue and the observable $b_j$. For simplicity, we refer to this process as the measurement replacing the stabilizer generator. These conditions also imply that, up to the first success, no logical information has been obtained, since an observable that completes a logical operator would commute with the current stabilizer. Note that condition~1 does not apply to the last qubit pair. The reason is that condition~5 ensures that, when all but the last qubit have been measured, any operator on the last qubit will always complete a logical measurement. Condition~3 states that all other single-qubit Pauli operators on the $j$-th qubit pair either anticommute with the current stabilizer or complete a logical operator. Since the logical variables $l_x$, $l_y$, and $l_z$ are symmetric Bernoulli random variables with outcomes $\{-1, +1\}$, the measurement outcomes of the operators that complete a logical operator are equally likely. Therefore, Lem.~\ref{lem:successful-bell-measurment} implies that condition~3 ensures a success probability of $\mathbb{P}_B$, provided that no transversal BM has yet succeeded. Note that the third condition is trivially satisfied for $j=1$, but we include it in the index domain for completeness. Condition~4 ensures that the logical operators $\overline{X}_j$ and $\overline{Z}_j$ do not conflict with any prior measurement, and condition~5 ensures that they do not conflict with each other on any unmeasured qubit. It is worth noting that conditions~4 and~5 imply that the logical operators $\overline{X}_j$ and $\overline{Z}_j$ only commute in every qubit except for the $j$-th qubit pair,
\begin{equation}
	\forall k,j \in \{ 1, \dots, n \} \wedge k \neq j : \comm{u_{j,k}}{v_{j,k}} = 0.
\end{equation}
Therefore, given that a successful BM occurred on the $j$-th qubit pair, they can be measured with probability one. Since in the two-code picture the logical operators $\overline{X}_j \otimes \overline{X}_j$ and $\overline{Z}_j \otimes \overline{Z}_j$ are transversal by definition, the unmeasured portions of these operators can either be obtained with probability one through transversal BMs, which are guaranteed to yield the partial information, or via single-qubit Pauli measurements.

In conclusion, each transversal BM up to the first success has a success probability of $\mathbb{P}_B$, and if any of them succeeds, the logical BM can be completed with probability one. For two identical $n$-qubit codes, we achieve the upper-bound success probability given in Thm.~\ref{thm:bound}, $1-(1-\mathbb{P}_B)^{n}$.

For a more detailed understanding of the conditions, we refer the reader to the proof provided in App.~\ref{app:proof-thm:sufficient}.

\subsection{Heuristics for finding optimal logical Bell measurements}
\label{sec:heuristics-for-finding-optimal-logical-bell-measurements}
In this section, we will explore heuristics that have been instrumental in discovering optimal schemes. We present two rules that we argue are necessary for a scheme to be optimal, offering guidelines for finding such schemes. Although we attempted to prove that these rules can also serve as sufficient conditions, we were unable to do so. Therefore, it remains an open question whether these rules are indeed sufficient for an optimal Bell measurement scheme. Proving these conditions to be sufficient would provide a more elegant method for demonstrating the optimality of our schemes than the conditions presented in Thm.~\ref{thm:sufficient}.

The first rule is to never measure logical information without a successful physical BM. The necessity of this rule was essentially argued in Lem.~\ref{lem:acomm-logicals} and the subsequent explanation following it. Upon measuring a logical operator $\overline{XX}$ or $\overline{ZZ}$ without a successful physical BM, the other one becomes unobtainable.

To present the second rule, we introduce the term ``almost measure an operator.'' Almost measuring an operator refers to measuring the operator except for one single-qubit measurement. For instance, if we measure the operator $XXII$, we have almost measured the operator $XXXI$. The second rule is defined in the single-code picture. The second rule is to never almost measure a code stabilizer. The necessity of the second rule can be clarified with the following argument, which is similar to an argument used in the proof of Thm.~\ref{thm:bound}. As discussed in the previous section, the pivotal part of our schemes is to use any successful transversal BM as the double information necessary to complete a logical BM. Let us assume we almost measured the code stabilizer $g$. We define $g=\prod_k \mu_k m$, where $\{ \mu_k \}$ is a subset of the already measured single qubit measurements and $m$ the single-qubit measurement missing to complete the measurement of the stabilizer. Now, let us assume that a successful transversal BM occurs on the support of $m$. Since we need to leverage the double information on the support of $m$, one of the two logical operators comprising the logical BM $\overline{X}$ and $\overline{Z}$ anticommutes with $m$. Without loss of generality, let us assume it is $\overline{X}$. Multiple choices for $\overline{X}$ may exist, but the argument holds for any choice that anticommutes with $m$. Since $\overline{X}$ and $g$ commute with each other, $\overline{X}$ also anticommutes with another operator $m' \in \{ \mu_k \}$. Thus, the measurement of $m'$ conflicts with $\overline{X}$, rendering the logical $\overline{XX}$ information $l_x$ unobtainable.

\section{Optimal logical Bell measurements for specific codes}
\label{sec:optimal-logical-bell-measurements-for-specific-codes}

In this section, we build on the general results discussed thus far by introducing specific optimal Bell measurement schemes applied to individual stabilizer codes.  Recall that we characterize a scheme as optimal if it satisfies the bound established in Thm.~\ref{thm:bound}. We continue to focus on the scenario where two logical qubits are encoded in identical stabilizer codes, so that $n_1=n_2=\frac{n_c}{2}\eqqcolon n$. We demonstrate that the schemes we devised for quantum parity, five-qubit, standard and rotated planar surface, tree, and seven-qubit Steane codes all achieve the bound of Thm.~\ref{thm:bound}.

For each code the explanation is structured similarly. We begin each section by briefly introducing the code and the measurement scheme up to the first successful BM. Then, the explanation is divided into two parts. First, we describe how logical information can be obtained if a successful physical BM occurs. In the second part, we demonstrate that the success probability for every physical BM is given by $\mathbb{P}_B$. As a prerequisite for this second part, we investigate the transformation of the stabilizer generators throughout the scheme. In Sec.~\ref{sec:sufficient-conditions-for-an-optimal-logical-Bell-measurement-with-feedforward-based-linear-optics} we established that this is sufficient to identify an optimal Bell measurement scheme. The order in which the two parts are addressed may vary between the codes.

In Sec.~\ref{sec:quantum-parity-code} we will generalize the scheme discussed in Sec.~\ref{sec:exemplary-optimal-measurement-scheme-for-qpc(2,2)} to the QPC code of arbitrary size. In Sec.~\ref{sec:five-qubit-code} we present our optimal scheme for the five-qubit code. The small number of qubits in the five-qubit code makes the proof both straightforward and easy to understand, providing an excellent example for how to apply Thm.~\ref{thm:sufficient}. In Secs.~\ref{sec:standard-planar-surface-code}, \ref{sec:rotated-planar-surface-code}, \ref{sec:tree-code}, and \ref{sec:steane-code} we present our schemes for the the standard and rotated planar surface code, the tree code, and the Steane code, respectively. Additionally, in Sec.~\ref{sec:optimization-and-comparionn-of-static-logical-bell-measurements-for-rotated-planar-surface-codes} we present an optimized static measurement scheme for the rotated planar surface code and compare various schemes for planar surface codes.

\subsection{Quantum parity code}
\label{sec:quantum-parity-code}
We will now generalize the QPC($2,2$) measurement scheme which we presented in Sec.~\ref{sec:exemplary-optimal-measurement-scheme-for-qpc(2,2)} to the QPC($r,m$) code of arbitrary size. This code consists of $r$ rows (or blocks in the original terminology), each containing $m$ qubits. The structure of the QPC($r,m$) code naturally leads to a double-index notation, where each qubit is indexed by a pair $(i, j)$, with $i \in \{1, \dots, r\}$ denoting the row and $j \in \{1, \dots, m\}$ enumerating the qubits within each row.

The QPC($r,m$) code is stabilized by two types of operators. First, each row $i$ of $m$ qubits is stabilized by $m-1$ stabilizer generators of the form $Z_{i,j}Z_{i,j+1}$ for all $j \in \{1, \dots, m-1\}$. Second, adjacent pairs of rows $i$ and $i+1$ are stabilized by the operator $\prod_{t=1}^{m} X_{i,t} X_{i+1,t}$. In conclusion, the QPC($r,m$) code is defined by the stabilizer generators:
\begin{equation}
	\begin{aligned}
	G_c = & \{ Z_{i,j}Z_{i,j+1} \}_{(i,j) \in \{(i,j) \mid 1 \leq i \leq r, 1 \leq j \leq m-1 \}} \\
		& \cup \{\prod_{t=1}^{m} X_{i,t} X_{i+1,t} \}_{i \in \{1, \dots, r-1\}}.
	\end{aligned}
\end{equation}

A logical $\overline{X}$ operator covers one row with $X$ operators:
\begin{equation}
	\prod_{t=1}^m X_{i,t} \in \left[ \overline{X} \right], \quad \text{where } i \in \{ 1, \dots , r \},
	\label{eq:qpc-logical-x}
\end{equation}
and a logical $\overline{Z}$ operator covers one qubit in every row with a $Z$ operator:
\begin{equation}
	\prod_{t=1}^r Z_{t,j_t} \in \left[ \overline{Z} \right],
	\label{eq:qpc-logical-z}
\end{equation}
where the indices $j_t$ can be chosen arbitrarily. Examples of these logical operators are illustrated in Fig.~\ref{fig:qpc-logical}.

\begin{figure}[tb]
	\def\svgwidth{0.45\textwidth}
	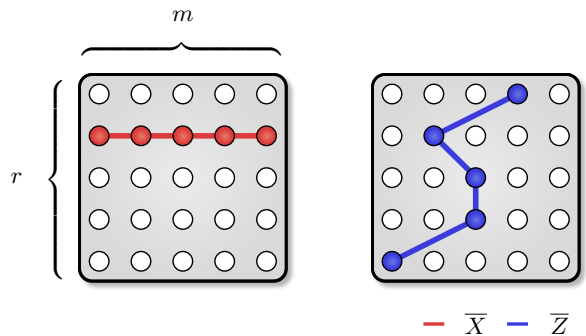
	\caption{\label{fig:qpc-logical}Examples of supports of a logical $\overline{X}$ (left) and $\overline{Z}$ (right) operators for QPC. Qubits on which the logical operators act with $X$, $Z$ or the identitiy are shown in red, blue, and white, respectively. The red and blue strings represent $\overline{X}$ and $\overline{Z}$, respectively. The code dimensions $(r, m)$ are shown in the left example.}
\end{figure}

In Sec.~\ref{sec:exemplary-optimal-measurement-scheme-for-qpc(2,2)} we already discussed our scheme for QPC($2,2$). We will now extend this discussion to the QPC code with arbitrary parameters.

For the measurement scheme to succeed we want to measure an element of each of the sets $[\overline{X}]$ and $[\overline{Z}]$. From Eqs.~\eqref{eq:qpc-logical-x} and~\eqref{eq:qpc-logical-z} we can deduce how to obtain these operators. To measure an element of $[\overline{X}]$ we need to obtain the $X$ information of every qubit of one row. To measure an element of $[\overline{Z}]$ we need to obtain $Z$ information of at least one qubit in each row. Thus, the strategy is as follows. In each row, we need to ensure that $Z$ information is obtained from at least one qubit while maximizing the probability of obtaining $X$ information from all qubits in the row. To achieve this, we measure the qubits sequentially from left to right using $X$-BMs, reserving the last qubit in each row for a $Z$-BM. If one of the $X$-BMs succeeds, we have acquired the necessary $Z$ information for that row. Consequently, we can complete the $\overline{X}$ measurement by performing $X$-BMs on all remaining qubits in the row including the last qubit. (Only if all $X$-BMs in the row fail, we stick to a $Z$-BM for the last qubit.) The $\overline{Z}$ measurement can then be completed by performing a $Z$-BM on one qubit from each remaining row. In the case where none of the BMs of the row succeeds, we proceed with the subsequent row in the same way as with the measured row. In conclusion, once a successful transversal BM is achieved, the logical BM can be completed with probability one. Note that the essence of the above explanation is captured in conditions~4 and~5 of Thm.~\ref{thm:sufficient}. The scheme is illustrated in Fig.~\ref{fig:qpc-scheme}.

\begin{figure}[tb]
	\def\svgwidth{0.45\textwidth}
\begingroup%
  \makeatletter%
  \providecommand\color[2][]{%
    \errmessage{(Inkscape) Color is used for the text in Inkscape, but the package 'color.sty' is not loaded}%
    \renewcommand\color[2][]{}%
  }%
  \providecommand\transparent[1]{%
    \errmessage{(Inkscape) Transparency is used (non-zero) for the text in Inkscape, but the package 'transparent.sty' is not loaded}%
    \renewcommand\transparent[1]{}%
  }%
  \providecommand\rotatebox[2]{#2}%
  \newcommand*\fsize{\dimexpr\f@size pt\relax}%
  \newcommand*\lineheight[1]{\fontsize{\fsize}{#1\fsize}\selectfont}%
  \ifx\svgwidth\undefined%
    \setlength{\unitlength}{2466.14173228bp}%
    \ifx\svgscale\undefined%
      \relax%
    \else%
      \setlength{\unitlength}{\unitlength * \real{\svgscale}}%
    \fi%
  \else%
    \setlength{\unitlength}{\svgwidth}%
  \fi%
  \global\let\svgwidth\undefined%
  \global\let\svgscale\undefined%
  \makeatother%
  \begin{picture}(1,0.63218391)%
    \lineheight{1}%
    \setlength\tabcolsep{0pt}%
    \put(0,0){\includegraphics[width=\unitlength,page=1]{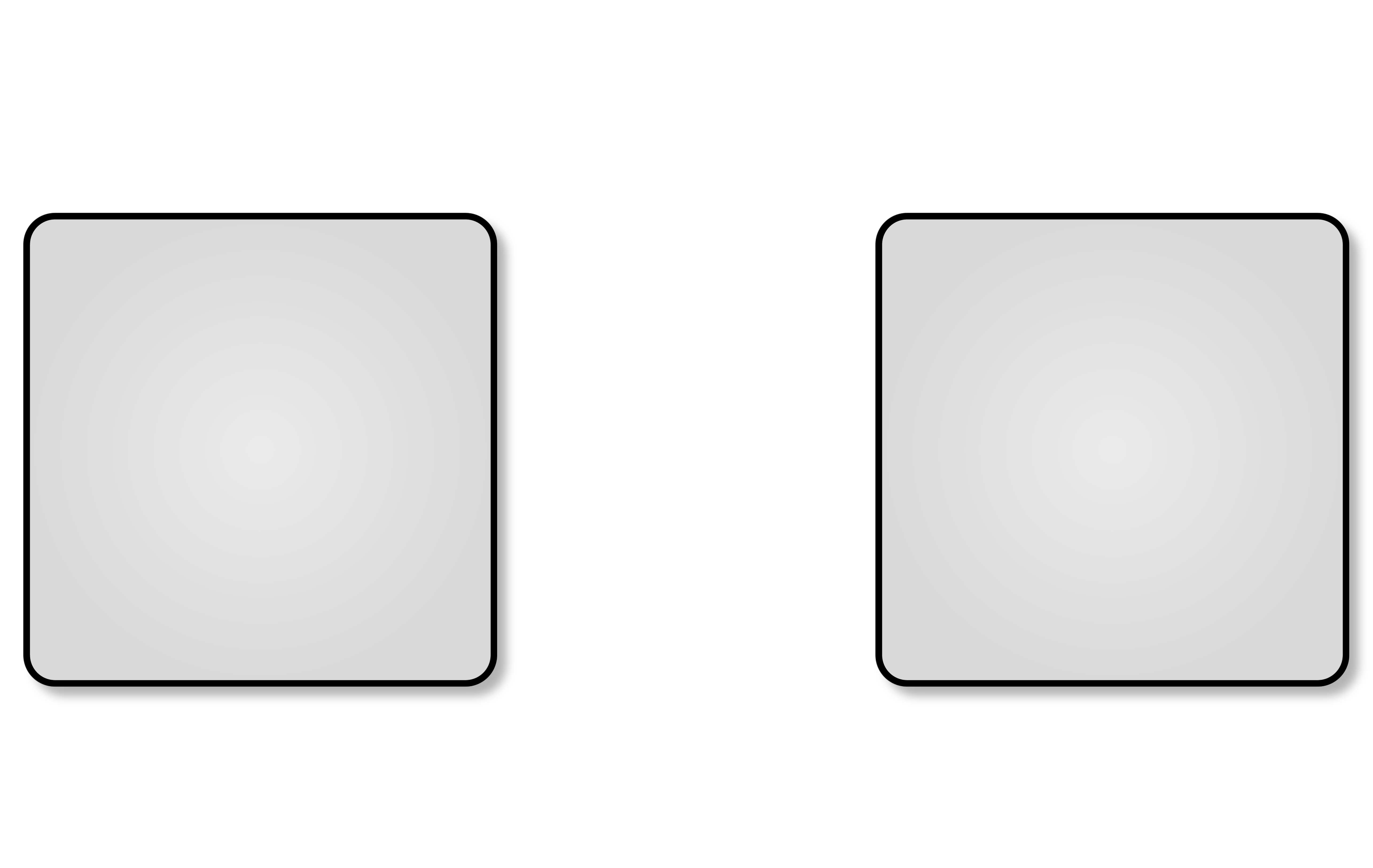}}%
    \put(0.24107442,0.06300642){\color[rgb]{0,0,0}\makebox(0,0)[lt]{\lineheight{0}\smash{\begin{tabular}[t]{l}\makebox[0pt][l]{\raisebox{-0.5\height}{$X$-BM}}\end{tabular}}}}%
    \put(0.48475258,0.06300642){\color[rgb]{0,0,0}\makebox(0,0)[lt]{\lineheight{0}\smash{\begin{tabular}[t]{l}\makebox[0pt][l]{\raisebox{-0.5\height}{$Z$-BM}}\end{tabular}}}}%
    \put(0.72958016,0.06532698){\color[rgb]{0,0,0}\makebox(0,0)[lt]{\lineheight{0}\smash{\begin{tabular}[t]{l}\makebox[0pt][l]{\raisebox{-0.5\height}{successful BM}}\end{tabular}}}}%
    \put(0,0){\includegraphics[width=\unitlength,page=2]{QPCSolutionExample.pdf}}%
    \put(0.5037349,0.33230275){\color[rgb]{0,0,0}\makebox(0,0)[t]{\lineheight{0}\smash{\begin{tabular}[t]{c}\makebox(0,0){success}\end{tabular}}}}%
    \put(0,0){\includegraphics[width=\unitlength,page=3]{QPCSolutionExample.pdf}}%
  \end{picture}%
\endgroup%

	\caption{\label{fig:qpc-scheme}Example for the measurement scheme for QPC($5,5$). In this example, the first two rows were measured without any successful BM. A successful BM occurs on the second qubit of the third row.  The scheme is then completed by performing $X$-BMs on the remaining qubits of that row and $Z$-BMs on one qubit from each of the remaining rows. In this figure, we have chosen the last qubit of each remaining row as an example. The red and blue strings represent the measured $\overline{X}$ and $\overline{Z}$, respectively.}
\end{figure}

We will now argue that each BM up to the first success has a success probability of $\mathbb{P}_B$. While the essence of this argument is captured by conditions~1 through~3 of Thm.~\ref{thm:sufficient}, the explanation here offers a more intuitive perspective. The transformation of the stabilizer generators is illustrated in Fig.~\ref{fig:qpc-stab}.

\newcommand{\boxSizeQPCstab}{0.22}

\begin{figure}
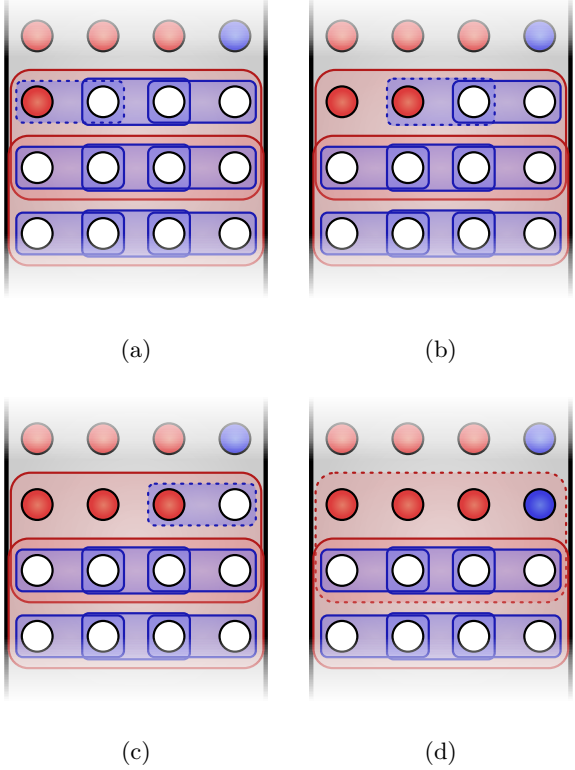

		\begin{subfigure}[c]{\boxSizeQPCstab\textwidth}
			\def\svgwidth{\textwidth}
			\input{QPCStabA.pdf_tex}
			\caption{}
		\end{subfigure}									
		\begin{subfigure}[c]{\boxSizeQPCstab\textwidth}
			\def\svgwidth{\textwidth}
			\input{QPCStabB.pdf_tex}
			\caption{}
		\end{subfigure}	
		\begin{subfigure}[c]{\boxSizeQPCstab\textwidth}
			\def\svgwidth{\textwidth}
			\input{QPCStabC.pdf_tex}
			\caption{}
		\end{subfigure}								
		\begin{subfigure}[c]{\boxSizeQPCstab\textwidth}
			\def\svgwidth{\textwidth}
			\input{QPCStabD.pdf_tex}
			\caption{}
		\end{subfigure}
	\caption{\label{fig:qpc-stab}Illustration of the transformation of stabilizer generators during the measurement of a single row for QPC. Each subfigure (a) through (d) represents a step of measuring the qubits sequentially from left to right. Red and blue qubits indicate $X$- and $Z$-BMs, respectively. Correspondingly, red and blue boxes represent $X$ and $Z$~stabilizers. In each step, the stabilizer generator that is replaced by the measurement is indicated by a dashed contour.}
\end{figure}

We begin by considering the $X$-BMs. We will construct an inductive argument to demonstrate that $X_{i,j}$ always exclusively anticommute with a single element of the current stabilizer generators, specifically $Z_{i,j} Z_{i,j+1}$. As a base case, $X_{i,1}$ anticommutes exclusively with $Z_{i,1} Z_{i,2}$, replacing it as a stabilizer generator after the measurement. For $1 < j < m$, $X_{i,j}$ anticommutes with both $Z_{i,j-1} Z_{i,j}$ and $Z_{i,j} Z_{i,j+1}$.  However, for each of these operators, the stabilizer generator $Z_{i,j-1} Z_{i,j}$ has already been replaced by the the preceding $X_{i,j-1}$ measurement, ensuring that $X_{i,j}$ exclusively anticommutes with $Z_{i,j} Z_{i,j+1}$. A similar argument can be applied to the $Z$-BMs on the last qubit of each row, $Z_{i,m}$. As a base case, $Z_{1,m}$ anticommutes exclusively with $\prod_{t=1}^{m} X_{1,t} X_{2,t}$, replacing it as a stabilizer generator after the measurement. For \text{$1<i<m$}, $Z_{i,m}$ anticommutes with both $\prod_{t=1}^{m} X_{i-1,t} X_{i,t}$ and $\prod_{t=1}^{m} X_{i,t} X_{i+1,t}$. However, in each of these cases, the stabilizer generator $\prod_{t=1}^{m} X_{i-1,t} X_{i,t}$ has already been replaced by the $Z$-BM $Z_{i-1,m}$ of the preceding row, ensuring, that $Z_{i,m}$ exclusively anticommutes with $\prod_{t=1}^{m} X_{i,t} X_{i+1,t}$. Recall from Sec.~\ref{sec:sufficient-conditions-for-an-optimal-logical-Bell-measurement-with-feedforward-based-linear-optics} that this argument does not need to apply to the very last qubit of the code.

We can now conclude the transformation of the stabilizer generators through the measurement scheme using Lem.~\ref{lem:transversal-bm}. Each $X$-BM on qubit $(i,j)$ replaces the stabilizer generator $Z_{i,j} Z_{i,j+1}$ with its corresponding measurement outcome. Similarly, each $Z$-BM on qubit $(i,m)$ replaces the stabilizer generator $\prod_{t=1}^{m} X_{i,t} X_{i+1,t}$. Applying Lem.~\ref{lem:successful-bell-measurment} to show that the success probability of each BM is $\mathbb{P}_B$ requires addressing the remaining observables.

First, we consider $Y_{i,j}$, which anticommutes with the same operators as in the previous discussion of the stabilizer transformation. For $j < m$, each $Y_{i,j}$ anticommutes with the stabilizer generator $Z_{i,j} Z_{i,j+1}$. When $j = m$ and $i \neq r$, $Y_{i,j}$ anticommutes with the stabilizer generator $\prod_{t=1}^{m} X_{i,t} X_{i+1,t}$. In the final measurement, $Y_{r,m}$ completes a logical operator, such as $\prod_{t=1}^r Z_{t,m} \prod_{s=1}^m X_{r,s} \in [\overline{Y}]$.

Next, we examine the $Z$ operators on the qubits of the $X$-BMs. Excluding the last row, for $i < r$ and $j < m$, $Z_{i,j}$ anticommutes with the stabilizer generator $\prod_{t=1}^{m} X_{i,t} X_{i+1,t}$. In the last row, each $Z_{r,j}$ completes a logical operator, for instance $Z_{r,j} \prod_{t=1}^{r-1} Z_{t,m} \in [\overline{Z}]$.

Finally, the $X$ operators on the qubits of the $Z$-BMs must be addressed. In this case, $X_{i,m}$ always completes the logical operator $\prod_{t=1}^m X_{i,t} \in \left[ \overline{X} \right]$.

Having confirmed that all single-qubit observables on the measured qubits either anticommute with an element of the current stabilizer or complete a logical operator, we can now apply Lem.~\ref{lem:successful-bell-measurment}. This leads us to conclude that the success probability for each of the $rm$ transversal BMs is $\mathbb{P}_B$ and the success probability for the logical BM scheme is $1- \left( 1 - \mathbb{P}_B \right)^{r m}$. For the standard linear-optics BM with $\mathbb{P}_B = \frac{1}{2}$, our scheme achieves a success probability of $1 - 2^{-r m}$, matching the performance of the scheme reported in Ref.~\cite{PhysRevA.100.052303}.

A detailed algebraic proof of the scheme's optimality, based on Thm.~\ref{thm:sufficient}, is provided in App.~\ref{app:proof-qpc}.

\subsection{Five-qubit code}
\label{sec:five-qubit-code}
For the five-qubit code~\cite{PhysRevLett.77.198}, the simplest approach of a static logical BM, measuring all qubit pairs with identical transversal BMs, is already optimal.

The five-qubit code is defined via the stabilizer generators:
\begin{equation}
	G_c = \{ XZZXI , IXZZX , XIXZZ , ZXIXZ \}.
\end{equation}
Our optimal scheme for the five-qubit code achieves the bound without requiring feedforward. In the following, we will demonstrate that simply performing $Y$-BMs on every qubit pair is sufficient to reach this bound. The calculation using $Z$-BMs or $X$-BMs measurements works analogously.

We begin our argument by transforming the stabilizer generators $G_c$ into $\tilde{G_c}$:
\begin{equation}
	\tilde{G_c} =  \{ XXYIY, YXXYI, IYXXY, YIYXX \},
\end{equation}
where it is straightforward to verify that this set generates the same stabilizer group:
\begin{equation}
	S_c = \langle G_c \rangle = \langle \tilde{G_c} \rangle.
\end{equation}
Following the approach outlined in Thm.~\ref{thm:sufficient}, we organize the elements of the set $\tilde{G_c}$ into a sequence:
\begin{equation}
	\begin{aligned}
	\mathbb{C} 	& = \left( c_{j} \right)_{j \in \{ 1, \dots, n-1 \} } \\
		& = \left( XXYIY, YXXYI, IYXXY, YIYXX \right).
	\end{aligned}
	\label{eq:5qubit-c}
\end{equation}
In this scheme, our goal is to measure elements of the following logical operators:
\begin{equation}
	\begin{aligned}
	\left( \overline{X}_{j} \right)_{j \in \{ 1, \dots, n \} } = ( 	& XIYYI, IXIYY, YIXIY, \\
																	& YYIXI, IYYIX ) \subset [\overline{X}],
	\end{aligned}
	\label{eq:5qubit-x-logical}
\end{equation}
\begin{equation}
	\begin{aligned}
	\left( \overline{Z}_{j} \right)_{j \in \{ 1, \dots, n \} } = ( 	& ZYIIY, YZYII, IYZYI, \\
																	& IIYZY, YIIYZ ) \subset [\overline{Z}].
	\end{aligned}	
	\label{eq:5qubit-z-logical}
\end{equation}

We will now argue that if any of the transversal BMs succeeds, the logical Bell information can be obtained with probability one. To illustrate this, let us consider an example. Suppose a successful BM occurs on the second qubit. In this case, the relevant logical operators are $\overline{X}_2 = IXIYY$ and $\overline{Z}_2 = YZYII$. For both operators, only the $Y$ information is required from all qubits except the second one. This $Y$ information is guaranteed to be obtained through the $Y$-BMs. As for the second qubit, the $X$ and $Z$ information is necessary to fully determine $\overline{X}_2$ and $\overline{Z}_2$, respectively. However, this information is acquired through the successful BM on the second qubit. This argument extends to all qubits. For every qubit $j$ the respective logical operators $\overline{X}_j$ and $\overline{Z}_j$ require only the $Y$ information on the other qubits. Note that the essence of the above explanation is captured in conditions~4 and~5 of Thm.~\ref{thm:sufficient}.

We will now argue that each BM up to the first success has a success probability of $\mathbb{P}_B$. The essence of this argument is captured by conditions~1 through~3 of Thm.~\ref{thm:sufficient}. Although the scheme is static, we can simplify the following argument by considering the measurements as being performed sequentially. This approach is equivalent to the static scheme, because the measurements commute, ensuring that the physical outcomes remain unchanged by whether the measurements are performed simultaneously or in sequence.

From Eq.~\eqref{eq:5qubit-c}, we observe that the operators $Y_j$ for $j \in \{1, \dots, n-1\}$ anticommute with their corresponding stabilizer generators $c_j$:
\begin{equation}
	\acomm{Y_j}{c_j} = 0,
\end{equation}
and they do not anticommute with any subsequent stabilizer generator $c_k$ for $k > j$:
\begin{equation}
	\comm{Y_j}{c_k} = 0.
\end{equation}
Thus, by Lem.~\ref{lem:transversal-bm}, each transversal BM on qubit $i$ replaces $c_i$ in the current stabilizer. Applying Lem.~\ref{lem:successful-bell-measurment} to show that the success probability of each BM is $\mathbb{P}_B$ requires addressing the remaining observables. The operators $Z_j$ also anticommute with the corresponding stabilizer generators $c_j$:
\begin{equation}
	\acomm{Z_j}{c_j} = 0.
\end{equation}
For the operators $X_j$, we must consider each index individually. The operator $X_1$ anticommutes with $c_2$, $X_2$ with $c_3$, and $X_3$ with $c_4$. Meanwhile, $X_4$ completes $\overline{X}_4$, and $X_5$ completes $\overline{X}_5$.

Having confirmed that all single-qubit observables on the measured qubits either anticommute with an element of the current stabilizer or complete a logical operator, we can now apply Lem.~\ref{lem:successful-bell-measurment}. This leads us to conclude that the success probability for each transversal BM is $\mathbb{P}_B$ and the success probability for the logical BM scheme is $1- \left( 1 - \mathbb{P}_B \right)^{5}$.

A detailed algebraic proof of the scheme's optimality, based on Thm.~\ref{thm:sufficient}, is provided in App.~\ref{app:proof-five-qubit}.

To the best of the authors' knowledge, no logical BM schemes for the five-qubit code has been published to date. Interestingly, since our scheme does not require feedforward this result implies that, in general, there is no tighter bound for static schemes than for feedforward-based ones.

\subsection{Standard planar surface code}
\label{sec:standard-planar-surface-code}
In this section, we introduce our measurement scheme for the standard planar surface code~\cite{KITAEV20032}. In our work, this code refers to the perhaps most studied code in topological quantum computing, and we briefly review it in the following. In this code, qubits reside on the edges of a square lattice. We denote the set of all vertices of the lattice $V$, the set of all edges of the lattice $E$, and the set of all faces of the lattice $F$. Stabilizer generators are associated with the vertices and faces of the lattice, as illustrated in Fig.~\ref{fig:planar-surface-code}. Given a vertex $v \in V$ the associated vertex operator $X_v$ is defined as:
\begin{equation}
	X_v = \prod_{e \mid v \in \partial e} X_e.
\end{equation}
Given a face $f \in F$ the associated face operator $Z_f$ is defined as:
\begin{equation}
	Z_f = \prod_{e \in \partial f} Z_e.
\end{equation}
Note that the lattice vertex or face in the subscript makes this notation always distinguishable from a Pauli operator acting on a single qubit. The stabilizer group $S_c$ of the standard planar surface code is generated by the combined set of vertex and face operators:
\begin{equation}
    G_c = \{ X_v \}_{v \in V} \cup \{ Z_f \}_{f \in F}.
\end{equation}

\begin{figure}[tb]
	\def\svgwidth{0.45\textwidth}
\begingroup%
  \makeatletter%
  \providecommand\color[2][]{%
    \errmessage{(Inkscape) Color is used for the text in Inkscape, but the package 'color.sty' is not loaded}%
    \renewcommand\color[2][]{}%
  }%
  \providecommand\transparent[1]{%
    \errmessage{(Inkscape) Transparency is used (non-zero) for the text in Inkscape, but the package 'transparent.sty' is not loaded}%
    \renewcommand\transparent[1]{}%
  }%
  \providecommand\rotatebox[2]{#2}%
  \newcommand*\fsize{\dimexpr\f@size pt\relax}%
  \newcommand*\lineheight[1]{\fontsize{\fsize}{#1\fsize}\selectfont}%
  \ifx\svgwidth\undefined%
    \setlength{\unitlength}{1955.90551181bp}%
    \ifx\svgscale\undefined%
      \relax%
    \else%
      \setlength{\unitlength}{\unitlength * \real{\svgscale}}%
    \fi%
  \else%
    \setlength{\unitlength}{\svgwidth}%
  \fi%
  \global\let\svgwidth\undefined%
  \global\let\svgscale\undefined%
  \makeatother%
  \begin{picture}(1,0.82608696)%
    \lineheight{1}%
    \setlength\tabcolsep{0pt}%
    \put(0,0){\includegraphics[width=\unitlength,page=1]{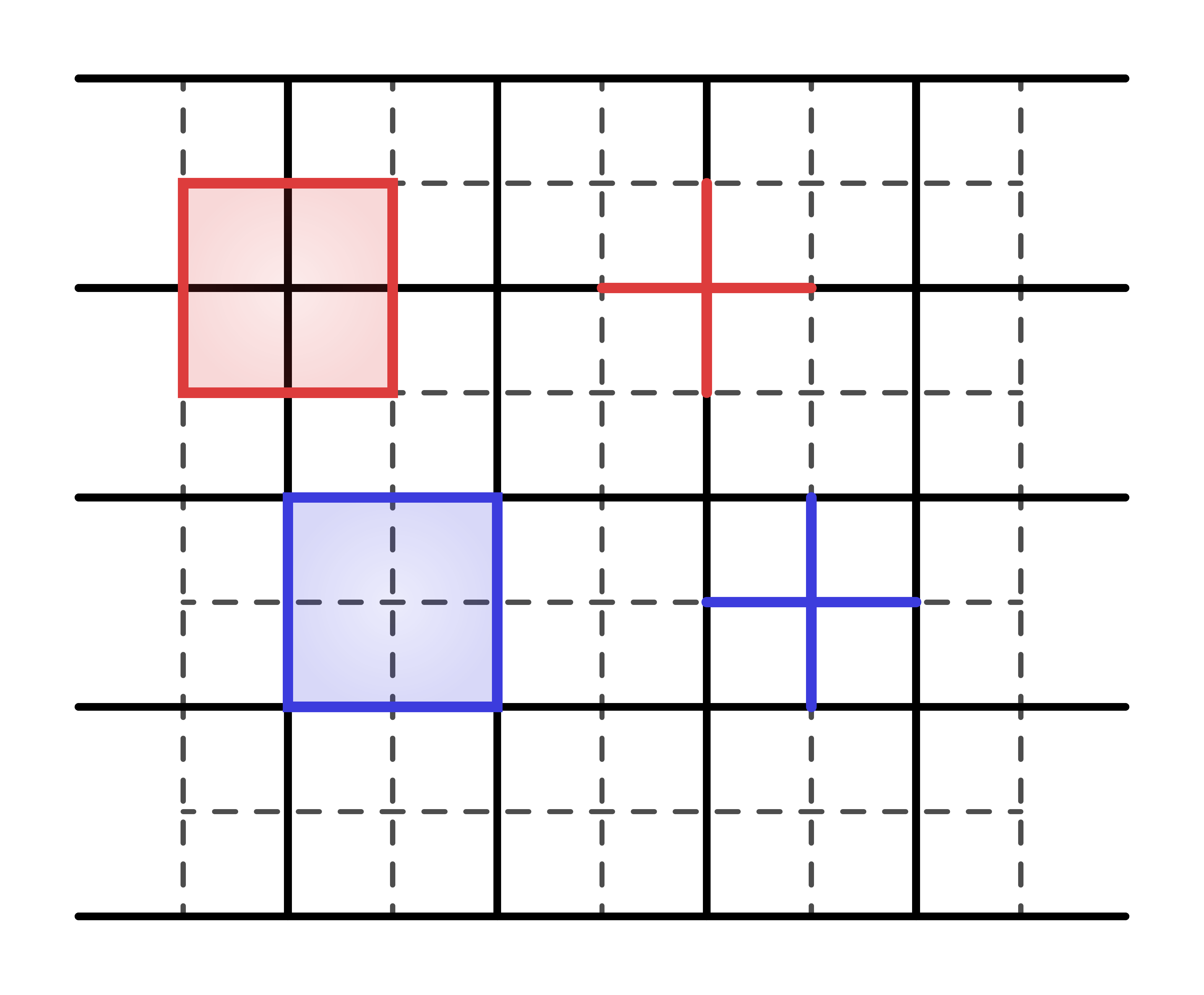}}%
    \put(0.28616074,0.61894551){\color[rgb]{0,0,0}\makebox(0,0)[t]{\lineheight{1.25}\smash{\begin{tabular}[t]{c}\makebox(0,0){$X_{v^*}$}\end{tabular}}}}%
    \put(0.62357865,0.61894551){\color[rgb]{0,0,0}\makebox(0,0)[t]{\lineheight{1.25}\smash{\begin{tabular}[t]{c}\makebox(0,0){$X_{v}$}\end{tabular}}}}%
    \put(0.361339,0.35807595){\color[rgb]{0,0,0}\makebox(0,0)[t]{\lineheight{1.25}\smash{\begin{tabular}[t]{c}\makebox(0,0){$Z_{f}$}\end{tabular}}}}%
    \put(0.71922381,0.35807595){\color[rgb]{0,0,0}\makebox(0,0)[t]{\lineheight{1.25}\smash{\begin{tabular}[t]{c}\makebox(0,0){$Z_{f^*}$}\end{tabular}}}}%
    \put(0,0){\includegraphics[width=\unitlength,page=2]{VertexAndFaceOperators.pdf}}%
  \end{picture}%
\endgroup%

	\caption{\label{fig:planar-surface-code}Examples of vertex and face operators for the standard planar surface code. For each operator, the edges belonging to the respective vertex or face, as well as their supports, are colored: blue for $Z$ operators and red for $X$ operators. The dashed lines indicate the dual lattice.}
\end{figure}

To introduce the necessary non-trivial topologies the lattice has smooth boundaries at the top and bottom and rough boundaries at the left and right. Smooth boundaries consist of three-qubit vertex operators and rough boundaries of three-qubit face operators. 

The dimensions of the lattice determine the distance of the code. We denote the dimensions of the planar surface code by $(r,m)$, where $r$ is the number of edges on the rough (left and right) boundaries and $m$ is the number of edges on the smooth (top and bottom) boundaries. Equivalently, $r$ is the code distance for $\overline{X}$ operators and $m$ is the code distance for $\overline{Z}$ operators. Alternative parametrizations are also used in the literature, e.g., in terms of the number of vertical edges per row.

In order to introduce logical operators it is instructive to introduce the dual lattice, which is displayed in Fig.~\ref{fig:planar-surface-code}. We will refer to the original lattice as the primal lattice. The dual lattice is obtained by mapping each face $f$ of the primal lattice to the corresponding vertex $f^*$ of the dual lattice, and similarly, edges $e$ to dual edges $e^*$, and vertices $v$ to dual faces $v^*$. We denote the sets of the dual vertices, dual edges, and dual faces by $F^*$, $E^*$, and $V^*$, respectively. A special property of the square lattice is that it is self-dual. Therefore, vertex operators map to dual face operators and face operators map to dual vertex operators, and the standard planar surface code can be equivalently defined on the dual lattice:
\begin{equation}
	X_v = X_{v^*} = \prod_{e^* \in \partial v^*} X_{e^*},
\end{equation}
\begin{equation}
	Z_f = Z_{f^*} = \prod_{e^* \mid f^* \in \partial e^*} Z_{e^*}.
\end{equation}

In other words, on the dual lattice the roles of $X$ and $Z$ are inverted, and the smooth and rough boundaries are interchanged. With these observations we can now describe the logical operators of the standard planar surface code. Any trail on the lattice can be naturally interpreted as a Pauli operator by assigning its support to all qubits along the trail. In topological quantum error correction, these trails are commonly called string operators, and we follow this terminology here. For $\overline{X}$ operators the string operators act with $X$ on their qubits, while for $\overline{Z}$ operators they act with $Z$. In this sense, we regard strings and logical operators as two equivalent descriptions of the same object, one topological and the other algebraic. In particular, any string operator on the primal lattice connecting the two rough boundaries (i.e., the left and right boundaries) corresponds to a $\overline{Z}$ operator. To verify this, observe that face operators ${Z_f}$ with ${f\in F}$ commute trivially with such a $\overline{Z}$ string, while vertex operators ${X_v}$ with ${v\in V}$ commute with it since each vertex is touched by the string an even number of times. Analogously, $\overline{X}$ operators can be represented on the dual lattice. Note, especially, that dual edges correspond to edges in the primal lattice, i.e., to qubits. Therefore, strings in the dual lattice can be used as a representation of operators just as in the primal lattice. Likewise, in the dual lattice picture, any string connecting the two rough boundaries (which in the dual lattice are the top and bottom boundaries) corresponds to an $\overline{X}$ operator, by the same reasoning.

We now define layers and columns as coordinate groupings in the planar surface code. A layer refers to all lattice elements, i.e., edges, vertices, or faces, that share the same vertical position in the lattice. Layers are assigned a layer index $l \in \{1, \dots, 2r-1\}$, starting from the top. Odd-numbered layers contain $m$ horizontal edges, and even-numbered layers contain $m-1$ vertical edges connecting adjacent horizontal layers. Analogously, a column refers to all lattice elements that share the same horizontal position. Columns are assigned a column index $c \in \{1, \dots, 2m-1\}$, starting from the left. Notably, in this coordinate system, qubits (edges) correspond exactly to positions $(l, c)$ for which the sum $l + c$ is even. For clarity, we will use the indices $(i,j) \in \left( \{ (1,1), \dots, (r,m) \} \right) \cap \{ (i,j) \mid i+j \text{ even} \} $, when indexing qubits.

We now describe our measurement scheme for the standard planar surface code, as illustrated in Fig.~\ref{fig:standard-planar-surface-5-scheme}. The measurement scheme iterates over the qubits sequentially, first layer by layer, proceeding from left to right within odd-numbered layers and from right to left within even-numbered layers. Odd-numbered layers are measured with $Z$-BMs, except for the last qubit in each layer, which is measured with an $X$-BM if no success has occurred on the other qubits of the layer. Although the position of this last qubit may be chosen anywhere in each row, for convenience we always place it in the last column. Even-numbered layers are measured using $Z$-BMs. For convenience, we will refer to odd- and even-numbered layers as $Z$- and $X$-layers, respectively.

\begin{figure}[tb]
	\def\svgwidth{0.45\textwidth}
	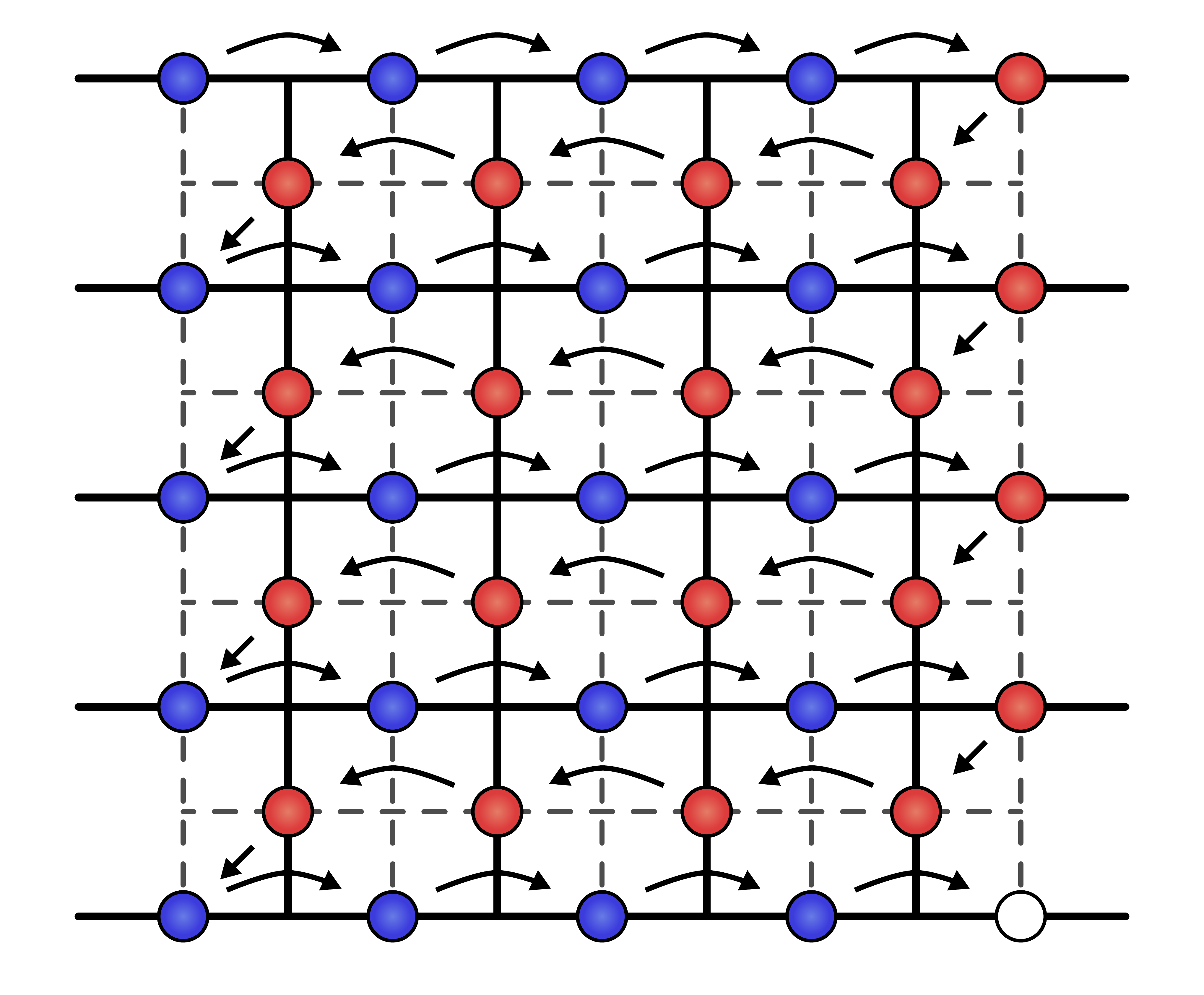
	\caption{\label{fig:standard-planar-surface-5-scheme}Measurement scheme for the standard planar surface code. The qubits at the vertices are measured following the black arrows. Red and blue vertices are measured with $X$- and $Z$-BMs, respectively. The measurement type on the last qubit is inconsequential and thus remains uncolored.}
\end{figure}

$\overline{X}$ strings connect the smooth boundaries, i.e., the top and bottom boundaries, while $\overline{Z}$ strings connect the rough boundaries, i.e., the left and right boundaries. As a direct consequence, in our scheme the two strings $\overline{X}$ and $\overline{Z}$ must intersect at the vertex where the successful BM occurs. For brevity, we will refer to this vertex as the success vertex.

We begin by considering a success vertex at an $X$-layer, as illustrated in Fig.~\ref{fig:standard-planar-surface-5-solution-x}. We denote the coordinate on which the success occurred as $(l_s,c_s)$. In this case, the following $\overline{Z}$ string in the primal lattice can be completed. The string starts on layer $l_s-1$ at the left boundary and traverses this layer to the column $c_s$ of the success vertex, then two steps downwards, passing through the success vertex and connecting to the unmeasured $Z$-layer on level $l_s+1$. From there the string connects on layer $l_s+1$ to the right boundary. The following $\overline{X}$ string can be completed in the dual lattice. Above the layer of the success vertex the string traverses the qubits which where measured with $X$-BMs on every $Z$-layer, which is the last qubit of every $Z$-layer above the success. Then the string traverses in the $X$-layer with index $l_s-1$ to the column $c_s-1$, which is left of the success vertex, and thus crosses the success vertex. From there it connects to the bottom boundary by measuring this column $c_s-1$ straight downward.

We now consider a success vertex at a $Z$-layer, as illustrated in Fig.~\ref{fig:standard-planar-surface-5-solution-z}. Again, the coordinate on which the success occurred is denoted as $(l_s,c_s)$. In this case, the $\overline{Z}$ string is completed by performing a $Z$-BM on the last qubit at $(l_s,2m-1)$ of this layer. If the success was on this last qubit this step is unnecessary and the $\overline{Z}$ string is already completed. The following $\overline{X}$ string can be completed in the dual lattice. Above the layer of the success vertex the string traverses the qubits which where measured with $X$-BMs on every $Z$-layer, which is the last qubit of every $Z$-layer above the success, down to layer $l_s-1$. Then the string traverses this $X$-layer with index $l_s-1$ to the column $c_s$ of the success vertex, then two steps downwards, passing through the success vertex and connecting to the unmeasured $Z$-layer on level $l_s+1$. From there the string can be completed by measuring the rest of the column $c_s$ straight down with $X$-BMs, connecting the string to the bottom boundary.

Importantly, for both cases we discussed, the $\overline{X}$ and $\overline{Z}$ string intersect in exactly one qubit, which is the qubit where the success occurred.

From our discussion of the completion of the logical operators, it is clear how the feedforward can be simplified. If the success vertex is on a $Z$-layer, the remainder of that layer is still measured with $Z$-BMs. Therefore, all $Z$-BMs on the $Z$-layer can be performed in a single step. If the success vertex is on an $X$-layer, no $Z$-BM is performed on this layer. Moreover, if the success vertex is on the final qubit of a $Z$-layer, the $\overline{Z}$ string is complete, and no additional $Z$ information is required. Consequently, an entire $X$-layer and the final $X$-BM of the preceding $Z$-layer can be performed within a single step.

\begin{figure*}
		\begin{subfigure}[c]{0.45\textwidth}
			\def\svgwidth{\textwidth}
			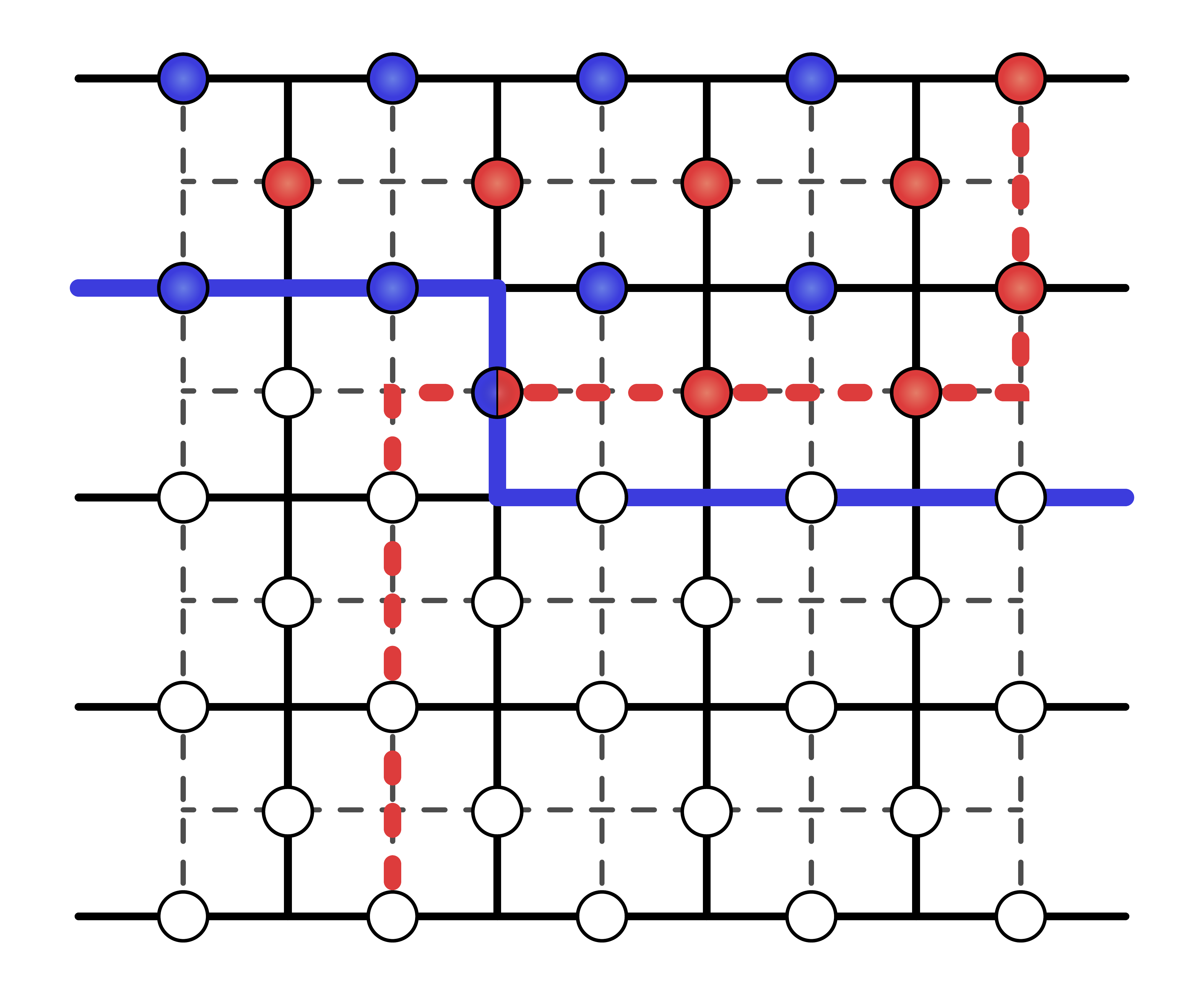
			\caption{\label{fig:standard-planar-surface-5-solution-x}}
		\end{subfigure}									
		\begin{subfigure}[c]{0.45\textwidth}
			\def\svgwidth{\textwidth}
			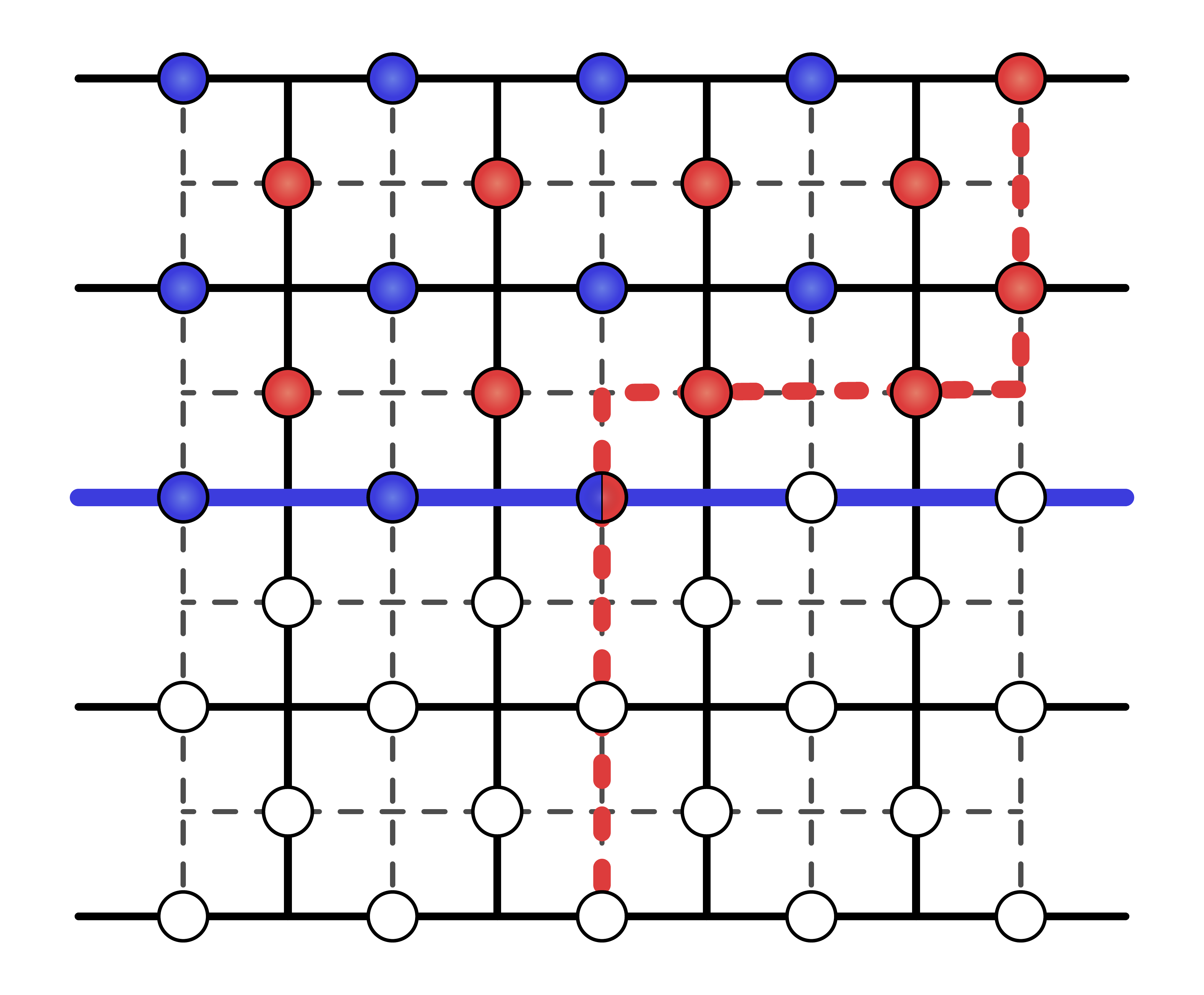
			\caption{\label{fig:standard-planar-surface-5-solution-z}}
		\end{subfigure}
		\caption{{\label{fig:standard-planar-surface-5-solution}}Examples illustrating how the scheme measures the logical operators in the event of a successful BM for the standard surface code. Red and blue qubits indicate $X$-BMs and $Z$-BMs, respectively. Qubits filled with both red and blue indicate a successful physical BM. Red and blue strings indicate the measured $\overline{X}$ and $\overline{Z}$ string, respectively. The $\overline{X}$ string is dashed, as it resides on the dual lattice. The white qubits on the red and blue strings are measured with $X$- and $Z$-BMs after the success occurred, respectively. Success vertices occurring on an (a) $X$-layer and (b) $Z$-layer are shown.}
\end{figure*}

We now turn to the transformation of stabilizer generators through the measurement scheme, as illustrated in Fig.~\ref{fig:standard-planar-surface-5-generators}. We begin by considering the $Z$-BMs on $Z$-layers. We will construct an inductive argument to demonstrate that the observables $Z_{i,j}$ exclusively anticommute with a single element of the current stabilizer generators, specifically $X_{i,j+1}$. As a base case, $Z_{i,1}$ anticommutes exclusively with $X_{i,2}$, replacing it as a stabilizer generator after the measurement. For $1 < j < 2m-1$, $Z_{i,j}$ anticommutes with both $X_{i,j-1}$ and $X_{i,j+1}$. However, for each of these operators, the stabilizer generator $X_{i,j-1}$ has already been replaced by the preceding $Z_{i,j-2}$ measurement, ensuring that $Z_{i,j}$ exclusively anticommutes with $X_{i,j+1}$. The final measurement of each $Z$-layer, $X_{i,2m-1}$, anticommutes exclusively with $Z_{i+1,2m-1}$, and thus replaces it as a stabilizer generator after the measurement. A similar argument can be applied to $X$-layers. The first measurement $X_{i,2m-2}$, which is spatially on the last qubit of this layer, anticommutes with both adjacent face operators $Z_{i,2m-3}$ and $Z_{i,2m-1}$. Again, one of these, specifically $Z_{i,2m-1}$ was replaced by the final $X_{i-1,2m-1}$ measurement of the previous layer. Finally, all remaining measurements $X_{i,j}$ anticommute with $Z_{i,j-1}$ and $Z_{i,j+1}$, where $Z_{i,j+1}$ was replaced by the preceding measurement $X_{i,j+2}$.

\begin{figure}
	\def\svgwidth{0.45\textwidth}
	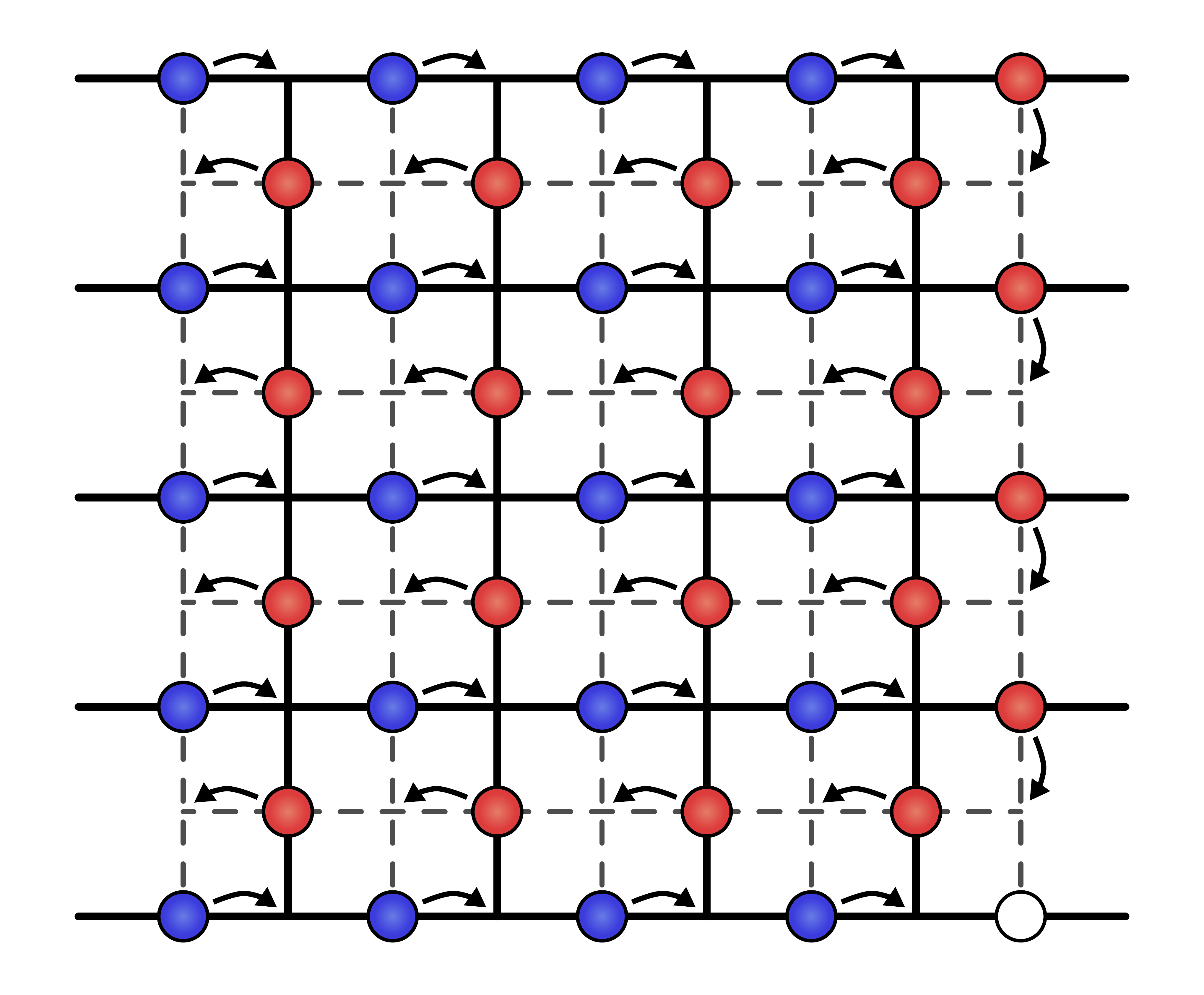
	\caption{\label{fig:standard-planar-surface-5-generators}Schematic representation of the transformation of stabilizer generators for the standard surface code. For each qubit, the black arrow points to the vertex and face (i.e., vertex of the dual lattice) whose stabilizer generator is replaced by the $Z$-BM and $X$-BM, respectively.}
\end{figure}

We now demonstrate that the success probability for every physical BM is given by $\mathbb{P}_B$. For every qubit except the last one at $(2r-1,2m-1)$, the same argument applies to $Y_{i,j}$ as in the discussion of the stabilizer generator transformations. Consequently, each $Y_{i,j}$ anticommutes with the current stabilizer.  For the final qubit, $Y_{2r-1,2m-1}$ completes the measurement of a logical $\overline{Y}$ operator, consisting of the product of the $\overline{X}$ string along the column $2m-1$ and the $\overline{Z}$ string along the layer $2r-1$. 

Next we consider the remaining observables for the $Z$-layers, for now excluding the last layer $2r-1$. Every qubit but the last in a $Z$-layer is measured using an $Z$-BM and $X_{i,j}$ always anticommutes with the face stabilizer generator $Z_{i+1,j}$ below it. On the last $Z$-layer every qubit $(2r+1,j)$ but the last, which is measured with a $Z$-BM completes an $\overline{X}$ string: Starting at the top left corner $(1,2m-1)$, traversing straight downward to the second to last layer $2r-2$, traversing through that layer leftwards to the column of the measurement $j$, from where it connects in $(2r-1,j)$ to the bottom boundary. For the last qubit in each $Z$-layer, which is measured with an $X$-BM, $Z_{i,2m-1}$ completes the $\overline{Z}$ string along the layer $i$. Finally, we consider the remaining observables for the $X$-layers. These qubits are measured using $X$-BMs, and $Z_{i,j}$ always anticommutes with the vertex stabilizer generator $X_{i+1,j}$ directly below it.

Having confirmed that all single-qubit observables on the measured qubits either anticommute with an element of the current stabilizer or complete a logical operator, we can now apply Lem.~\ref{lem:successful-bell-measurment}. This leads us to conclude that the success probability for each transversal BM is $\mathbb{P}_B$. Noting that an $(r,m)$ standard planar surface code consists of $2mr-m-r+1$ qubits, the success probability for the logical BM scheme is
\begin{equation}
	1- \left( 1 - \mathbb{P}_B \right)^{2mr-m-r+1}.
\end{equation}

A more algebraic approach to the proof of the scheme's optimality, based on Thm.~\ref{thm:sufficient}, is provided in App.~\ref{app:proof-standard-planar-surface-code}.

To our knowledge, the only existing work on logical BM schemes for the planar surface code is Ref.~\cite{PhysRevA.99.062308} where an optimized static linear-optics scheme for the standard planar surface code was presented. This static scheme, assuming $\mathbb{P}_B = \frac{1}{2}$, achieves a no-loss success probability of $1 - 2^{-2 \max (r,m) + 1 }$. Thus, we conclude that, in the absence of loss, our scheme achieves a significantly higher success probability than the scheme presented in Ref.~\cite{PhysRevA.99.062308}, at the cost of requiring feedforward.

\subsection{Rotated planar surface code}
\label{sec:rotated-planar-surface-code}
In this section, we introduce our measurement scheme for the rotated planar surface code~\cite{Horsman_2012}, which is an adaptation of the standard planar surface code requiring fewer qubits for the same code distance. In this code, qubits reside on the vertices of a lattice, and plaquettes are defined by the faces of the lattice. The code consists of dark (brown) and light (yellow) plaquettes in an alternating checkerboard pattern. Boundary plaquettes are truncated, with dark plaquettes at the top and bottom boundaries and light plaquettes at the left and right boundaries, so that each boundary plaquette contains two qubits, in contrast to interior plaquettes, which contain four. Similar to the QPC code, this code consists of $r$ rows, each containing $m$ vertices. The vertices are indexed by pairs $(i, j)$, where $i \in \{1, \dots, r\}$ denotes the row, and $j \in \{1, \dots, m\}$ denotes the column position within each row. Here we focus on the quadratic case where $r = m \eqqcolon N$, i.e., a quadratic $N \times N$ lattice. For the rectangular code with arbitrary lattice dimensions, see App.~\ref{app:rectangular-rotated-planar-surface-code}.

We denote the set of vertices in a plaquette $p$ by $\p(p)$ and the sets of dark and light plaquettes by $P_D$ and $P_L$, respectively. The stabilizer group $S_c$ of the rotated planar surface code is generated by $X$-type stabilizers associated with the dark plaquettes and $Z$-type stabilizers associated with the light plaquettes:
\begin{equation}
    G_c = \Big\{ \prod_{v \in \p(p)} X_v \mid p \in P_D \Big\} \cup \Big\{ \prod_{v \in \p(p)} Z_v \mid p \in P_L \Big\}.
\end{equation}
Consequently, we denote the top and bottom, i.e., the dark boundaries as $X$-boundaries and the left and right, i.e., the light boundaries as $Z$-boundaries. Furthermore, we will refer to dark plaquettes as $X$-plaquettes and light plaquettes as $Z$-plaquettes. In this context, the term opposite type refers to the complementary relationship between $X$ and $Z$ operators and their corresponding plaquettes.

Any path on the lattice can be naturally interpreted as a Pauli operator by assigning its support to all qubits along the path. In topological quantum error correction, these paths are commonly called string operators, and we follow this terminology here. For $\overline{X}$ operators the string acts with $X$ on its qubits, while for $\overline{Z}$ operators it acts with $Z$.  In particular, $\overline{X}$ operators correspond to strings connecting the two $X$-boundaries, and $\overline{Z}$ operators to strings connecting the two $Z$-boundaries. To commute with all stabilizer generators, a string must touch plaquettes of the opposite type an even number of times, which enforces that strings can only traverse plaquettes of opposite type diagonally.

We now consider specific types of strings that define logical operators, as illustrated in Fig.~\ref{fig:rotated-planar-surface-5-logicals}. We start by considering a $\overline{Z}$ string that starts at the top-right corner. From each vertex along the string, we move either one step to the left or diagonally to the lower-left. After exactly $N$ steps, the string reaches the left boundary. Recall that diagonal steps are only possible when crossing an $X$-plaquette. Let us examine an arbitrary step along the string. Due to the checkerboard pattern of the lattice, there is exactly one $X$-plaquette that is touched and passed by the string in each step. By the imposed constraints, this plaquette is either traversed along its edge or diagonally. In both cases, the string touches the $X$-plaquette exactly twice. Since all subsequent steps continue leftward, the plaquette will not be touched again. Thus, every $X$-plaquette is touched exactly twice along the string. Since the string consists of $Z$ operators, it trivially commutes with all $Z$-plaquettes. As we have shown, the string also commutes with all $X$-plaquettes, because each is touched twice. Moreover, this string connects the $Z$-boundaries. In conclusion, it represents a $\overline{Z}$ operator for the code.

\begin{figure}
	\def\svgwidth{0.45\textwidth}
\begingroup%
  \makeatletter%
  \providecommand\color[2][]{%
    \errmessage{(Inkscape) Color is used for the text in Inkscape, but the package 'color.sty' is not loaded}%
    \renewcommand\color[2][]{}%
  }%
  \providecommand\transparent[1]{%
    \errmessage{(Inkscape) Transparency is used (non-zero) for the text in Inkscape, but the package 'transparent.sty' is not loaded}%
    \renewcommand\transparent[1]{}%
  }%
  \providecommand\rotatebox[2]{#2}%
  \newcommand*\fsize{\dimexpr\f@size pt\relax}%
  \newcommand*\lineheight[1]{\fontsize{\fsize}{#1\fsize}\selectfont}%
  \ifx\svgwidth\undefined%
    \setlength{\unitlength}{3855.11811024bp}%
    \ifx\svgscale\undefined%
      \relax%
    \else%
      \setlength{\unitlength}{\unitlength * \real{\svgscale}}%
    \fi%
  \else%
    \setlength{\unitlength}{\svgwidth}%
  \fi%
  \global\let\svgwidth\undefined%
  \global\let\svgscale\undefined%
  \makeatother%
  \begin{picture}(1,1.17647059)%
    \lineheight{1}%
    \setlength\tabcolsep{0pt}%
    \put(0,0){\includegraphics[width=\unitlength,page=1]{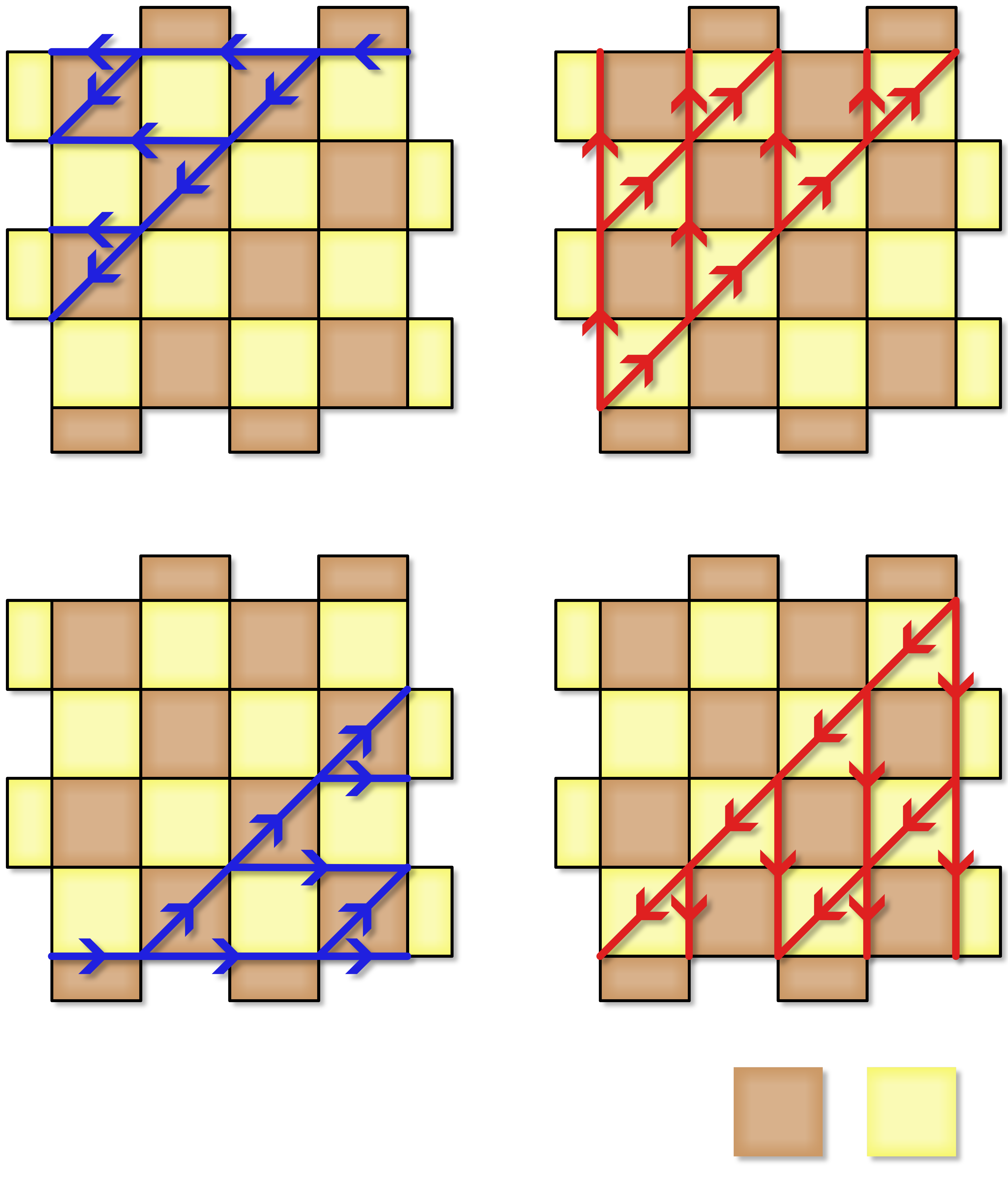}}%
    \put(0.77200183,0.07181172){\color[rgb]{0,0,0}\makebox(0,0)[t]{\lineheight{0}\smash{\begin{tabular}[t]{c}\makebox(0,0){$\contour{white}{X}$}\end{tabular}}}}%
    \put(0.90435468,0.07181172){\color[rgb]{0,0,0}\makebox(0,0)[t]{\lineheight{0}\smash{\begin{tabular}[t]{c}\makebox(0,0){$\contour{white}{Z}$}\end{tabular}}}}%
    \put(0,0){\includegraphics[width=\unitlength,page=2]{RotatedSurfaceStrings.pdf}}%
  \end{picture}%
\endgroup%

	\caption{\label{fig:rotated-planar-surface-5-logicals}Relevant logical operators for the $5 \times 5$ rotated planar surface code. Each blue string, which follows the direction of the arrows and connects the left and right boundaries, represents a $\overline{Z}$ operator. Each red string, following the arrows and connecting the top and bottom boundaries, represents an $\overline{X}$ operator. Note that the string directions are not significant for the logical operators but are used here to compactly illustrate all relevant strings.}
\end{figure}

A similar reasoning applies to $\overline{X}$ operators due to the symmetry of the code. Specifically, the code is symmetric under the simultaneous exchange of $X$- and $Z$-plaquettes and rotation by $90^{\circ}$. Similarly to the $\overline{Z}$ string which moves to the left and the lower-left, an $\overline{X}$ string move to the right and the upper-right. It starts at the bottom-left corner and traverses rightward. Note that for $\overline{X}$ strings the plaquettes which can be diagonally traversed exchange. Along this string, each $Z$-plaquette is touched exactly twice, ensuring that the string commutes with all $X$-plaquettes, thus defining a valid $\overline{X}$ operator.

Additional strings can be obtained from the code’s symmetry under a $180^{\circ}$ rotation. Specifically, we can identify $\overline{Z}$ strings that start from the bottom-left corner and traverse to the right and upper-right, as well as $\overline{X}$ strings that begin from the top-right corner and move downward and to the lower-left.

We now describe our measurement scheme for the rotated planar surface code as illustrated in Fig.~\ref{fig:rotated-planar-surface-5-scheme}. In our scheme, qubits are addressed based on the diagonals of the lattice, with each vertex assigned to a diagonal defined by the sum of its indices, $k = i + j$. For example, the first diagonal corresponds to $k = 2$ and consists of the single vertex $(1, 1)$, while the second diagonal corresponds to $k = 3$ and consists of two vertices, $(1, 2)$ and $(2, 1)$.

The scheme starts at the top-left corner at vertex $(1,1)$ and proceeds along the diagonals. For each diagonal, the vertices are addressed based on the parity of the diagonal sum. Specifically, for diagonals with even index sums, qubits are measured iteratively in a downward direction from the upper boundary using $X$-BMs. Conversely, for diagonals with odd index sums, qubits are measured iteratively in an upward direction from the left boundary using $Z$-BMs. Thus, we will refer to these diagonals as $X$- and $Z$-diagonals, respectively.

Simultaneously, a mirrored process starts at the bottom-right corner at vertex $(N, N)$, iterating over the diagonals in the opposite direction. Here, the same measurements are applied, specifically $X$-BMs for qubits with even index sums and $Z$-BMs for qubits with odd index sums. However, in this mirrored part of the scheme, the order within the diagonals is reversed. Diagonals with even index sums are traversed upwards, starting from the bottom boundary, while diagonals with odd index sums are traversed downwards, starting from the right boundary. Notably, the direction of the middle diagonal is inconsequential, as the scheme does not depend on the traversal order here.

\begin{figure}
	\def\svgwidth{0.45\textwidth}
	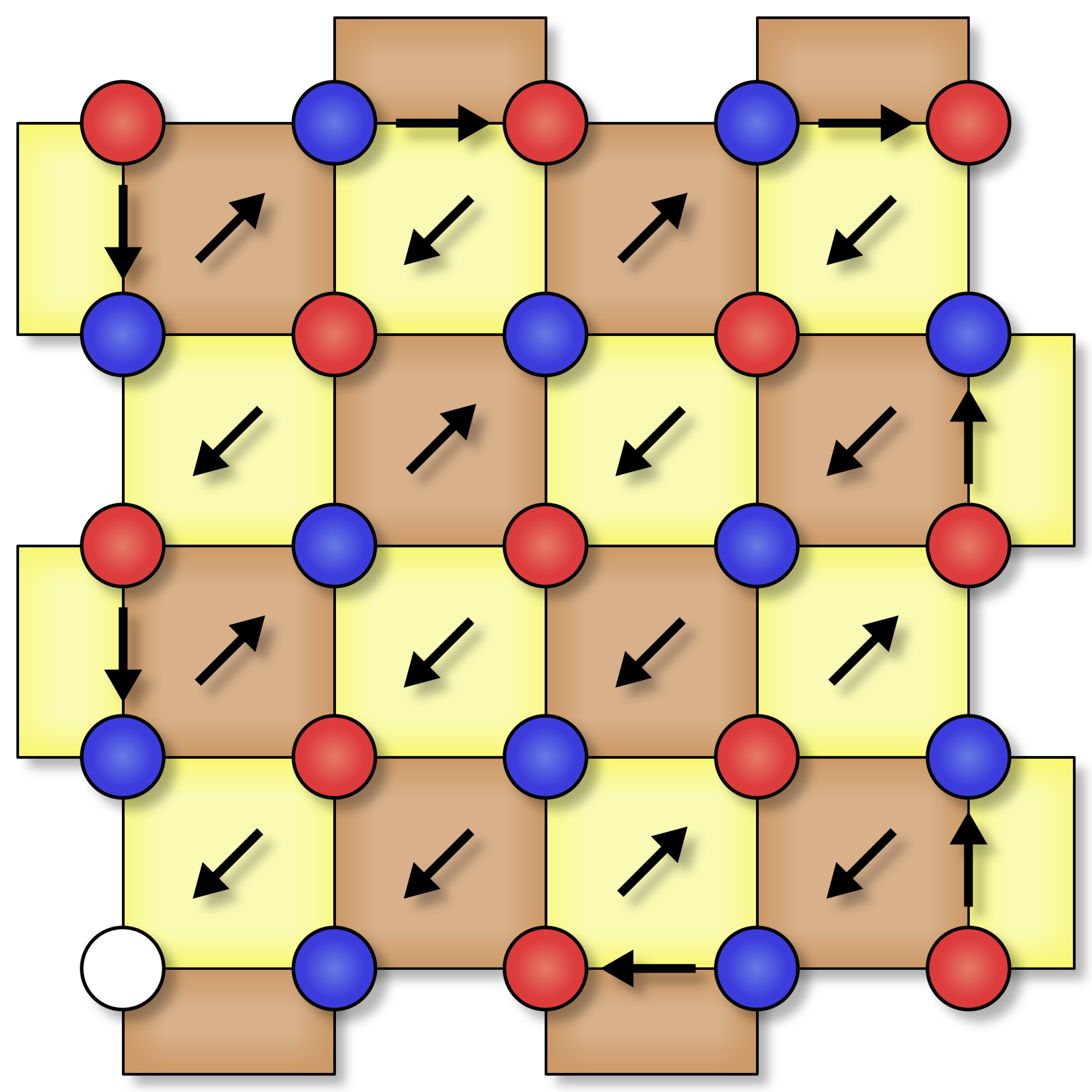
	\caption{\label{fig:rotated-planar-surface-5-scheme}Measurement scheme for the rotated planar surface code. The scheme starts simultaneously at the top-left and the bottom-right vertex. The qubits at the vertices are measured following the black arrows. Red and blue vertices are measured with $X$- and $Z$-BMs, respectively. The measurement type on the last qubit is inconsequential and thus remains uncolored.}
\end{figure}

$\overline{X}$ strings connect the $X$-boundaries, i.e., the top and bottom boundaries, while $\overline{Z}$ strings connect the $Z$-boundaries, i.e., the left and right boundaries. As a direct consequence, in our scheme,the two strings $\overline{X}$ and $\overline{Z}$ must intersect at the vertex where the successful BM occurs. For brevity, we will refer to this vertex as the success vertex. In the following explanations, we describe the strings measured to complete the logical BM. These strings are often decomposed into multiple parts. For qubits near the boundaries, some of these parts may vanish. However, these strings nevertheless extend naturally to such cases, since they remain valid even without these parts.

We begin by considering a success vertex on an $X$-diagonal in the top-left triangle, i.e., a qubit with an even index sum less than or equal to $N+1$, as illustrated in Figs.~\ref{fig:rotated-planar-surface-7-solution-x-a} and~\ref{fig:rotated-planar-surface-7-solution-middle}. All other cases follow from the symmetry of the code. Recall that this diagonal is measured using $X$-BMs, and the plaquettes along this diagonal are $Z$-plaquettes. Therefore, the $\overline{X}$ string can traverse this diagonal fully. This diagonal touches the top $X$-boundary, thus we complete the string by moving from the vertex where the diagonal touches the left boundary straight down towards the bottom-left corner of the lattice.

\begin{figure}
		\begin{subfigure}[c]{0.22\textwidth}
			\def\svgwidth{\textwidth}
			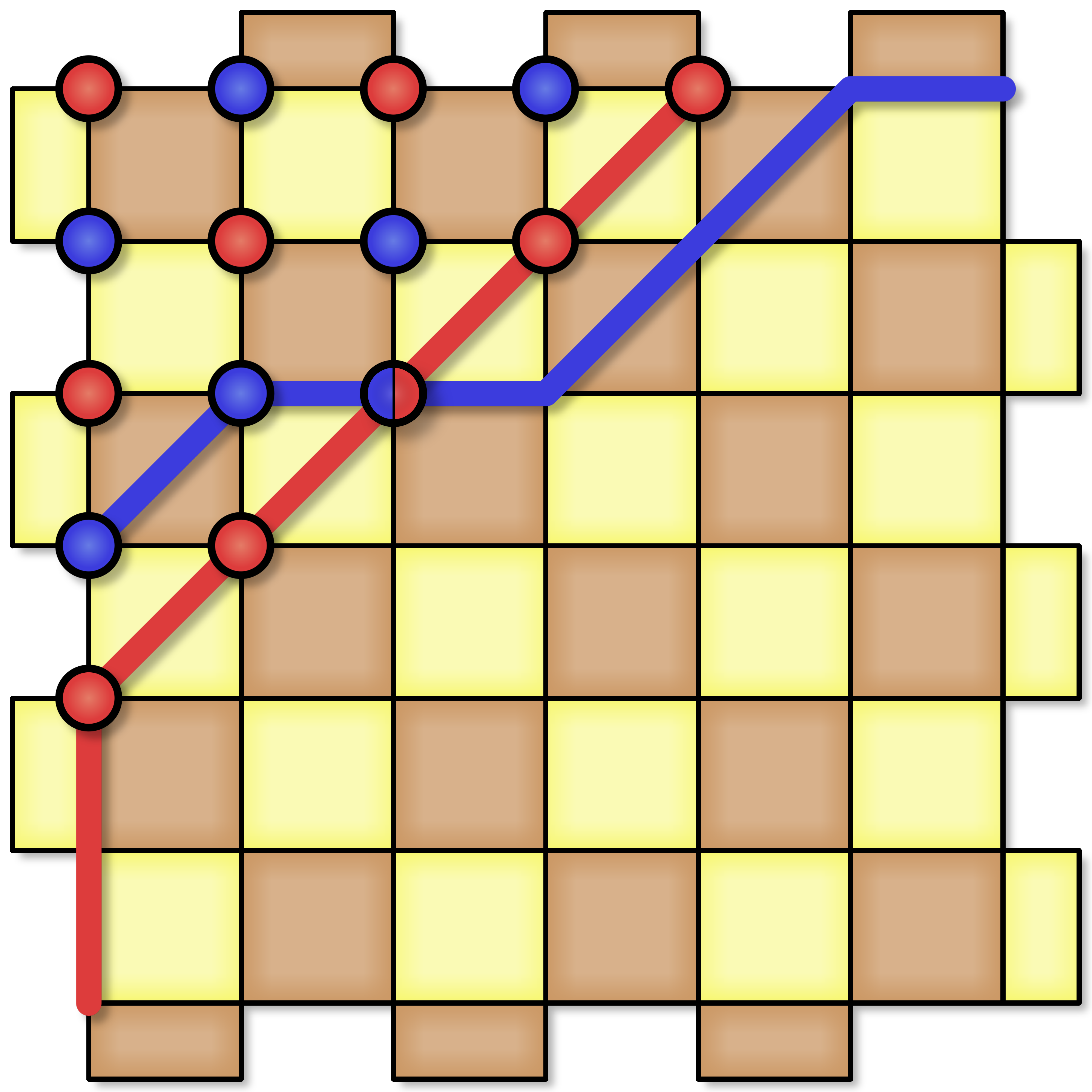
			\caption{\label{fig:rotated-planar-surface-7-solution-x-a}}
		\end{subfigure}									
		\begin{subfigure}[c]{0.22\textwidth}
			\def\svgwidth{\textwidth}
			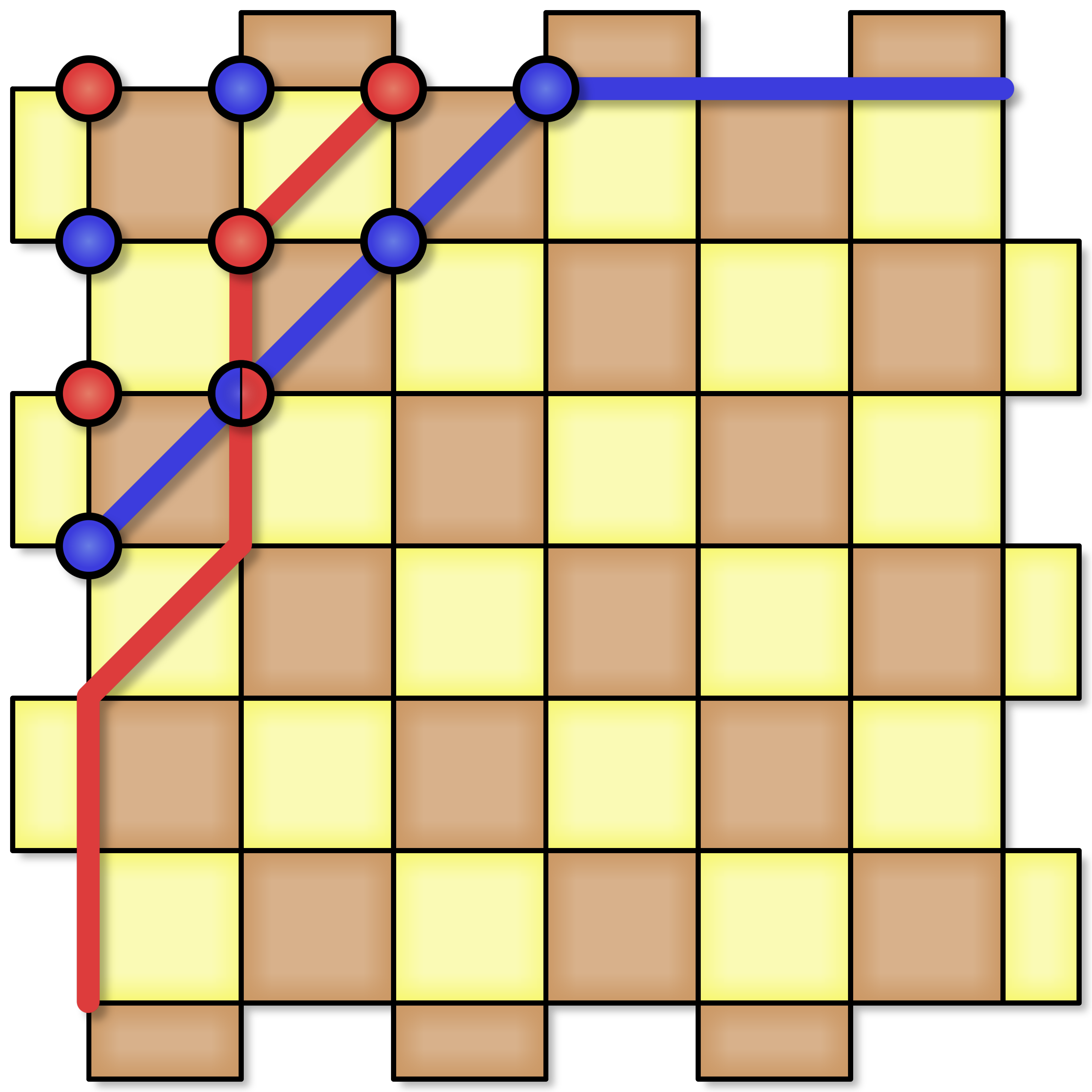
			\caption{\label{fig:rotated-planar-surface-7-solution-z-a}}
		\end{subfigure}
		\begin{subfigure}[c]{0.22\textwidth}
			\def\svgwidth{\textwidth}
			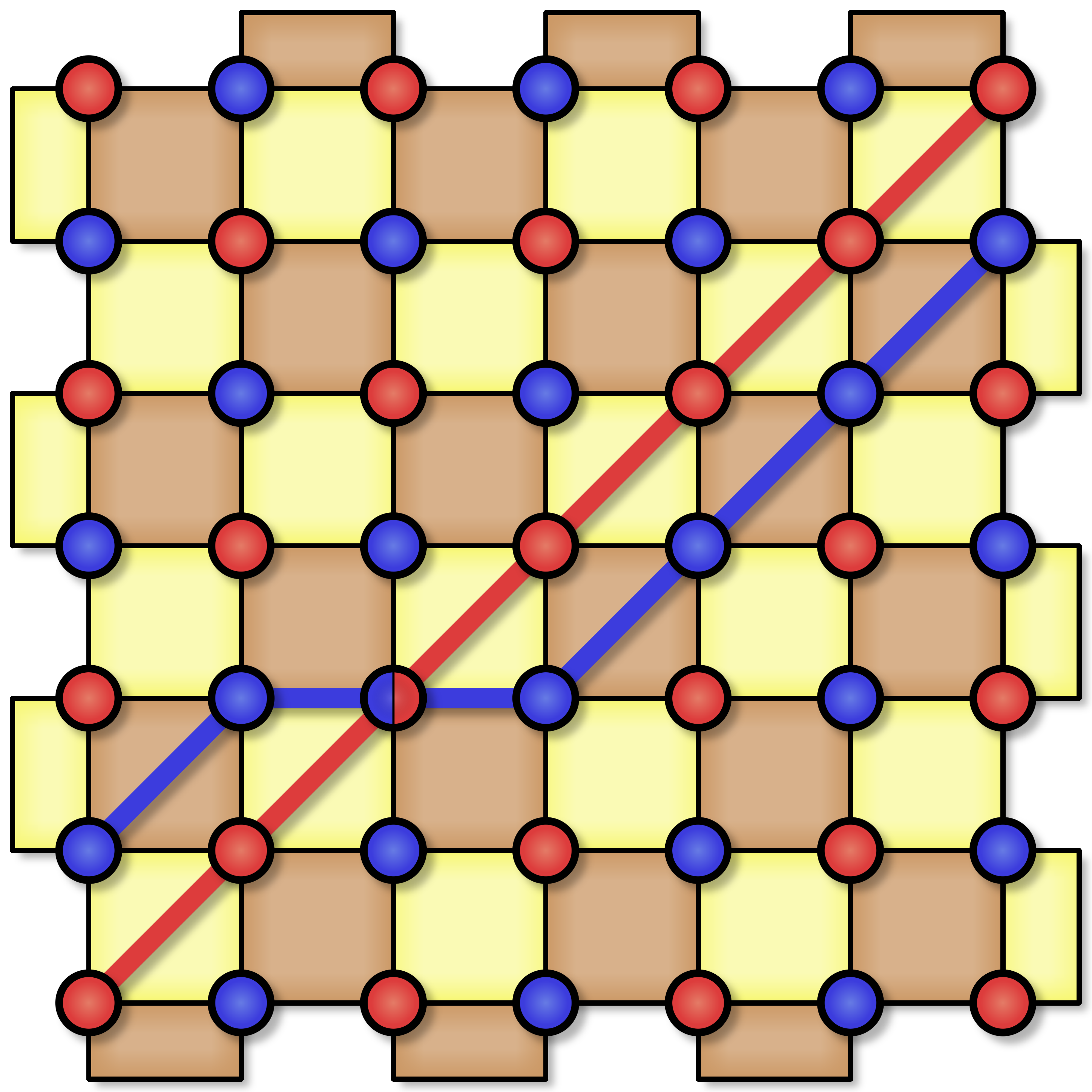
			\caption{\label{fig:rotated-planar-surface-7-solution-middle}}
		\end{subfigure}
		\caption{{\label{fig:rotated-planar-surface-7-solution}}Examples illustrating how the scheme measures the logical operators in the event of a successful BM for the rotated planar surface code. Red and blue qubits indicate $X$-BMs and $Z$-BMs, respectively. Qubits filled with both red and blue indicate a successful physical BM. Red and blue strings indicate the measured $\overline{X}$ and $\overline{Z}$ string, respectively. All possible cases where the success vertex occurred are covered: (a) upper-left triangle, $X$-diagonal; (b) upper-left triangle, $Z$-diagonal. The cases for the lower-right triangle mirror the cases in (a) and (b). The solution for success on the middle diagonal is shown in (c) and is essentially identical to the case in (a).}
\end{figure}

Now, we examine the $\overline{Z}$ string in this case. Let the index sum of the success vertex be $k$.  The $\overline{Z}$ string starts from the vertex on the previous diagonal $k-1$ that touches the left boundary. From there, it moves diagonally upwards to the row of the success vertex, then two steps horizontally to the right, crossing the $\overline{X}$ string at the success vertex. The string continues along the diagonal with index sum $k+1$ up to the top boundary and then extends rightward along the top boundary to the top-right corner.

For the case that the success vertex is on a $Z$-diagonal in the top-left triangle, i.e., a qubit with an odd index sum smaller than or equal to $N+1$, as illustrated in Fig.~\ref{fig:rotated-planar-surface-7-solution-z-a}, a similar argument applies, due to the symmetry of the code under the simultaneous exchange of $X$- and $Z$-plaquettes and a $90^{\circ}$ rotation.

Now, the $\overline{Z}$ string is measured along the diagonal of the success vertex and connected horizontally along the top boundary with the right boundary. The $\overline{X}$ string connects the diagonals on either side of the success vertex by moving two steps vertically through the success vertex and connecting the lower diagonal (i.e., with higher index sum $k+1$) along the left boundary to the bottom boundary. More precisely, assuming the index sum of the success vertex is $k$, the string starts at the vertex where the diagonal with index sum $k-1$ touches the top boundary, and then it traverses downward to the column of the success vertex. At this point, it takes two vertical steps downward, crossing the $\overline{Z}$ string at the success vertex. The remainder of the string follows the diagonal with index sum $k+1$ downwards to the left boundary and continues along the left boundary to the bottom left corner.

The solutions for the bottom-right triangle mirror those of the top-left triangle under a $180^{\circ}$ rotation, as the code is symmetric under such a transformation. The solution naturally extends to the middle diagonal with index sum $N+1$.

We now turn to the transformation of stabilizer generators through the measurement scheme, as illustrated in Fig.~\ref{fig:rotated-planar-surface-5-generators}. A single-qubit measurement operator anticommutes with a stabilizer generator associated with a plaquette if and only if it is of the opposite type and touches the plaquette. For brevity, we will simply say that the measurement anticommutes with the plaquette. For instance, the $X$-BM on the first qubit, $X_{1,1}$, anticommutes solely with the boundary plaquette $Z_{1,1} Z_{2,1}$.

\begin{figure}
	\def\svgwidth{0.45\textwidth}
	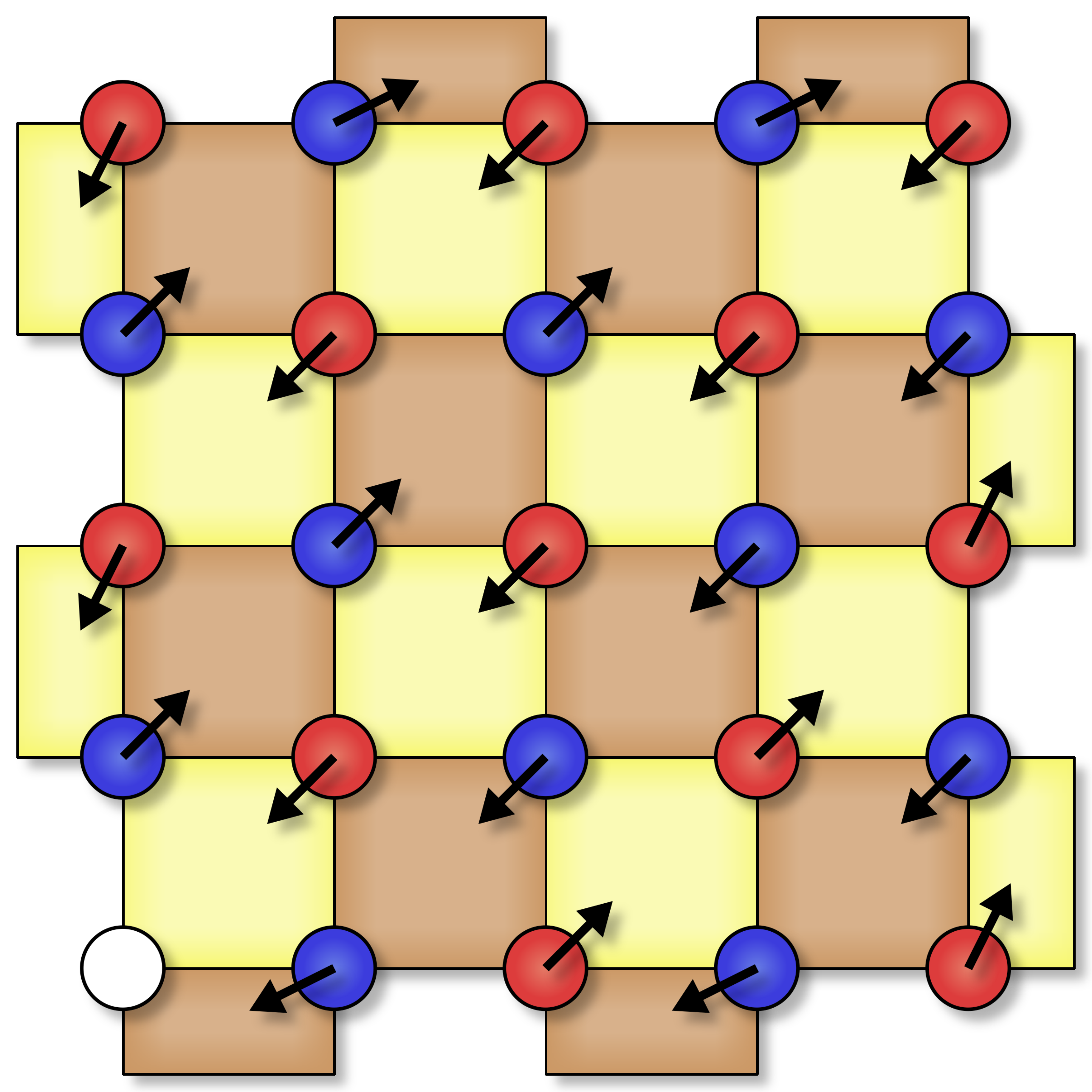
	\caption{\label{fig:rotated-planar-surface-5-generators}Schematic representation of the transformation of stabilizer generators for the rotated planar surface code. For each vertex, the black arrow points to the plaquette whose stabilizer generator is replaced by the BM on that vertex. Note that each vertex touches only one plaquette of the opposite type that has not been replaced by a previous measurement.}
\end{figure}

Consider now any diagonal. The only plaquettes that touch the vertices of the diagonal and are of the opposite type to the BMs applied along the diagonal are the plaquettes on this diagonal, including the boundary plaquette. Therefore, we can apply an inductive argument to show that each measurement anticommutes exclusively with a single element of the current stabilizer generators. The measurements begin on the side opposite to the boundary plaquette, ensuring that the first measurement touches only one plaquette it anticommutes with. This measurement replaces the corresponding stabilizer generator. Each subsequent measurement touches exactly two plaquettes of the opposite type, specifically the one replaced by the previous measurement and the next plaquette along the diagonal. Therefore, the transformation of the stabilizer generators is straightforward to track, as each measurement successively replaces the next plaquette along the measured diagonal. Recall from Sec.~\ref{sec:sufficient-conditions-for-an-optimal-logical-Bell-measurement-with-feedforward-based-linear-optics} that this argument does not need to apply to the very last qubit of the code.

It is possible to simplify the feedforward in this scheme. From our previous discussion of the logical operators, we know that in the event of a successful BM, no vertex on the current diagonal is measured with a different type BM from what it would have been measured with if no success had occurred. Thus, while the order of measurements within a diagonal is useful for understanding the transformation of the stabilizer generators, it is not strictly necessary. In fact, the entire diagonal, as well as its mirrored counterpart, can be measured simultaneously, as the measurements within a single diagonal are independent of each other.

We now demonstrate that the success probability for every physical BM is given by $\mathbb{P}_B$. First, we consider any diagonal except the middle diagonal. For each measurement, the $Y$ operator on this vertex anticommutes with the plaquette with which the type of the BM also anticommutes. The Pauli operator not measured along the diagonal (specifically, $Z$ for $X$-diagonals and $X$ for $Z$-diagonals) anticommutes with the adjacent plaquette towards the center of the lattice, which is unmeasured at this point.

Next, we consider the middle diagonal. The previous argument for the $Y$ operators extends to the middle diagonal, except for the final qubit at the bottom-left corner. The reasoning is analogous if we instead take the final qubit to be at the top-right corner. At this vertex, $Y$ operator completes a $\overline{Y}$ operator. This can be verified by decomposing $\overline{Y}$ into the product of an $\overline{X}$ and a $\overline{Z}$ operator. For example, if the middle diagonal has an even index sum (i.e., $N$ is odd), then the $X$ operators along the middle diagonal form an $\overline{X}$ operator, while the product of $Z$ operators along the adjacent diagonal $N+2$ and the bottom-left vertex $(N,1)$ constitutes a $\overline{Z}$ operator. Their product, which does not conflict with prior measurements, is a $\overline{Y}$ operator in $[\overline{Y}]$ and acts with a $Y$ operator at the bottom-left corner. If the middle diagonal instead has an odd index sum (i.e., $N$ is even), a similar argument applies, namely that the $Z$ operators along the middle diagonal constitute a $\overline{Z}$ operator, while the preceding diagonal $N$ and an $X$ operator at the bottom-left vertex $(N,1)$ form an $\overline{X}$ operator. Again, their product is a $\overline{Y}$ operator in $[\overline{Y}]$, acting with $Y$ on the final measured vertex at the bottom-left corner.

Finally, we consider the operators of the type opposite to the BMs along the middle diagonal. From our previous discussion of the logical operators, it is evident that these operators connect the two diagonals adjacent to the middle diagonal, thereby completing a logical measurement, as these two diagonals touch all boundaries. Based on these observations and Lem.~\ref{lem:successful-bell-measurment}, we conclude that the success probability for each BM is $\mathbb{P}_B$. As a result, we obtain an optimal scheme.

Due to the complexity and structure of the scheme, applying the algebraic conditions of Thm.~\ref{thm:sufficient} would lead to unwieldy and impractical expressions that are beyond the scope of a manual calculation. However, our topological treatment is equally general. Therefore, as the sole exception, we omit a detailed algebraic treatment for the rotated planar surface code.

At the end of Sec.~\ref{sec:optimization-and-comparionn-of-static-logical-bell-measurements-for-rotated-planar-surface-codes}, we compare our scheme with both our optimized static scheme and the scheme presented in Ref.~\cite{PhysRevA.99.062308}.

\subsection{Optimization and comparison of static logical Bell measurements for the rotated planar surface code}
\label{sec:optimization-and-comparionn-of-static-logical-bell-measurements-for-rotated-planar-surface-codes}
In this section, we devise an optimized static scheme for the rotated planar surface code and conclude by comparing schemes for planar surface codes. We begin by briefly outlining the improved static scheme for the standard planar surface code presented in Ref.~\cite{PhysRevA.99.062308}. Similar to our feedforward-based schemes, this approach relies exclusively on transversal guaranteed partial information BMs. Consequently, the single-code reduction can be applied for this discussion. For simplicity, we assume $r \geq m$, noting that the discussion for $r < m$ proceeds analogously by exchanging the roles of the $X$ and $Z$ operators. In this scheme, all qubits along a $\overline{Z}$ string are measured using $Z$-BMs, while all qubits outside the string are measured with $X$-BMs, ensuring that the logical $\overline{Z}$ information is always obtained. Under the assumption that standard linear-optics BMs with a success probability of $\mathbb{P}_B = \frac{1}{2}$ are used, it is shown that the scheme succeeds if and only if at least one of the $Z$-BMs along the $\overline{Z}$ string succeeds, with each measurement having an independent success probability of $\mathbb{P}_B = \frac{1}{2}$.

Given these properties, that a single successful BM along the $\overline{Z}$ string is both necessary and sufficient, and that the success probability for each BM up to the first success is $\mathbb{P}_B$, the success probability of the logical scheme increases with the weight of the $\overline{Z}$ string. As is standard, we define the weight of an operator to be the number of qubits on which it acts nontrivially. We use this approach to devise our scheme by finding a $\overline{Z}$ string which is as long as possible while still allowing the completion of an $\overline{X}$ string after a single successful BM along the $\overline{Z}$ string. Our solution is depicted in Figs.~\ref{fig:static-optimized-string} and~\ref{fig:static-optimized-string-trunc}.

Essentially, the string follows a wave-like pattern along the longer side of the surface. The pattern follows each second vertical edge, leaving the intermediate vertical edge available for the $\overline{X}$ string to connect through the wave-like structure, which will be discussed in greater detail later in this section. At the turning points along the boundaries, it passes diagonally or horizontally through the plaquettes, depending on the type of plaquette to ensure it commutes with all $X$-plaquettes. The period of this wave-like pattern spans four columns.

The validity of this solution is not restricted to cases where the length of the lattice is an exact multiple of the wave’s period (plus one for the starting point). Even in scenarios where the lattice length does not accommodate the last period completely, the solution remains valid. In such cases, the wave-like pattern can simply terminate at the lattice boundary by including only as much of the final period as fits within the given length. This is illustrated in Fig.~\ref{fig:static-optimized-string-trunc}.

\begin{figure*}
	\begin{subfigure}[c]{0.45\textwidth}
		\def\svgwidth{\textwidth}
		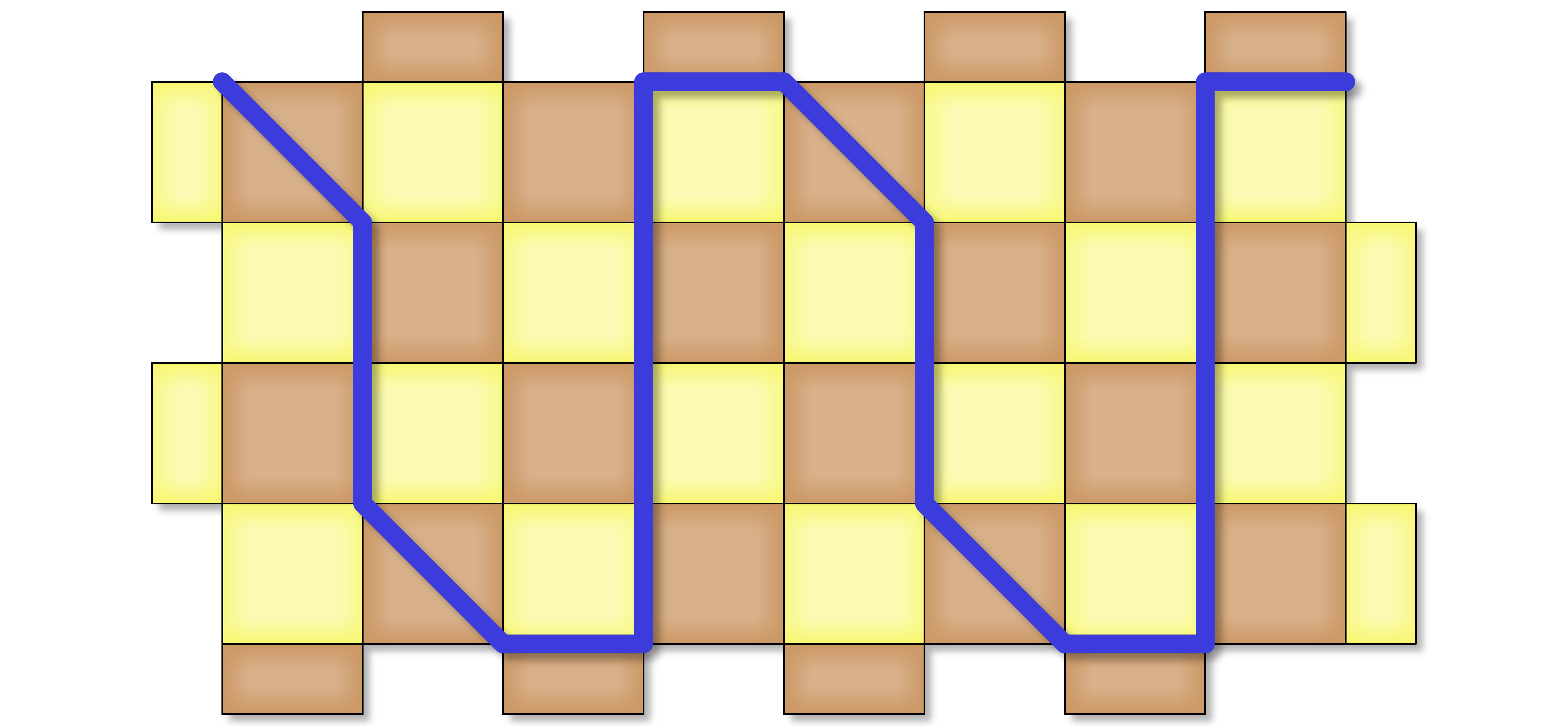
		\caption{}
	\end{subfigure}
	\begin{subfigure}[c]{0.45\textwidth}
		\def\svgwidth{\textwidth}
		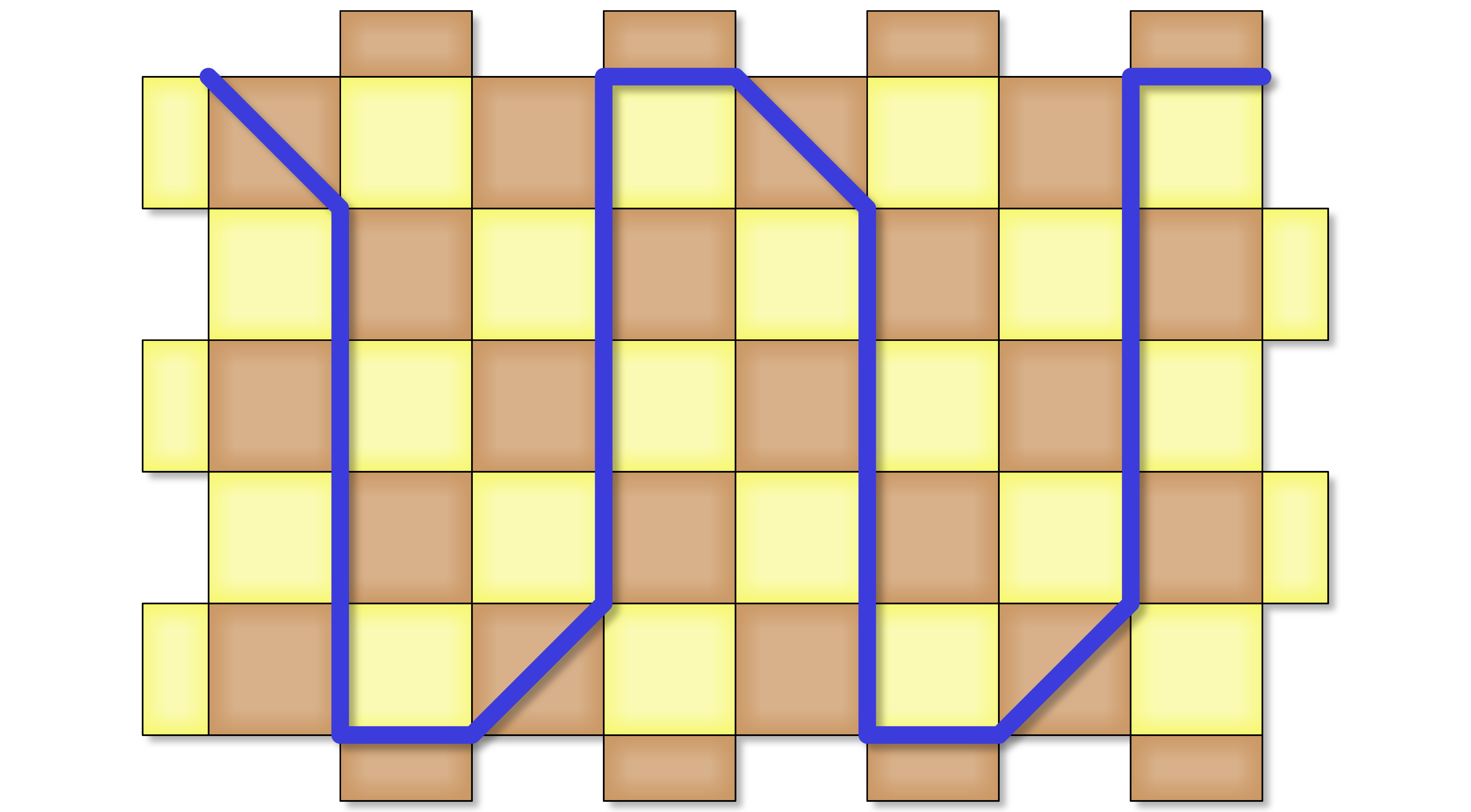
		\caption{}
	\end{subfigure}
	\caption{\label{fig:static-optimized-string}$\overline{Z}$ strings of the optimized static scheme for the rotated planar surface code with parameters $(5,9)$ in (a) and $(6,9)$ in (b). Note that the $X$-plaquettes (dark plaquettes) at the peaks of the wave must be crossed diagonally to ensure that the string commutes with all plaquettes.}
\end{figure*}

\begin{figure*}
	\def\svgwidth{\textwidth}
	\input{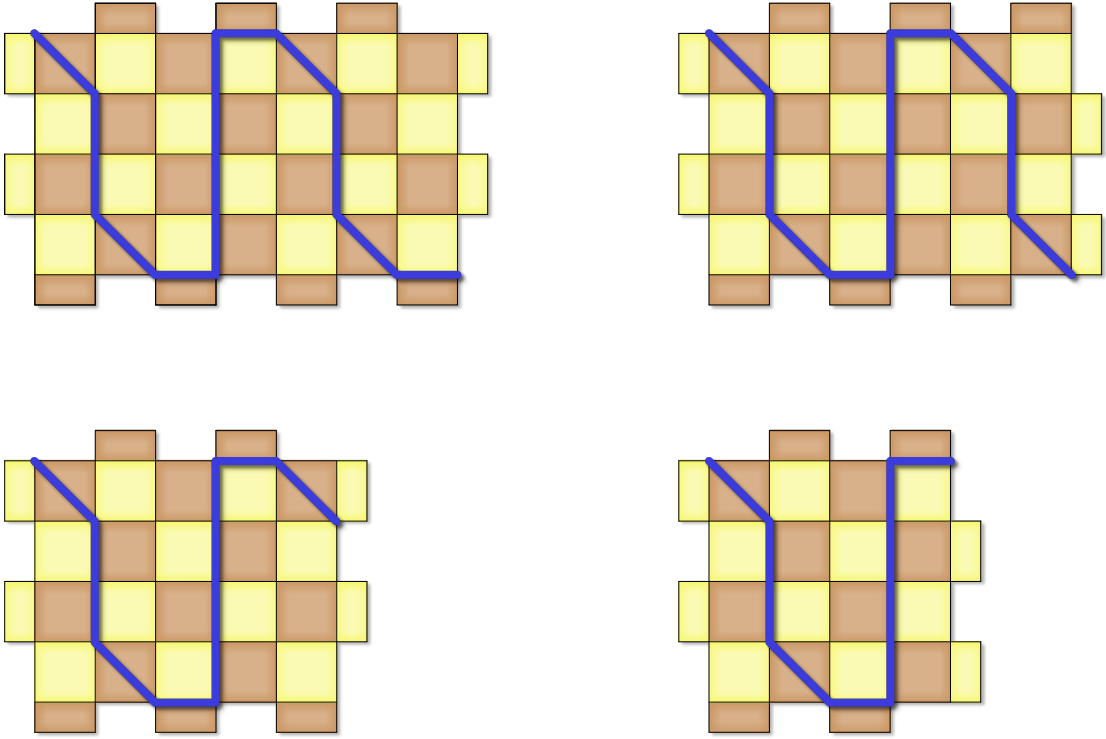}
	\caption{\label{fig:static-optimized-string-trunc}$\overline{Z}$ strings of the optimized static scheme for the rotated planar surface code where the lattice width $m$ does not fully accommodate the final wave period. The wave-like pattern terminates at the lattice boundary by including only as much of the final period as fits within the given length. Note that the string remains a valid $\overline{Z}$ operator.}
\end{figure*}

To argue that the success probability for each BM along the string, up to the first success, is $\mathbb{P}_B$, we observe the following. A static scheme is a special case of a feedforward-based scheme. Since all measurements commute, a static scheme is equivalent to any sequential scheme that applies fixed measurements to each qubit, independent of the results on other qubits. The defining characteristic of a static scheme is that the measurements are predetermined and do not depend on prior outcomes.

For the following argument, we assume that the BMs are performed sequentially along the string before all qubits outside the string are measured. In Fig.~\ref{fig:static-optimized-generators}, we illustrate the transformation of the stabilizer generators. Similar to the discussion of the feedforward-based scheme in the previous section, each $Z$-BM along the string anticommutes with two plaquettes, which are the one replaced by the previous measurement and one additional plaquette. Thus, each measurement replaces the next $X$-plaquette along the string. Since every vertex also touches at least one $Z$-plaquette, we can conclude that for each qubit along the string, every single-qubit Pauli operator anticommutes with at least one plaquette. The only exception is the very last qubit of the string, where the $Z$ operator completes the logical $\overline{Z}$ operator, which has equal probabilities for both of its possible outcomes. Therefore, using Lem.~\ref{lem:successful-bell-measurment}, we conclude that the success probability of each BM along the string up to the first success is $\mathbb{P}_B$, independent of the outcomes of other measurements along the string.

\begin{figure*}
	\begin{subfigure}[c]{0.45\textwidth}
		\def\svgwidth{\textwidth}
		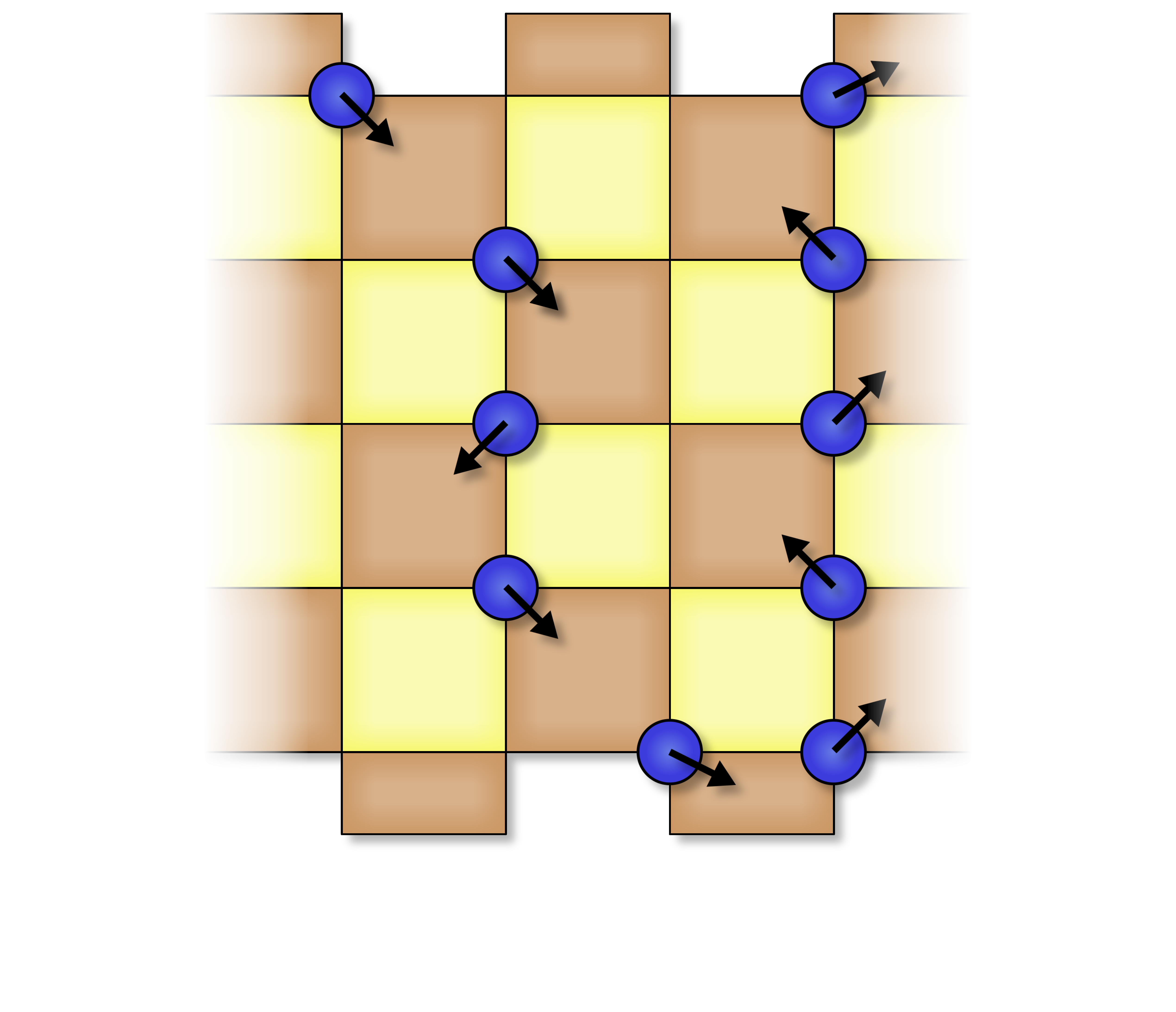
		\caption{}
	\end{subfigure}
	\begin{subfigure}[c]{0.45\textwidth}
		\def\svgwidth{\textwidth}
		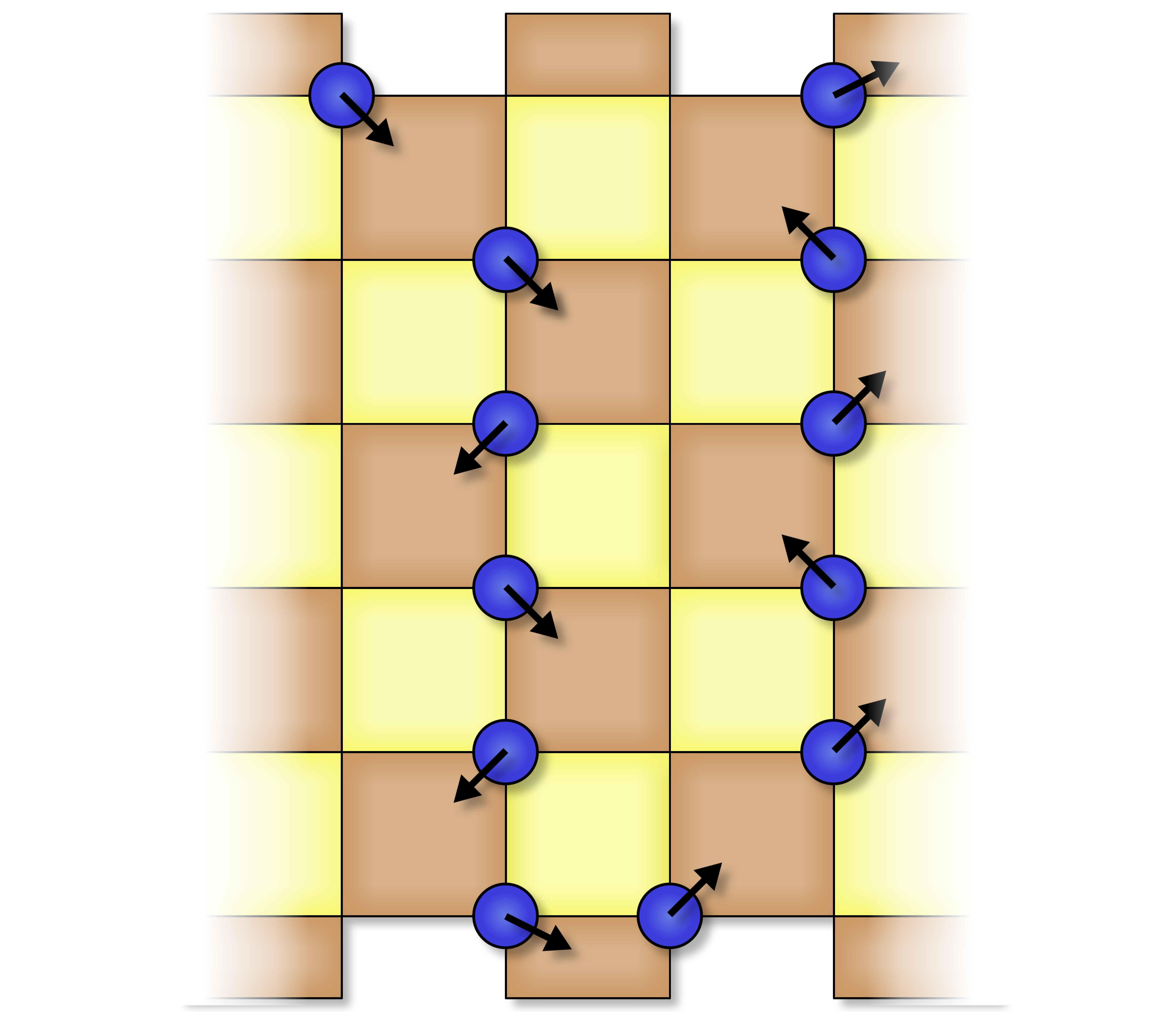
		\caption{}
	\end{subfigure}
	\caption{\label{fig:static-optimized-generators}Schematic representation of the transformation of stabilizer generators in the static scheme for the rotated planar surface code. For each vertex, the black arrow points to the plaquette whose stabilizer generator is replaced by the BM on that vertex. Note that each vertex touches only one plaquette of the opposite type that has not been replaced by a previous measurement. The surface fades toward the left and right to indicate that only a single period of the wave pattern is shown, while the code may extend arbitrarily far in both directions.}
\end{figure*}

If a successful BM occurs on any qubit along the $\overline{Z}$ string, the $\overline{X}$ measurement can be completed. This is illustrated in Fig.~\ref{fig:static-optimized-sol}, which shows, for each qubit along one period of the $\overline{Z}$ string, an $\overline{X}$ string that does not conflict with the $\overline{Z}$ string outside that qubit. With a slight modification for vertices on the right boundary, these $\overline{X}$ strings remain valid even when the $\overline{Z}$ string period is truncated at the right boundary, as we will show in the following argument. Note that for no vertex the $\overline{X}$ string extends more than one qubit-column to the right, thus, for all vertices not on the right boundary, the solutions remain valid. Note that, for any truncation of the $\overline{Z}$ string, it touches the right boundary at only one vertex. In this case, if a success occurs on the qubit of the truncated $\overline{Z}$ string at the right boundary, the $\overline{X}$ string can be measured by performing $X$-BMs on the entire right boundary, forming a valid $\overline{X}$ measurement.

Furthermore, since any $\overline{X}$ string intersects any $\overline{Z}$ string at least once, the logical BM will always fail if no transversal BM along the $\overline{Z}$ string succeeds.

\begin{figure*}
		\begin{subfigure}[c]{0.3\textwidth}
			\def\svgwidth{\textwidth}
			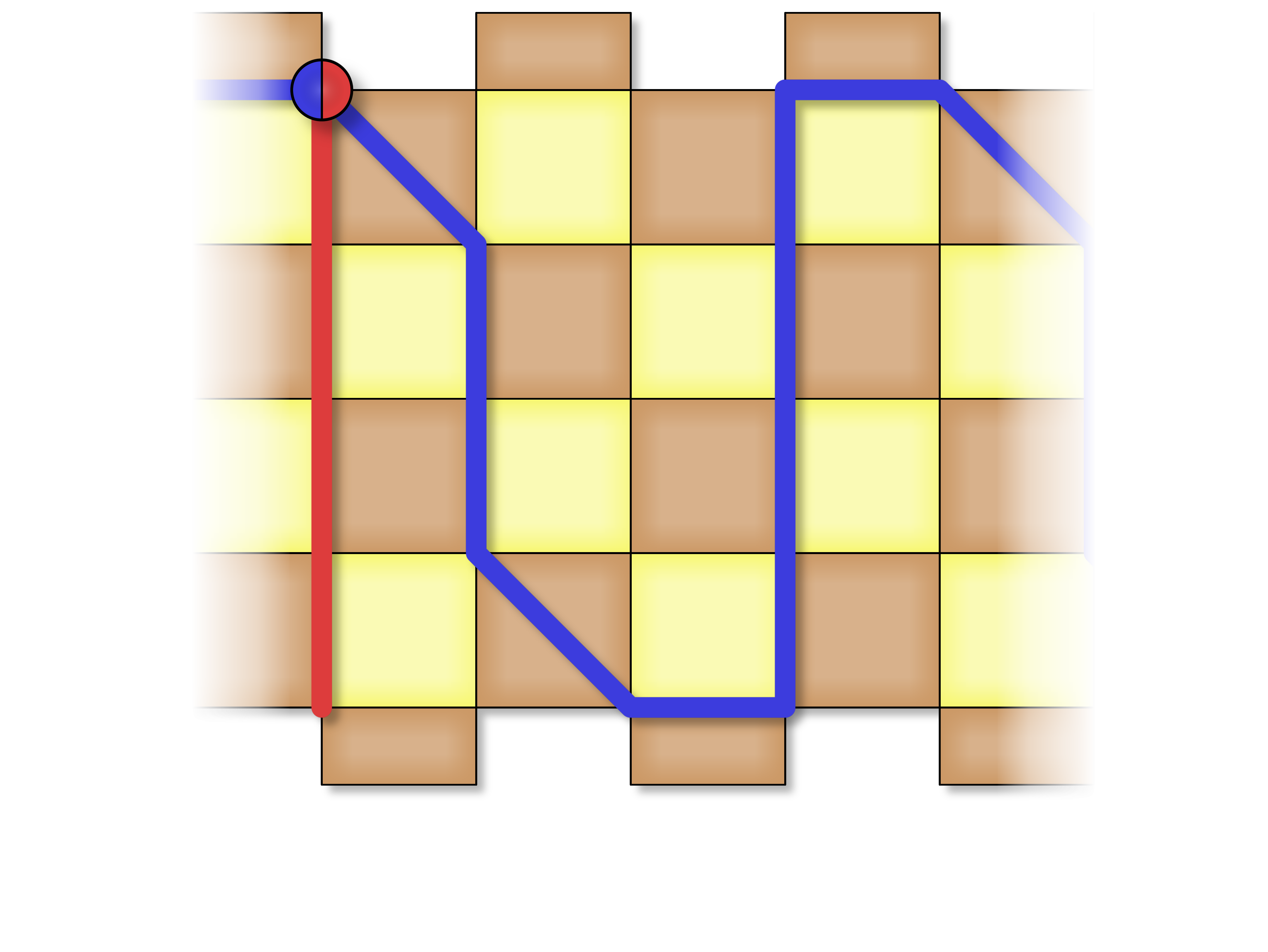
		\end{subfigure}
		\begin{subfigure}[c]{0.3\textwidth}
			\def\svgwidth{\textwidth}
			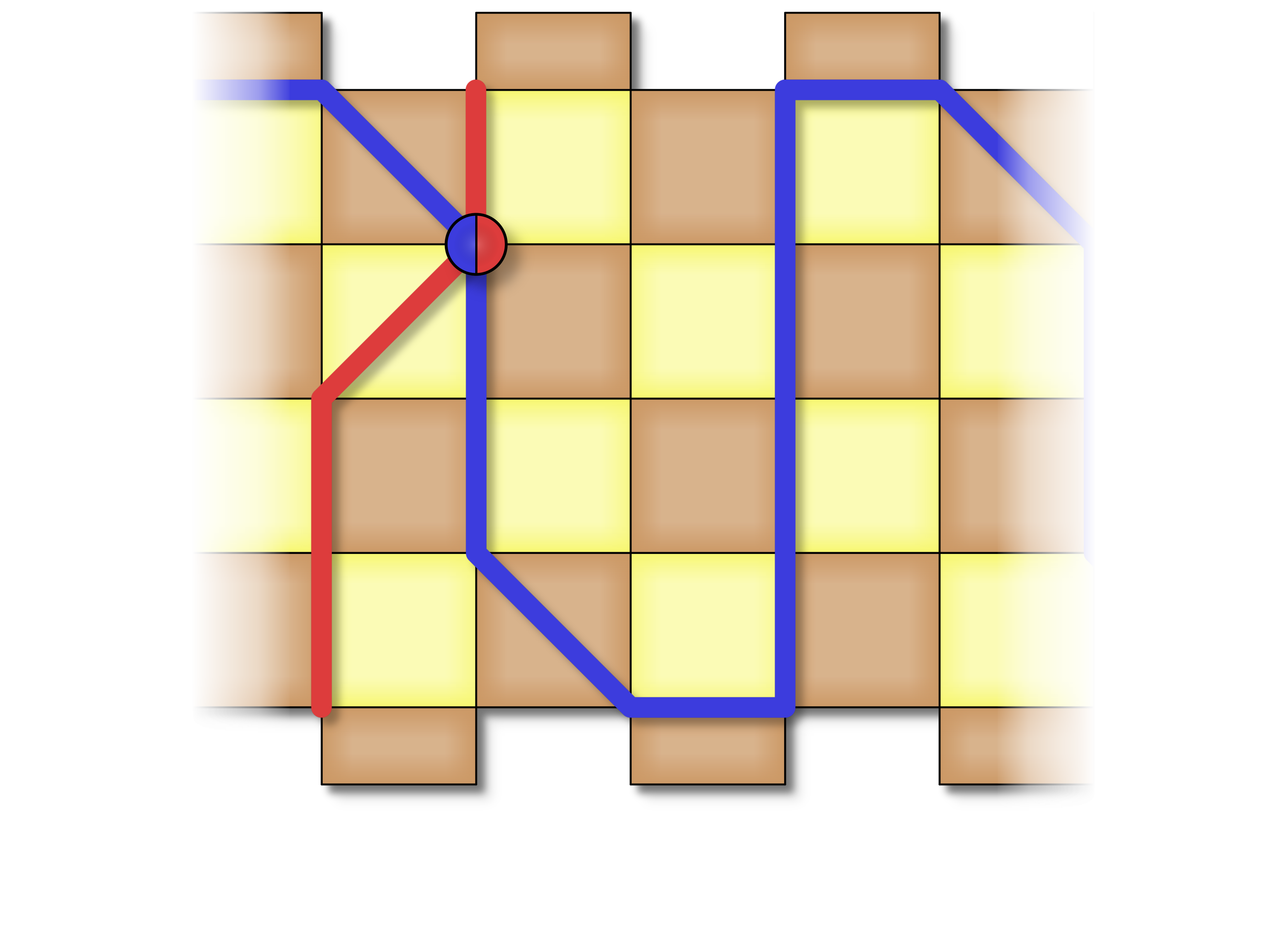
		\end{subfigure}
		\begin{subfigure}[c]{0.3\textwidth}
			\def\svgwidth{\textwidth}
			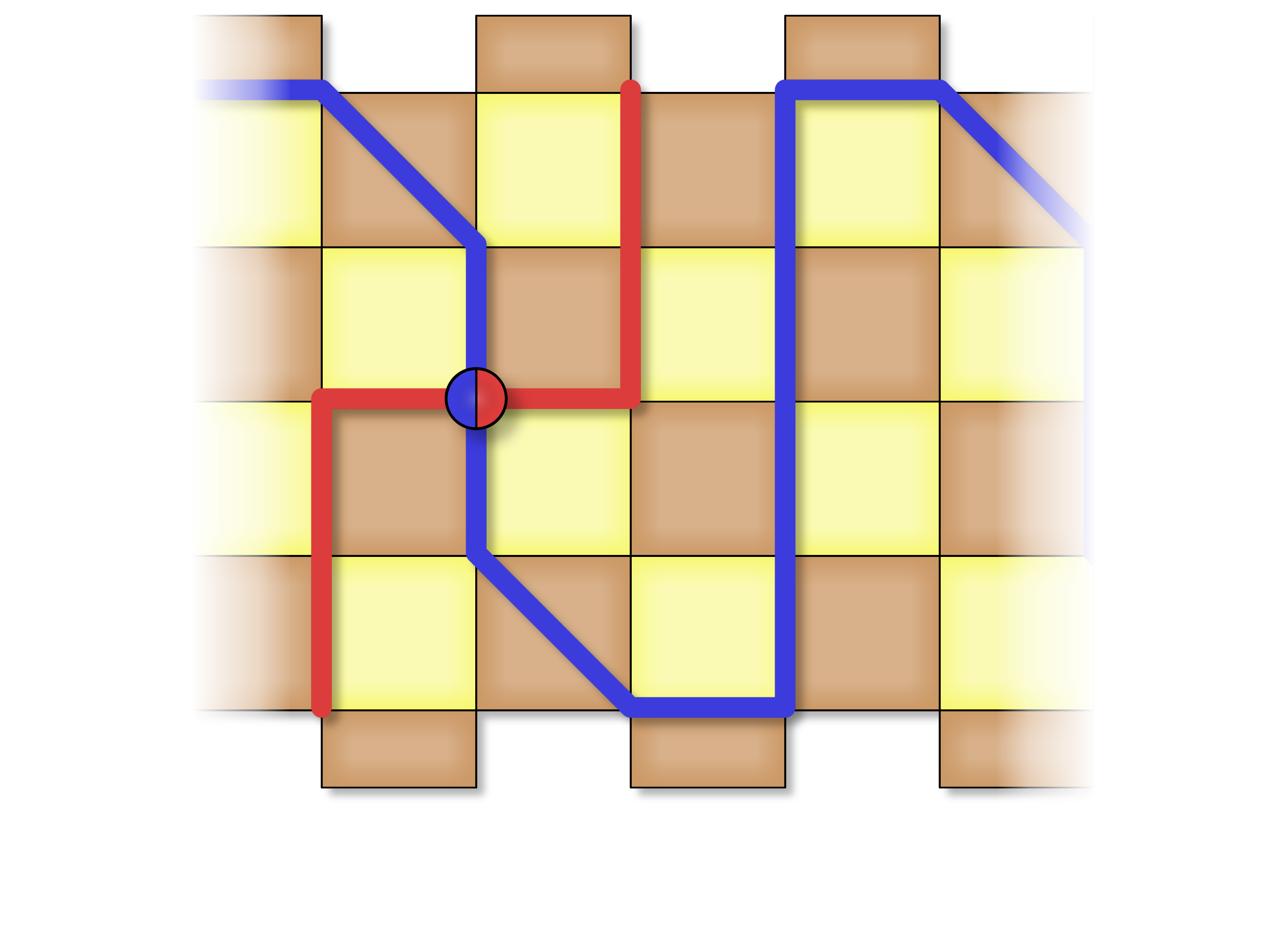
		\end{subfigure}
		\begin{subfigure}[c]{0.3\textwidth}
			\def\svgwidth{\textwidth}
			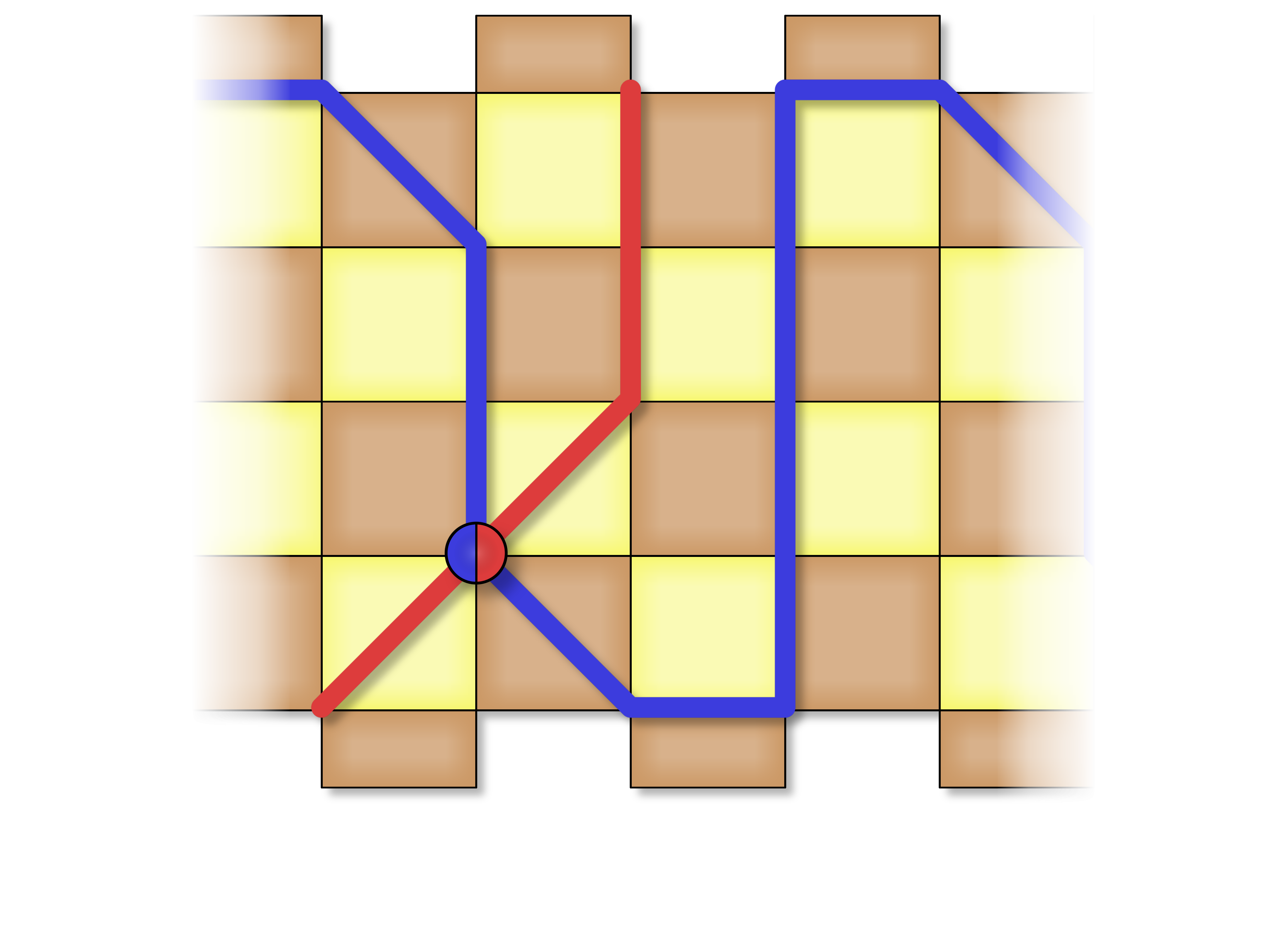
		\end{subfigure}
		\begin{subfigure}[c]{0.3\textwidth}
			\def\svgwidth{\textwidth}
			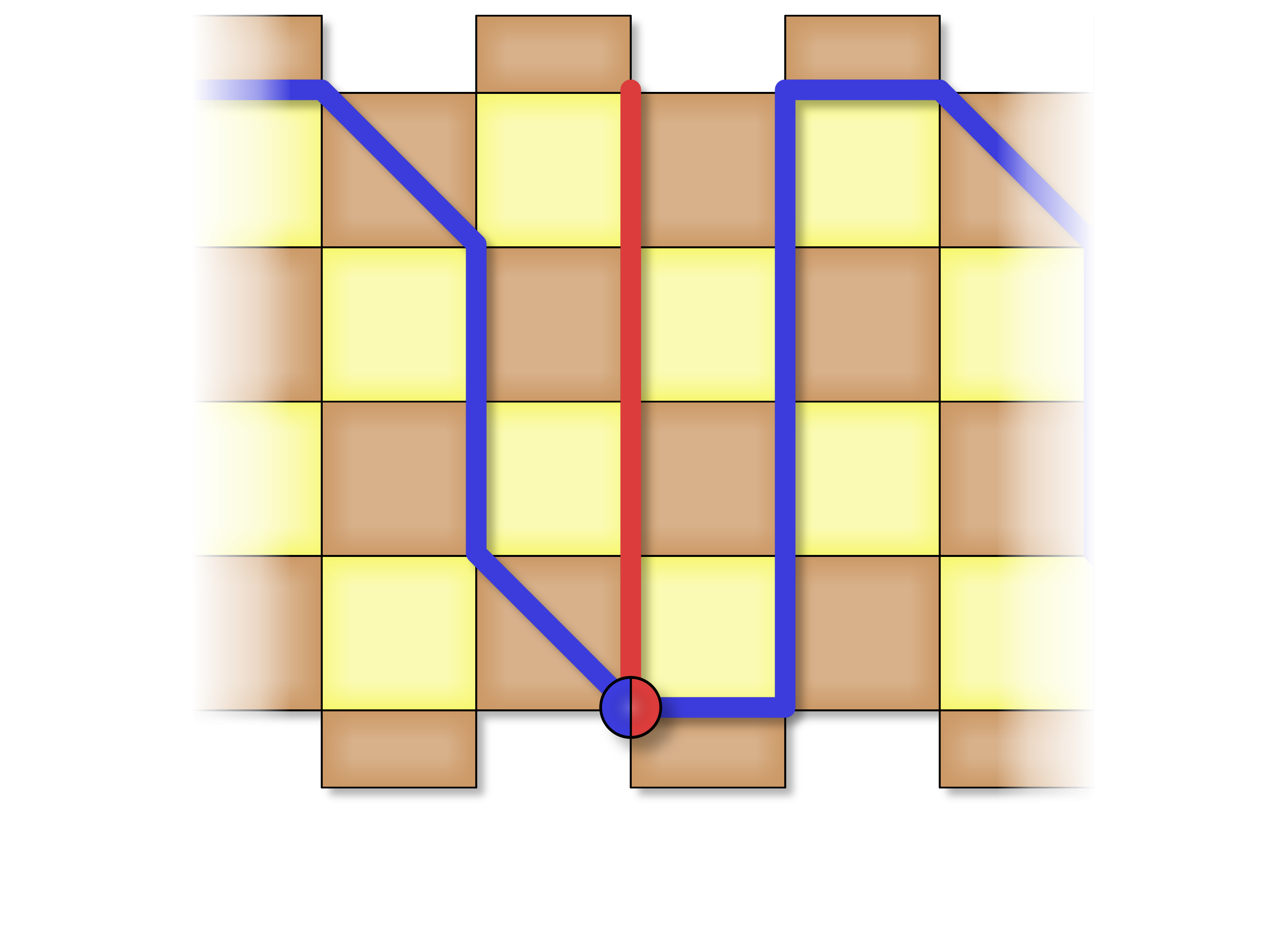
		\end{subfigure}
		\begin{subfigure}[c]{0.3\textwidth}
			\def\svgwidth{\textwidth}
			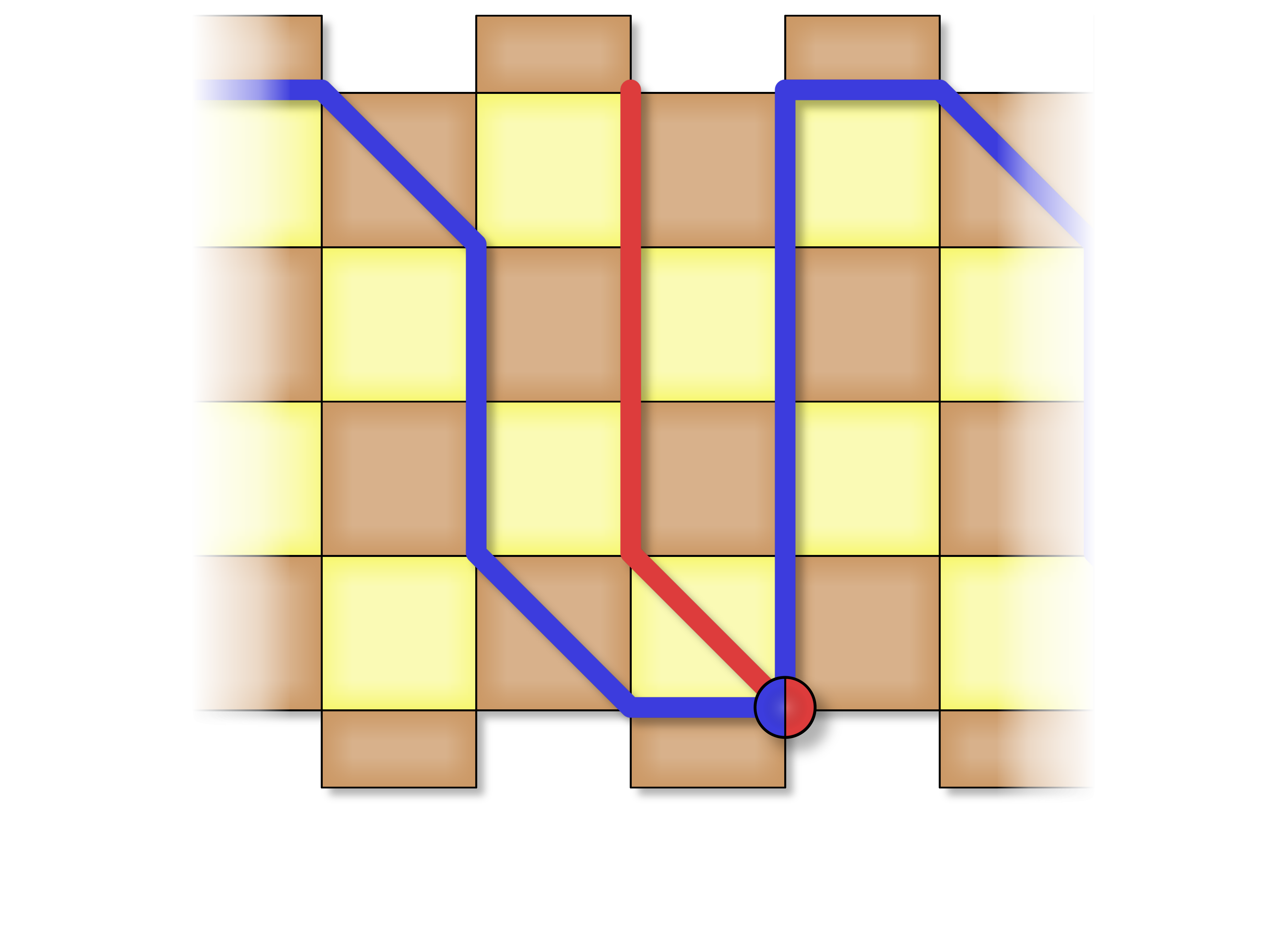
		\end{subfigure}
		\begin{subfigure}[c]{0.3\textwidth}
			\def\svgwidth{\textwidth}
			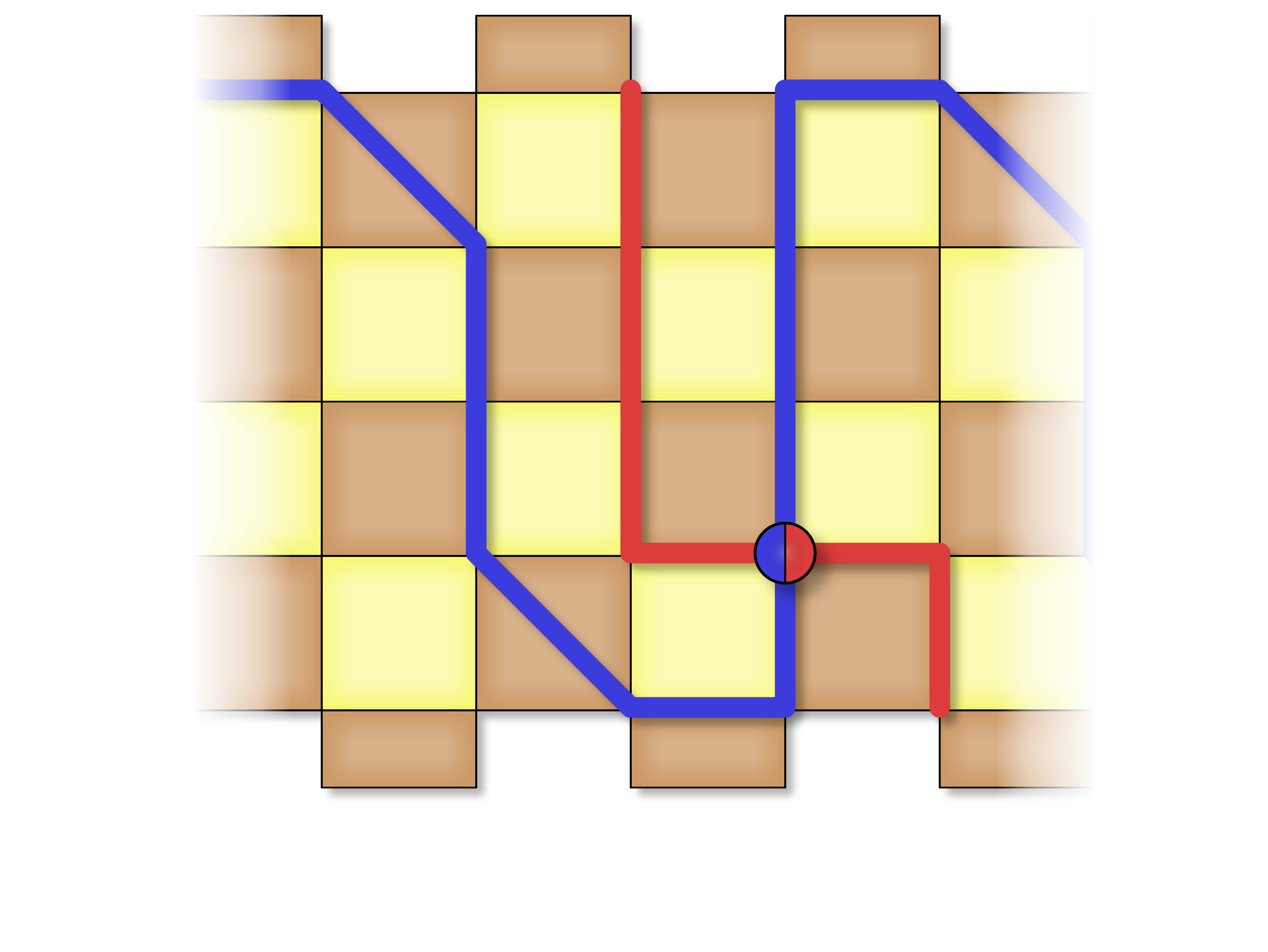
		\end{subfigure}
		\begin{subfigure}[c]{0.3\textwidth}
			\def\svgwidth{\textwidth}
			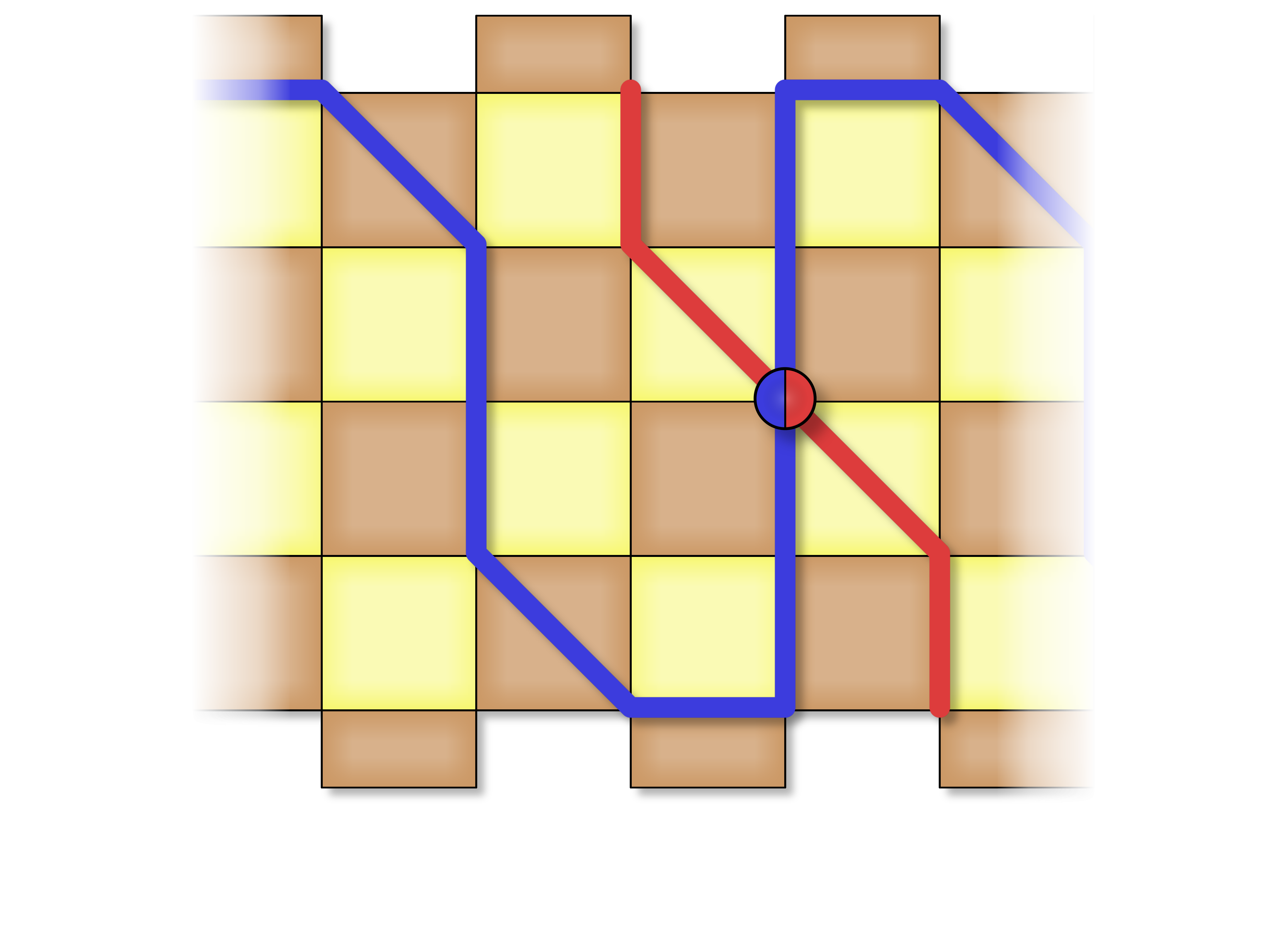
		\end{subfigure}
		\begin{subfigure}[c]{0.3\textwidth}
			\def\svgwidth{\textwidth}
			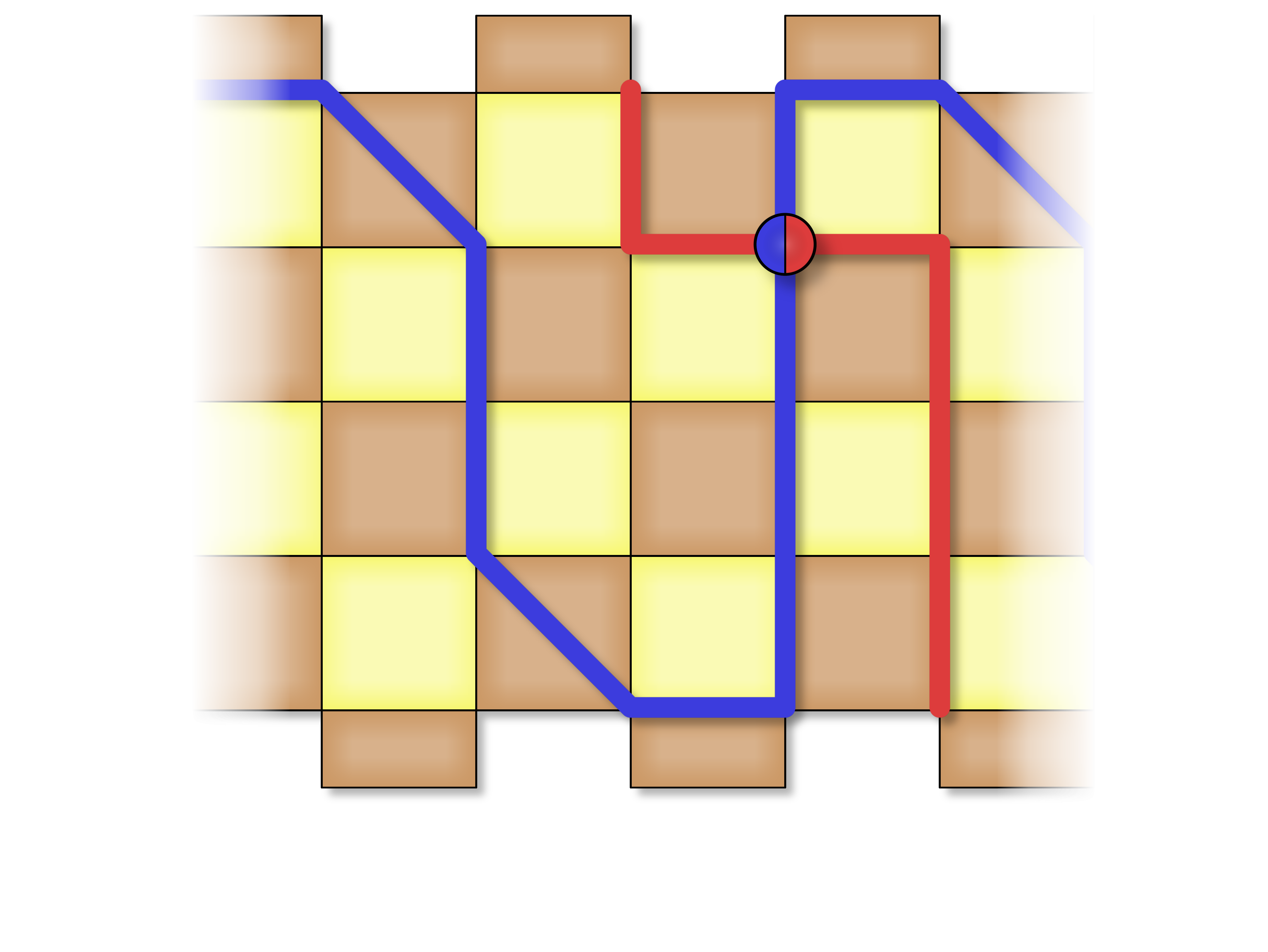
		\end{subfigure}
		\begin{subfigure}[c]{0.3\textwidth}
			\def\svgwidth{\textwidth}
			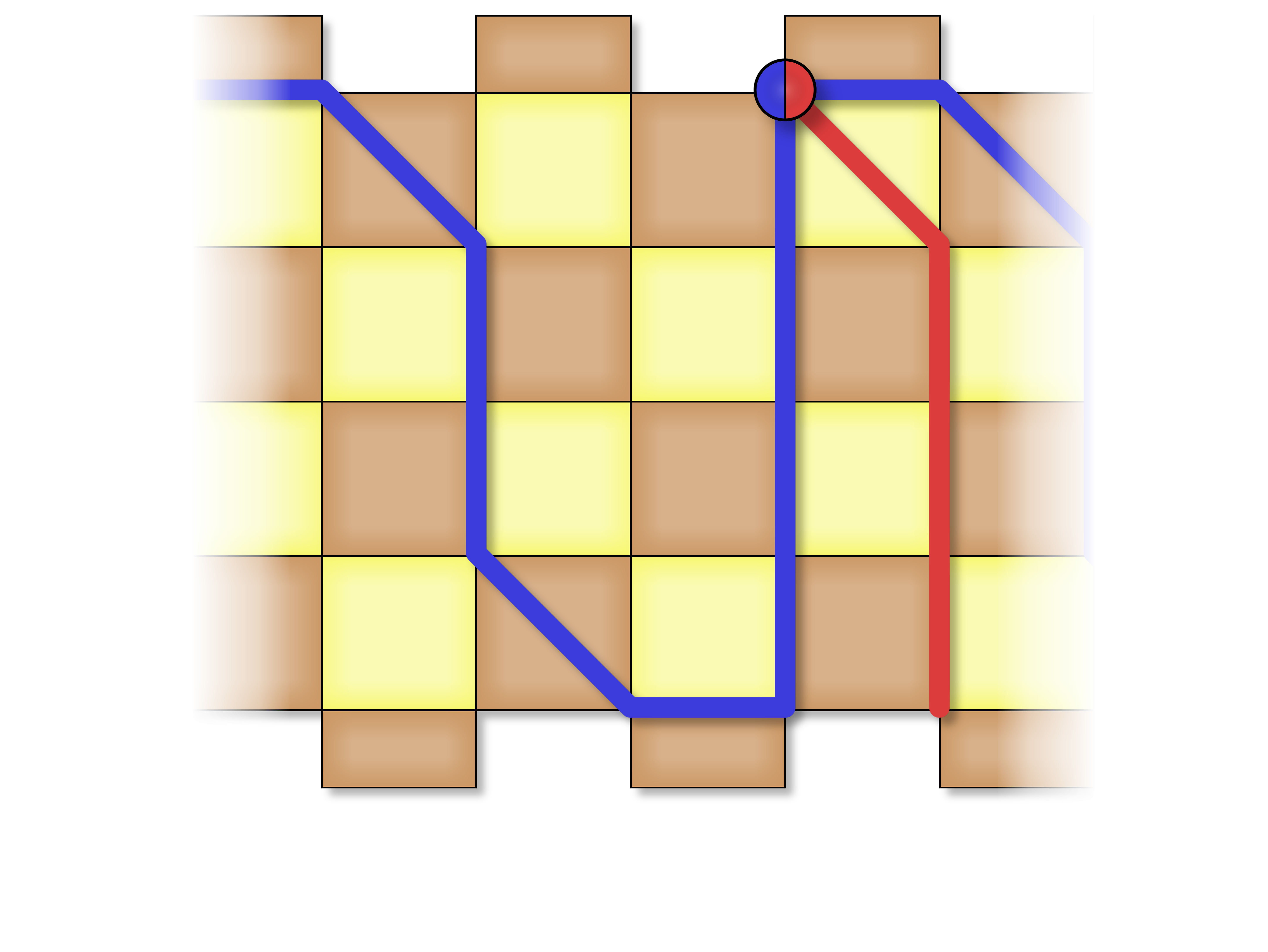
		\end{subfigure}
	\caption{\label{fig:static-optimized-sol}Logical $\overline{X}$ operator corresponding to each possible success vertex along the measured $\overline{Z}$ string for the optimized static scheme for the rotated planar surface code. Qubits filled with both red and blue indicate a successful physical BM. Red and blue strings indicate the measured $\overline{X}$ and $\overline{Z}$ string, respectively. The surface fades toward the left and right to indicate that only a single period of the wave pattern is shown, while the code may extend arbitrarily far in both directions.}
\end{figure*}

To calculate the success probability of the scheme, we compute the weight of the $\overline{Z}$ string, as illustrated in Fig.~\ref{fig:static-optimized-weight}. One period of the wave pattern is divided into four segments, and the weights of each segment is calculated individually. The total weight of the string is then obtained by summing the weight of each segment, multiplied by the number of times it occurs along the string. The number of occurrences for each segment can be determined by dividing the length of the lattice by four, accounting for the offset determined by the position in the wave period, and rounding down to the nearest integer. The resulting formula for the total weight of the string is given by:
\begin{equation}
	\begin{aligned}
	& W_Z(r,m) = & \\
		& \begin{dcases}
					\begin{aligned}
					1	& + \lfloor \frac{m+2}{4} \rfloor + (r-2) \lfloor \frac{m+1}{4} \rfloor	\\
						& + \lfloor \frac{m}{4} \rfloor + r \lfloor \frac{m-1}{4} \rfloor 		
					\end{aligned}	 
					& \quad \text{if } r \text{ odd,} 	\\
					\\
					\begin{aligned}
					1	& + \lfloor \frac{m+2}{4} \rfloor + (r-1) \lfloor \frac{m+1}{4} \rfloor	\\
						& + \lfloor \frac{m}{4} \rfloor + (r-1) \lfloor \frac{m-1}{4} \rfloor
					\end{aligned}
					& \quad \text{if } r \text{ even.} 	\\
		\end{dcases}
	\end{aligned}
\end{equation}

\begin{figure*}
		\begin{subfigure}[c]{0.35\textwidth}
			\def\svgwidth{\textwidth}
\begingroup%
  \makeatletter%
  \providecommand\color[2][]{%
    \errmessage{(Inkscape) Color is used for the text in Inkscape, but the package 'color.sty' is not loaded}%
    \renewcommand\color[2][]{}%
  }%
  \providecommand\transparent[1]{%
    \errmessage{(Inkscape) Transparency is used (non-zero) for the text in Inkscape, but the package 'transparent.sty' is not loaded}%
    \renewcommand\transparent[1]{}%
  }%
  \providecommand\rotatebox[2]{#2}%
  \newcommand*\fsize{\dimexpr\f@size pt\relax}%
  \newcommand*\lineheight[1]{\fontsize{\fsize}{#1\fsize}\selectfont}%
  \ifx\svgwidth\undefined%
    \setlength{\unitlength}{2777.95275591bp}%
    \ifx\svgscale\undefined%
      \relax%
    \else%
      \setlength{\unitlength}{\unitlength * \real{\svgscale}}%
    \fi%
  \else%
    \setlength{\unitlength}{\svgwidth}%
  \fi%
  \global\let\svgwidth\undefined%
  \global\let\svgscale\undefined%
  \makeatother%
  \begin{picture}(1,1)%
    \lineheight{1}%
    \setlength\tabcolsep{0pt}%
    \put(0,0){\includegraphics[width=\unitlength,page=1]{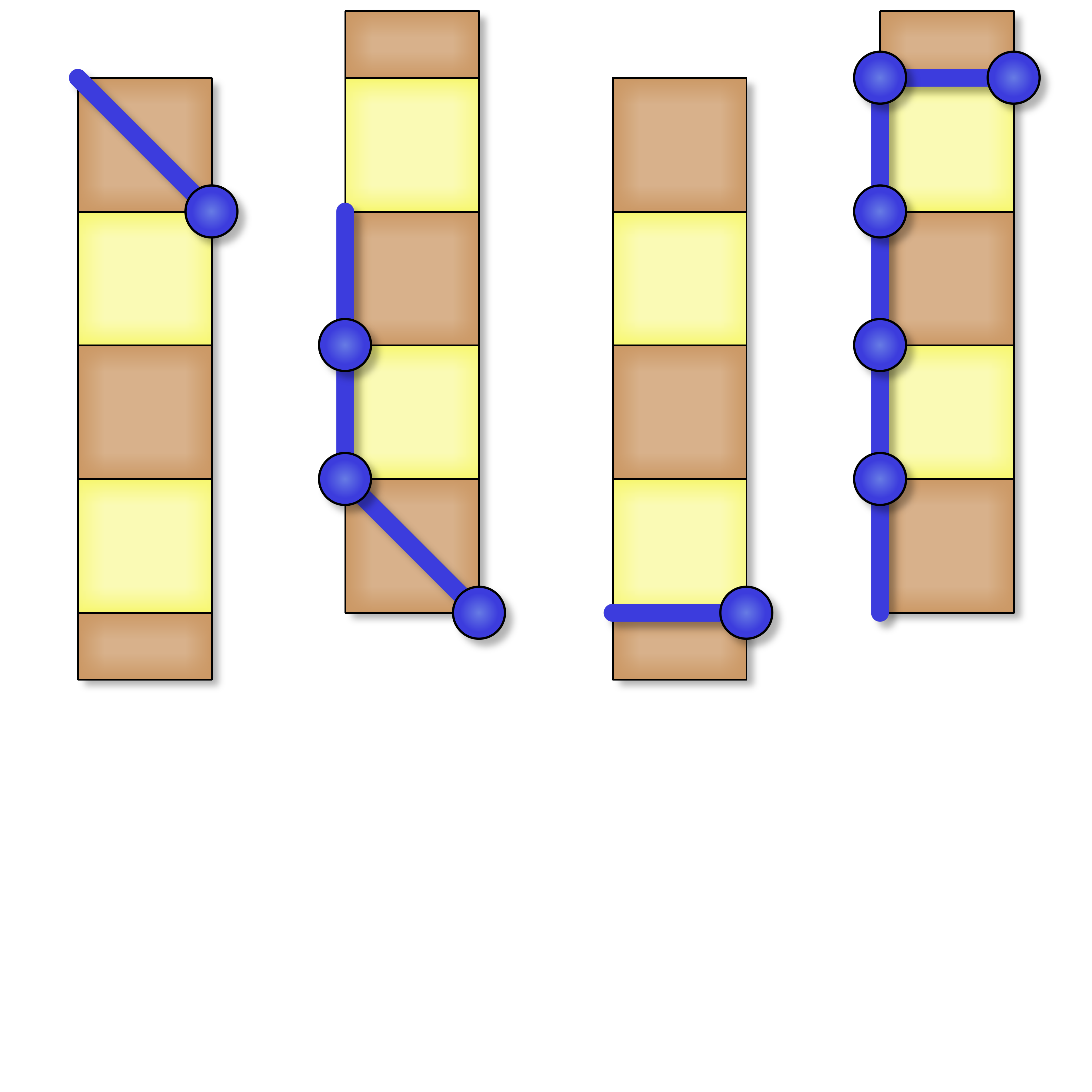}}%
    \put(0.13272433,0.13026933){\color[rgb]{0,0,0}\makebox(0,0)[t]{\lineheight{0}\smash{\begin{tabular}[t]{c}\makebox(0,0){$1$}\end{tabular}}}}%
    \put(0.3776223,0.13026933){\color[rgb]{0,0,0}\makebox(0,0)[t]{\lineheight{0}\smash{\begin{tabular}[t]{c}\makebox(0,0){$r-2$}\end{tabular}}}}%
    \put(0.86741819,0.13026933){\color[rgb]{0,0,0}\makebox(0,0)[t]{\lineheight{0}\smash{\begin{tabular}[t]{c}\makebox(0,0){$r$}\end{tabular}}}}%
    \put(0.62252023,0.13026933){\color[rgb]{0,0,0}\makebox(0,0)[t]{\lineheight{0}\smash{\begin{tabular}[t]{c}\makebox(0,0){$1$}\end{tabular}}}}%
  \end{picture}%
\endgroup%

			\caption{}
		\end{subfigure}
		\hspace{15mm}
		\begin{subfigure}[c]{0.35\textwidth}
			\def\svgwidth{\textwidth}
\begingroup%
  \makeatletter%
  \providecommand\color[2][]{%
    \errmessage{(Inkscape) Color is used for the text in Inkscape, but the package 'color.sty' is not loaded}%
    \renewcommand\color[2][]{}%
  }%
  \providecommand\transparent[1]{%
    \errmessage{(Inkscape) Transparency is used (non-zero) for the text in Inkscape, but the package 'transparent.sty' is not loaded}%
    \renewcommand\transparent[1]{}%
  }%
  \providecommand\rotatebox[2]{#2}%
  \newcommand*\fsize{\dimexpr\f@size pt\relax}%
  \newcommand*\lineheight[1]{\fontsize{\fsize}{#1\fsize}\selectfont}%
  \ifx\svgwidth\undefined%
    \setlength{\unitlength}{2777.95275591bp}%
    \ifx\svgscale\undefined%
      \relax%
    \else%
      \setlength{\unitlength}{\unitlength * \real{\svgscale}}%
    \fi%
  \else%
    \setlength{\unitlength}{\svgwidth}%
  \fi%
  \global\let\svgwidth\undefined%
  \global\let\svgscale\undefined%
  \makeatother%
  \begin{picture}(1,1)%
    \lineheight{1}%
    \setlength\tabcolsep{0pt}%
    \put(0,0){\includegraphics[width=\unitlength,page=1]{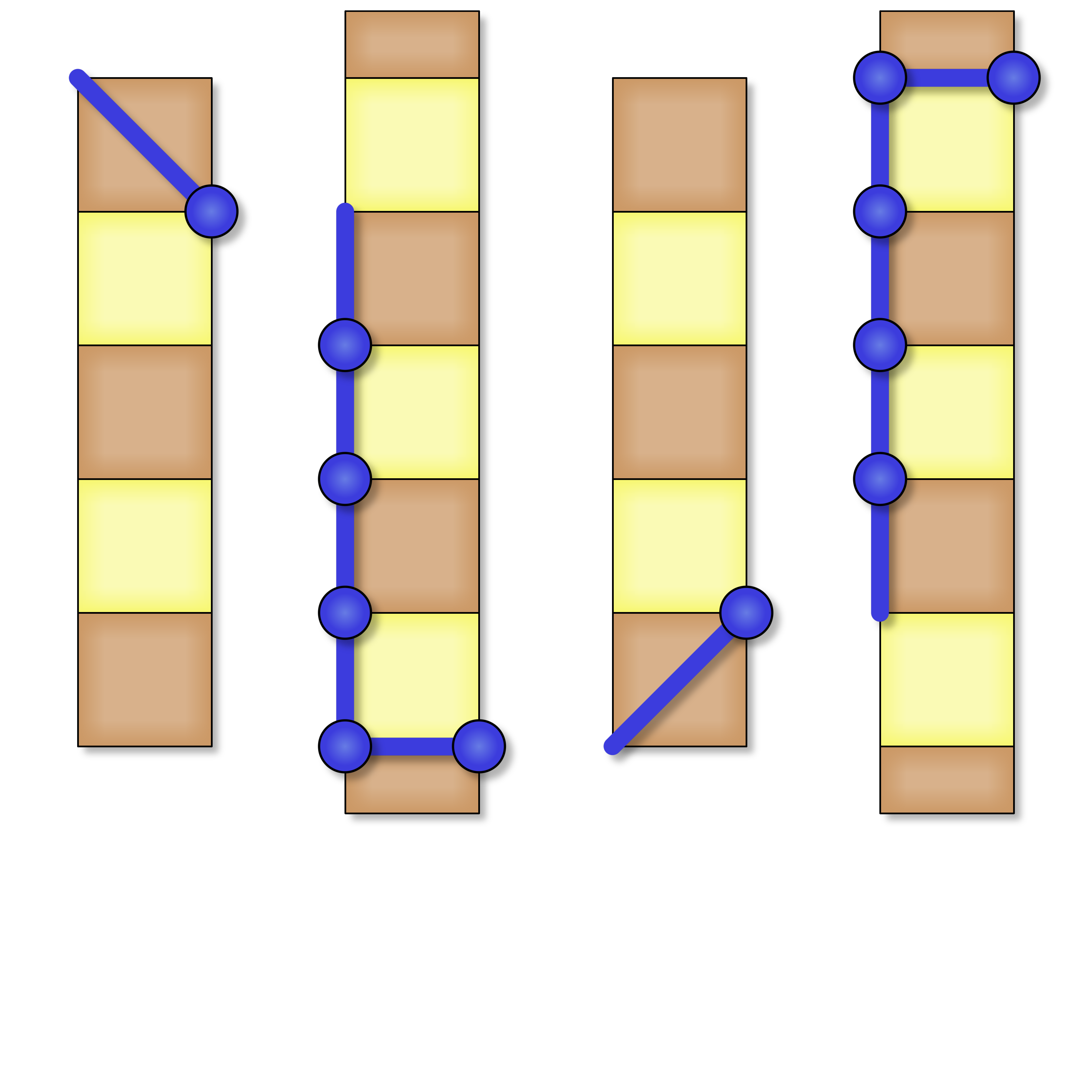}}%
    \put(0.13272434,0.13026933){\color[rgb]{0,0,0}\makebox(0,0)[t]{\lineheight{0}\smash{\begin{tabular}[t]{c}\makebox(0,0){$1$}\end{tabular}}}}%
    \put(0.3776223,0.13026933){\color[rgb]{0,0,0}\makebox(0,0)[t]{\lineheight{0}\smash{\begin{tabular}[t]{c}\makebox(0,0){$r-1$}\end{tabular}}}}%
    \put(0.86741819,0.13026933){\color[rgb]{0,0,0}\makebox(0,0)[t]{\lineheight{0}\smash{\begin{tabular}[t]{c}\makebox(0,0){$r-1$}\end{tabular}}}}%
    \put(0.62252023,0.13026933){\color[rgb]{0,0,0}\makebox(0,0)[t]{\lineheight{0}\smash{\begin{tabular}[t]{c}\makebox(0,0){$1$}\end{tabular}}}}%
  \end{picture}%
\endgroup%

			\caption{}
		\end{subfigure}
	\caption{\label{fig:static-optimized-weight}Weights for each one plaquette wide segment of a full period of the $\overline{Z}$ string for the optimized static scheme for the rotated planar surface code. The strings are shown for even (a) and odd (b) values of $r$. Note that the first qubit of the string is missing in each segment. This ensures that when the segments are connected, the qubit at the connection point is not counted twice.}
\end{figure*}

In conclusion the success probability, of the scheme is
\begin{equation}
	P_{\text{static, optimized}}(r,m) = 1 - (1-\mathbb{P}_B)^{W_Z(r,m)},
\end{equation}
which is the probability that at least one transversal BM along the $\overline{Z}$ string succeeds. While we do not claim that our optimized static scheme achieves the global optimum, it is the most efficient static scheme we have found.

Toward the end of this section, we compare schemes for planar surface codes. To the best of the authors' knowledge, no logical BM schemes for the rotated planar surface code have been published to date. Ref.~\cite{PhysRevA.99.062308} introduces a static logical BM scheme for the standard planar surface code and provides methods to calculate its success probability. The standard planar surface code requires significantly more qubits for the same code distance than the rotated planar surface code~\cite{Horsman_2012}. Both codes are characterized by two parameters, $r$ and $m$, which define the distances of the logical $\overline{X}$ and $\overline{Z}$ operators, respectively. However, while the standard planar surface code requires $2rm+1-r-m$ qubits, the rotated planar surface code requires only $rm$ qubits. To compare their performances, we consider the rotated $(5,5)$ and the standard $(4,4)$ code, as both utilize the same number of qubits, 25. For comparability, we assume standard linear-optics BMs with $\mathbb{P}_B = \frac{1}{2}$. In this case, the rotated code not only achieves a higher code distance but also exhibits a significantly higher success probability when using our optimized static scheme of $\frac{2047}{2048}$ compared to the improved scheme in Ref.~\cite{PhysRevA.99.062308}, which achieves only $\frac{127}{128}$. Both schemes significantly outperform the simplest scheme, which uses only $Z$-BMs and achieves a $\frac{6625}{8192}$ success probability for the rotated planar surface code.

In Fig.~\ref{fig:rotated-planar-surface-codes-comparison}, we compare the performance of the three approaches for the rotated planar surface code, namely the simplest scheme, which uses transversal $Z$-BMs on all qubits, our optimized static scheme, and the feedforward-based scheme described in Sec.~\ref{sec:rotated-planar-surface-code}. To evaluate the simplest scheme, we implemented an algorithm based on the methods of Ref.~\cite{PhysRevA.99.062308} to compute its logical BM success probability. 

\begin{figure}
	\def\svgwidth{0.45\textwidth}
\begingroup%
  \makeatletter%
  \providecommand\color[2][]{%
    \errmessage{(Inkscape) Color is used for the text in Inkscape, but the package 'color.sty' is not loaded}%
    \renewcommand\color[2][]{}%
  }%
  \providecommand\transparent[1]{%
    \errmessage{(Inkscape) Transparency is used (non-zero) for the text in Inkscape, but the package 'transparent.sty' is not loaded}%
    \renewcommand\transparent[1]{}%
  }%
  \providecommand\rotatebox[2]{#2}%
  \newcommand*\fsize{\dimexpr\f@size pt\relax}%
  \newcommand*\lineheight[1]{\fontsize{\fsize}{#1\fsize}\selectfont}%
  \ifx\svgwidth\undefined%
    \setlength{\unitlength}{637.79527559bp}%
    \ifx\svgscale\undefined%
      \relax%
    \else%
      \setlength{\unitlength}{\unitlength * \real{\svgscale}}%
    \fi%
  \else%
    \setlength{\unitlength}{\svgwidth}%
  \fi%
  \global\let\svgwidth\undefined%
  \global\let\svgscale\undefined%
  \makeatother%
  \begin{picture}(1,0.46666667)%
    \lineheight{1}%
    \setlength\tabcolsep{0pt}%
    \put(0,0){\includegraphics[width=\unitlength,page=1]{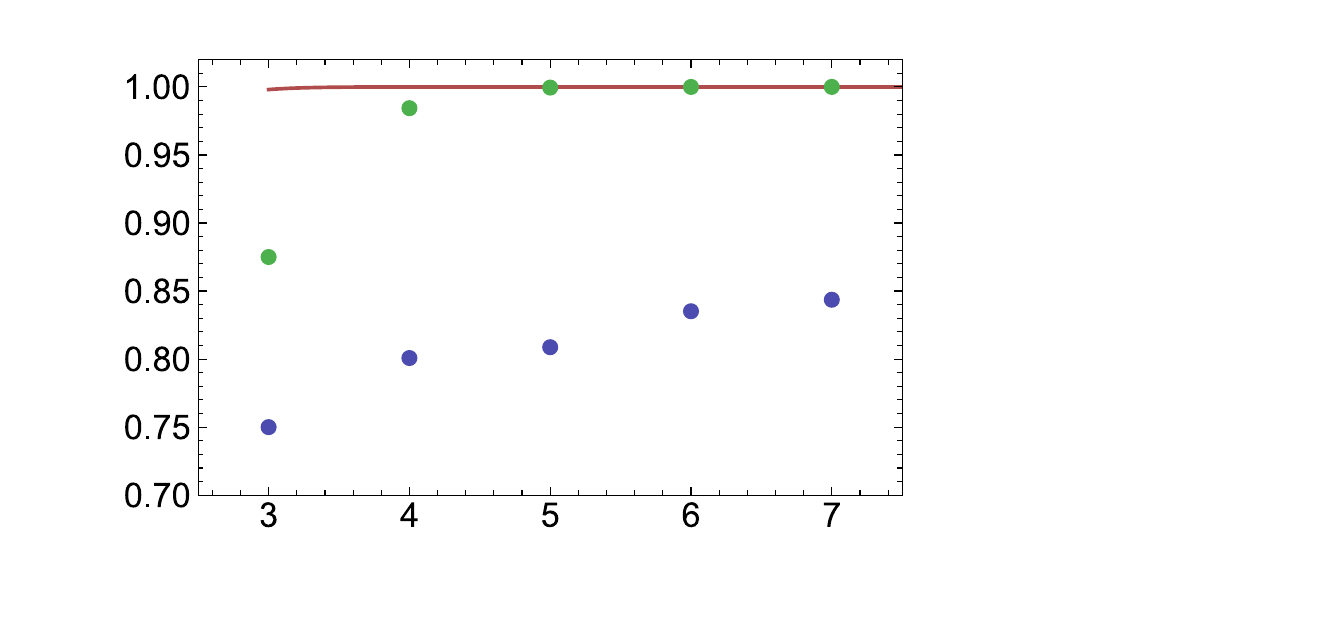}}%
    \put(0.41361164,0.03154475){\color[rgb]{0,0,0}\makebox(0,0)[t]{\lineheight{0}\smash{\begin{tabular}[t]{c}\makebox(0,0){$d$}\end{tabular}}}}%
    \put(0.05197525,0.25843341){\color[rgb]{0,0,0}\rotatebox{90}{\makebox(0,0)[t]{\lineheight{0}\smash{\begin{tabular}[t]{c}\makebox(0,0){$p$}\end{tabular}}}}}%
    \put(0,0){\includegraphics[width=\unitlength,page=2]{RotatedSurfaceComparisonPlot.pdf}}%
    \put(0.75520928,0.27103368){\color[rgb]{0,0,0}\makebox(0,0)[lt]{\lineheight{0}\smash{\begin{tabular}[t]{l}\makebox[0pt][l]{\raisebox{-0.5\height}{adaptive}}\end{tabular}}}}%
    \put(0.75520928,0.20436676){\color[rgb]{0,0,0}\makebox(0,0)[lt]{\lineheight{0}\smash{\begin{tabular}[t]{l}\makebox[0pt][l]{\raisebox{-0.5\height}{static, optimized}}\end{tabular}}}}%
    \put(0.75520928,0.13770009){\color[rgb]{0,0,0}\makebox(0,0)[lt]{\lineheight{0}\smash{\begin{tabular}[t]{l}\makebox[0pt][l]{\raisebox{-0.5\height}{static, simple}}\end{tabular}}}}%
  \end{picture}%
\endgroup%

	\caption{\label{fig:rotated-planar-surface-codes-comparison}Comparison for logical BM schemes for the quadratic rotated planar surface code. The plot displays the success probability $p$ as a function of the code dimension $r=m=d$.}			
\end{figure}

\subsection{Tree code}
\label{sec:tree-code}
In this section, we introduce our measurement scheme for the tree code~\cite{PhysRevLett.97.120501}. Tree states are a subclass of graph states. A graph state is defined from a graph with vertices $V$, where each vertex $v \in V$ corresponds to a qubit. The graph state is then characterized as the quantum state stabilized by the set of stabilizer generators $\{ K_v \mid v \in V \}$, where each stabilizer generator is given by:
\begin{equation}
    K_v = X_v \prod_{w \in \N(v)} Z_w,
\end{equation}
where $\N(v)$ denotes the set of neighbors of $v$.

A tree state is a graph state whose underlying graph is a tree. Let $r \in V$ denote the root of the tree. The stabilizer group $S_c$ of the tree code is generated by the set:
\begin{equation}
    G_c = \{ K_v \mid v \in V \setminus \{r\} \}.
\end{equation}
Removing the stabilizer generator associated with the root vertex $r$ from the graph state increases the degree of freedom of the quantum state, transforming it into a tree code.

The operator $Z_r$ is the Pauli $Z$ operator applied to the qubit at the root vertex $r$ of the tree. It commutes with all stabilizer generators in $G_c$ but is not an element of the stabilizer group $S_c$. Therefore, $Z_r$ is a logical operator of the tree code and specifically belongs to $[\overline{X}]$. This implies that, without further intervention, the tree code would effectively exhibit a distance of one. In particular, a phase flip on the root qubit would induce an undetectable logical error. Additionally, as we will discuss later in this section, any logical operator in $[\overline{X}]$ or $[\overline{Y}]$ requires the Pauli $Z$ information from the root qubit. Consequently, losing the root qubit would irreversibly destroy this logical information, making the code unsuitable for protecting against qubit loss.

Both of these issues can be addressed by measuring the root qubit $r$ immediately using an $X$ measurement, before any errors can occur. This approach effectively resolves the problem by ensuring that the remaining qubits encode information with a non-trivial distance, thereby enabling error correction capabilities. In our schemes this translates to either an $X$-BM or two single-qubit $X$ measurements on the roots of the trees.

To introduce our notation, we define several functions relevant to tree structures. The children function $\C(v)$ denotes the set of direct children of vertex $v$. The ancestor function $\anc(v,n)$ identifies the $n$-th ancestor of vertex $v$; for instance, $\anc(v,1)$ is the parent, $\anc(v,2)$ is the grandparent, and so on. Lastly, the depth function $\depth(v)$ gives the depth of vertex $v$ within the tree, defined as the number of edges on the path from the root to $v$. We define the root as the zero level of the tree, with the leaves occupying the highest level. The depth $d$ of a tree is the maximum depth of any node in the tree.

Let us look at examples of $\overline{X}$ and $\overline{Z}$ operators for the tree code. An example of an $\overline{X}$ operator is:
\begin{equation}
    X_v \prod_{w \in \C(v)} Z_w \in [\overline{X}], \quad \text{where } v \in \C(r).
    \label{eq:tree-X-example}
\end{equation}
An example of a $\overline{Z}$ operator is:
\begin{equation}
    X_r \prod_{w \in \C(r)} Z_w = K_r \in [\overline{Z}].
    \label{eq:tree-Z-example}
\end{equation}
These logical operators are illustrated in Fig.~\ref{fig:binary-tree-23-logical-base}. Note that the $\overline{X}$ operator, as shown in Eq.~\eqref{eq:tree-X-example}, is the product of $Z_r$ with the code stabilizer $K_v$ for a vertex $v \in \C(r)$. Similarly, the $\overline{Z}$ operator in Eq.~\eqref{eq:tree-Z-example} is the operator $K_r$ associated with the root vertex $r$. Consequently, defining the operators $Z_r$ and $K_r$ as the base operators for $\overline{X}$ and $\overline{Z}$ operators, respectively, allows us to express the sets $[\overline{X}]$ and $[\overline{Z}]$ in a compact form as follows:
\begin{equation}
	[\overline{X}] = \{ Z_r s \mid s \in S_c \},
	\label{eq:tree-logical-set-x}
\end{equation}
\begin{equation}
	[\overline{Z}] = \{ K_r s \mid s \in S_c \}.
	\label{eq:tree-logical-set-z}
\end{equation}

\definecolor{myred}{RGB}{221,60,60}
\definecolor{myblue}{RGB}{60,60,221}

\begin{figure*}
		\begin{subfigure}[c]{0.45\textwidth}
			\def\svgwidth{\textwidth}
\begingroup%
  \makeatletter%
  \providecommand\color[2][]{%
    \errmessage{(Inkscape) Color is used for the text in Inkscape, but the package 'color.sty' is not loaded}%
    \renewcommand\color[2][]{}%
  }%
  \providecommand\transparent[1]{%
    \errmessage{(Inkscape) Transparency is used (non-zero) for the text in Inkscape, but the package 'transparent.sty' is not loaded}%
    \renewcommand\transparent[1]{}%
  }%
  \providecommand\rotatebox[2]{#2}%
  \newcommand*\fsize{\dimexpr\f@size pt\relax}%
  \newcommand*\lineheight[1]{\fontsize{\fsize}{#1\fsize}\selectfont}%
  \ifx\svgwidth\undefined%
    \setlength{\unitlength}{935.43307087bp}%
    \ifx\svgscale\undefined%
      \relax%
    \else%
      \setlength{\unitlength}{\unitlength * \real{\svgscale}}%
    \fi%
  \else%
    \setlength{\unitlength}{\svgwidth}%
  \fi%
  \global\let\svgwidth\undefined%
  \global\let\svgscale\undefined%
  \makeatother%
  \begin{picture}(1,0.60606061)%
    \lineheight{1}%
    \setlength\tabcolsep{0pt}%
    \put(0,0){\includegraphics[width=\unitlength,page=1]{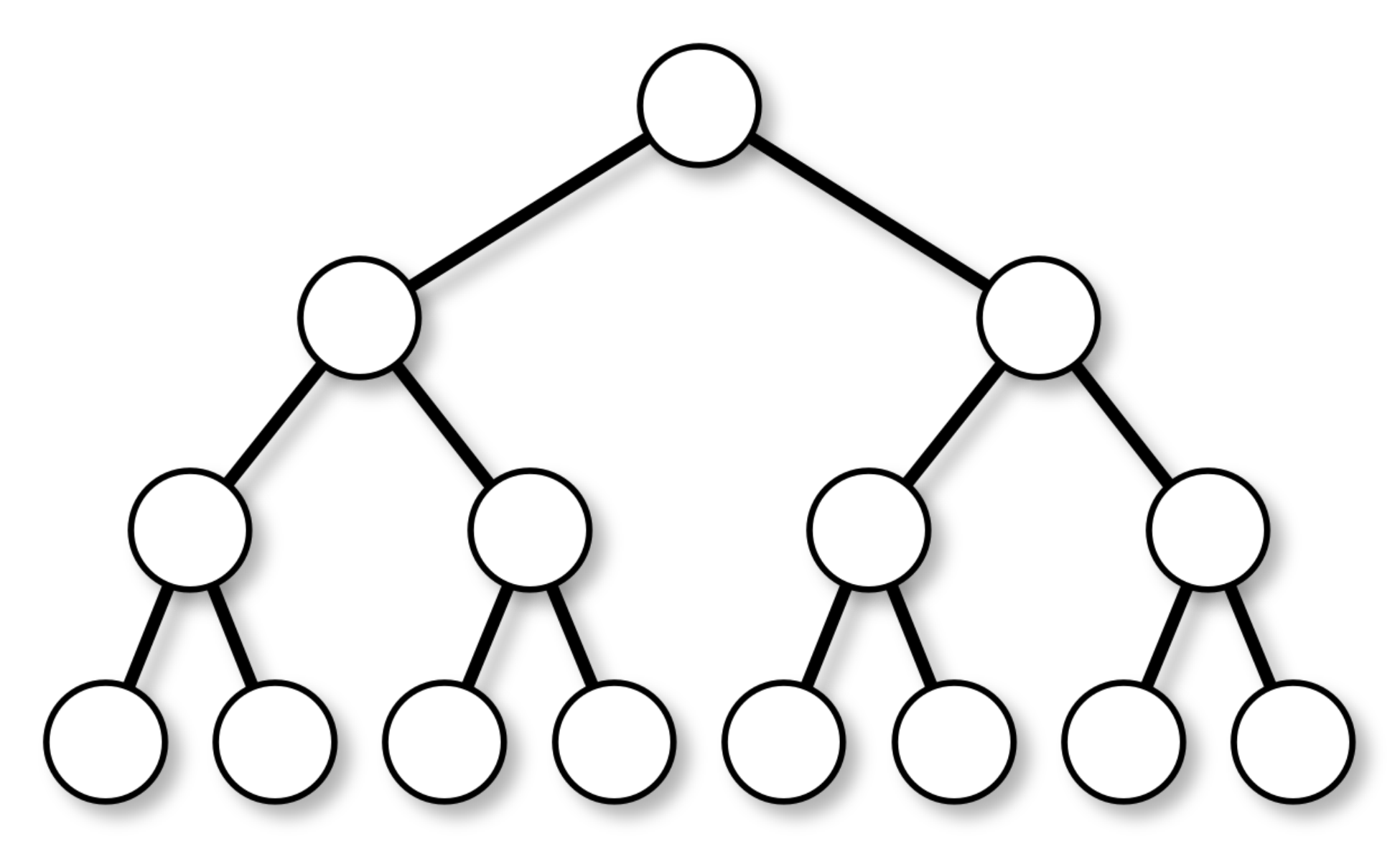}}%
    \put(0.37899954,0.2272199){\color[rgb]{0,0,0}\makebox(0,0)[t]{\lineheight{0}\smash{\begin{tabular}[t]{c}\makebox(0,0){$\color{myred}{Z}$}\end{tabular}}}}%
    \put(0.13657531,0.2272199){\color[rgb]{0,0,0}\makebox(0,0)[t]{\lineheight{0}\smash{\begin{tabular}[t]{c}\makebox(0,0){$\color{myred}{Z}$}\end{tabular}}}}%
    \put(0.25778742,0.37873507){\color[rgb]{0,0,0}\makebox(0,0)[t]{\lineheight{0}\smash{\begin{tabular}[t]{c}\makebox(0,0){$\color{myred}{X}$}\end{tabular}}}}%
  \end{picture}%
\endgroup%
			
			\caption{}
		\end{subfigure}									
		\begin{subfigure}[c]{0.45\textwidth}
			\def\svgwidth{\textwidth}
\begingroup%
  \makeatletter%
  \providecommand\color[2][]{%
    \errmessage{(Inkscape) Color is used for the text in Inkscape, but the package 'color.sty' is not loaded}%
    \renewcommand\color[2][]{}%
  }%
  \providecommand\transparent[1]{%
    \errmessage{(Inkscape) Transparency is used (non-zero) for the text in Inkscape, but the package 'transparent.sty' is not loaded}%
    \renewcommand\transparent[1]{}%
  }%
  \providecommand\rotatebox[2]{#2}%
  \newcommand*\fsize{\dimexpr\f@size pt\relax}%
  \newcommand*\lineheight[1]{\fontsize{\fsize}{#1\fsize}\selectfont}%
  \ifx\svgwidth\undefined%
    \setlength{\unitlength}{935.43307087bp}%
    \ifx\svgscale\undefined%
      \relax%
    \else%
      \setlength{\unitlength}{\unitlength * \real{\svgscale}}%
    \fi%
  \else%
    \setlength{\unitlength}{\svgwidth}%
  \fi%
  \global\let\svgwidth\undefined%
  \global\let\svgscale\undefined%
  \makeatother%
  \begin{picture}(1,0.60606061)%
    \lineheight{1}%
    \setlength\tabcolsep{0pt}%
    \put(0,0){\includegraphics[width=\unitlength,page=1]{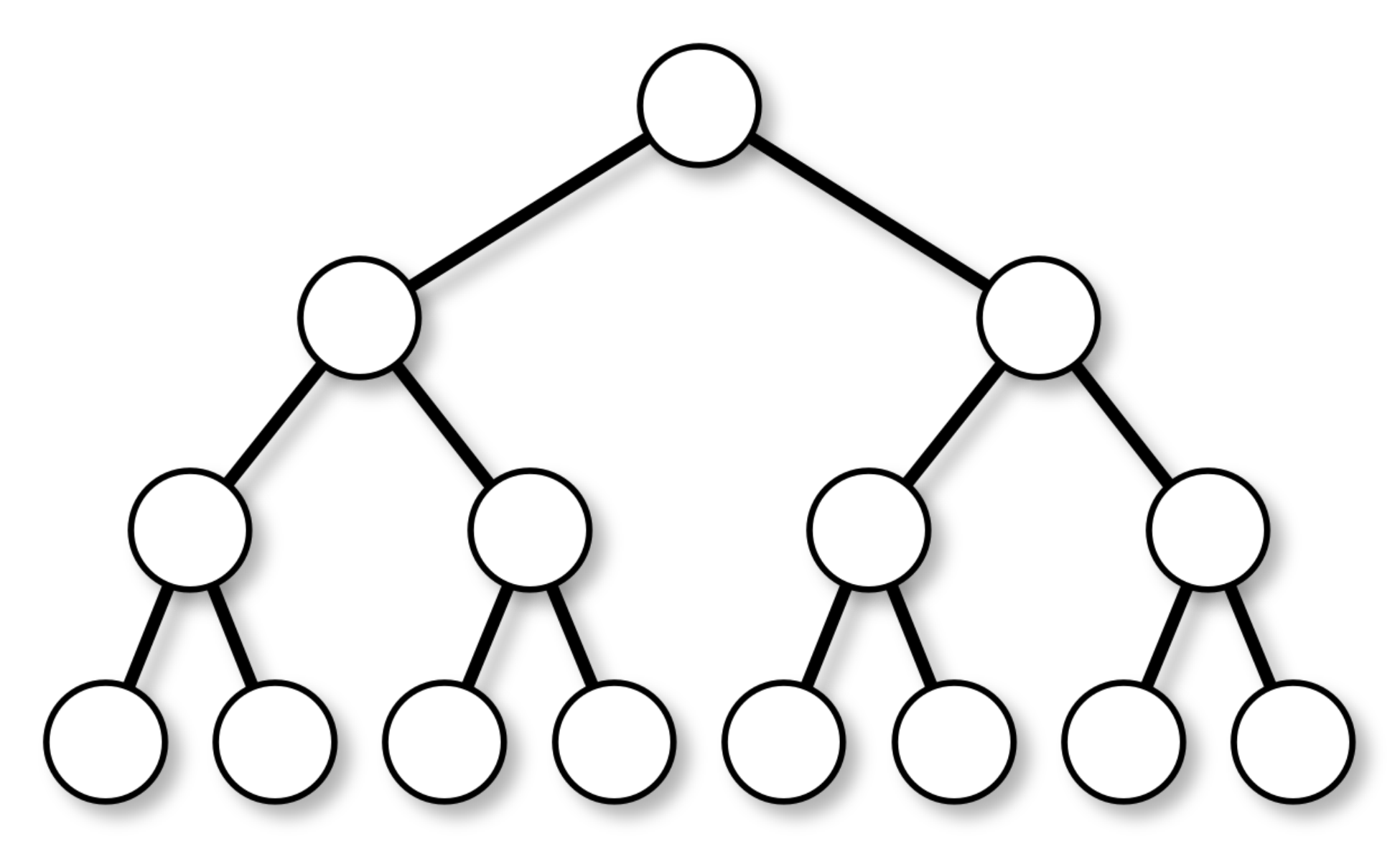}}%
    \put(0.74263591,0.37873507){\color[rgb]{0,0,0}\makebox(0,0)[t]{\lineheight{0}\smash{\begin{tabular}[t]{c}\makebox(0,0){$\color{myblue}{Z}$}\end{tabular}}}}%
    \put(0.25778741,0.37873507){\color[rgb]{0,0,0}\makebox(0,0)[t]{\lineheight{0}\smash{\begin{tabular}[t]{c}\makebox(0,0){$\color{myblue}{Z}$}\end{tabular}}}}%
    \put(0.50021167,0.53025023){\color[rgb]{0,0,0}\makebox(0,0)[t]{\lineheight{0}\smash{\begin{tabular}[t]{c}\makebox(0,0){$\color{myblue}{X}$}\end{tabular}}}}%
  \end{picture}%
\endgroup%

			\caption{}
		\end{subfigure}
	\caption{\label{fig:binary-tree-23-logical-base}Examples of logical operators of the tree code as defined in Eqs.~\eqref{eq:tree-X-example} and~\eqref{eq:tree-Z-example}. The logical $\overline{X}$ operator is shown in red letters in subfigure (a), while the logical $\overline{Z}$ operator is displayed in blue letters in subfigure (b).}
\end{figure*}
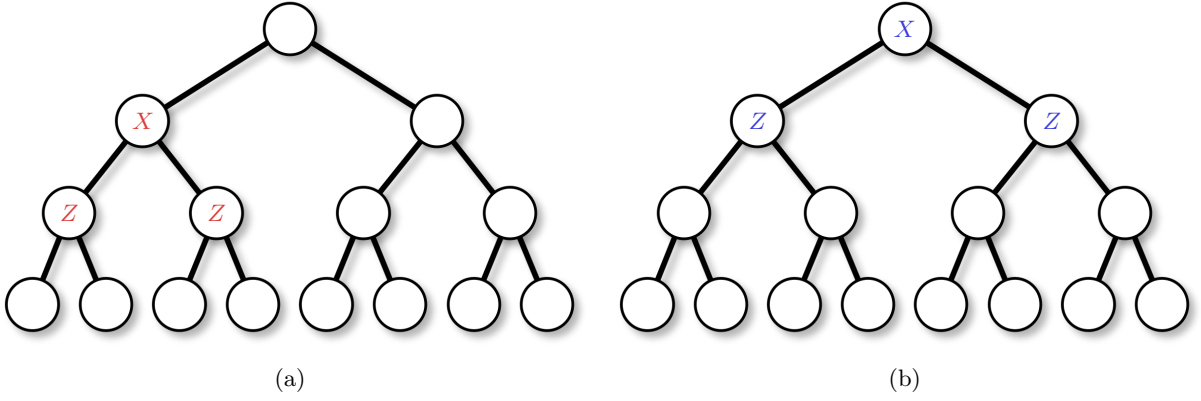

Next, we describe our measurement scheme for the tree code. The scheme begins at the leaf nodes and progresses towards the root, measuring each level in sequence. In other words, the measurements proceed from the highest to the lowest level. At each level, every qubit is measured using a $Z$-BM. All qubits within a level are measured independently of each other, so there is no need to define an order for measuring the qubits within a level.

Tracking the current stabilizer throughout this scheme is straightforward, as each $Z$-BM $Z_j$ on qubit $j$ anticommutes with exactly one stabilizer generator, specifically $K_j \in G_c$. Thus, each $Z$-BM $Z_j$ on qubit $j$ replaces $K_j$ in the current stabilizer generator. We now argue that each BM, up to the first successful one, has a success probability of $\mathbb{P}_B$. For every qubit $j$, $Y_j$ anticommutes with $K_j$, which is an element of the current stabilizer generators, as discussed for $Z_j$. For all levels except the first, $X_j$ anticommutes with the stabilizer generator associated with its parent, $K_{\anc(j,1)}$. At the first level $X_j$ completes the logical operator $X_j \prod_{w \in \C(j)} Z_w \in [\overline{X}]$. Having confirmed that all single-qubit observables on the measured qubits either anticommute with an element of the current stabilizer or complete a logical operator, we can now apply Lem.~\ref{lem:successful-bell-measurment}, concluding that the success probability for each transversal BM up to the first success is $\mathbb{P}_B$.

In the following, we will explain how the logical operators $\overline{X}_j$ and $\overline{Z}_j$ are measured if a successful BM occurs on qubit $j$. For clarity, we will analyze the cases from the lowest to the highest level, though it should be noted that the actual measurements are performed sequentially from the highest to the lowest level.

We begin with a simple example of a binary tree with height two. This tree can be characterized by the branching parameters $(b_0, b_1) = (2, 2)$. The branching parameters can be used to define a rooted tree in which, at every level, all nodes have the same number of children. Such a tree of depth $d$ is specified by the sequence $(b_0, b_1, \dots, b_{d-1})$, where $b_i$ denotes the number of children of a node at depth $i$. Since both the tree structure and our measurement scheme are symmetric for all qubits at the same level, it suffices to consider two cases, namely when the first successful BM occurs at either the first or the second level. The logical operators for these cases are illustrated in Fig.~\ref{fig:binary-tree-22-solution}.

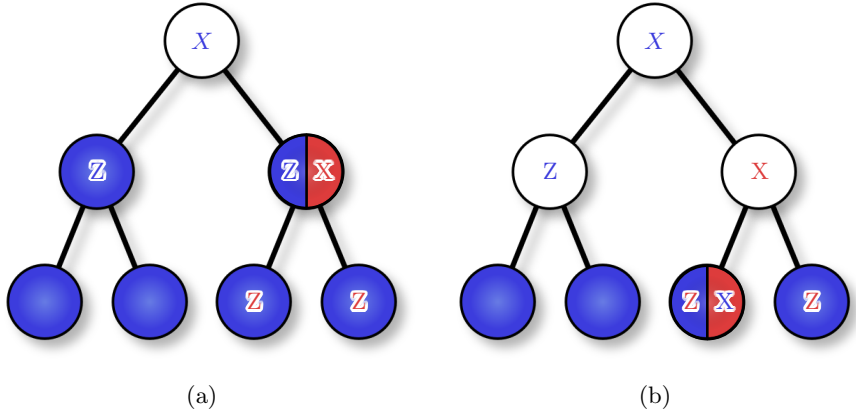
\begin{figure*}
		\begin{subfigure}[c]{0.33\textwidth}
			\def\svgwidth{\textwidth}
\begingroup%
  \makeatletter%
  \providecommand\color[2][]{%
    \errmessage{(Inkscape) Color is used for the text in Inkscape, but the package 'color.sty' is not loaded}%
    \renewcommand\color[2][]{}%
  }%
  \providecommand\transparent[1]{%
    \errmessage{(Inkscape) Transparency is used (non-zero) for the text in Inkscape, but the package 'transparent.sty' is not loaded}%
    \renewcommand\transparent[1]{}%
  }%
  \providecommand\rotatebox[2]{#2}%
  \newcommand*\fsize{\dimexpr\f@size pt\relax}%
  \newcommand*\lineheight[1]{\fontsize{\fsize}{#1\fsize}\selectfont}%
  \ifx\svgwidth\undefined%
    \setlength{\unitlength}{481.88972053bp}%
    \ifx\svgscale\undefined%
      \relax%
    \else%
      \setlength{\unitlength}{\unitlength * \real{\svgscale}}%
    \fi%
  \else%
    \setlength{\unitlength}{\svgwidth}%
  \fi%
  \global\let\svgwidth\undefined%
  \global\let\svgscale\undefined%
  \makeatother%
  \begin{picture}(1,0.88235302)%
    \lineheight{1}%
    \setlength\tabcolsep{0pt}%
    \put(0,0){\includegraphics[width=\unitlength,page=1]{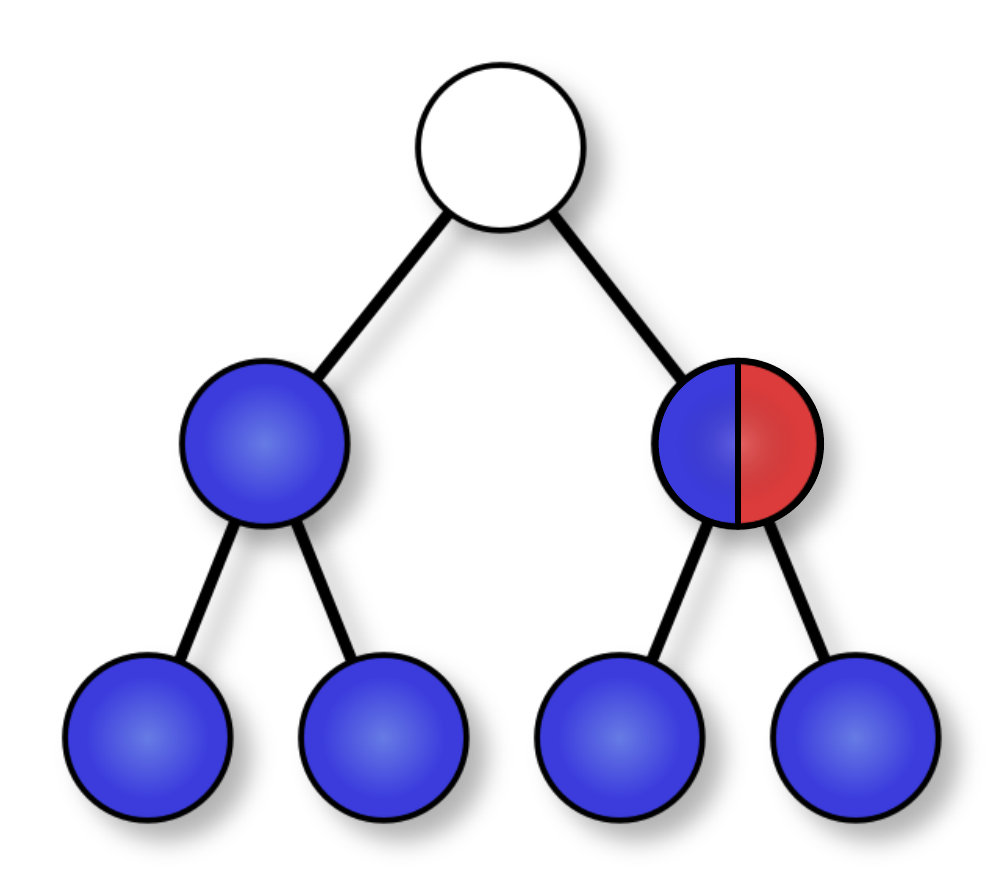}}%
    \put(0.77467743,0.44107132){\color[rgb]{1,1,1}\makebox(0,0)[t]{\lineheight{0}\smash{\begin{tabular}[t]{c}\makebox(0,0){$\contour{white}{\color{myred}{X}}$}\end{tabular}}}}%
    \put(0.69820686,0.44107132){\color[rgb]{1,1,1}\makebox(0,0)[t]{\lineheight{0}\smash{\begin{tabular}[t]{c}\makebox(0,0){$\contour{white}{\color{myblue}{Z}}$}\end{tabular}}}}%
    \put(0.61805797,0.14695629){\color[rgb]{0,0,0}\makebox(0,0)[t]{\lineheight{0}\smash{\begin{tabular}[t]{c}\makebox(0,0){$\contour{white}{\color{myred}{Z}}$}\end{tabular}}}}%
    \put(0.85335215,0.14695629){\color[rgb]{0,0,0}\makebox(0,0)[t]{\lineheight{0}\smash{\begin{tabular}[t]{c}\makebox(0,0){$\contour{white}{\color{myred}{Z}}$}\end{tabular}}}}%
    \put(0.2651168,0.44107401){\color[rgb]{0,0,0}\makebox(0,0)[t]{\lineheight{0}\smash{\begin{tabular}[t]{c}\makebox(0,0){$\contour{white}{\color{myblue}{Z}}$}\end{tabular}}}}%
    \put(0.50041089,0.73519169){\color[rgb]{0,0,0}\makebox(0,0)[t]{\lineheight{0}\smash{\begin{tabular}[t]{c}\makebox(0,0){$\color{myblue}{X}$}\end{tabular}}}}%
  \end{picture}%
\endgroup%
			
			\caption{}
		\end{subfigure}									
		\begin{subfigure}[c]{0.33\textwidth}
			\def\svgwidth{\textwidth}
\begingroup%
  \makeatletter%
  \providecommand\color[2][]{%
    \errmessage{(Inkscape) Color is used for the text in Inkscape, but the package 'color.sty' is not loaded}%
    \renewcommand\color[2][]{}%
  }%
  \providecommand\transparent[1]{%
    \errmessage{(Inkscape) Transparency is used (non-zero) for the text in Inkscape, but the package 'transparent.sty' is not loaded}%
    \renewcommand\transparent[1]{}%
  }%
  \providecommand\rotatebox[2]{#2}%
  \newcommand*\fsize{\dimexpr\f@size pt\relax}%
  \newcommand*\lineheight[1]{\fontsize{\fsize}{#1\fsize}\selectfont}%
  \ifx\svgwidth\undefined%
    \setlength{\unitlength}{481.88972053bp}%
    \ifx\svgscale\undefined%
      \relax%
    \else%
      \setlength{\unitlength}{\unitlength * \real{\svgscale}}%
    \fi%
  \else%
    \setlength{\unitlength}{\svgwidth}%
  \fi%
  \global\let\svgwidth\undefined%
  \global\let\svgscale\undefined%
  \makeatother%
  \begin{picture}(1,0.88235302)%
    \lineheight{1}%
    \setlength\tabcolsep{0pt}%
    \put(0,0){\includegraphics[width=\unitlength,page=1]{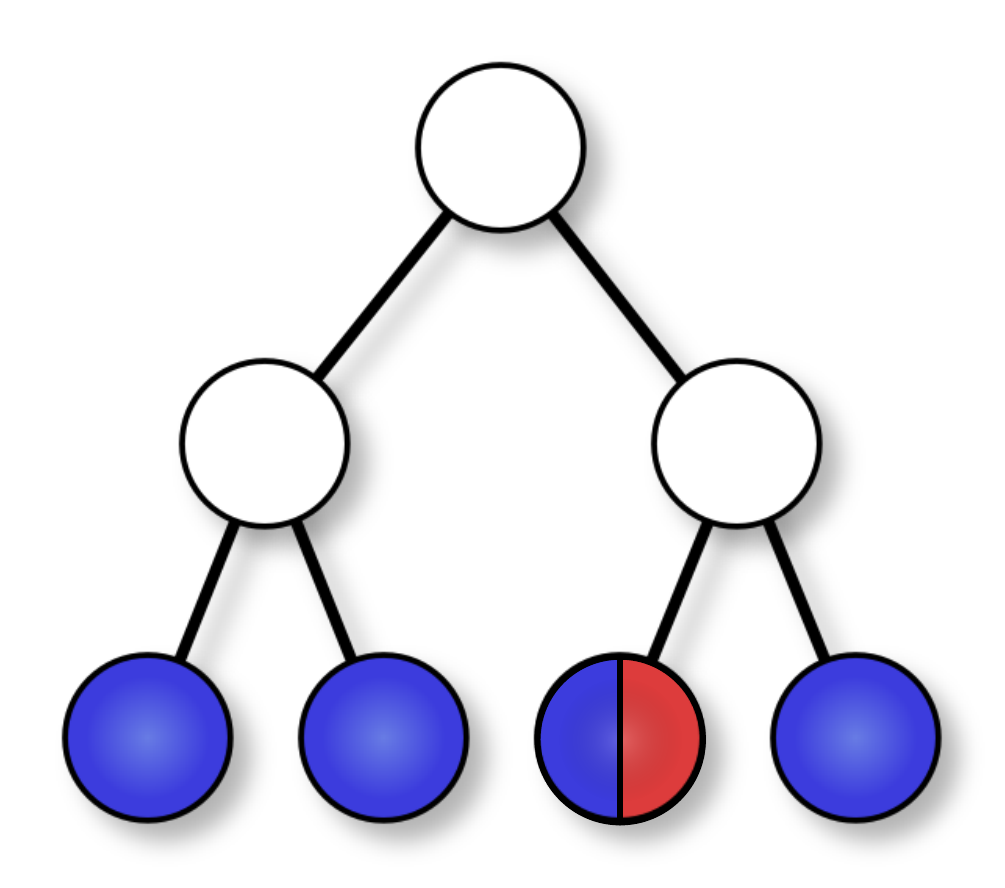}}%
    \put(0.65703067,0.14695363){\color[rgb]{1,1,1}\makebox(0,0)[t]{\lineheight{0}\smash{\begin{tabular}[t]{c}\makebox(0,0){$\contour{white}{\color{myblue}{X}}$}\end{tabular}}}}%
    \put(0.58056005,0.14695363){\color[rgb]{1,1,1}\makebox(0,0)[t]{\lineheight{0}\smash{\begin{tabular}[t]{c}\makebox(0,0){$\contour{white}{\color{myred}{Z}}$}\end{tabular}}}}%
    \put(0.85335215,0.14695629){\color[rgb]{0,0,0}\makebox(0,0)[t]{\lineheight{0}\smash{\begin{tabular}[t]{c}\makebox(0,0){$\contour{white}{\color{myred}{Z}}$}\end{tabular}}}}%
    \put(0.2651168,0.44107401){\color[rgb]{0,0,0}\makebox(0,0)[t]{\lineheight{0}\smash{\begin{tabular}[t]{c}\makebox(0,0){$\contour{white}{\color{myblue}{Z}}$}\end{tabular}}}}%
    \put(0.50041089,0.73519169){\color[rgb]{0,0,0}\makebox(0,0)[t]{\lineheight{0}\smash{\begin{tabular}[t]{c}\makebox(0,0){$\color{myblue}{X}$}\end{tabular}}}}%
    \put(0.73570504,0.44107401){\color[rgb]{0,0,0}\makebox(0,0)[t]{\lineheight{0}\smash{\begin{tabular}[t]{c}\makebox(0,0){$\contour{white}{\color{myred}{X}}$}\end{tabular}}}}%
  \end{picture}%
\endgroup%
			
			\caption{}
		\end{subfigure}
	\caption{\label{fig:binary-tree-22-solution}Measured logical operators in the event of a successful BM for a binary tree of height two, characterized by the branching parameters $(2,2)$. Panel (a) shows the case where the first successful BM occurs at the first level, while panel (b) depicts the case for the second level. Blue nodes represent qubits previously measured with a $Z$-BM. Qubits filled with both red and blue indicate a successful physical BM. Red letters denote the  measured $\overline{X}$ operator, and blue letters denote the measured $\overline{Z}$ operator.}
\end{figure*}

In the first case, if the successful BM occurs on a qubit $v_1$ at the first level (i.e., $\depth(v_1) = 1$), we complete the logical operators from Eqs.~\eqref{eq:tree-X-example} and~\eqref{eq:tree-Z-example}:
\begin{equation}
	\overline{X}_{v_1} = Z_r K_{v_1} \in [\overline{X}]
\end{equation}
and
\begin{equation}
	\overline{Z}_{v_1} = K_r \in [\overline{Z}].
\end{equation}

In the second case, if the successful BM occurs on a qubit $v_2$ at the second level (i.e., $\depth(v_2) = 2$), we complete:
\begin{equation}
	\overline{X}_{v_2} = Z_r K_{\anc(v_2,1)} \in [\overline{X}]
	\label{eq:tree-2-x-sol}
\end{equation}
and
\begin{equation}
	\overline{Z}_{v_2} = K_r K_{v_2} \in [\overline{Z}].
	\label{eq:tree-2-z-sol}
\end{equation}
It is straightforward to see that these two logical operators can always be measured with unit probability. Below the level of success, both logical operators require only $Z$ information, except for the qubit where the success occurs. Above the success level, the logical operators do not conflict on the required information on any qubit. It is important to note that if $v_2 \in C(v_1)$, the logical $\overline{X}$ operator remains identical, i.e., $\overline{X}_{v_1} = \overline{X}_{v_2}$. While the notation using the ancestor function may not make this immediately apparent, it is still a useful representation as it facilitates the generalization of the scheme in the subsequent discussion.

Before tackling the most general case, let us build on our previous example by examining a binary tree with one more level. We will now consider a binary tree of height three, characterized by the branching parameters $(2,2,2)$. The logical operators for this tree are illustrated in Fig.~\ref{fig:binary-tree-23-solution}. When the successful BM occurs on a qubit $v_1$ at the first level (i.e., $\depth(v_1) = 1$), the logical operators used to complete the measurement are identical to those for the binary tree of height two. This is because the structure of the trees at the relevant levels is the same. When the successful BM occurs on a qubit $v_2$ at the second level (i.e., $\depth(v_2) = 2$), the solution in Eqs.~\eqref{eq:tree-2-x-sol} and~\eqref{eq:tree-2-z-sol} remains essentially unchanged as well. However, we note that the operator $\overline{Z}_{v_2} = K_r K_{v_2} \in [\overline{Z}]$ now has support on nodes at the third level of the tree. When the successful BM occurs on a qubit $v_3$ at the third level (i.e., $\depth(v_3) = 3$), the logical $\overline{Z}$ operator can still be measured identically to the previous level:
\begin{equation}
	\overline{Z}_{v_3} = K_r K_{\anc(v_3,1)} \in [\overline{Z}].
\end{equation}
However, the $\overline{X}$ operator for $v_3$ differs from that of the previous level. It now includes an additional factor $K_{v_3}$:
\begin{equation}
	\overline{X}_{v_3} = Z_r K_{v_3} K_{\anc(v_3,2)} \in [\overline{X}].
\end{equation}

\begin{figure*}
		\begin{subfigure}[c]{0.5\textwidth}
			\def\svgwidth{\textwidth}
			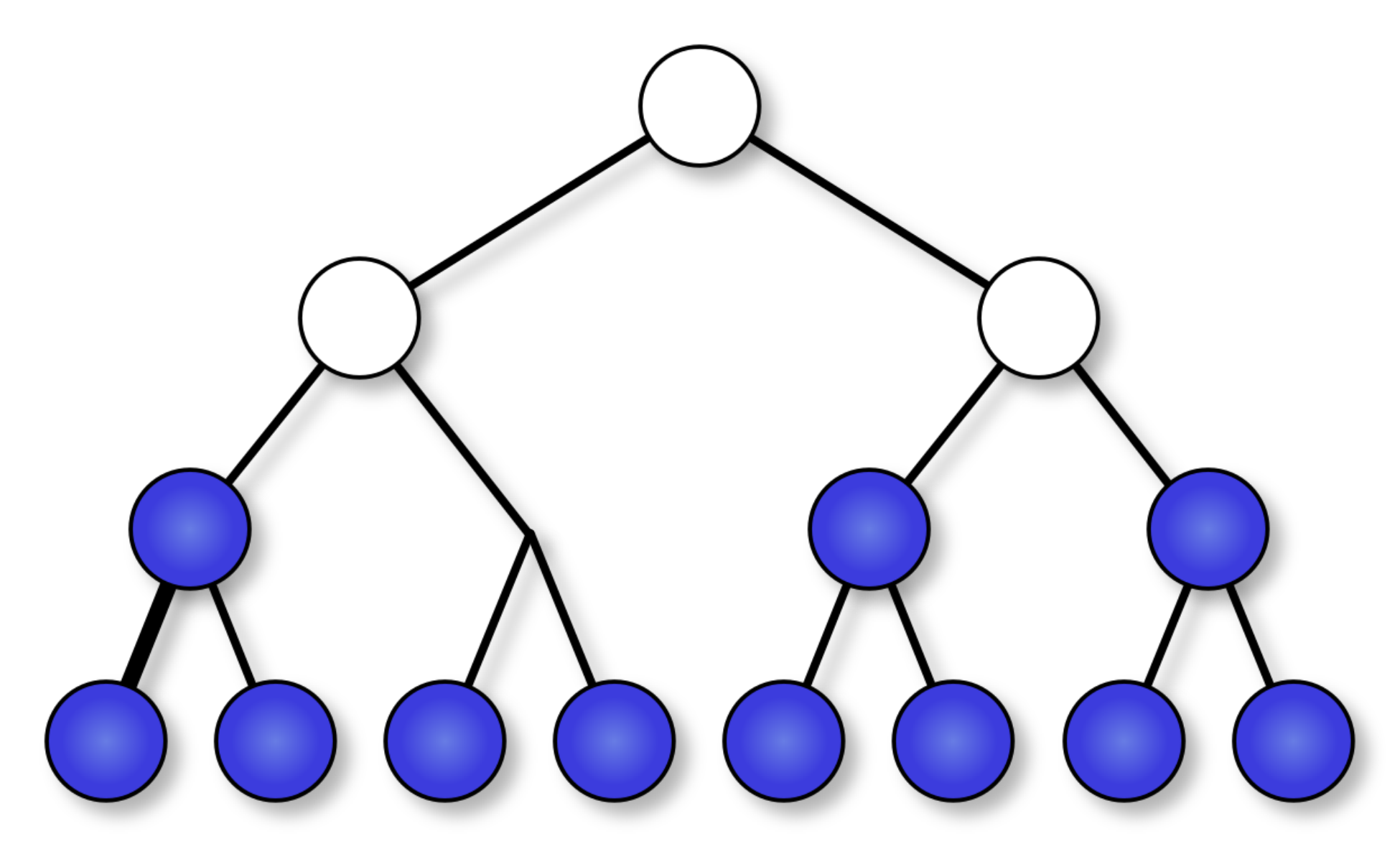			
			\caption{}
		\end{subfigure}									
		\begin{subfigure}[c]{0.5\textwidth}
			\def\svgwidth{\textwidth}
			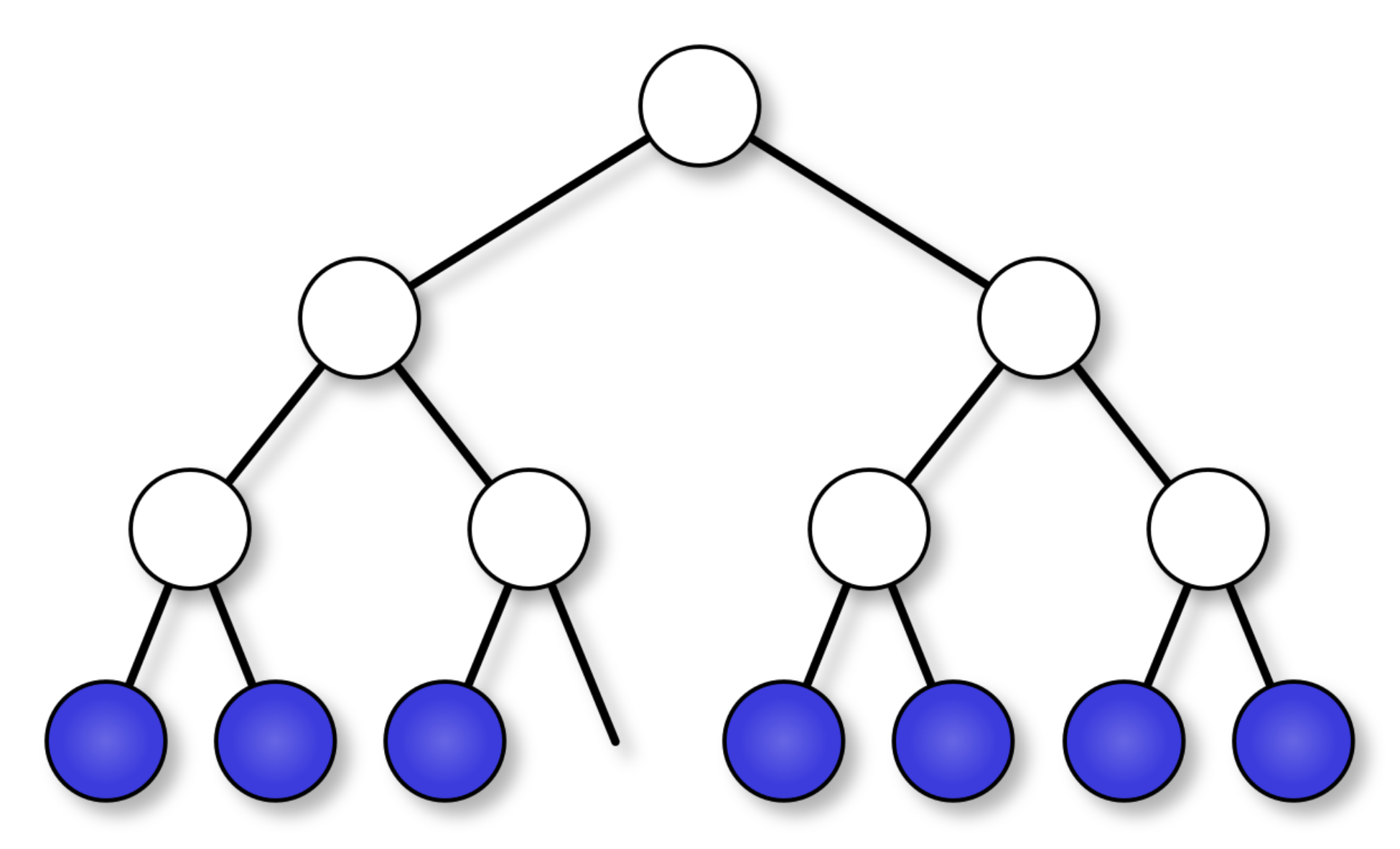			
			\caption{}
		\end{subfigure}
	\caption{\label{fig:binary-tree-23-solution}Measured logical operators in the event of a successful BM for a binary tree of height three, characterized by the branching parameters $(2,2,2)$. Panel (a) shows the scenario where the first successful BM occurs at the second level, while panel (b) depicts the case for the third level. The case for the second level is essentially identical to the solution for the binary tree of height two. Blue nodes represent qubits previously measured with a $Z$-BM. Qubits filled with both red and blue indicate a successful physical BM. Red letters denote the  measured $\overline{X}$ operator, and blue letters denote the measured $\overline{Z}$ operator.}
\end{figure*}

Building on the previous example, we now discuss the structure of the logical operators at each level for a general tree. Recall from Eqs.~\eqref{eq:tree-logical-set-x} and~\eqref{eq:tree-logical-set-z} that logical operators are represented by a base operator, $Z_r$ or $K_r$, multiplied by stabilizer generators corresponding to a set of vertices. The relevant logical operators at each level can be expressed as the product of a base operator and stabilizer generators associated with every second vertex along the path to the vertex where the successful BM occurred. For brevity, we will refer to the set of vertices on this path simply as the path. For the $\overline{X}$ operator, we include the vertices at odd levels, while for the $\overline{Z}$ operator, we include the vertices at even levels. Here and in the following including refers to the associated stabilizer being multiplied to the base operator to obtain the logical operator. For the example of a $(2,2,2)$ tree this is illustrated in Fig.~\ref{fig:binary-tree-23-logicals}.

We can now argue that the logical operators can always be obtained with probability one, based on two key observations. First, we note that the logical operators require only $X$ information along the path, excluding the final vertex of the path which is the vertex where the success occurred. This follows from the fact that the stabilizer generators are included from every second vertex. As a result, any intermediate, i.e., non-endpoint, vertex is a neighbor of exactly two included vertices. The associated stabilizer of these two included vertices each act with a $Z$ operator on the intermediate vertex. These $Z$ operators cancel out, leaving no contribution on the intermediate vertex. The $\overline{X}$ operator never requires $Z$ information on the start of the path, the root, because the base operator $Z_r$ cancels with the first included stabilizer, which is at level one. Similarly, the $\overline{Z}$ operator never requires $Z$ information on the root because the base operator acts with $X$ on the root and no adjacent vertex is included. Recall that this $X$ information on the root is always available, since we included an $X$-BM (or, alternatively, two single-qubit $X$ measurements) on the root of the tree in our code construction.

Secondly, we observe that no $X$ information is required from any vertex outside the path, as no logical operator includes a vertex which is not on the path. Thus, the logical operators do not conflict on vertices outside the path.

We conclude that once a successful BM occurs, both logical operators can be obtained with probability one by performing $X$-BMs on each vertex of the path and $Z$-BMs on the remaining vertices.
\begin{figure*}
		\begin{subfigure}[c]{\textwidth}
			\def\svgwidth{\textwidth}
\begingroup%
  \makeatletter%
  \providecommand\color[2][]{%
    \errmessage{(Inkscape) Color is used for the text in Inkscape, but the package 'color.sty' is not loaded}%
    \renewcommand\color[2][]{}%
  }%
  \providecommand\transparent[1]{%
    \errmessage{(Inkscape) Transparency is used (non-zero) for the text in Inkscape, but the package 'transparent.sty' is not loaded}%
    \renewcommand\transparent[1]{}%
  }%
  \providecommand\rotatebox[2]{#2}%
  \newcommand*\fsize{\dimexpr\f@size pt\relax}%
  \newcommand*\lineheight[1]{\fontsize{\fsize}{#1\fsize}\selectfont}%
  \ifx\svgwidth\undefined%
    \setlength{\unitlength}{2069.29133858bp}%
    \ifx\svgscale\undefined%
      \relax%
    \else%
      \setlength{\unitlength}{\unitlength * \real{\svgscale}}%
    \fi%
  \else%
    \setlength{\unitlength}{\svgwidth}%
  \fi%
  \global\let\svgwidth\undefined%
  \global\let\svgscale\undefined%
  \makeatother%
  \begin{picture}(1,0.2739726)%
    \lineheight{1}%
    \setlength\tabcolsep{0pt}%
    \put(0,0){\includegraphics[width=\unitlength,page=1]{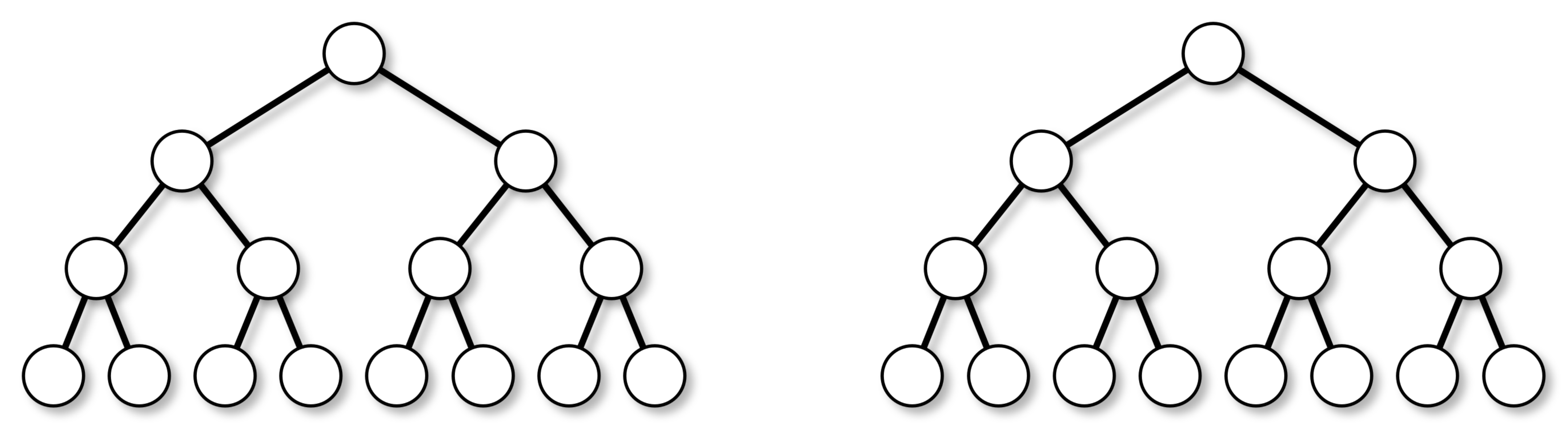}}%
    \put(0.77406825,0.23970216){\color[rgb]{0,0,0}\makebox(0,0)[t]{\lineheight{0}\smash{\begin{tabular}[t]{c}\makebox(0,0){$\color{myblue}{X}$}\end{tabular}}}}%
    \put(0.88365729,0.17120901){\color[rgb]{0,0,0}\makebox(0,0)[t]{\lineheight{0}\smash{\begin{tabular}[t]{c}\makebox(0,0){$\color{myblue}{Z}$}\end{tabular}}}}%
    \put(0.22612309,0.23970216){\color[rgb]{0,0,0}\makebox(0,0)[t]{\lineheight{0}\smash{\begin{tabular}[t]{c}\makebox(0,0){$\color{myred}{Z}$}\end{tabular}}}}%
    \put(0.66447925,0.17120901){\color[rgb]{0,0,0}\makebox(0,0)[t]{\lineheight{0}\smash{\begin{tabular}[t]{c}\makebox(0,0){$\color{myblue}{Z}$}\end{tabular}}}}%
  \end{picture}%
\endgroup%
			
			\caption{Base operators}
		\end{subfigure}									
		\begin{subfigure}[c]{\textwidth}
			\def\svgwidth{\textwidth}
			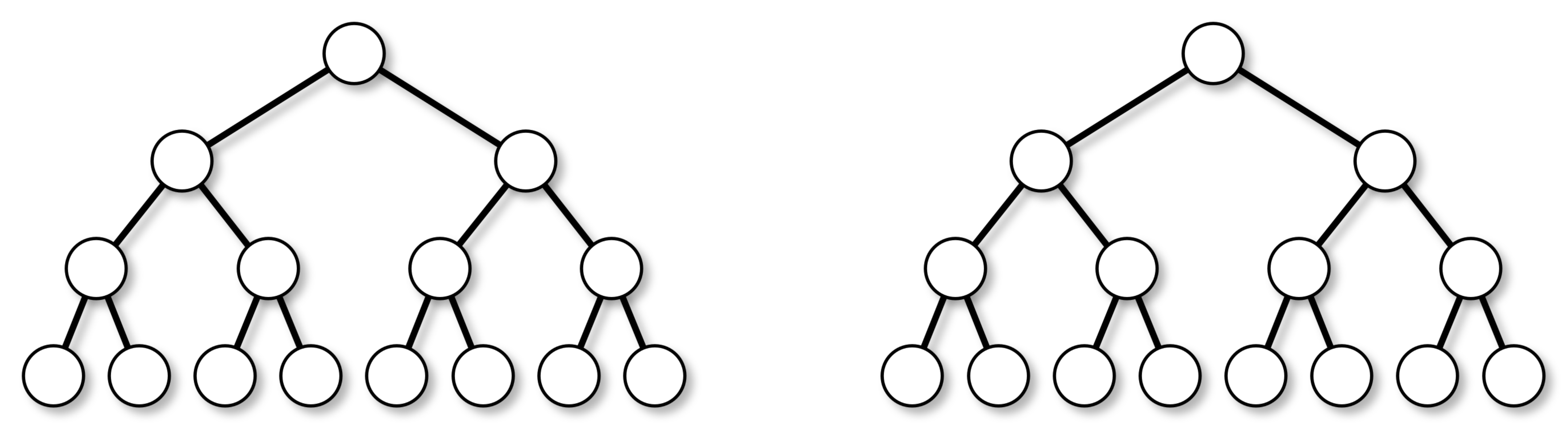			
			\caption{1st level logical operators}
		\end{subfigure}
				\begin{subfigure}[c]{\textwidth}
			\def\svgwidth{\textwidth}
			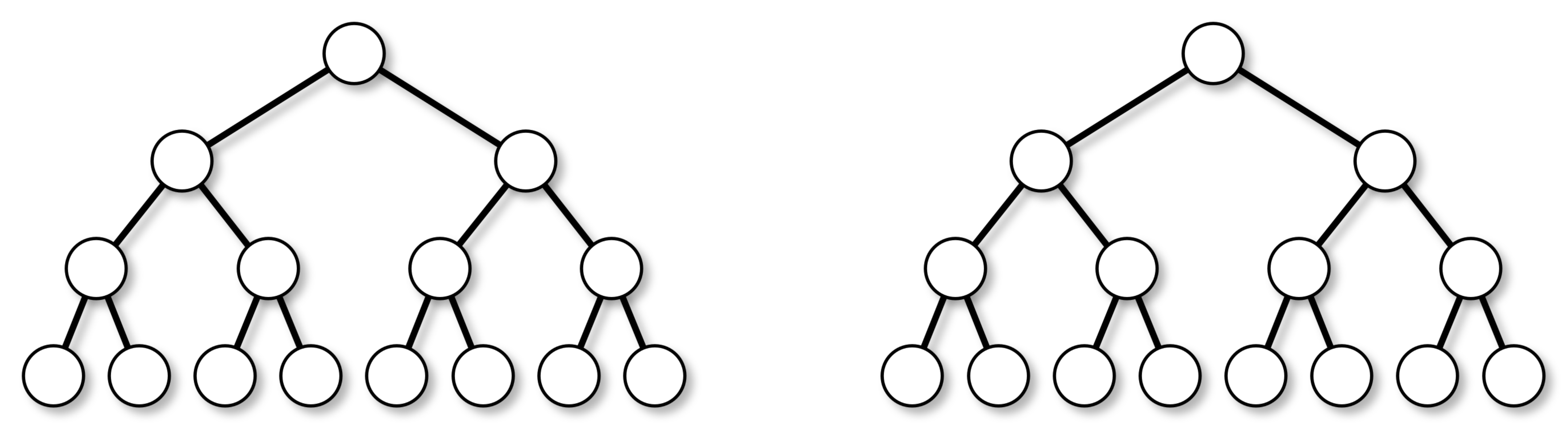			
			\caption{2nd level logical operators}
		\end{subfigure}									
		\begin{subfigure}[c]{\textwidth}
			\def\svgwidth{\textwidth}
			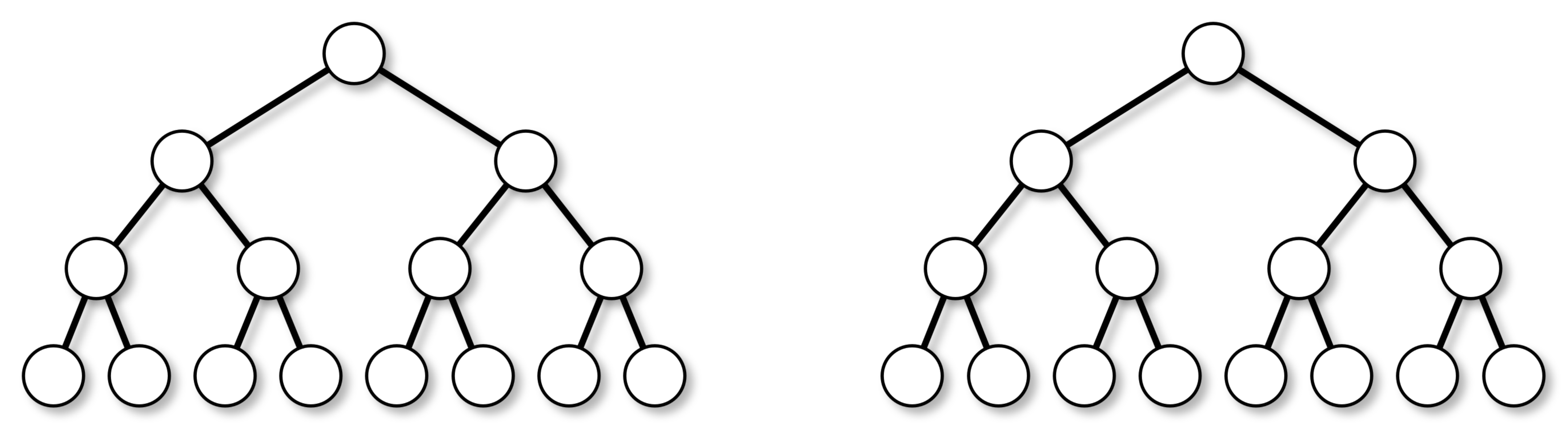			
			\caption{3rd level logical operators}
		\end{subfigure}
	\caption{\label{fig:binary-tree-23-logicals}Logical operators of a binary tree of height three, characterized by the branching parameters $(2,2,2)$. The $\overline{X}$ operators are displayed in red letters on the left, while the $\overline{Z}$ operators are displayed in blue letters on the right. Red and blue boxes around the vertices represent the stabilizer generators used to construct the $\overline{X}$ and $\overline{Z}$ operators, respectively. To illustrate this, the terms $Z^2$ are not simplified to $I$.}
\end{figure*}

It is important to emphasize that our scheme is not restricted to binary trees or trees where every node from one level has the same number of children. The scheme remains optimal for any rooted tree structure. To formulate the general scheme, we formalize the structure of the logical operators. Assume a successful BM occurred on vertex $j$ at level $\depth(j)$. As explained above, the $\overline{X}$ and $\overline{Z}$ operators are obtained by including every vertex along the path at odd and even levels, respectively:
\begin{equation}
	\overline{X}_j = 
	\begin{dcases}
		Z_r \prod_{i=0}^{\frac{\depth(j)-2}{2}} K_{\anc(j,2i+1)}	& \quad \text{if } \depth(j) \ \text{even,} \\ \\
		Z_r \prod_{i=0}^{\frac{\depth(j)-1}{2}} K_{\anc(j,2i)} 		& \quad \text{if } \depth(j) \ \text{odd,}
	\end{dcases}
	\label{eq:tree-logical-x-j}
\end{equation}
\begin{equation}
	\overline{Z}_j = 
	\begin{dcases}
		K_r \prod_{i=0}^{\frac{\depth(j)-2}{2}} K_{\anc(j,2i)}		& \quad \text{if } \depth(j) \ \text{even,} \\ \\
		K_r \prod_{i=0}^{\frac{\depth(j)-3}{2}} K_{\anc(j,2i+1)} 	& \quad \text{if } \depth(j) \ \text{odd.}
	\end{dcases}
	\label{eq:tree-logical-z-j}
\end{equation}

This generality is further illustrated by the example in Fig.~\ref{fig:large-tree}. In conclusion, the optimality for this general scheme is achieved, because the logical operators require only $X$ information along the path from the root to the vertex where the success occurred. Thus, once a successful BM occurs, both logical operators can be obtained with probability one by performing $X$-BMs on each vertex of the path and $Z$-BMs on the remaining vertices.

\begin{figure*}
	\def\svgwidth{\textwidth}
	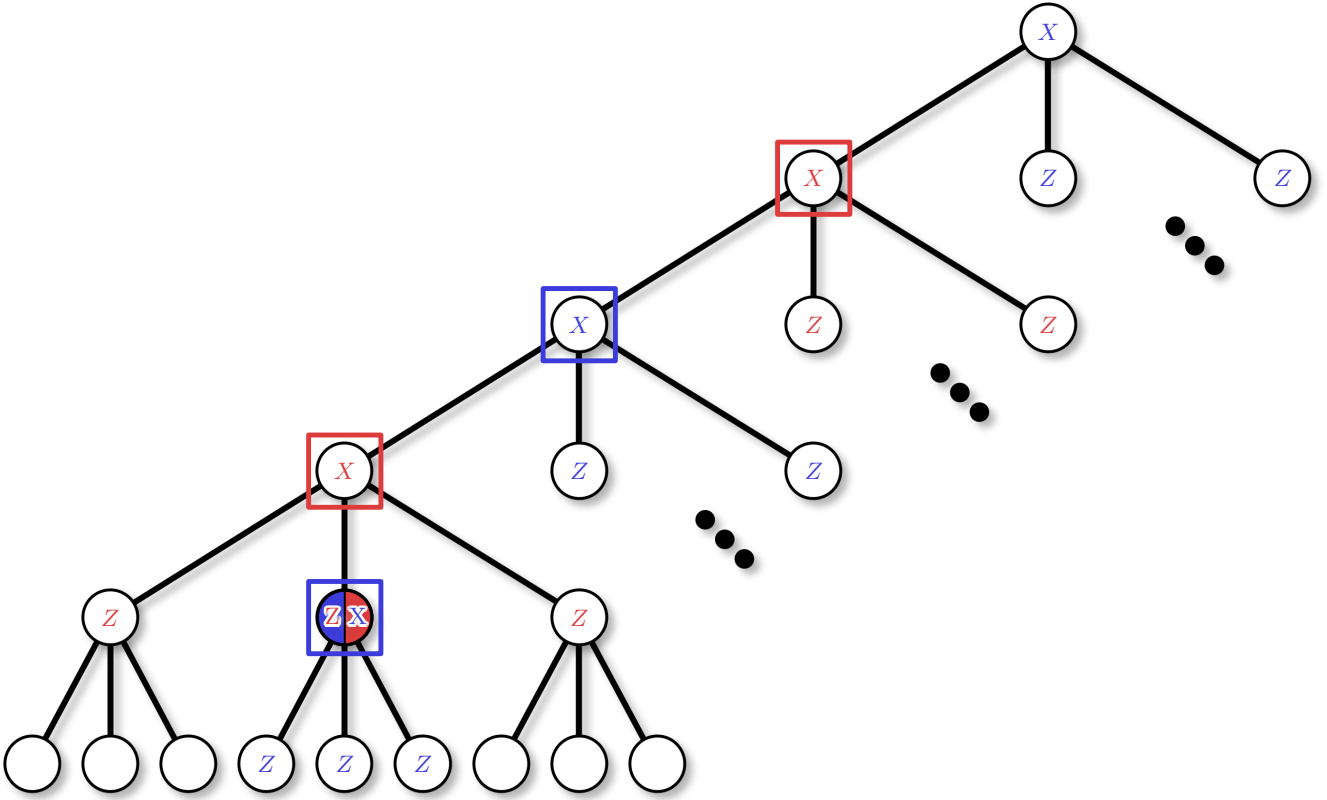	
	\caption{\label{fig:large-tree}Logical operators for a large tree, where the successful BM occurred on the fourth level. The qubit filled with both red and blue indicates a successful physical BM. Red letters denote the $\overline{X}$ operator, while blue letters denote the $\overline{Z}$ operator. Red and blue boxes around the vertices represent the stabilizer generators used to construct the $\overline{X}$ and $\overline{Z}$ operators, respectively.}
\end{figure*}

A detailed algebraic proof of the scheme's optimality, based on Thm.~\ref{thm:sufficient}, is provided in App.~\ref{app:proof-tree}.

To the best of the authors' knowledge, the only logical BM schemes for the tree code published so far appear in Refs.~\cite{Hilaire_2021} and~\cite{patil2024improveddesignallphotonicquantum}. To conclude this section we briefly compare the performance of our scheme with those presented in Refs.~\cite{Hilaire_2021} and~\cite{patil2024improveddesignallphotonicquantum}. Ref.~\cite{Hilaire_2021} introduces two schemes, namely a static and a feedforward-based scheme. The static scheme is the simplest scheme, that measures all qubits using transversal $Z$-BMs. Its success probability can be determined via a simple argument.

Recall that a static scheme is a special case of a feedforward-based scheme (see Sec.~\ref{sec:optimization-and-comparionn-of-static-logical-bell-measurements-for-rotated-planar-surface-codes}). Furthermore, similar to our feedforward-based schemes, this approach relies exclusively on transversal guaranteed partial information BMs. Consequently, the single-code reduction can be applied for this discussion.

Thus, for the sake of this argument, let us assume that all measurements on the first level of the tree are performed sequentially before any measurements on higher levels. Each single-qubit operator $Z_{v_1}$ at the first level ($\depth(v_1) = 1$) anticommutes with its corresponding stabilizer generator $K_{v_1}$, while each single-qubit operator $X_{v_1}$ and $Y_{v_1}$ anticommutes with a stabilizer generator $K_{v_2}$ from the next level ($\depth(v_2) = 2$). Therefore, using Lem.~\ref{lem:successful-bell-measurment}, we conclude that the success probability of each BM at the first level of the tree is $\mathbb{P}_B$, independent of the measurement outcomes of other qubits at the same level.

From Eq.~\eqref{eq:tree-Z-example} we immediately see that a $\overline{Z}$ operator is always measured in this scheme. Furthermore, Eq.~\eqref{eq:tree-logical-set-x} shows that any $\overline{X}$ operator contains at least one single-qubit operator, either $X_{v_1}$ or $Y_{v_1}$, at the first level ($\depth(v_1)=1$). Moreover, there always exists an $\overline{X}$ operator that requires $X$ information from only a single qubit at the first level. (Recall that $Z_r$ was excluded from the set of valid $\overline{X}$ operators to preserve the error correction properties of the code.) Therefore, the scheme succeeds if and only if at least one of the transversal BMs at the first level succeeds. Consequently, we conclude that the success probability of the static scheme is given by $1-(1-\mathbb{P}_B)^{b_0}$. This result was previously derived in Ref.~\cite{Hilaire_2021} using a recursive approach.

The feedforward-based schemes introduced in Refs.~\cite{Hilaire_2021} and~\cite{patil2024improveddesignallphotonicquantum} employ single-qubit measurements below the first level, conditioned on the outcomes of the transversal BMs at the first level. Furthermore, the feedforward-based scheme presented in Ref.~\cite{patil2024improveddesignallphotonicquantum} uses single-qubit measurements on the remaining qubits of the first level after a successful BM on the first level. While these approaches enhance loss tolerance and error robustness, the argument for the no-loss success probability of the static scheme applies to these feedforward-based schemes as well. Consequently, both the static and feedforward-based schemes in Refs.~\cite{Hilaire_2021} and~\cite{patil2024improveddesignallphotonicquantum} share the same no-loss success probability: $1-(1-\mathbb{P}_B)^{b_0}$. The success probability of our scheme for a tree consisting of $n$ nodes without the root is given by $1-(1-\mathbb{P}_B)^{n}$. Thus, in the absence of loss, our scheme significantly outperforms the schemes presented in Refs.~\cite{Hilaire_2021,patil2024improveddesignallphotonicquantum}.

\subsection{Steane code}
\label{sec:steane-code}
In this section, we introduce our measurement scheme for the seven-qubit Steane code~\cite{doi:10.1098/rspa.1996.0136}. The Steane code, displayed in Fig.~\ref{fig:steane-code}, is the smallest triangular color code~\cite{PhysRevLett.97.180501}. For an instructive and comprehensive introduction to color codes, we refer the reader to Ref.~\cite{Lidar_Brun_2013_TopologicalCodes}. A 2D color code is built from a 3-valent lattice with 3-colorable faces embedded in a closed surface. The faces are typically colored red, green, and blue.

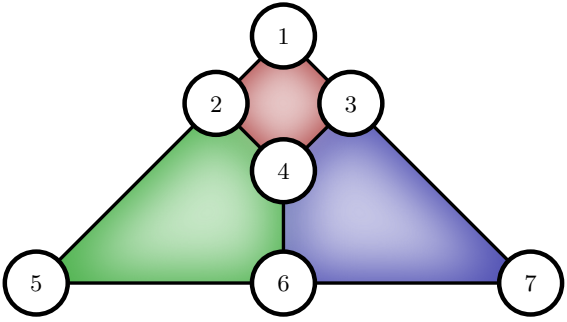
\begin{figure}
	\def\svgwidth{0.45\textwidth}
\begingroup%
  \makeatletter%
  \providecommand\color[2][]{%
    \errmessage{(Inkscape) Color is used for the text in Inkscape, but the package 'color.sty' is not loaded}%
    \renewcommand\color[2][]{}%
  }%
  \providecommand\transparent[1]{%
    \errmessage{(Inkscape) Transparency is used (non-zero) for the text in Inkscape, but the package 'transparent.sty' is not loaded}%
    \renewcommand\transparent[1]{}%
  }%
  \providecommand\rotatebox[2]{#2}%
  \newcommand*\fsize{\dimexpr\f@size pt\relax}%
  \newcommand*\lineheight[1]{\fontsize{\fsize}{#1\fsize}\selectfont}%
  \ifx\svgwidth\undefined%
    \setlength{\unitlength}{765.35433071bp}%
    \ifx\svgscale\undefined%
      \relax%
    \else%
      \setlength{\unitlength}{\unitlength * \real{\svgscale}}%
    \fi%
  \else%
    \setlength{\unitlength}{\svgwidth}%
  \fi%
  \global\let\svgwidth\undefined%
  \global\let\svgscale\undefined%
  \makeatother%
  \begin{picture}(1,0.59259259)%
    \lineheight{1}%
    \setlength\tabcolsep{0pt}%
    \put(0,0){\includegraphics[width=\unitlength,page=1]{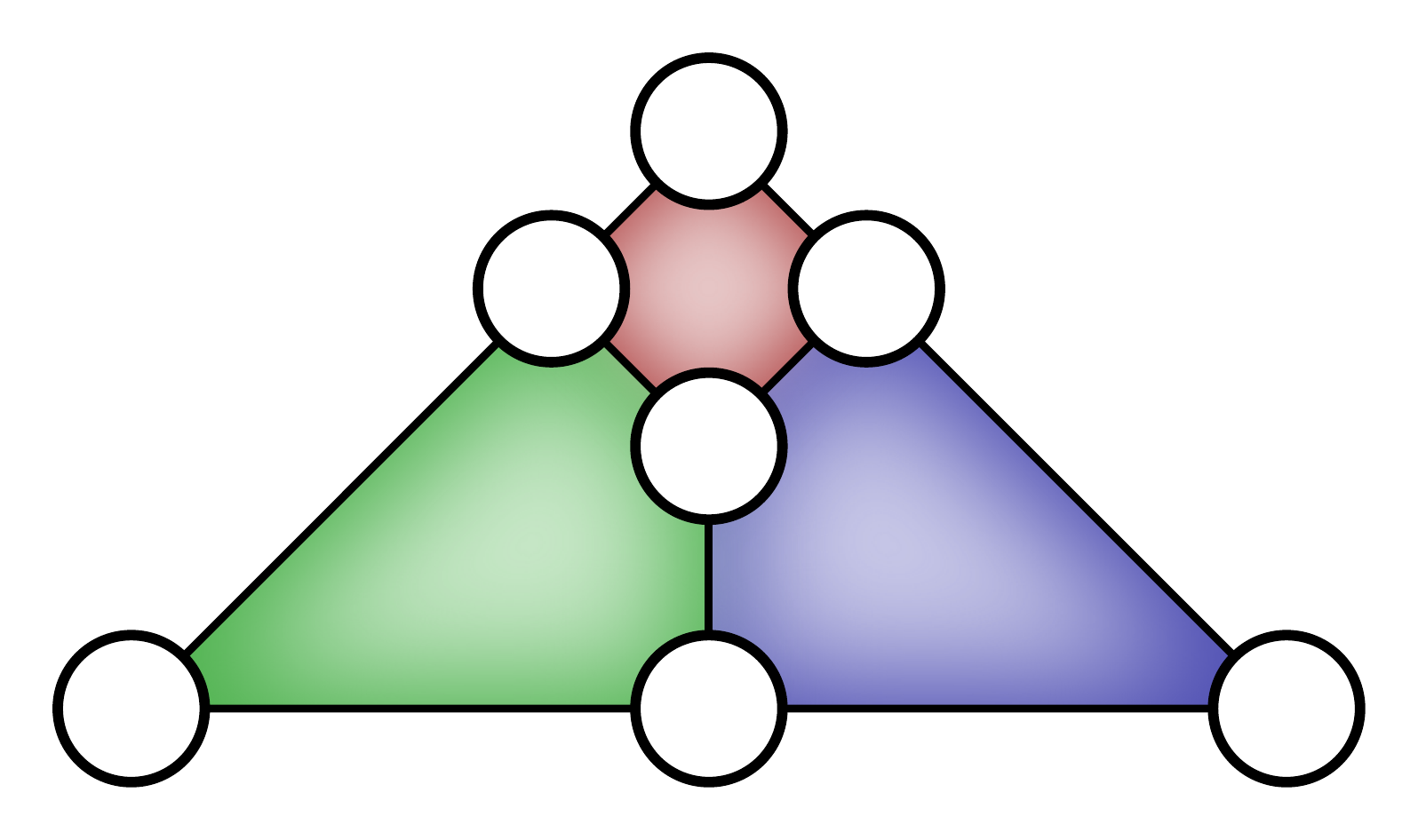}}%
    \put(0.50025872,0.49875399){\color[rgb]{0,0,0}\makebox(0,0)[t]{\lineheight{0}\smash{\begin{tabular}[t]{c}\makebox(0,0){$1$}\end{tabular}}}}%
    \put(0.50025872,0.27653176){\color[rgb]{0,0,0}\makebox(0,0)[t]{\lineheight{0}\smash{\begin{tabular}[t]{c}\makebox(0,0){$4$}\end{tabular}}}}%
    \put(0.38914758,0.38764287){\color[rgb]{0,0,0}\makebox(0,0)[t]{\lineheight{0}\smash{\begin{tabular}[t]{c}\makebox(0,0){$2$}\end{tabular}}}}%
    \put(0.61136983,0.38764287){\color[rgb]{0,0,0}\makebox(0,0)[t]{\lineheight{0}\smash{\begin{tabular}[t]{c}\makebox(0,0){$3$}\end{tabular}}}}%
    \put(0.09285128,0.09134657){\color[rgb]{0,0,0}\makebox(0,0)[t]{\lineheight{0}\smash{\begin{tabular}[t]{c}\makebox(0,0){$5$}\end{tabular}}}}%
    \put(0.50025872,0.09134657){\color[rgb]{0,0,0}\makebox(0,0)[t]{\lineheight{0}\smash{\begin{tabular}[t]{c}\makebox(0,0){$6$}\end{tabular}}}}%
    \put(0.90766612,0.09134657){\color[rgb]{0,0,0}\makebox(0,0)[t]{\lineheight{0}\smash{\begin{tabular}[t]{c}\makebox(0,0){$7$}\end{tabular}}}}%
  \end{picture}%
\endgroup%
	
	\caption{\label{fig:steane-code}Steane code. Each vertex hosts a qubit and is assigned a unique index to identify individual qubits. Each of the three faces of this triangular color code is associated with both an $X$- and a $Z$-stabilizer generator.}
\end{figure}

Triangular color codes are embedded in the orientable surface of a triangle. We place a qubit at each vertex of the lattice, and for each face $f$ of the lattice we define two stabilizer generators:
\begin{equation}
	X_f = \prod_{v \in f} X_v,
	\label{eq:X-f}
\end{equation}
and
\begin{equation}
	Z_f = \prod_{v \in f} Z_v.
	\label{eq:Z-f}
\end{equation}
The seven-qubit Steane code is defined by the combined set of the two types of stabilizer generators:
\begin{equation}
	G_c = \{ X_f \}_{f \in F} \cup \{ Z_f \}_{f \in F},
	\label{eq:steane-C}
\end{equation}
where $F$ denotes the set of all faces. From Eqs.~\eqref{eq:X-f}, \eqref{eq:Z-f}, and~\eqref{eq:steane-C} we see that color codes are symmetric under exchange of $X$ and $Z$ operators.

Logical operators in color codes are naturally associated with 0-chains, that is, subsets of the lattice vertices. Since each vertex hosts a qubit, a 0-chain specifies a set of qubits, and a logical operator is constructed by applying a single-qubit Pauli operator to each of them. Because color codes are CSS codes, $\overline{X}$ operators consist solely of tensor products of $X$ operators, and $\overline{Z}$ operators solely of tensor products of $Z$ operators. Thus, for any 0-chain that supports a valid logical operator, applying $X$ or $Z$ operators to the corresponding qubits yields a $\overline{X}$ or $\overline{Z}$ operator, respectively.

Therefore, we do not need to distinguish between $\overline{X}$ and $\overline{Z}$ operators at the level of 0-chains. We can treat all logical operators uniformly in terms of vertex subsets alone. Thus, we refer to the vertices contained in a 0-chain as the support of the 0-chain.

Although logical operators in color codes can be understood topologically, for the small Steane code it is sufficient, and instructive, to derive them by identifying one representative and generating the others through multiplication with stabilizer elements. A natural choice for a representative logical operator is the 0-chain consisting of all vertices of the code. It commutes with every stabilizer generator, since each face consists of four vertices and thus overlaps with the full vertex set in an even number of qubits. It is also not itself a stabilizer, since the $\overline{X}$ and $\overline{Z}$ operators corresponding to this 0-chain act on all qubits in the code, which is an odd number of qubits. Thus, they anticommute and cannot both be stabilizers, as all stabilizers must commute. Due to the symmetry of the code under exchange of $X$ and $Z$, neither of the two operators is a stabilizer. All logical operators of the Steane code are displayed in Fig.~\ref{fig:steane-logicals}.

\begin{figure*}
	\begin{subfigure}[c]{0.32\textwidth}
		\def\svgwidth{\textwidth}
		
		\caption{}
	\end{subfigure}
	\begin{subfigure}[c]{0.32\textwidth}
		\def\svgwidth{\textwidth}
\begingroup%
  \makeatletter%
  \providecommand\color[2][]{%
    \errmessage{(Inkscape) Color is used for the text in Inkscape, but the package 'color.sty' is not loaded}%
    \renewcommand\color[2][]{}%
  }%
  \providecommand\transparent[1]{%
    \errmessage{(Inkscape) Transparency is used (non-zero) for the text in Inkscape, but the package 'transparent.sty' is not loaded}%
    \renewcommand\transparent[1]{}%
  }%
  \providecommand\rotatebox[2]{#2}%
  \newcommand*\fsize{\dimexpr\f@size pt\relax}%
  \newcommand*\lineheight[1]{\fontsize{\fsize}{#1\fsize}\selectfont}%
  \ifx\svgwidth\undefined%
    \setlength{\unitlength}{765.35433071bp}%
    \ifx\svgscale\undefined%
      \relax%
    \else%
      \setlength{\unitlength}{\unitlength * \real{\svgscale}}%
    \fi%
  \else%
    \setlength{\unitlength}{\svgwidth}%
  \fi%
  \global\let\svgwidth\undefined%
  \global\let\svgscale\undefined%
  \makeatother%
  \begin{picture}(1,0.59259259)%
    \lineheight{1}%
    \setlength\tabcolsep{0pt}%
    \put(0,0){\includegraphics[width=\unitlength,page=1]{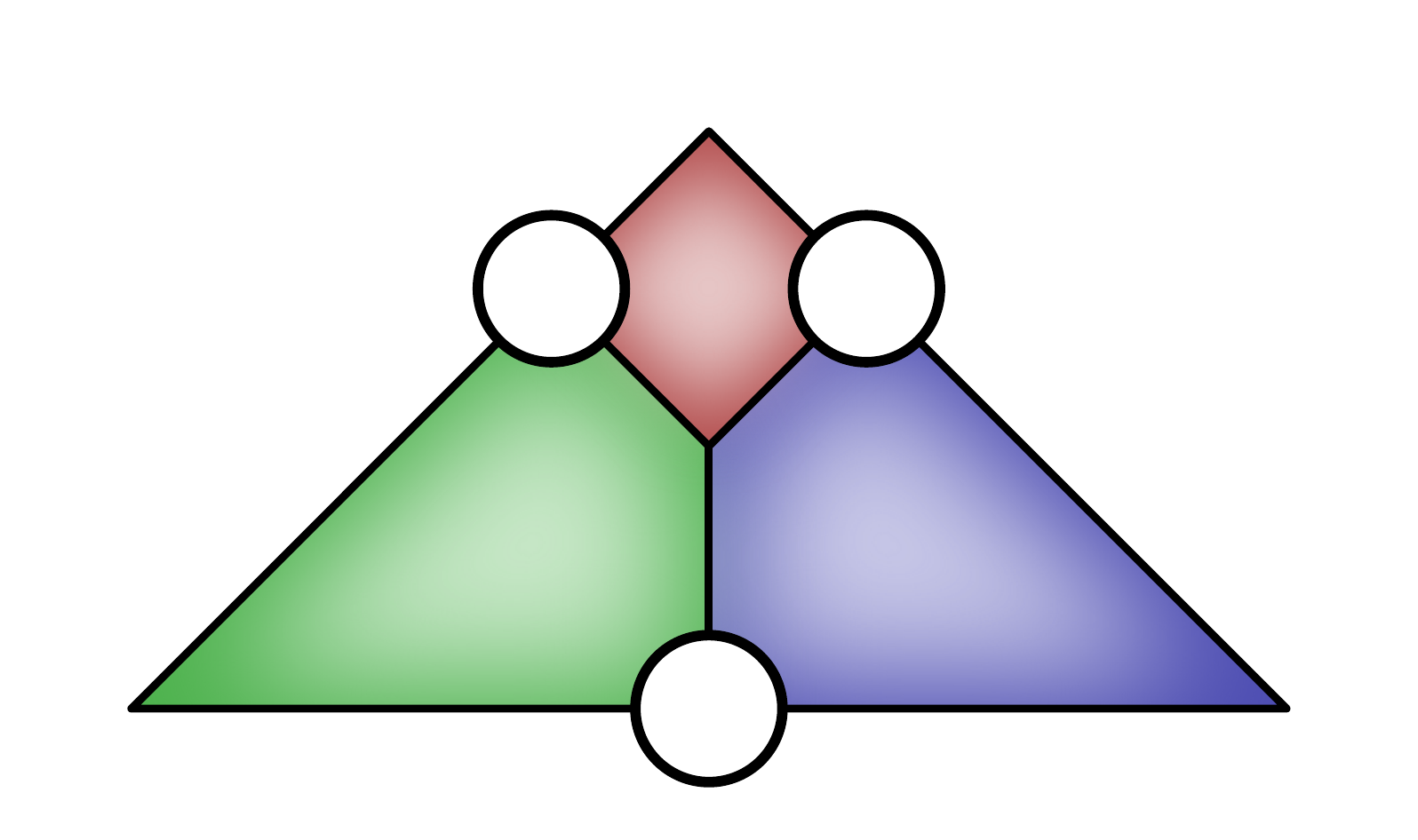}}%
    \put(0.38914758,0.38764287){\color[rgb]{0,0,0}\makebox(0,0)[t]{\lineheight{0}\smash{\begin{tabular}[t]{c}\makebox(0,0){$2$}\end{tabular}}}}%
    \put(0.61136983,0.38764287){\color[rgb]{0,0,0}\makebox(0,0)[t]{\lineheight{0}\smash{\begin{tabular}[t]{c}\makebox(0,0){$3$}\end{tabular}}}}%
    \put(0.50025872,0.09134657){\color[rgb]{0,0,0}\makebox(0,0)[t]{\lineheight{0}\smash{\begin{tabular}[t]{c}\makebox(0,0){$6$}\end{tabular}}}}%
  \end{picture}%
\endgroup%

		\caption{}
	\end{subfigure}
	\begin{subfigure}[c]{0.32\textwidth}
		\def\svgwidth{\textwidth}
\begingroup%
  \makeatletter%
  \providecommand\color[2][]{%
    \errmessage{(Inkscape) Color is used for the text in Inkscape, but the package 'color.sty' is not loaded}%
    \renewcommand\color[2][]{}%
  }%
  \providecommand\transparent[1]{%
    \errmessage{(Inkscape) Transparency is used (non-zero) for the text in Inkscape, but the package 'transparent.sty' is not loaded}%
    \renewcommand\transparent[1]{}%
  }%
  \providecommand\rotatebox[2]{#2}%
  \newcommand*\fsize{\dimexpr\f@size pt\relax}%
  \newcommand*\lineheight[1]{\fontsize{\fsize}{#1\fsize}\selectfont}%
  \ifx\svgwidth\undefined%
    \setlength{\unitlength}{765.35433071bp}%
    \ifx\svgscale\undefined%
      \relax%
    \else%
      \setlength{\unitlength}{\unitlength * \real{\svgscale}}%
    \fi%
  \else%
    \setlength{\unitlength}{\svgwidth}%
  \fi%
  \global\let\svgwidth\undefined%
  \global\let\svgscale\undefined%
  \makeatother%
  \begin{picture}(1,0.59259259)%
    \lineheight{1}%
    \setlength\tabcolsep{0pt}%
    \put(0,0){\includegraphics[width=\unitlength,page=1]{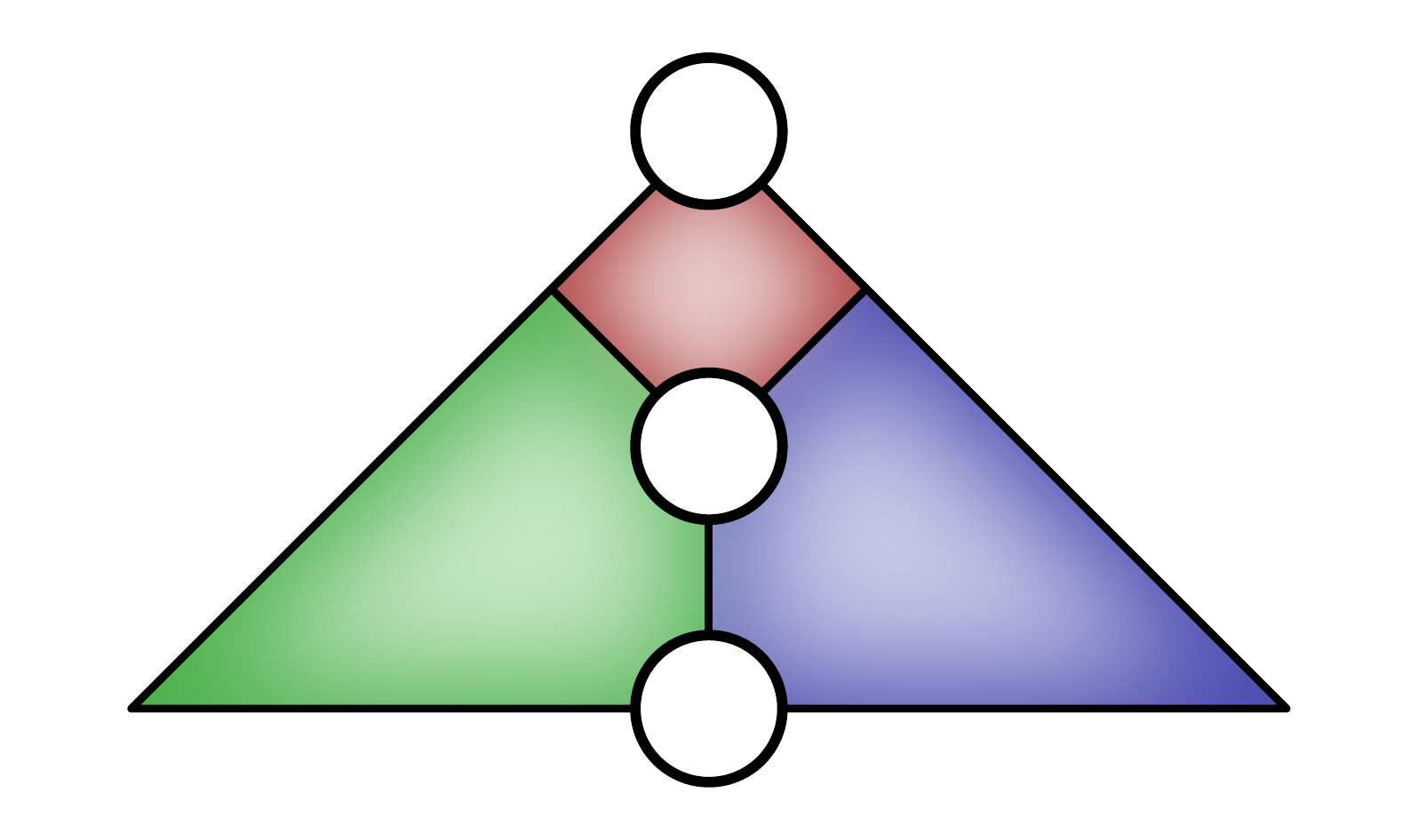}}%
    \put(0.50025872,0.49875399){\color[rgb]{0,0,0}\makebox(0,0)[t]{\lineheight{0}\smash{\begin{tabular}[t]{c}\makebox(0,0){$1$}\end{tabular}}}}%
    \put(0.50025872,0.27653176){\color[rgb]{0,0,0}\makebox(0,0)[t]{\lineheight{0}\smash{\begin{tabular}[t]{c}\makebox(0,0){$4$}\end{tabular}}}}%
    \put(0.50025872,0.09134657){\color[rgb]{0,0,0}\makebox(0,0)[t]{\lineheight{0}\smash{\begin{tabular}[t]{c}\makebox(0,0){$6$}\end{tabular}}}}%
  \end{picture}%
\endgroup%

		\caption{}
	\end{subfigure}
	\begin{subfigure}[c]{0.32\textwidth}
		\def\svgwidth{\textwidth}
\begingroup%
  \makeatletter%
  \providecommand\color[2][]{%
    \errmessage{(Inkscape) Color is used for the text in Inkscape, but the package 'color.sty' is not loaded}%
    \renewcommand\color[2][]{}%
  }%
  \providecommand\transparent[1]{%
    \errmessage{(Inkscape) Transparency is used (non-zero) for the text in Inkscape, but the package 'transparent.sty' is not loaded}%
    \renewcommand\transparent[1]{}%
  }%
  \providecommand\rotatebox[2]{#2}%
  \newcommand*\fsize{\dimexpr\f@size pt\relax}%
  \newcommand*\lineheight[1]{\fontsize{\fsize}{#1\fsize}\selectfont}%
  \ifx\svgwidth\undefined%
    \setlength{\unitlength}{765.35433071bp}%
    \ifx\svgscale\undefined%
      \relax%
    \else%
      \setlength{\unitlength}{\unitlength * \real{\svgscale}}%
    \fi%
  \else%
    \setlength{\unitlength}{\svgwidth}%
  \fi%
  \global\let\svgwidth\undefined%
  \global\let\svgscale\undefined%
  \makeatother%
  \begin{picture}(1,0.59259259)%
    \lineheight{1}%
    \setlength\tabcolsep{0pt}%
    \put(0,0){\includegraphics[width=\unitlength,page=1]{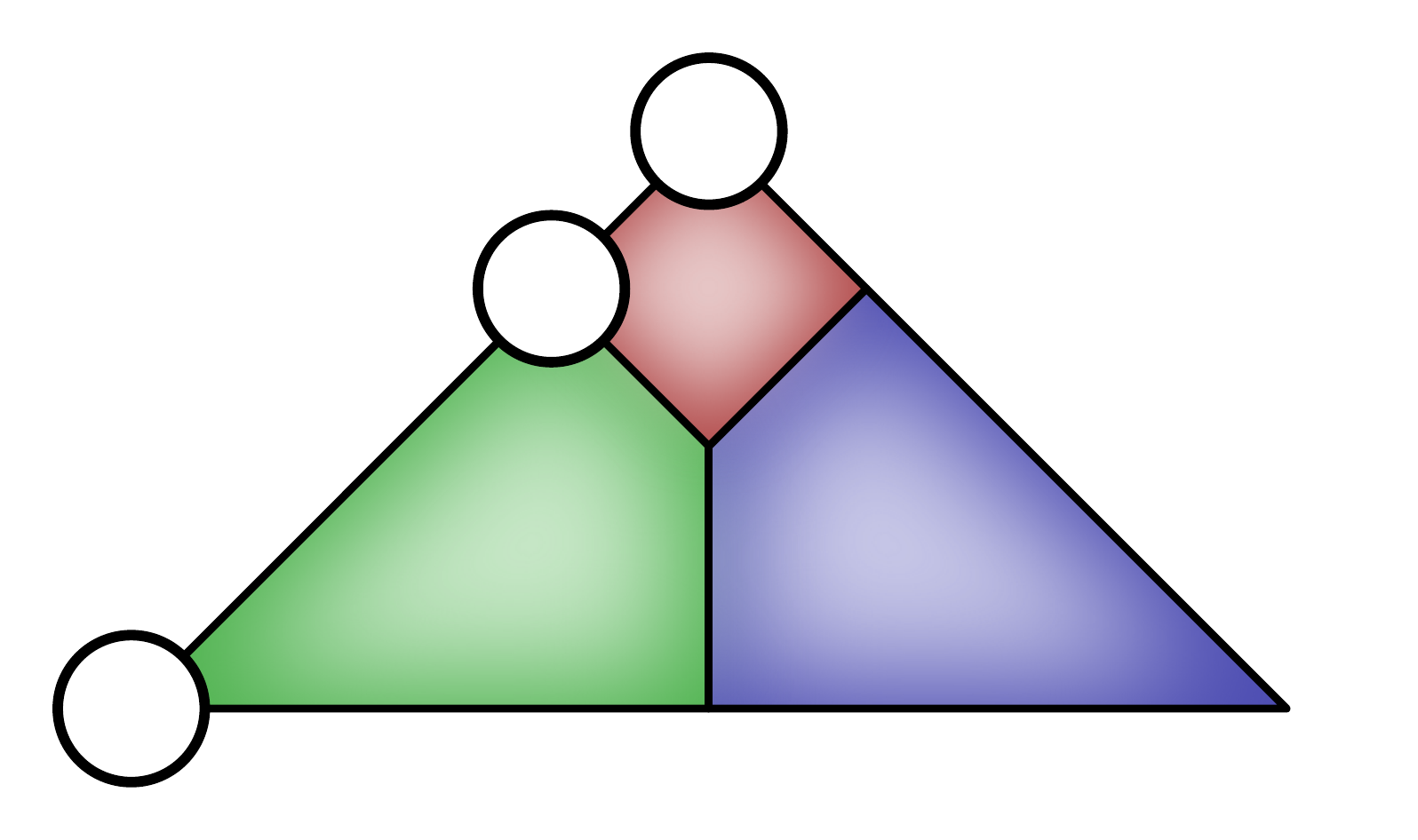}}%
    \put(0.50025872,0.49875399){\color[rgb]{0,0,0}\makebox(0,0)[t]{\lineheight{0}\smash{\begin{tabular}[t]{c}\makebox(0,0){$1$}\end{tabular}}}}%
    \put(0.38914758,0.38764287){\color[rgb]{0,0,0}\makebox(0,0)[t]{\lineheight{0}\smash{\begin{tabular}[t]{c}\makebox(0,0){$2$}\end{tabular}}}}%
    \put(0.09285128,0.09134657){\color[rgb]{0,0,0}\makebox(0,0)[t]{\lineheight{0}\smash{\begin{tabular}[t]{c}\makebox(0,0){$5$}\end{tabular}}}}%
  \end{picture}%
\endgroup%

		\caption{}
	\end{subfigure}
	\begin{subfigure}[c]{0.32\textwidth}
		\def\svgwidth{\textwidth}
\begingroup%
  \makeatletter%
  \providecommand\color[2][]{%
    \errmessage{(Inkscape) Color is used for the text in Inkscape, but the package 'color.sty' is not loaded}%
    \renewcommand\color[2][]{}%
  }%
  \providecommand\transparent[1]{%
    \errmessage{(Inkscape) Transparency is used (non-zero) for the text in Inkscape, but the package 'transparent.sty' is not loaded}%
    \renewcommand\transparent[1]{}%
  }%
  \providecommand\rotatebox[2]{#2}%
  \newcommand*\fsize{\dimexpr\f@size pt\relax}%
  \newcommand*\lineheight[1]{\fontsize{\fsize}{#1\fsize}\selectfont}%
  \ifx\svgwidth\undefined%
    \setlength{\unitlength}{765.35433071bp}%
    \ifx\svgscale\undefined%
      \relax%
    \else%
      \setlength{\unitlength}{\unitlength * \real{\svgscale}}%
    \fi%
  \else%
    \setlength{\unitlength}{\svgwidth}%
  \fi%
  \global\let\svgwidth\undefined%
  \global\let\svgscale\undefined%
  \makeatother%
  \begin{picture}(1,0.59259259)%
    \lineheight{1}%
    \setlength\tabcolsep{0pt}%
    \put(0,0){\includegraphics[width=\unitlength,page=1]{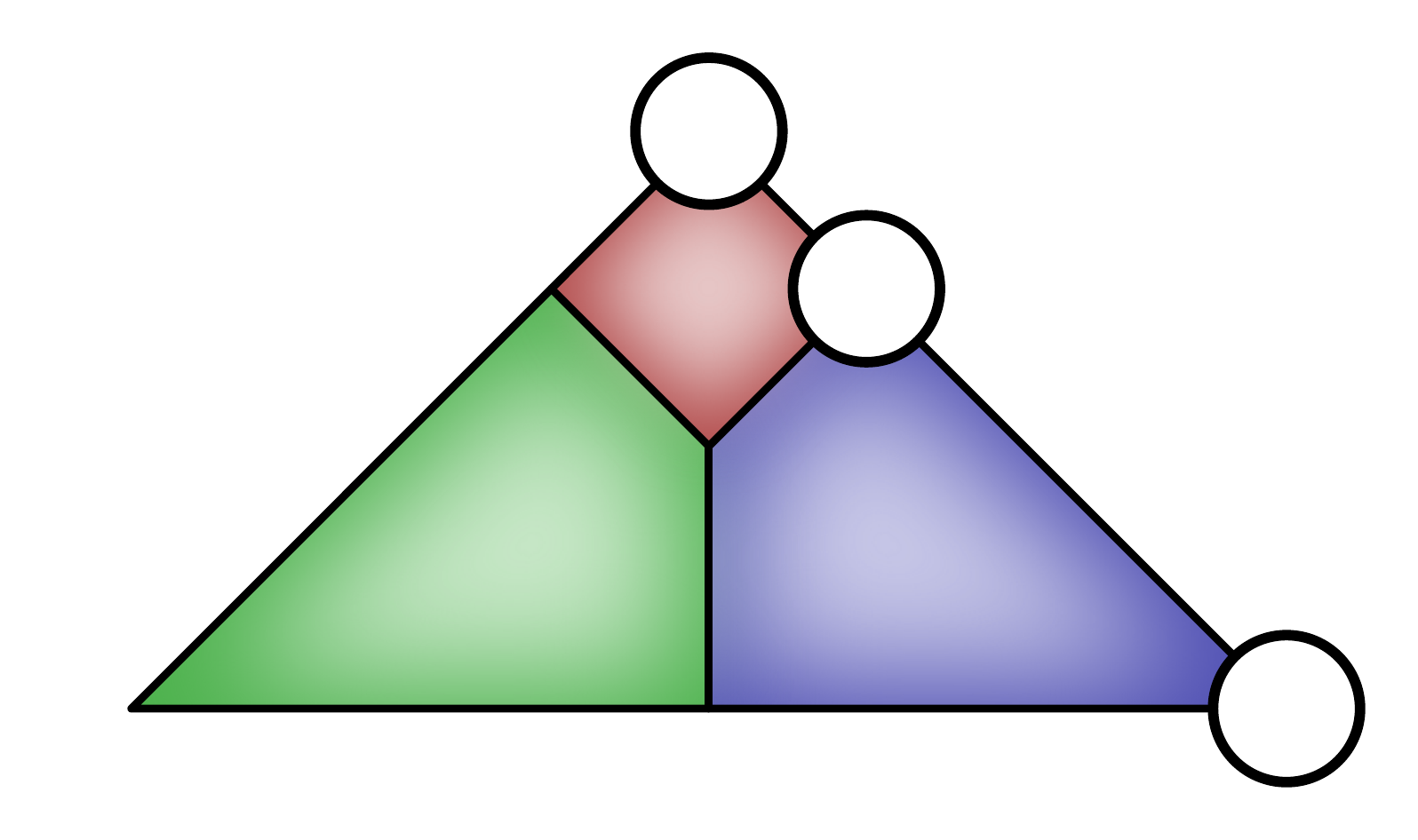}}%
    \put(0.50025872,0.49875399){\color[rgb]{0,0,0}\makebox(0,0)[t]{\lineheight{0}\smash{\begin{tabular}[t]{c}\makebox(0,0){$1$}\end{tabular}}}}%
    \put(0.61136983,0.38764287){\color[rgb]{0,0,0}\makebox(0,0)[t]{\lineheight{0}\smash{\begin{tabular}[t]{c}\makebox(0,0){$3$}\end{tabular}}}}%
    \put(0.90766612,0.09134657){\color[rgb]{0,0,0}\makebox(0,0)[t]{\lineheight{0}\smash{\begin{tabular}[t]{c}\makebox(0,0){$7$}\end{tabular}}}}%
  \end{picture}%
\endgroup%

		\caption{}
	\end{subfigure}
	\begin{subfigure}[c]{0.32\textwidth}
		\def\svgwidth{\textwidth}
\begingroup%
  \makeatletter%
  \providecommand\color[2][]{%
    \errmessage{(Inkscape) Color is used for the text in Inkscape, but the package 'color.sty' is not loaded}%
    \renewcommand\color[2][]{}%
  }%
  \providecommand\transparent[1]{%
    \errmessage{(Inkscape) Transparency is used (non-zero) for the text in Inkscape, but the package 'transparent.sty' is not loaded}%
    \renewcommand\transparent[1]{}%
  }%
  \providecommand\rotatebox[2]{#2}%
  \newcommand*\fsize{\dimexpr\f@size pt\relax}%
  \newcommand*\lineheight[1]{\fontsize{\fsize}{#1\fsize}\selectfont}%
  \ifx\svgwidth\undefined%
    \setlength{\unitlength}{765.35433071bp}%
    \ifx\svgscale\undefined%
      \relax%
    \else%
      \setlength{\unitlength}{\unitlength * \real{\svgscale}}%
    \fi%
  \else%
    \setlength{\unitlength}{\svgwidth}%
  \fi%
  \global\let\svgwidth\undefined%
  \global\let\svgscale\undefined%
  \makeatother%
  \begin{picture}(1,0.59259259)%
    \lineheight{1}%
    \setlength\tabcolsep{0pt}%
    \put(0,0){\includegraphics[width=\unitlength,page=1]{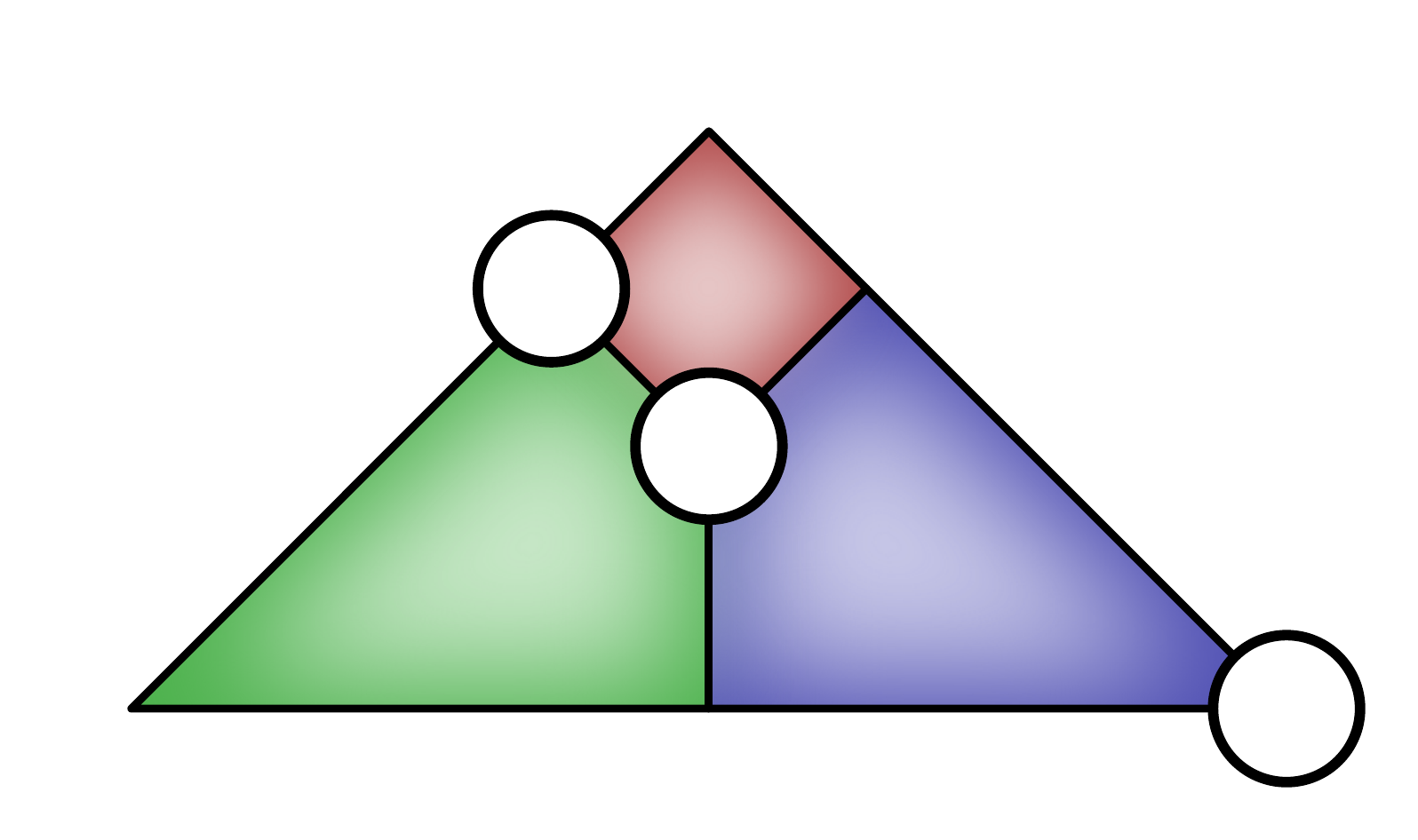}}%
    \put(0.50025872,0.27653176){\color[rgb]{0,0,0}\makebox(0,0)[t]{\lineheight{0}\smash{\begin{tabular}[t]{c}\makebox(0,0){$4$}\end{tabular}}}}%
    \put(0.38914758,0.38764287){\color[rgb]{0,0,0}\makebox(0,0)[t]{\lineheight{0}\smash{\begin{tabular}[t]{c}\makebox(0,0){$2$}\end{tabular}}}}%
    \put(0.90766612,0.09134657){\color[rgb]{0,0,0}\makebox(0,0)[t]{\lineheight{0}\smash{\begin{tabular}[t]{c}\makebox(0,0){$7$}\end{tabular}}}}%
  \end{picture}%
\endgroup%

		\caption{}
	\end{subfigure}
	\begin{subfigure}[c]{0.32\textwidth}
		\def\svgwidth{\textwidth}
\begingroup%
  \makeatletter%
  \providecommand\color[2][]{%
    \errmessage{(Inkscape) Color is used for the text in Inkscape, but the package 'color.sty' is not loaded}%
    \renewcommand\color[2][]{}%
  }%
  \providecommand\transparent[1]{%
    \errmessage{(Inkscape) Transparency is used (non-zero) for the text in Inkscape, but the package 'transparent.sty' is not loaded}%
    \renewcommand\transparent[1]{}%
  }%
  \providecommand\rotatebox[2]{#2}%
  \newcommand*\fsize{\dimexpr\f@size pt\relax}%
  \newcommand*\lineheight[1]{\fontsize{\fsize}{#1\fsize}\selectfont}%
  \ifx\svgwidth\undefined%
    \setlength{\unitlength}{765.35433071bp}%
    \ifx\svgscale\undefined%
      \relax%
    \else%
      \setlength{\unitlength}{\unitlength * \real{\svgscale}}%
    \fi%
  \else%
    \setlength{\unitlength}{\svgwidth}%
  \fi%
  \global\let\svgwidth\undefined%
  \global\let\svgscale\undefined%
  \makeatother%
  \begin{picture}(1,0.59259259)%
    \lineheight{1}%
    \setlength\tabcolsep{0pt}%
    \put(0,0){\includegraphics[width=\unitlength,page=1]{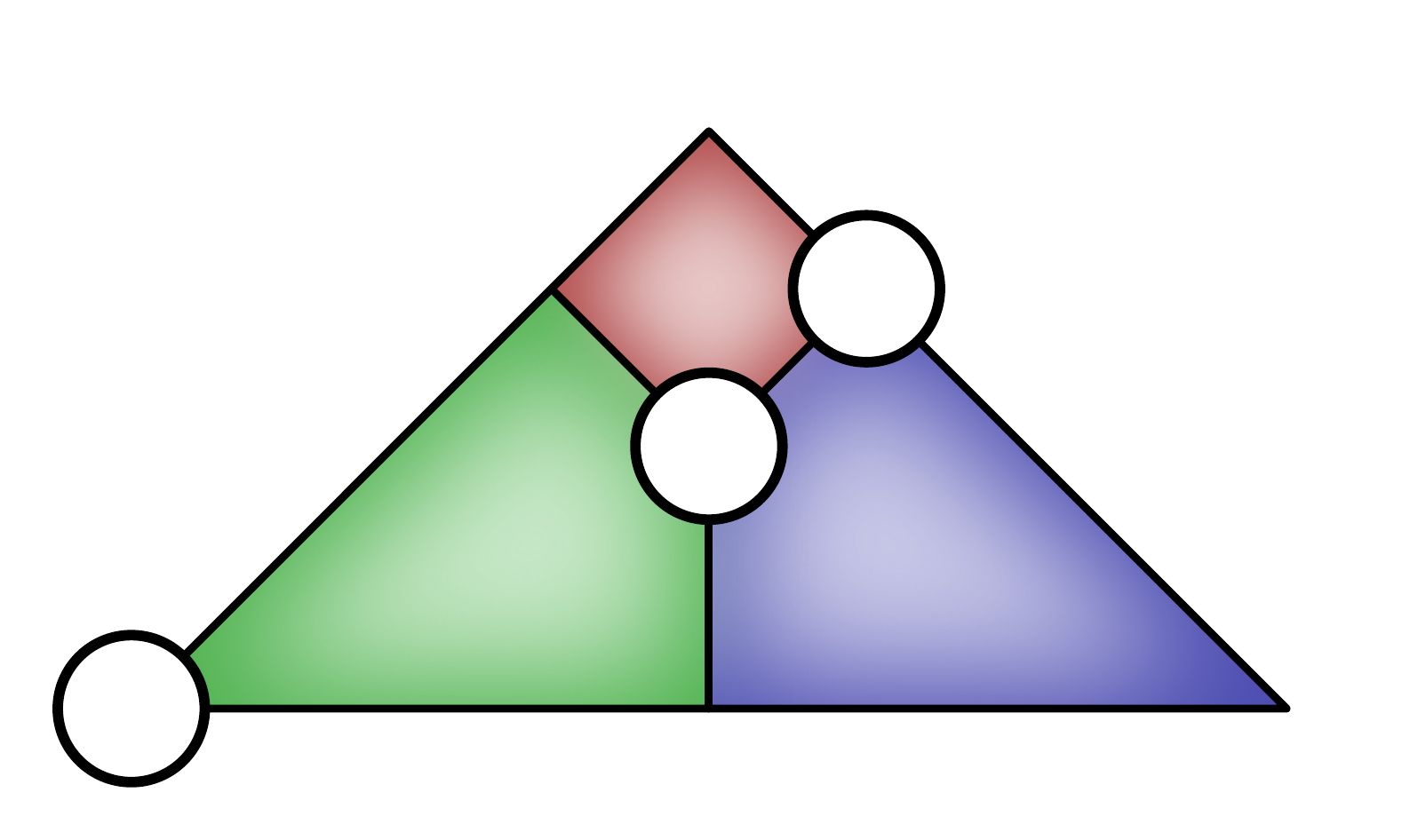}}%
    \put(0.50025872,0.27653176){\color[rgb]{0,0,0}\makebox(0,0)[t]{\lineheight{0}\smash{\begin{tabular}[t]{c}\makebox(0,0){$4$}\end{tabular}}}}%
    \put(0.61136983,0.38764287){\color[rgb]{0,0,0}\makebox(0,0)[t]{\lineheight{0}\smash{\begin{tabular}[t]{c}\makebox(0,0){$3$}\end{tabular}}}}%
    \put(0.09285128,0.09134657){\color[rgb]{0,0,0}\makebox(0,0)[t]{\lineheight{0}\smash{\begin{tabular}[t]{c}\makebox(0,0){$5$}\end{tabular}}}}%
  \end{picture}%
\endgroup%

		\caption{}
	\end{subfigure}
	\begin{subfigure}[c]{0.32\textwidth}
		\def\svgwidth{\textwidth}
\begingroup%
  \makeatletter%
  \providecommand\color[2][]{%
    \errmessage{(Inkscape) Color is used for the text in Inkscape, but the package 'color.sty' is not loaded}%
    \renewcommand\color[2][]{}%
  }%
  \providecommand\transparent[1]{%
    \errmessage{(Inkscape) Transparency is used (non-zero) for the text in Inkscape, but the package 'transparent.sty' is not loaded}%
    \renewcommand\transparent[1]{}%
  }%
  \providecommand\rotatebox[2]{#2}%
  \newcommand*\fsize{\dimexpr\f@size pt\relax}%
  \newcommand*\lineheight[1]{\fontsize{\fsize}{#1\fsize}\selectfont}%
  \ifx\svgwidth\undefined%
    \setlength{\unitlength}{765.35433071bp}%
    \ifx\svgscale\undefined%
      \relax%
    \else%
      \setlength{\unitlength}{\unitlength * \real{\svgscale}}%
    \fi%
  \else%
    \setlength{\unitlength}{\svgwidth}%
  \fi%
  \global\let\svgwidth\undefined%
  \global\let\svgscale\undefined%
  \makeatother%
  \begin{picture}(1,0.59259259)%
    \lineheight{1}%
    \setlength\tabcolsep{0pt}%
    \put(0,0){\includegraphics[width=\unitlength,page=1]{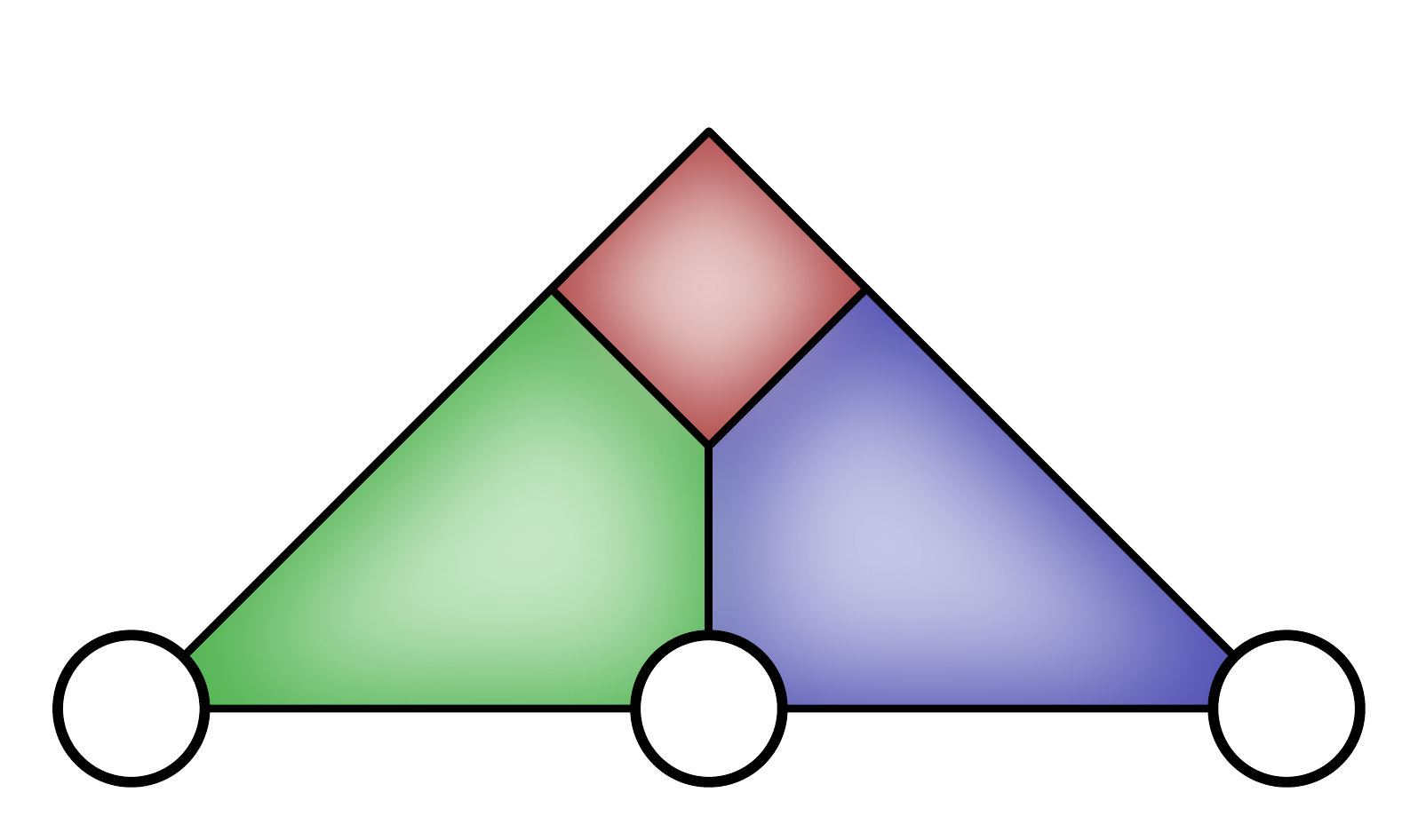}}%
    \put(0.09285128,0.09134657){\color[rgb]{0,0,0}\makebox(0,0)[t]{\lineheight{0}\smash{\begin{tabular}[t]{c}\makebox(0,0){$5$}\end{tabular}}}}%
    \put(0.50025872,0.09134657){\color[rgb]{0,0,0}\makebox(0,0)[t]{\lineheight{0}\smash{\begin{tabular}[t]{c}\makebox(0,0){$6$}\end{tabular}}}}%
    \put(0.90766612,0.09134657){\color[rgb]{0,0,0}\makebox(0,0)[t]{\lineheight{0}\smash{\begin{tabular}[t]{c}\makebox(0,0){$7$}\end{tabular}}}}%
  \end{picture}%
\endgroup%

		\caption{}
	\end{subfigure}
	\caption{\label{fig:steane-logicals}All $2^3 = 8$ 0-chains that represent logical operators in the Steane code. The chains are labeled from (a) to (h) for later reference.}
\end{figure*}
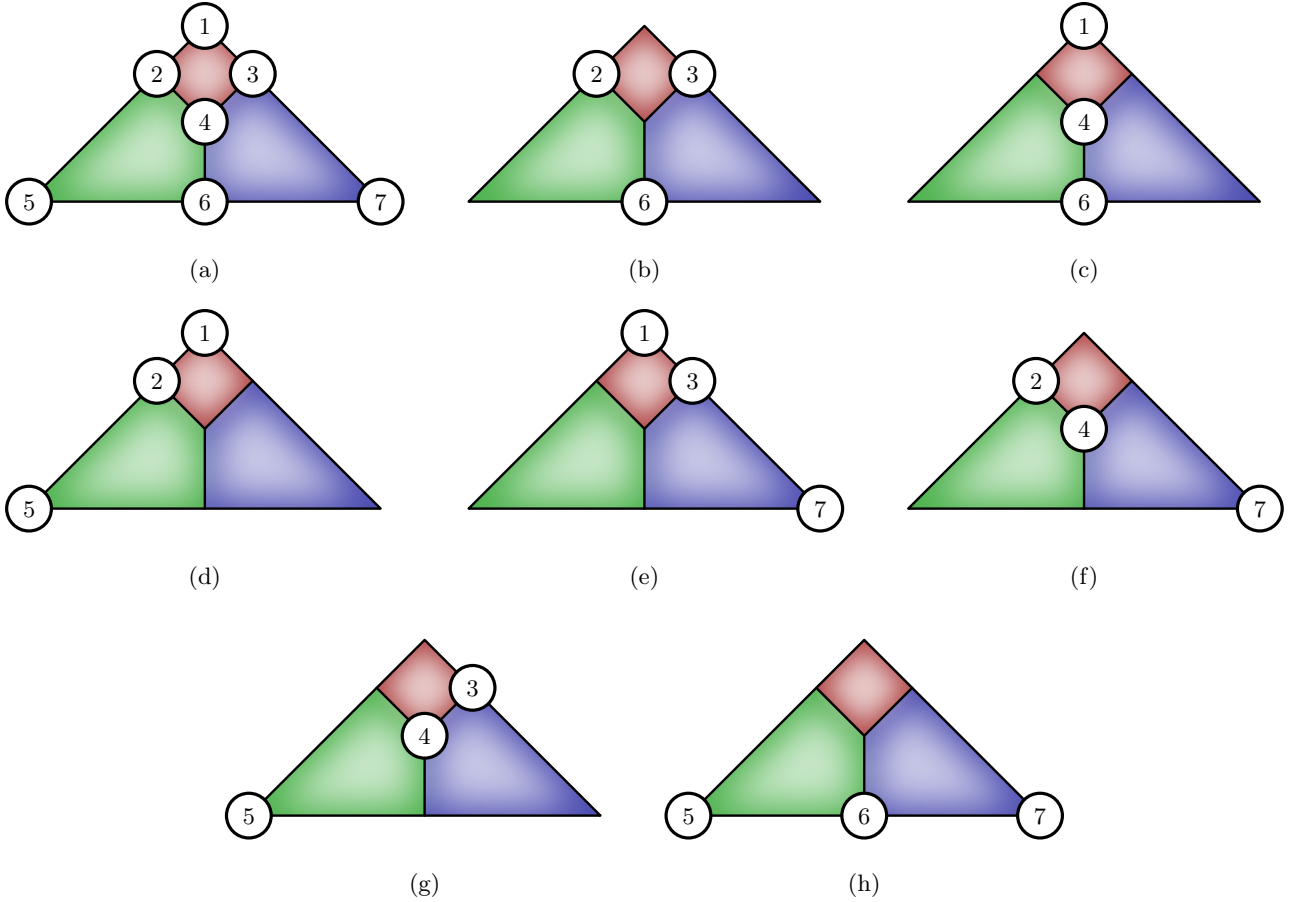

Next, we describe our measurement scheme, illustrated in Fig.~\ref{fig:steane-scheme}. We label individual qubits according to the numbering introduced in Fig.~\ref{fig:steane-code}. The protocol begins with an $X$-BM on the tip at the top of the triangle, corresponding to qubit~$1$. This is followed by $Z$-BMs on the remaining three qubits of the red face, namely qubits~$2$, $3$, and~$4$. Finally, $X$-BMs are performed on the bottom three qubits, i.e., qubits~$5$, $6$, and~$7$.

\begin{figure*}
	\begin{subfigure}[c]{0.32\textwidth}
		\def\svgwidth{\textwidth}
\begingroup%
  \makeatletter%
  \providecommand\color[2][]{%
    \errmessage{(Inkscape) Color is used for the text in Inkscape, but the package 'color.sty' is not loaded}%
    \renewcommand\color[2][]{}%
  }%
  \providecommand\transparent[1]{%
    \errmessage{(Inkscape) Transparency is used (non-zero) for the text in Inkscape, but the package 'transparent.sty' is not loaded}%
    \renewcommand\transparent[1]{}%
  }%
  \providecommand\rotatebox[2]{#2}%
  \newcommand*\fsize{\dimexpr\f@size pt\relax}%
  \newcommand*\lineheight[1]{\fontsize{\fsize}{#1\fsize}\selectfont}%
  \ifx\svgwidth\undefined%
    \setlength{\unitlength}{765.35433071bp}%
    \ifx\svgscale\undefined%
      \relax%
    \else%
      \setlength{\unitlength}{\unitlength * \real{\svgscale}}%
    \fi%
  \else%
    \setlength{\unitlength}{\svgwidth}%
  \fi%
  \global\let\svgwidth\undefined%
  \global\let\svgscale\undefined%
  \makeatother%
  \begin{picture}(1,0.59259259)%
    \lineheight{1}%
    \setlength\tabcolsep{0pt}%
    \put(0,0){\includegraphics[width=\unitlength,page=1]{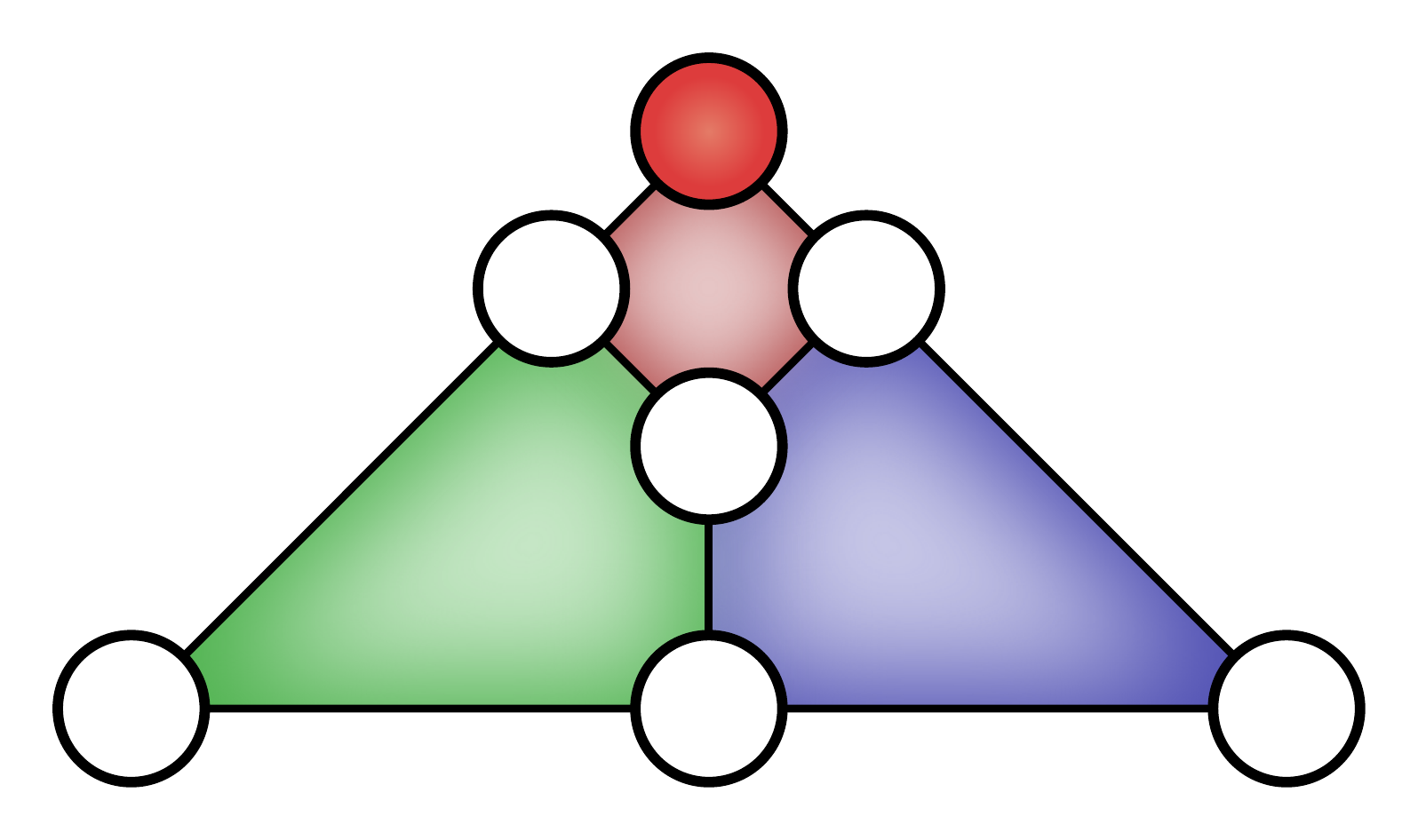}}%
    \put(0.50025872,0.49875399){\color[rgb]{0,0,0}\makebox(0,0)[t]{\lineheight{0}\smash{\begin{tabular}[t]{c}\makebox(0,0){$\contour{white}{1}$}\end{tabular}}}}%
    \put(0.50025872,0.27653176){\color[rgb]{0,0,0}\makebox(0,0)[t]{\lineheight{0}\smash{\begin{tabular}[t]{c}\makebox(0,0){$4$}\end{tabular}}}}%
    \put(0.38914758,0.38764287){\color[rgb]{0,0,0}\makebox(0,0)[t]{\lineheight{0}\smash{\begin{tabular}[t]{c}\makebox(0,0){$2$}\end{tabular}}}}%
    \put(0.61136983,0.38764287){\color[rgb]{0,0,0}\makebox(0,0)[t]{\lineheight{0}\smash{\begin{tabular}[t]{c}\makebox(0,0){$3$}\end{tabular}}}}%
    \put(0.09285128,0.09134657){\color[rgb]{0,0,0}\makebox(0,0)[t]{\lineheight{0}\smash{\begin{tabular}[t]{c}\makebox(0,0){$5$}\end{tabular}}}}%
    \put(0.50025872,0.09134657){\color[rgb]{0,0,0}\makebox(0,0)[t]{\lineheight{0}\smash{\begin{tabular}[t]{c}\makebox(0,0){$6$}\end{tabular}}}}%
    \put(0.90766612,0.09134657){\color[rgb]{0,0,0}\makebox(0,0)[t]{\lineheight{0}\smash{\begin{tabular}[t]{c}\makebox(0,0){$7$}\end{tabular}}}}%
  \end{picture}%
\endgroup%

		\caption{}
	\end{subfigure}
	\begin{subfigure}[c]{0.32\textwidth}
		\def\svgwidth{\textwidth}
\begingroup%
  \makeatletter%
  \providecommand\color[2][]{%
    \errmessage{(Inkscape) Color is used for the text in Inkscape, but the package 'color.sty' is not loaded}%
    \renewcommand\color[2][]{}%
  }%
  \providecommand\transparent[1]{%
    \errmessage{(Inkscape) Transparency is used (non-zero) for the text in Inkscape, but the package 'transparent.sty' is not loaded}%
    \renewcommand\transparent[1]{}%
  }%
  \providecommand\rotatebox[2]{#2}%
  \newcommand*\fsize{\dimexpr\f@size pt\relax}%
  \newcommand*\lineheight[1]{\fontsize{\fsize}{#1\fsize}\selectfont}%
  \ifx\svgwidth\undefined%
    \setlength{\unitlength}{765.35433071bp}%
    \ifx\svgscale\undefined%
      \relax%
    \else%
      \setlength{\unitlength}{\unitlength * \real{\svgscale}}%
    \fi%
  \else%
    \setlength{\unitlength}{\svgwidth}%
  \fi%
  \global\let\svgwidth\undefined%
  \global\let\svgscale\undefined%
  \makeatother%
  \begin{picture}(1,0.59259259)%
    \lineheight{1}%
    \setlength\tabcolsep{0pt}%
    \put(0,0){\includegraphics[width=\unitlength,page=1]{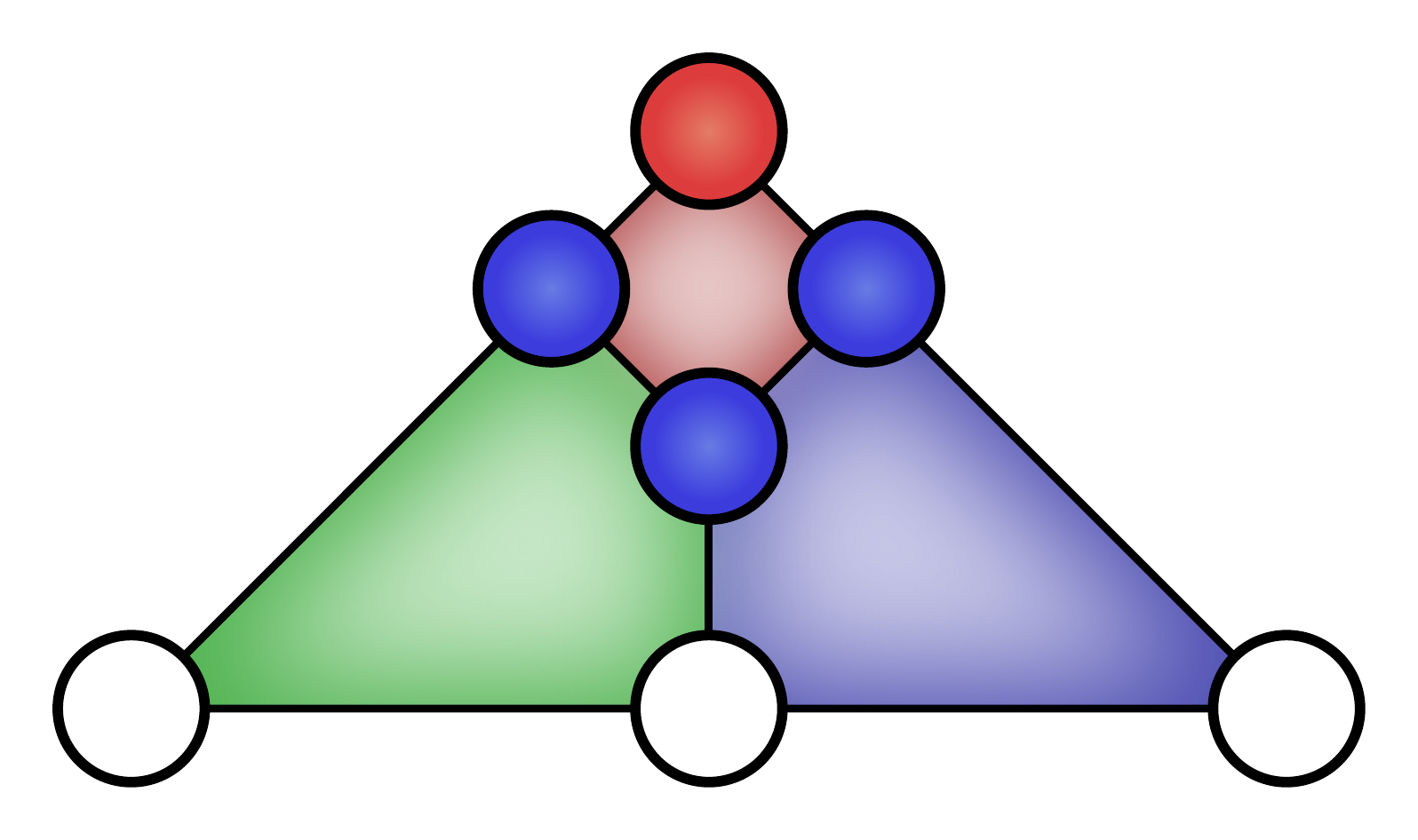}}%
    \put(0.50025872,0.49875399){\color[rgb]{0,0,0}\makebox(0,0)[t]{\lineheight{0}\smash{\begin{tabular}[t]{c}\makebox(0,0){$\contour{white}{1}$}\end{tabular}}}}%
    \put(0.50025872,0.27653176){\color[rgb]{0,0,0}\makebox(0,0)[t]{\lineheight{0}\smash{\begin{tabular}[t]{c}\makebox(0,0){$\contour{white}{4}$}\end{tabular}}}}%
    \put(0.38914758,0.38764287){\color[rgb]{0,0,0}\makebox(0,0)[t]{\lineheight{0}\smash{\begin{tabular}[t]{c}\makebox(0,0){$\contour{white}{2}$}\end{tabular}}}}%
    \put(0.61136983,0.38764287){\color[rgb]{0,0,0}\makebox(0,0)[t]{\lineheight{0}\smash{\begin{tabular}[t]{c}\makebox(0,0){$\contour{white}{3}$}\end{tabular}}}}%
    \put(0.09285128,0.09134657){\color[rgb]{0,0,0}\makebox(0,0)[t]{\lineheight{0}\smash{\begin{tabular}[t]{c}\makebox(0,0){$5$}\end{tabular}}}}%
    \put(0.50025872,0.09134657){\color[rgb]{0,0,0}\makebox(0,0)[t]{\lineheight{0}\smash{\begin{tabular}[t]{c}\makebox(0,0){$6$}\end{tabular}}}}%
    \put(0.90766612,0.09134657){\color[rgb]{0,0,0}\makebox(0,0)[t]{\lineheight{0}\smash{\begin{tabular}[t]{c}\makebox(0,0){$7$}\end{tabular}}}}%
  \end{picture}%
\endgroup%

		\caption{}
	\end{subfigure}
	\begin{subfigure}[c]{0.32\textwidth}
		\def\svgwidth{\textwidth}
		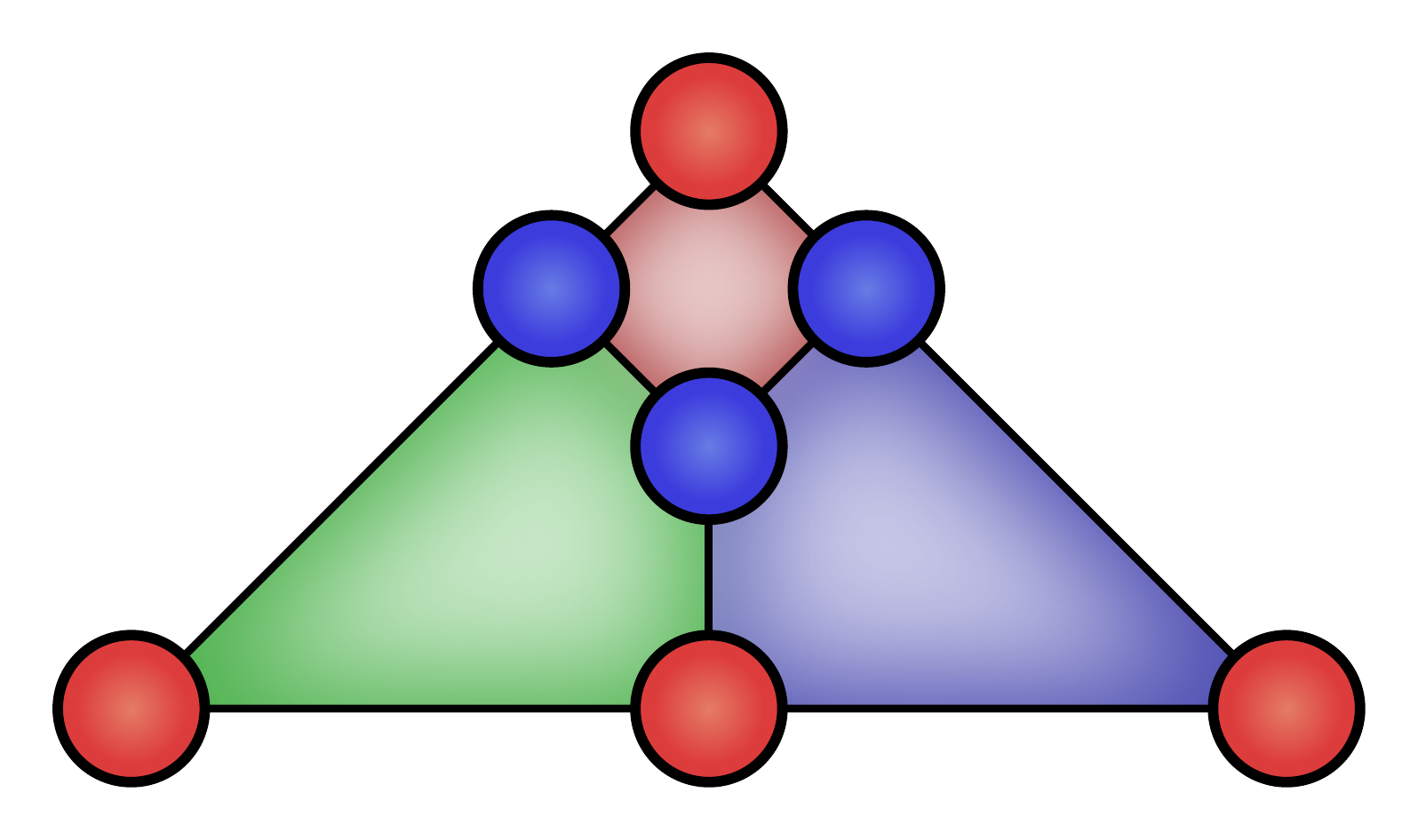
		\caption{}
	\end{subfigure}
	\caption{\label{fig:steane-scheme}Measurement scheme for the Steane code. $X$-BMs and $Z$-BMs are indicated by red and blue vertices, respectively. The three steps of the scheme are: (a) measuring the first qubit, at the tip of the triangle using an $X$-BM; (b) measuring the remaining three qubits of the red plaquette with $Z$-BMs; (c) measuring the remaining three qubits at the bottom with $X$-BMs.}
\end{figure*}

In the following, we explain how the logical operators $\overline{X}$ and $\overline{Z}$ are measured if any of the transversal BMs succeeds, as illustrated in Fig.~\ref{fig:steane-sol}. We refer to the logical operators by their label in Fig.~\ref{fig:steane-logicals}; for example, logical operator~(e) corresponds to the 0-chain consisting of vertices $1$, $3$, and $7$, i.e., the logical operators $X_1X_3X_7 \in [\overline{X}]$ and $Z_1Z_3Z_7 \in [\overline{Z}]$. Thus, we refer to these two operators as $\overline{X}$~(e) and $\overline{Z}$~(e), respectively.

The first measurement is an $X$-BM on the tip at the top of the triangle, i.e., qubit~$1$. Three logical 0-chains have support on this qubit: (c), (d), and (e). (We exclude~(a), since it contains all qubits and thus requires unnecessary measurements.) If this BM succeeds, we can choose any two of these three logical 0-chains to complete the logical $\overline{X}$ and $\overline{Z}$ measurements.

Next, the remaining three qubits of the red face, i.e., qubits $2$, $3$, and $4$, are measured using $Z$-BMs. If a success occurs on qubit~$2$, we can complete $\overline{X}$~(d) and either $\overline{Z}$~(b) or~(f).

The case for qubit~$3$ mirrors the one for qubit~$2$. We can complete $\overline{X}$~(e), and either $\overline{Z}$~(b) or~(g). If the BM on qubit~$4$ succeeds, we complete $\overline{X}$~(c), and either $\overline{Z}$~(f) or~(g).

In the final step, the remaining qubits, $5$, $6$, and $7$, at the bottom of the triangle, are measured with $X$-BMs. $\overline{X}$~(h) is thus completed in all cases. If a success occurs on the left bottom tip of the triangle, i.e., on qubit~$5$, $\overline{Z}$~(g) is completed. If the success is on qubit~$7$, the right bottom tip, $\overline{Z}$~(f) is completed. Finally, if the success occurs on qubit~$6$, $\overline{Z}$~(b) is completed.

\begin{figure*}
	\begin{subfigure}[c]{0.32\textwidth}
		\def\svgwidth{\textwidth}
\begingroup%
  \makeatletter%
  \providecommand\color[2][]{%
    \errmessage{(Inkscape) Color is used for the text in Inkscape, but the package 'color.sty' is not loaded}%
    \renewcommand\color[2][]{}%
  }%
  \providecommand\transparent[1]{%
    \errmessage{(Inkscape) Transparency is used (non-zero) for the text in Inkscape, but the package 'transparent.sty' is not loaded}%
    \renewcommand\transparent[1]{}%
  }%
  \providecommand\rotatebox[2]{#2}%
  \newcommand*\fsize{\dimexpr\f@size pt\relax}%
  \newcommand*\lineheight[1]{\fontsize{\fsize}{#1\fsize}\selectfont}%
  \ifx\svgwidth\undefined%
    \setlength{\unitlength}{765.35433071bp}%
    \ifx\svgscale\undefined%
      \relax%
    \else%
      \setlength{\unitlength}{\unitlength * \real{\svgscale}}%
    \fi%
  \else%
    \setlength{\unitlength}{\svgwidth}%
  \fi%
  \global\let\svgwidth\undefined%
  \global\let\svgscale\undefined%
  \makeatother%
  \begin{picture}(1,0.59259259)%
    \lineheight{1}%
    \setlength\tabcolsep{0pt}%
    \put(0,0){\includegraphics[width=\unitlength,page=1]{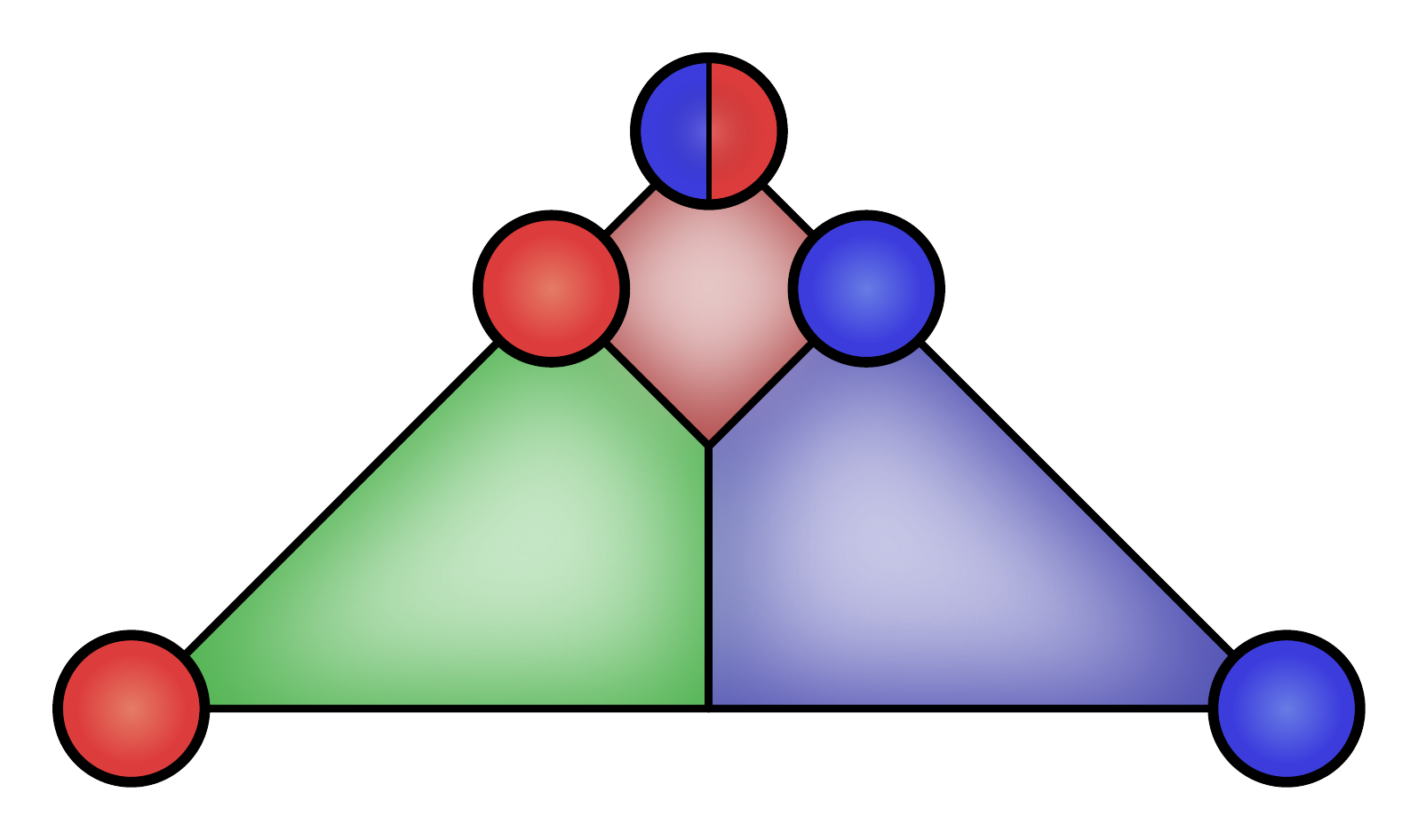}}%
    \put(0.50025872,0.49875399){\color[rgb]{0,0,0}\makebox(0,0)[t]{\lineheight{0}\smash{\begin{tabular}[t]{c}\makebox(0,0){$\contour{white}{1}$}\end{tabular}}}}%
    \put(0.38914758,0.38764287){\color[rgb]{0,0,0}\makebox(0,0)[t]{\lineheight{0}\smash{\begin{tabular}[t]{c}\makebox(0,0){$\contour{white}{2}$}\end{tabular}}}}%
    \put(0.61136983,0.38764287){\color[rgb]{0,0,0}\makebox(0,0)[t]{\lineheight{0}\smash{\begin{tabular}[t]{c}\makebox(0,0){$\contour{white}{3}$}\end{tabular}}}}%
    \put(0.0928513,0.09134657){\color[rgb]{0,0,0}\makebox(0,0)[t]{\lineheight{0}\smash{\begin{tabular}[t]{c}\makebox(0,0){$\contour{white}{5}$}\end{tabular}}}}%
    \put(0.90766613,0.09134657){\color[rgb]{0,0,0}\makebox(0,0)[t]{\lineheight{0}\smash{\begin{tabular}[t]{c}\makebox(0,0){$\contour{white}{7}$}\end{tabular}}}}%
  \end{picture}%
\endgroup%

		\caption{}
	\end{subfigure}
	\begin{subfigure}[c]{0.32\textwidth}
		\def\svgwidth{\textwidth}
\begingroup%
  \makeatletter%
  \providecommand\color[2][]{%
    \errmessage{(Inkscape) Color is used for the text in Inkscape, but the package 'color.sty' is not loaded}%
    \renewcommand\color[2][]{}%
  }%
  \providecommand\transparent[1]{%
    \errmessage{(Inkscape) Transparency is used (non-zero) for the text in Inkscape, but the package 'transparent.sty' is not loaded}%
    \renewcommand\transparent[1]{}%
  }%
  \providecommand\rotatebox[2]{#2}%
  \newcommand*\fsize{\dimexpr\f@size pt\relax}%
  \newcommand*\lineheight[1]{\fontsize{\fsize}{#1\fsize}\selectfont}%
  \ifx\svgwidth\undefined%
    \setlength{\unitlength}{765.35433071bp}%
    \ifx\svgscale\undefined%
      \relax%
    \else%
      \setlength{\unitlength}{\unitlength * \real{\svgscale}}%
    \fi%
  \else%
    \setlength{\unitlength}{\svgwidth}%
  \fi%
  \global\let\svgwidth\undefined%
  \global\let\svgscale\undefined%
  \makeatother%
  \begin{picture}(1,0.59259259)%
    \lineheight{1}%
    \setlength\tabcolsep{0pt}%
    \put(0,0){\includegraphics[width=\unitlength,page=1]{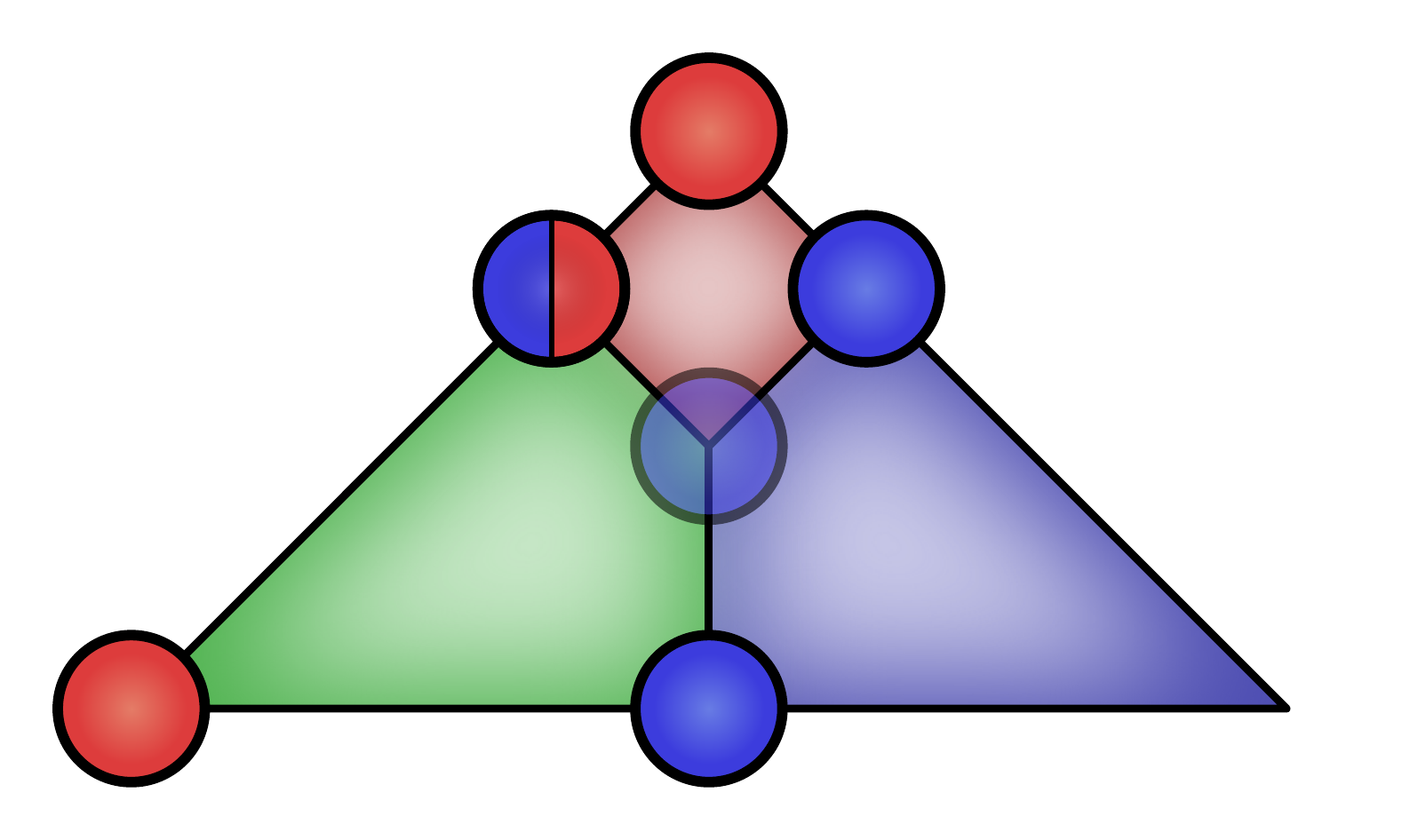}}%
    \put(0.50025872,0.49875399){\color[rgb]{0,0,0}\makebox(0,0)[t]{\lineheight{0}\smash{\begin{tabular}[t]{c}\makebox(0,0){$\contour{white}{1}$}\end{tabular}}}}%
    \put(0.50025872,0.27653176){\color[rgb]{0,0,0}\makebox(0,0)[t]{\lineheight{0}\smash{\begin{tabular}[t]{c}\makebox(0,0){$\contour{white}{4}$}\end{tabular}}}}%
    \put(0.38914758,0.38764287){\color[rgb]{0,0,0}\makebox(0,0)[t]{\lineheight{0}\smash{\begin{tabular}[t]{c}\makebox(0,0){$\contour{white}{2}$}\end{tabular}}}}%
    \put(0.61136983,0.38764287){\color[rgb]{0,0,0}\makebox(0,0)[t]{\lineheight{0}\smash{\begin{tabular}[t]{c}\makebox(0,0){$\contour{white}{3}$}\end{tabular}}}}%
    \put(0.0928513,0.09134657){\color[rgb]{0,0,0}\makebox(0,0)[t]{\lineheight{0}\smash{\begin{tabular}[t]{c}\makebox(0,0){$\contour{white}{5}$}\end{tabular}}}}%
    \put(0.50025872,0.09134657){\color[rgb]{0,0,0}\makebox(0,0)[t]{\lineheight{0}\smash{\begin{tabular}[t]{c}\makebox(0,0){$\contour{white}{6}$}\end{tabular}}}}%
  \end{picture}%
\endgroup%

		\caption{}
	\end{subfigure}
	\begin{subfigure}[c]{0.32\textwidth}
		\def\svgwidth{\textwidth}
\begingroup%
  \makeatletter%
  \providecommand\color[2][]{%
    \errmessage{(Inkscape) Color is used for the text in Inkscape, but the package 'color.sty' is not loaded}%
    \renewcommand\color[2][]{}%
  }%
  \providecommand\transparent[1]{%
    \errmessage{(Inkscape) Transparency is used (non-zero) for the text in Inkscape, but the package 'transparent.sty' is not loaded}%
    \renewcommand\transparent[1]{}%
  }%
  \providecommand\rotatebox[2]{#2}%
  \newcommand*\fsize{\dimexpr\f@size pt\relax}%
  \newcommand*\lineheight[1]{\fontsize{\fsize}{#1\fsize}\selectfont}%
  \ifx\svgwidth\undefined%
    \setlength{\unitlength}{765.35433071bp}%
    \ifx\svgscale\undefined%
      \relax%
    \else%
      \setlength{\unitlength}{\unitlength * \real{\svgscale}}%
    \fi%
  \else%
    \setlength{\unitlength}{\svgwidth}%
  \fi%
  \global\let\svgwidth\undefined%
  \global\let\svgscale\undefined%
  \makeatother%
  \begin{picture}(1,0.59259259)%
    \lineheight{1}%
    \setlength\tabcolsep{0pt}%
    \put(0,0){\includegraphics[width=\unitlength,page=1]{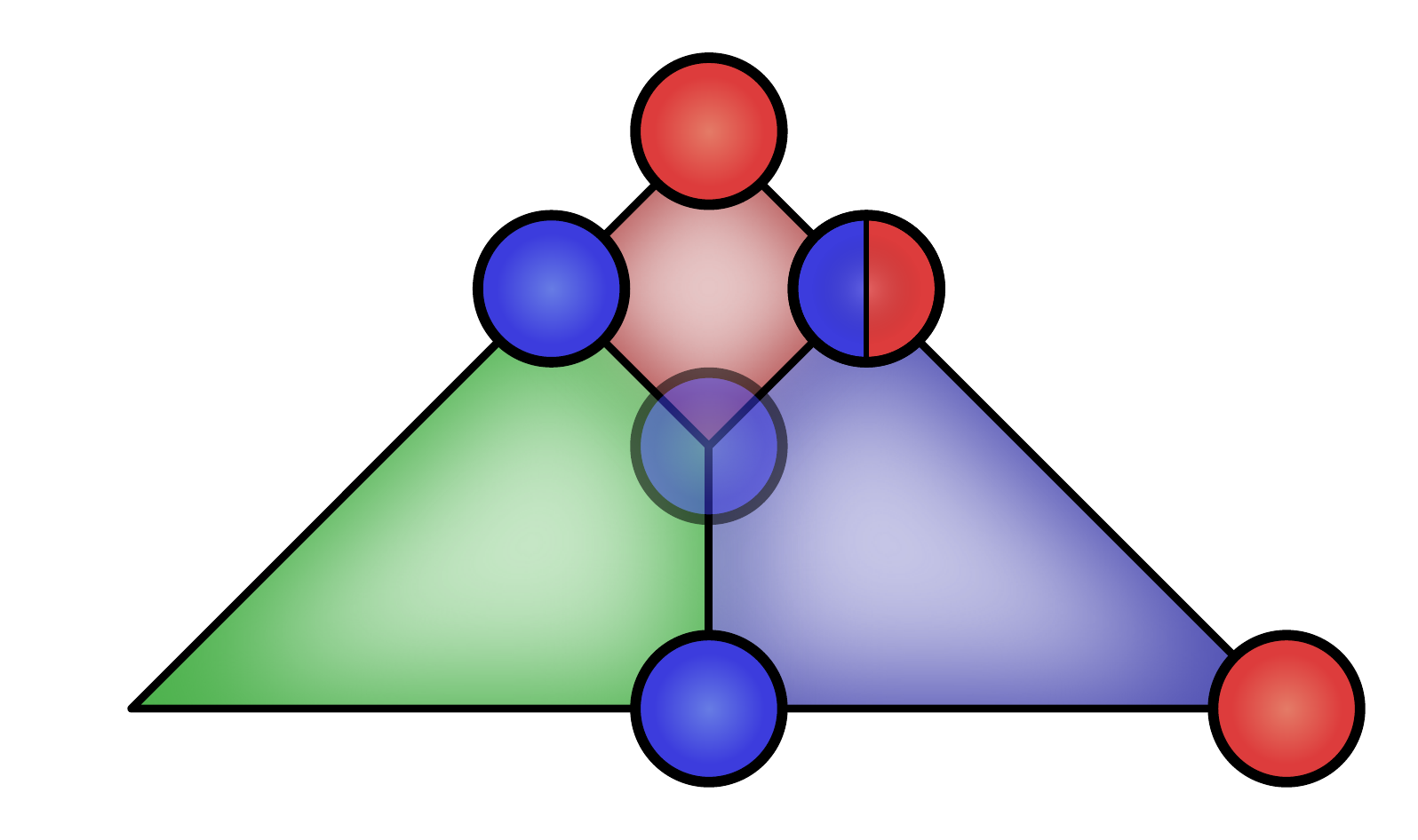}}%
    \put(0.50025872,0.49875399){\color[rgb]{0,0,0}\makebox(0,0)[t]{\lineheight{0}\smash{\begin{tabular}[t]{c}\makebox(0,0){$\contour{white}{1}$}\end{tabular}}}}%
    \put(0.50025872,0.27653176){\color[rgb]{0,0,0}\makebox(0,0)[t]{\lineheight{0}\smash{\begin{tabular}[t]{c}\makebox(0,0){$\contour{white}{4}$}\end{tabular}}}}%
    \put(0.38914758,0.38764287){\color[rgb]{0,0,0}\makebox(0,0)[t]{\lineheight{0}\smash{\begin{tabular}[t]{c}\makebox(0,0){$\contour{white}{2}$}\end{tabular}}}}%
    \put(0.61136983,0.38764287){\color[rgb]{0,0,0}\makebox(0,0)[t]{\lineheight{0}\smash{\begin{tabular}[t]{c}\makebox(0,0){$\contour{white}{3}$}\end{tabular}}}}%
    \put(0.50025872,0.09134657){\color[rgb]{0,0,0}\makebox(0,0)[t]{\lineheight{0}\smash{\begin{tabular}[t]{c}\makebox(0,0){$\contour{white}{6}$}\end{tabular}}}}%
    \put(0.90766613,0.09134657){\color[rgb]{0,0,0}\makebox(0,0)[t]{\lineheight{0}\smash{\begin{tabular}[t]{c}\makebox(0,0){$\contour{white}{7}$}\end{tabular}}}}%
  \end{picture}%
\endgroup%

		\caption{}
	\end{subfigure}
	\begin{subfigure}[c]{0.32\textwidth}
		\def\svgwidth{\textwidth}
\begingroup%
  \makeatletter%
  \providecommand\color[2][]{%
    \errmessage{(Inkscape) Color is used for the text in Inkscape, but the package 'color.sty' is not loaded}%
    \renewcommand\color[2][]{}%
  }%
  \providecommand\transparent[1]{%
    \errmessage{(Inkscape) Transparency is used (non-zero) for the text in Inkscape, but the package 'transparent.sty' is not loaded}%
    \renewcommand\transparent[1]{}%
  }%
  \providecommand\rotatebox[2]{#2}%
  \newcommand*\fsize{\dimexpr\f@size pt\relax}%
  \newcommand*\lineheight[1]{\fontsize{\fsize}{#1\fsize}\selectfont}%
  \ifx\svgwidth\undefined%
    \setlength{\unitlength}{765.35433071bp}%
    \ifx\svgscale\undefined%
      \relax%
    \else%
      \setlength{\unitlength}{\unitlength * \real{\svgscale}}%
    \fi%
  \else%
    \setlength{\unitlength}{\svgwidth}%
  \fi%
  \global\let\svgwidth\undefined%
  \global\let\svgscale\undefined%
  \makeatother%
  \begin{picture}(1,0.59259259)%
    \lineheight{1}%
    \setlength\tabcolsep{0pt}%
    \put(0,0){\includegraphics[width=\unitlength,page=1]{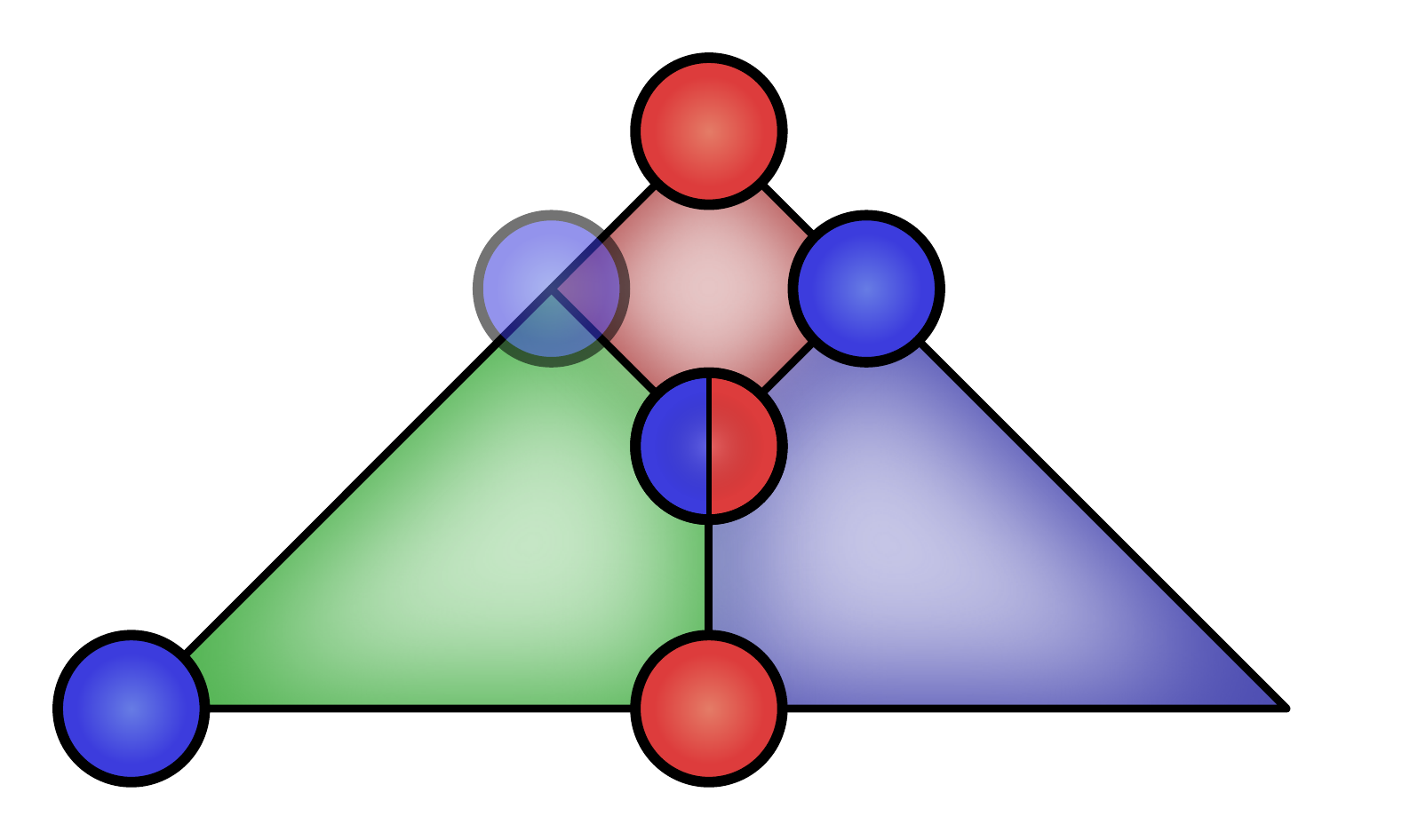}}%
    \put(0.50025872,0.49875399){\color[rgb]{0,0,0}\makebox(0,0)[t]{\lineheight{0}\smash{\begin{tabular}[t]{c}\makebox(0,0){$\contour{white}{1}$}\end{tabular}}}}%
    \put(0.50025872,0.27653176){\color[rgb]{0,0,0}\makebox(0,0)[t]{\lineheight{0}\smash{\begin{tabular}[t]{c}\makebox(0,0){$\contour{white}{4}$}\end{tabular}}}}%
    \put(0.38914758,0.38764287){\color[rgb]{0,0,0}\makebox(0,0)[t]{\lineheight{0}\smash{\begin{tabular}[t]{c}\makebox(0,0){$\contour{white}{2}$}\end{tabular}}}}%
    \put(0.61136983,0.38764287){\color[rgb]{0,0,0}\makebox(0,0)[t]{\lineheight{0}\smash{\begin{tabular}[t]{c}\makebox(0,0){$\contour{white}{3}$}\end{tabular}}}}%
    \put(0.0928513,0.09134657){\color[rgb]{0,0,0}\makebox(0,0)[t]{\lineheight{0}\smash{\begin{tabular}[t]{c}\makebox(0,0){$\contour{white}{5}$}\end{tabular}}}}%
    \put(0.50025872,0.09134657){\color[rgb]{0,0,0}\makebox(0,0)[t]{\lineheight{0}\smash{\begin{tabular}[t]{c}\makebox(0,0){$\contour{white}{6}$}\end{tabular}}}}%
  \end{picture}%
\endgroup%

		\caption{}
	\end{subfigure}
	\begin{subfigure}[c]{0.32\textwidth}
		\def\svgwidth{\textwidth}
		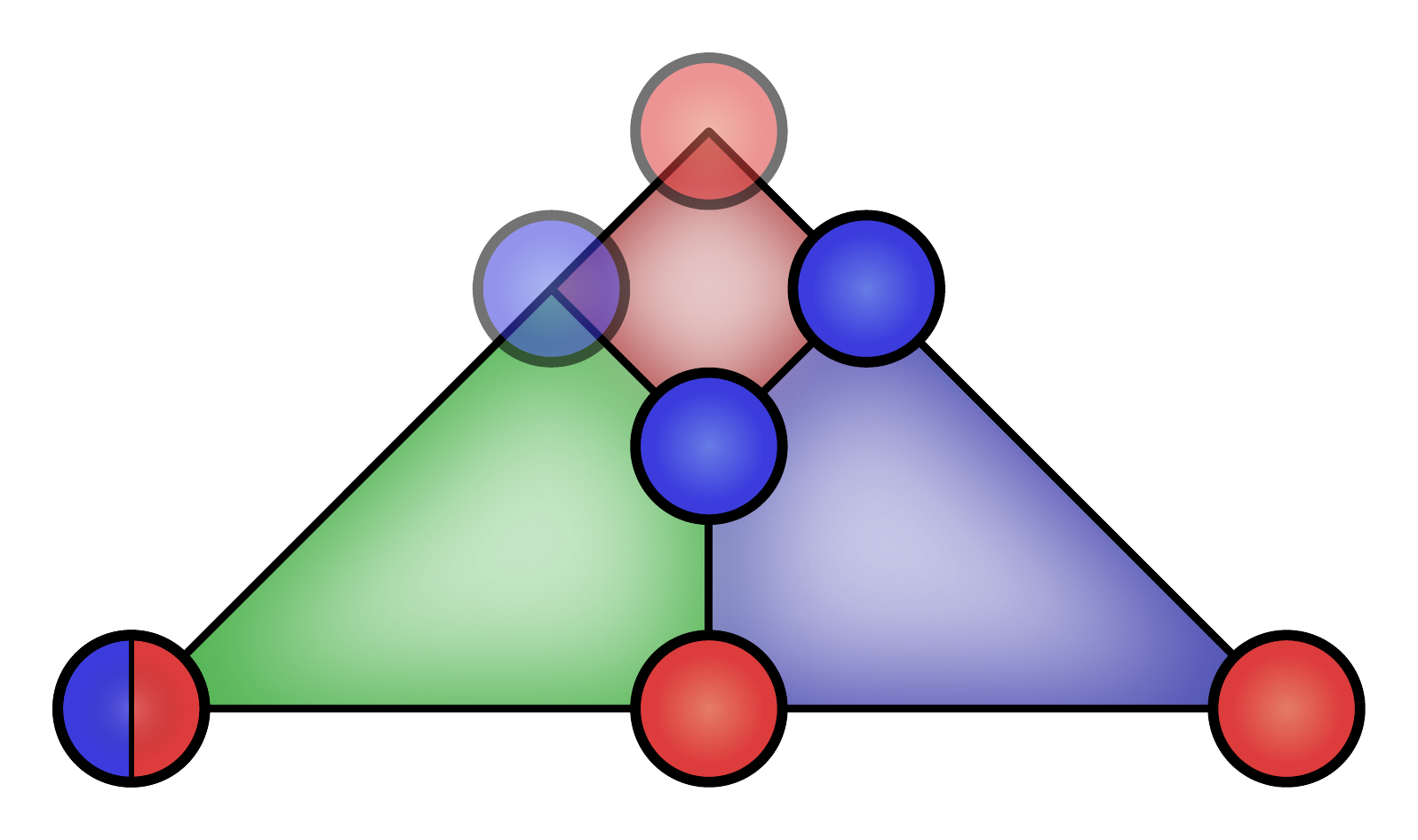
		\caption{}
	\end{subfigure}
	\begin{subfigure}[c]{0.32\textwidth}
		\def\svgwidth{\textwidth}
		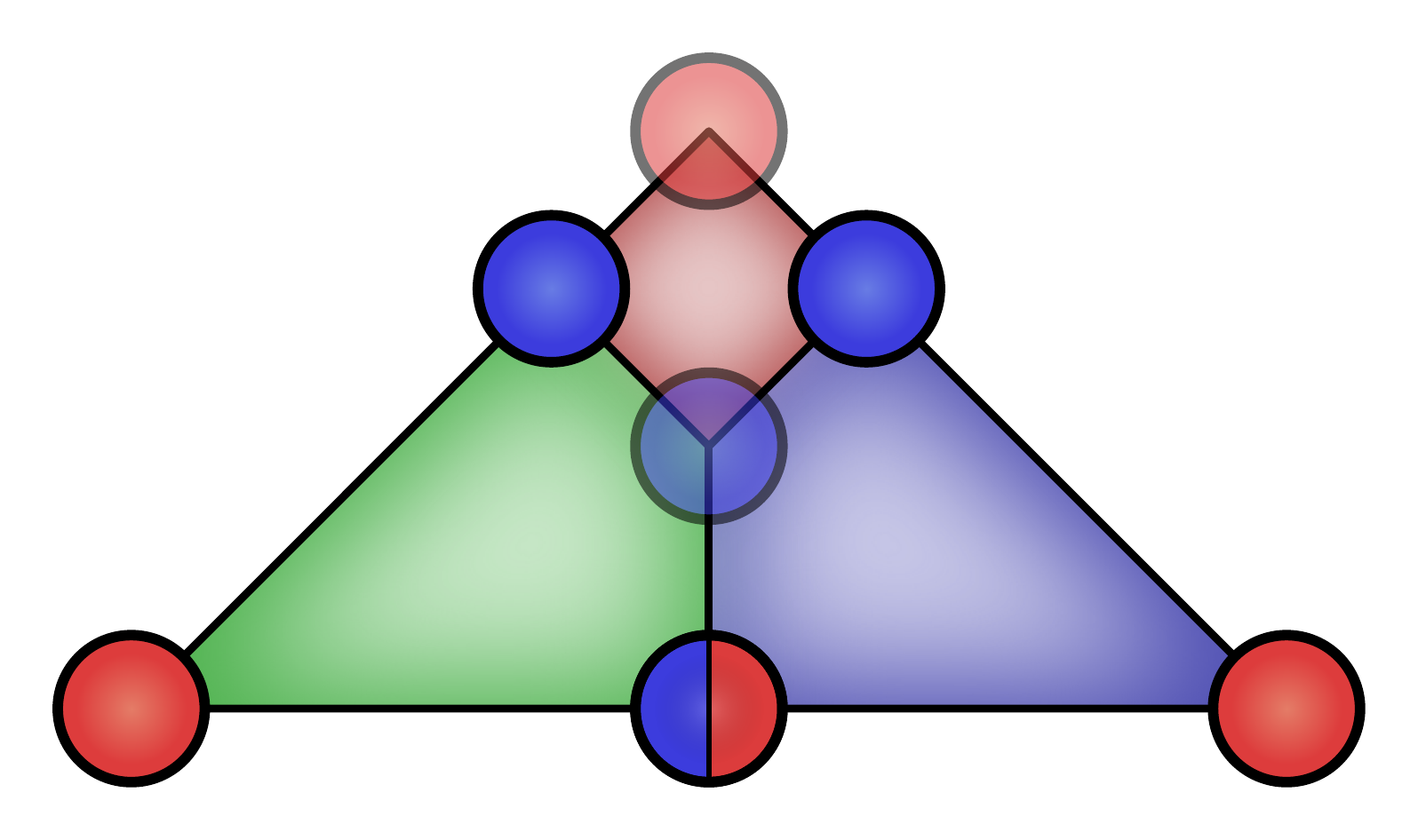
		\caption{}
	\end{subfigure}
	\begin{subfigure}[c]{0.32\textwidth}
		\def\svgwidth{\textwidth}
		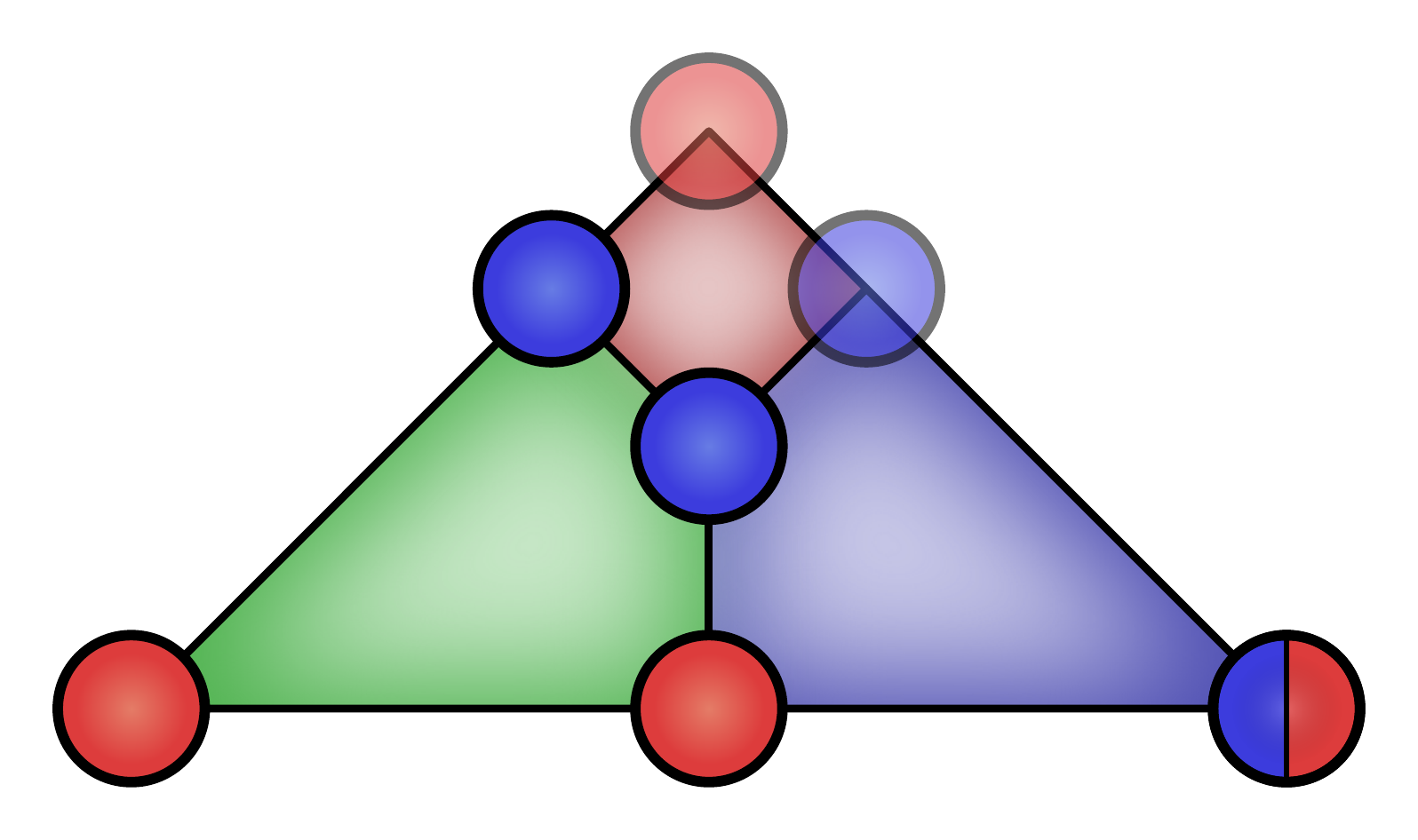
		\caption{}
	\end{subfigure}
	\caption{\label{fig:steane-sol}Logical operators measured upon a successful transversal BM for the Steane code. Vertices colored with both red and blue represent a successful BM. Red and blue vertices correspond to $X$- and $Z$-BMs, which form the logical $\overline{X}$ and $\overline{Z}$ operators, respectively. Transparent red and blue vertices indicate transversal BMs that have already been performed in the measurement scheme, but are not part of the measured logical operators. For each qubit $j$, the corresponding tuple $\left( \overline{X}_j , \overline{Z}_j \right)$ is shown in (a) to (g), respectively.}
\end{figure*}

We now demonstrate that the success probability for every physical BM is given by $\mathbb{P}_B$. Since we are dealing with a CSS code, we can treat $X$ and $Z$ stabilizers separately: $G_c = G_{c,X} \cup G_{c,Z}$. Furthermore, we label the stabilizer generators according to the color of their corresponding plaquette: the operators $X_r$, $X_g$, and $X_b$ correspond to the $X$ stabilizers of the red, green, and blue plaquettes, respectively. First we transform the $X$ stabilizers of the code,
\begin{equation}
	\begin{aligned}
		G_{c,X} & = \{ X_r, X_g, X_b \} \\
				& = \{ XXXXIII, IXIXXXI, IIXXIXX \},
	\end{aligned}
\end{equation}
into the form:
\begin{equation}
	\begin{aligned}
		\tilde{G}_{c,X} & = \{ X_g, X_b, X_r X_g X_b \} \\
						& = \{ IXIXXXI, IIXXIXX, XIIXXIX \},
	\end{aligned}
\end{equation}
where the red plaquette stabilizer $X_r$ is replaced with the product of all three $X$-stabilizers and we reordered the set. It is straightforward to check that:
\begin{equation}
	\langle G_{c,X} \rangle = \langle \tilde{G}_{c,X} \rangle.
\end{equation}
To connect our argument to Thm.~\ref{thm:sufficient} we define the ordered code stabilizers:
\begin{equation}
	\begin{aligned}
		\mathbb{C} 		& = ( 	c_{j} )_{j \in \{ 1, \dots, n-1 \} } \\
				& = ( 	Z_r, X_g, X_b, X_r X_g X_b, Z_g, Z_b ) \\
				& = ( 	ZZZZIII, IXIXXXI, IIXXIXX, \\
				&	 	XIIXXIX, IZIZZZI, IIZZIZZ ).
	\end{aligned}
\label{eq:steane-sequence-c}
\end{equation}
The measurement sequence, which was explained above, is given by:
\begin{equation}
\begin{aligned}
	\mathbb{B} 	& = ( b_{j} \}_{j \in \{ 1, \dots, n-1 ) } \\
				& = ( X_1, Z_2, Z_3, Z_4, X_5, X_6 ).
\end{aligned}
\end{equation}
In this definition, we impose an ordering on the measurements performed within one step of the scheme. Recall that this is physically equivalent because no feedforward is applied between these measurements (see Sec.~\ref{sec:optimization-and-comparionn-of-static-logical-bell-measurements-for-rotated-planar-surface-codes}). It is straightforward to verify that, for all $j$, $c_j$ is the only element of the current stabilizer generators that anticommutes with the measurement $b_j$. Therefore, the stabilizer generators are easy to track, and each measurement $b_j$ replaces $c_j$ in the current stabilizer generators as long as no success occurs.

For every measurement, $Y_j$ anticommutes with the current stabilizer by the same reasoning. For the first $X$-BM on qubit $1$, the operator $Z_1$ anticommutes with $c_4 = X_r X_g X_b = XIIXXIX$. For the three $Z$-BMs on qubits $2$, $3$, and $4$, the corresponding $X$ operators anticommute with one of the two remaining $Z$ stabilizers, specifically $X_2$ with $c_5=Z_g$, $X_3$ with $c_6=Z_b$, and $X_4$ with both $c_5=Z_g$ and $c_6=Z_b$.

For the subsequent two $XX$-BMs on qubits $5$ and $6$ the two respective $Z$ operators $Z_5$ and $Z_6$ complete the measurement of $\overline{Z}$~(g) and~(b), respectively, as previously discussed. Finally, on the last qubit $7$, $X_7$ completes $\overline{X}$~(h), $Z_7$ completes $\overline{Z}$~(f), and $Y_7$ completes the $\overline{Y}$ operator~$IZIZXXY$, which is the product of $\overline{X}$~(h) and $\overline{Z}$~(f).

Having established that all single-qubit operators on the measured qubits either anticommute with an element of the current stabilizer or complete a logical operator, we can apply Lem.~\ref{lem:successful-bell-measurment} to conclude that each transversal BM up to the first success succeeds with probability $\mathbb{P}_B$, and the overall success probability of the logical BM scheme is $1 - (1 - \mathbb{P}_B)^7$.

A detailed algebraic proof of the scheme's optimality, based on Thm.~\ref{thm:sufficient}, is provided in App.~\ref{app:proof-steane}.

To the best of the authors' knowledge, the only logical BM scheme for the Steane code published so far appears in Ref.~\cite{PhysRevA.99.062308}, which is a static scheme that does not use feedforward. The scheme in Ref.~\cite{PhysRevA.99.062308} achieves a success probability of $1 - 2^{-5}$, assuming standard linear-optics BMs with $\mathbb{P}_B = \frac{1}{2}$.

\section{Conclusion}
\label{sec:conclusion}
In this work, we investigated feedforward-based linear-optics logical Bell measurements (BMs) on stabilizer codes, restricted to a toolbox consisting of physical BMs, single-qubit Clifford gates, and single-qubit Pauli measurements. To identify fundamental limits, we focused on an idealized, error-free, i.e., especially loss-free setting. This error-free case serves as a necessary foundational step, providing the basis upon which any treatment of imperfections can be built. We have shown that at least a single successful physical BM is required for a successful logical BM. As a necessary condition this provides the general upper bound of $1- \left( 1 - \mathbb{P}_B \right)^{ \min(n_1,n_2) }$ for the success probability of logical linear-optics BMs on stabilizer encoded qubits, where $\mathbb{P}_B$ is the success probability of a physical guaranteed partial information BM, and $n_1$ and $n_2$ are the numbers of the physical qubits used to encode the first and second logical qubit, respectively. For standard linear-optics BMs with $\mathbb{P}_B = \frac{1}{2}$ and identically encoded logical qubits, this simplifies to $1-2^{-n}$. This improves upon a proof previously given in Ref.~\cite{PhysRevA.100.052303} by extending it from static linear optics to feedforward-based schemes and by circumventing the restrictive assumption that photon-number-resolving detectors can distinguish only up to two-photon events.

We derived methods to find feedforward-based schemes satisfying this bound which are generally applicable to any stabilizer code, and we demonstrated this for quantum parity, five-qubit, standard and rotated planar surface, tree, and seven-qubit Steane codes. Our schemes attain the general upper bound for all these codes, while this bound had previously only been reached for the quantum parity code in Ref.~\cite{PhysRevA.100.052303}. Additionally, we presented an optimized static scheme for the rotated planar surface code. While this scheme does not achieve the success probability of the feedforward-based bound, it still performs significantly better than a simple static scheme.

Interestingly, the scheme we developed for the five-qubit code does not require feedforward, thus it can be fully implemented using static operations alone. This observation implies that, in general, there is no tighter bound for static schemes than for feedforward-based ones. However, for certain codes, the standard toolbox, which relies on guaranteed partial information BMs, fails to achieve the bound when constrained to static operations~\cite{PhysRevA.99.062308}. Therefore, even though we have disproved the existence of a tighter bound for static schemes for general stabilizer codes, it seems unlikely that the bound can be achieved with static means in many cases.

The results presented here, fully based on the stabilizer formalism, deepen the theoretical foundations of linear-optics logical BMs by introducing a rigorous framework to formally describe logical measurement processes. Within this framework, we not only characterized the ultimate performance limits but also demonstrated how they can be reached in practice, thereby obtaining schemes that significantly improve upon the efficiencies reported in the existing literature. Because logical BMs are a common primitive in both MBQC or FBQC and all-optical quantum communication, these advances apply directly across both domains. In computation, higher success probabilities and more efficient logical schemes strengthen the prospects for implementing encoded operations that are essential for scalable fault-tolerant architectures. More specifically, in FBQC, the universal resource states can be built more efficiently and with larger loss thresholds~\cite{Bartolucci2023,PhysRevLett.133.050604,PhysRevLett.133.050605}. In communication, the same improvements enable more resource-efficient entanglement swapping at the logical level, directly supporting the construction of long-distance, loss-resilient, all-optical quantum repeaters. In conclusion, these contributions advance the efficiency, viability, and scalability of photonic quantum technologies by providing a rigorous theoretical foundation for logical BMs that benefits both computation and communication. Viewed in this broader context, our results represent concrete progress toward the overarching goal of realizing scalable, fault-tolerant, optical quantum technologies.

An important direction for future work is the incorporation of photon loss, which plays a critical role in realistic settings. As this requires a dedicated analysis beyond the scope of the present study, it will be pursued in subsequent work. Furthermore, while a set of sufficient conditions has been rigorously established, they are cumbersome in form. A more elegant set of heuristics introduced here performs well across all examined cases, and it remains open whether they can be formally proven sufficient. Finally, while the scheme developed here has been demonstrated for the Steane code, the smallest instance of a color code, it would be natural to explore its extension to more general color codes.

\begin{acknowledgments}
We thank Frank Schmidt for useful discussions and we also thank Evgeny Shchukin for assistance with the efficient implementation of the algorithm used to generate the data points of the simple static scheme in Fig.~\ref{fig:rotated-planar-surface-codes-comparison}. We further thank the BMFTR in Germany for support via PhotonQ, QR.N, QuKuK, and QuaPhySI.
\end{acknowledgments}

\appendix

\section{Formal derivation of $\mathbb{P}_B$ and perfect specificity of physical Bell measurements}
\label{app:generalized-physical-bell-measurements}

In the following, we briefly present a complete formulation of the probability $\mathbb{P}_B$ of successfully performing a BM on a uniform mixture of Bell states.

We start by defining the set of unambiguous measurement outcomes $ \{ s \} \subseteq \{ m \}$ which is the set of all measurement results satisfying the condition in Eq.~\eqref{eq:s}:
\begin{equation}
\begin{aligned}
	\{ s \} = \{	& s \in \{ m \} \mid \quad \exists \sigma : \\
					& P_s U \ket{\Phi_\sigma} \ket{A} \neq 0 \quad \\
					& \wedge \quad \forall i \neq \sigma \quad P_s U \ket{\Phi_i} \ket{A} = 0 \}.
\end{aligned}
	\label{eq:set-s}
\end{equation}
Using this definition we can write the success probability $\mathbb{P}_\varrho(\{s\})$ of a physical BM on a quantum state $\varrho$ as
\begin{equation}
	\mathbb{P}_\varrho(\{s\}) = \sum_{s \in \{s\}} p_\varrho(s),
\end{equation}
where $p_\varrho (m)$ is the probability of the measurement outcome $m$ assuming the initial quantum state $\varrho$,
\begin{equation}
	p_{\varrho} (m) = \Tr \{ P_m U \left( \varrho \otimes \ket{A} \bra{A} \right) U^\dagger \}.
	\label{eq:p-rho-m}
\end{equation}
Hence, the success probability $\mathbb{P}_\varrho(\{s\})$ is the probability to obtain any outcome in $\{s\}$ for a given initial state $\varrho$. When examining a uniform mixture of Bell states, defined as
\begin{equation}
	\varrho_B = \frac{1}{4} \sum_{j=1}^4 \ket{\Phi_j} \bra{\Phi_j} = \frac{I}{4},
\end{equation}
we obtain:
\begin{equation}
\begin{aligned}
	p_{\varrho_B} (m) & = \Tr \{ P_m U \left( \varrho_B \otimes \ket{A} \bra{A} \right) U^\dagger \} \\
		& = \frac{1}{4} \sum_{j=1}^4 \bra{A} \bra{\Phi_j} U^\dagger P_m U \ket{\Phi_j} \ket{A}.
\end{aligned}
\end{equation}
In the special case where no ancillary state is used, again we simply remove the ancilla vectors, while $U$ and $P_m$ only act on $\mathcal{H}_B$. We easily seew how the probability $p_{\varrho_B}(m)$ simplifies for a measurement outcome $s \in \{ s \}$. By inserting Eq.~\eqref{eq:set-s}, we obtain
\begin{equation}
	p_{\varrho_B} (s) = \frac{1}{4} \bra{A} \bra{\Phi_\sigma} U^\dagger P_s U \ket{\Phi_\sigma} \ket{A},
	\label{eq:def-pb-a}
\end{equation}
for some unique $\sigma \in \{1,2,3,4\}$, which we use to calculate:
\begin{equation}
	\mathbb{P}_{\varrho_B} (\{s\}) = \sum_{s \in \{s\}} p_{\varrho_B} (s) \eqqcolon \mathbb{P}_B.
	\label{eq:def-pb-b}
\end{equation}
The following simple lemma captures an important property of this definition of a successful BM. Using this corollary alongside with Cor.~\ref{cor:pm-success}, which we will present in App.~\ref{app:physical-bell-measurements-on-entangled-quantum-states}, we will formally verify that Eq.~\eqref{eq:s} defines an unambiguous BM.

\begin{lemma}
\label{lem:perfect-specificity}
(Perfect specificity of BMs) Let us assume a physical BM as defined in Sec.~\ref{sec:generalized-physical-bell-measurements} on an encoded quantum state $\ket{\psi} \in \mathcal{H}_B$, which is entirely within the code space. Excluding any errors, the probability of a successful measurement result $p_{\varrho} (s)$, where $s \in \{ s \}$, is proportional to the probability to project the quantum states onto $\ket{\Phi_\sigma}$ for some unique $\sigma \in \{1,2,3,4\}$. Specifically, it is impossible to have measurement result $s \in \{ s \}$, when the projection of the measured quantum state onto $\ket{\Phi_\sigma}$ is zero.
\end{lemma}
\begin{proof}
We consider a general quantum state:
\begin{equation}
	\varrho = \sum_{i,j=1}^4 \alpha_{ij} \ket{\Phi_i}\bra{\Phi_j}.
	\label{eq:general-density-op}
\end{equation}
We insert the conditions from Eq.~\eqref{eq:s} and the state from Eq.~\eqref{eq:general-density-op} into Eq.~\eqref{eq:p-rho-m} and immediately arrive at:
\begin{equation}
\begin{aligned}
	p_\varrho (s) & = \Tr\{ P_s U \left( \varrho \otimes \ket{A}\bra{A} \right) U^\dagger \} \\
	& = \Tr\{ \alpha_{\sigma \sigma} P_s U \left( \ket{\Phi_\sigma}\bra{\Phi_\sigma} \otimes \ket{A}\bra{A} \right) U^\dagger \} \\
	& \propto \alpha_{\sigma \sigma}.
\end{aligned} 
\end{equation}
\end{proof}

\section{Derivations for physical Bell measurements on entangled quantum states}
\label{app:physical-bell-measurements-on-entangled-quantum-states}
In the first lemma of this appendix we treat the most general form of a measurement outcome. We will show that for any given measurement result $\alpha \in \{m\}$ we can obtain the post-measurement state by replacing the physical BM with a simpler effective projection on $\mathcal{H}_B$ without the need of a unitary operation $U$ or an ancillary state.

\begin{lemma}
\label{lem:pm-general}
(Post-measurement state for a physical BM) Let us assume a quantum state $\ket{\psi} \in \mathcal{H}_B \otimes \mathcal{H}_R$. Furthermore, we assume that the quantum state in $\mathcal{H}_B$ is entirely within the two-qubit code space. We consider a physical BM on $\mathcal{H}_B$ and a measurement outcome $\alpha \in \{m\}$ with corresponding projector $P_\alpha = \ket{\alpha} \bra{\alpha}$. Then the measurement outcome $\alpha$ can be fully characterized by complex coefficients $\alpha_j$, where $j \in \{ 1, 2, 3, 4 \}$ via,
\begin{equation}
	P_\alpha U \ket{\Phi_j}_B \ket{A}_A = \alpha_j \ket{\alpha}_{BA},
	\label{eq:alpha}
\end{equation}
where the coefficients $\alpha_j$ satisfy $\sum_{j=1}^4 | \alpha_j |^2 \leq 1$.

Then after performing the physical measurement and obtaining the outcome $\alpha$ the post-measurement state on $\mathcal{H}_R$ is identical to the post-measurement state from the effective projection
\begin{equation}
		\Pi_\alpha = \ket{B_\alpha}_B \bra{B_\alpha} \otimes I_R,
		\label{eq:p-alpha}
\end{equation}
where
\begin{equation}
	\ket{B_\alpha} = \frac{1}{\sqrt{\sum_{j=1}^4 | \alpha_j |^2 }} \sum_{j=1}^4 \alpha_j^* \ket{\Phi_j}.
\end{equation}
Hence the post-measurement state of the physical measurement is
\begin{equation}
	\frac{\Tr_B \{ \Pi_\alpha \ket{\psi} \bra{\psi} \}}{\Tr \{ \Pi_\alpha \ket{\psi} \bra{\psi} \} },
\end{equation}
where $\Tr \equiv \Tr_{BR}$.

\end{lemma}
\begin{proof}
We use a general decomposition of $\ket{\psi}$,
\begin{equation}
	\ket{\psi} = \sum_{j=1}^4 \sum_{k=1}^d \lambda_{jk} \ket{\Phi_j}_B \ket{R_k}_R,
	\label{eq:psi-comp}
\end{equation}
with normalization
\begin{equation}
	\sum_{j=1}^4 \sum_{k=1}^d |\lambda_{jk}|^2 = 1,
\end{equation}
where we use $\{ \ket{\Phi_j}_B \}_{j \in \{1,2,3,4\}}$ as a basis for the encoded qubits in $\mathcal{H}_B$ and an arbitrary orthonormal basis $\{ \ket{R_k} \}_{k \in \{1,\dots,d\}}$ for the remaining quantum state in $\mathcal{H}_R$. We can use $\{ \ket{\Phi_j}_B \}_{j \in \{1,2,3,4\}}$ as a basis for the quantum state in $\mathcal{H}_B$ since the lemma demands the qubits on $\mathcal{H}_B$ to be entirely within the code space. We use Eq.~\eqref{eq:psi-comp} and Eq.~\eqref{eq:alpha} to obtain:
\begin{equation}
\begin{aligned}
	P_\alpha U \ket{\psi}_{BR} \ket{A}_A	& = \sum_{j=1}^4 \sum_{k=1}^d \lambda_{jk} P_\alpha U \ket{\Phi_j}_B \ket{A}_A \ket{R_k}_R \\
											& = \sum_{j=1}^4 \sum_{k=1}^d \lambda_{jk} \alpha_j \ket{\alpha}_{BA} \ket{R_k}_R. \\
\end{aligned}
\label{eq:physical-alpha-projection}
\end{equation}
We continue by computing the probability of this measurement outcome $\alpha$ using Eq.~\eqref{eq:physical-alpha-projection}:
\begin{equation}
\begin{aligned}
	\bra{A} \bra{\psi} U^\dagger P_\alpha U \ket{\psi} \ket{A} \\
	= \sum_{j,i=1}^4 \alpha_j \alpha_i^* \braket{\alpha | \alpha} \sum_{k,l=1}^d \lambda_{jk} \lambda_{ik}^* \braket{R_l | R_k} \\
	= \sum_{j,i=1}^4 \alpha_j \alpha_i^* \sum_{k=1}^d \lambda_{jk} \lambda_{ik}^*\\
	\coloneqq p.
\end{aligned}
\end{equation}
Next, we again use Eq.~\eqref{eq:physical-alpha-projection} to compute the post-measurement state for the outcome $\alpha$:
\begin{equation}
	\begin{split}
		\varrho_R^\prime & = \frac{1}{p} \Tr_{BA} \{ P_\alpha U ( \ket{\psi} \bra{\psi} \otimes \ket{A} \bra{A} ) U^\dagger P_\alpha \} \\
		& = \frac{1}{p} \Tr_{BA} \{ \sum_{j,i=1}^4 \alpha_j \alpha_i^* \ket{\alpha} \bra{\alpha} \otimes \sum_{k,l=1}^{d} \lambda_{jk} \lambda_{il}^* \ket{R_k} \bra{R_l} \} \\
		& = \frac{1}{p} \sum_{j,i=1}^4 \alpha_j \alpha_i^* \sum_{k,l=1}^{d} \lambda_{jk} \lambda_{il}^* \ket{R_k} \bra{R_l} \\
		& \eqqcolon \ket{\psi^\prime}_R \bra{\psi^\prime},
	\end{split}
\end{equation}
where we defined
\begin{equation}
	\ket{\psi^\prime}_R = \frac{1}{\sqrt{p}} \sum_{j=1}^4 \alpha_j \sum_{k=1}^{d} \lambda_{jk} \ket{R_k}.
\end{equation}

In the remainder of this proof we will show that the post-measurement state of the effective projection is identical to this state. We define the normalized state
\begin{equation}
	\begin{split}
		\ket{B_\alpha} & = \frac{1}{\sqrt{\sum_{m=1}^4 | \alpha_m |^2 }} \sum_{m=1}^4 \alpha_m^* \ket{\Phi_m} \\
		& \eqqcolon N_\alpha \sum_{m=1}^4 \alpha_m^* \ket{\Phi_m} \\
	\end{split}
\end{equation}
which defines the effective projector
\begin{equation}
	\begin{split}
		\Pi_\alpha & = \ket{B_\alpha} \bra{B_\alpha} \\
		& = N_\alpha^2 \sum_{m,n=1}^4 \alpha_m^* \alpha_n \ket{\Phi_m} \bra{\Phi_n}.
		\end{split}
\end{equation}

\begin{widetext}

We now compute the properties of this effective projective measurement defined by the projector $\Pi_\alpha$. We start by computing the probability of this projection:
\begin{equation}
	\begin{split}
		\Tr \{ \Pi_\alpha \ket{\psi} \bra{\psi} \} & = \bra{\psi} \Pi_\alpha \ket{\psi} \\
		& = N_\alpha^2 \sum_{m,n=1}^4 \sum_{j,i=1}^4 \sum_{k,l=1}^{d} \alpha_m^* \alpha_n \lambda_{jk} \lambda_{il}^* \braket{\Phi_i | \Phi_m} \braket{\Phi_n | \Phi_j} \braket{R_l | R_k} \\
		& = N_\alpha^2 \sum_{j,i=0}^3 \alpha_j \alpha_i^* \sum_{k=1}^{d-1} \lambda_{jk} \lambda_{ik}^* \\
		& = N_\alpha^2 \cdot p \\
		& \eqqcolon p_\alpha.
		\end{split}
\end{equation}
Thus, the post-measurement state of the effective projection is given by:
\begin{equation}
	\begin{split}
		\varrho_{\text{effective}}^\prime & = \frac{1}{p_\alpha} \Tr_B \{ \Pi_\alpha \ket{\psi} \bra{\psi} \Pi_\alpha \} \\
		& = \frac{N_\alpha^4}{p_\alpha} \Tr_B \{ \sum_{m,n,o,p=1}^4 \sum_{j,i=1}^4 \sum_{k,l=1}^{d} \alpha_m^* \alpha_n \alpha_o^* \alpha_p \lambda_{jk} \lambda_{il}^* \ket{\Phi_m} \braket{\Phi_n | \Phi_j} \braket{\Phi_i | \Phi_o} \bra{\Phi_p} \otimes \ket{R_k} \bra{R_l} \} \\
		& = \frac{N_\alpha^4}{p_\alpha}	\Tr_B \{ \sum_{m,p=1}^4 \sum_{j,i=1}^4 \sum_{k,l=1}^{d} \alpha_m^* \alpha_j \alpha_i^* \alpha_p \lambda_{jk} \lambda_{il}^* \ket{\Phi_m} \bra{\Phi_p} \otimes \ket{R_k} \bra{R_l} \} \\
		& = \frac{N_\alpha^4}{p_\alpha} \sum_{m=1}^4 \sum_{j,i=1}^4 \sum_{k,l=1}^{d} \alpha_m^* \alpha_j \alpha_i^* \alpha_m \lambda_{jk} \lambda_{il}^* \ket{R_k}\bra{R_l} \\
		& = \frac{N_\alpha^4}{ p N_\alpha^2 N_\alpha^2 } \sum_{j,i=1}^4 \sum_{k,l=1}^{d} \alpha_j \alpha_i^* \lambda_{jk} \lambda_{il}^* \ket{R_k}\bra{R_l} \\
		& = \frac{1}{p} \sum_{j,i=1}^4 \sum_{k,l=1}^{d} \alpha_j \alpha_i^* \lambda_{jk} \lambda_{il}^* \ket{R_k}\bra{R_l} \\
		& = \varrho_R^\prime.
	\end{split}
	\label{eq:proof-cor-pm-general}
\end{equation}
\end{widetext}

This final Eq.~\eqref{eq:proof-cor-pm-general} concludes the proof, since it proves that on the remaining Hilbert space $\mathcal{H}_R$ the physical measurement result defined by Eq.~\eqref{eq:alpha} is indistinguishable from the projection in Eq.~\eqref{eq:p-alpha}. Note that the probabilities $p_\alpha$ and $p$ differ by a factor of $N_\alpha^2$ but this does not interfere with Lem.~\ref{lem:pm-general}, which makes no claim about the probability of the measurement outcome.
\end{proof}

From Lem.~\ref{lem:pm-general} directly follows our next corollary. It states that a successful BM result will project the post-measurement state in $\mathcal{H}_R$ as we would expect: it projects it onto the same state a perfect Bell projection would do. While this assertion appears intuitive and not surprising, we have rigorously verified it.

\begin{corollary}
\label{cor:pm-success}
(Post-measurement state of a successful destructive physical BM on an entangled state) Let us assume a quantum state $\ket{\psi} \in \mathcal{H}_B \otimes \mathcal{H}_R$. Furthermore, we assume that the quantum state in $\mathcal{H}_B$ is entirely within the two-qubit code space. We consider a physical BM on $\mathcal{H}_B$ and a measurement outcome $s \in \{s\} $.

Then the post-measurement state in $\mathcal{H}_R$ is identical to the post-measurement state of the initial state $\ket{\psi}$ projected onto $\ket{\Phi_\sigma}_B \bra{\Phi_\sigma} \otimes I_R$.
\end{corollary}
\begin{proof}
Follows directly from Lem.~\ref{lem:pm-general} and Eq.~\eqref{eq:s}.
\end{proof}

Previously, in Sec.~\ref{sec:generalized-physical-bell-measurements}, we postulated the definition of an unambiguous measurement outcome of a physical BM. In Lem.~\ref{lem:perfect-specificity}, which we presented in the previous App.~\ref{app:generalized-physical-bell-measurements}, we showed that the probability to identify a Bell state is proportional to the probability of the projection of this Bell state on the measured quantum state and specifically, that it is impossible to identify a Bell state when the projection of the measured state on this Bell state is zero. In Cor.~\ref{cor:pm-success} we proved that the post-measurement state on the surviving qubits is indistinguishable from the projection of a complete, ideal BM on the Bell state identified by the physical measurement. Thus, Lem.~\ref{lem:perfect-specificity} and Cor.~\ref{cor:pm-success} give a formal notion of an unambiguous physical BM result.

As a next step, in Lem.~\ref{lem:prob-success} below, we address the special case where the measurement probabilities of a perfect BM are uniformly distributed. It states that the success probability of the physical BM $\mathbb{P}_{\ket{\psi} \bra{\psi}} (\{s\})$ in this case is identical to physically measuring a uniform mixture of Bell states. Again, this matches the intuition that the local state in $\mathcal{H}_B$ on which the BM is performed mimics a uniform mixture of Bell states. However, recall that this intuition is insufficient as a rigorous argument, since the condition for the state to have uniform probability for the Bell projections could be fulfilled by local states that are not a uniform mixture of Bell states (see Sec.~\ref{sec:physical-bell-measurements}).

\begin{lemma}
\label{lem:prob-success}
(Success probability on uniformly distributed outcomes) Let us assume a quantum state $\ket{\psi} \in \mathcal{H}_B \otimes \mathcal{H}_R$, where the quantum state in $\mathcal{H}_B$ is entirely within the two-qubit code space. We assume that measuring $ZZ$ and $XX$ on the subspace $\mathcal{H}_B$ has uniform probability for the four outcomes, where the observables act on the encoded qubits in $\mathcal{H}_B$.

Then, if we perform a physical BM on $\mathcal{H}_B$ the probability of a successful measurement result $s \in \{s\}$ is $p_{\varrho_B} (s)$. Therefore, the success probability of the physical BM $\mathbb{P}_{\ket{\psi} \bra{\psi}} (\{s\})$ is identical to a measurement on a uniform mixture of Bell states, $\mathbb{P}_{\ket{\psi} \bra{\psi}} (\{s\}) = \mathbb{P}_B$.
\end{lemma}
\begin{proof}
We use a general decomposition of $\ket{\psi}$:
\begin{equation}
	\ket{\psi} = \sum_{j=1}^4 \sum_{k=1}^d \lambda_{jk} \ket{\Phi_j}_B \ket{R_k}_R,
	\label{eq:psi-comp-b}
\end{equation}
with normalization
\begin{equation}
	\sum_{j=1}^4 \sum_{k=1}^d |\lambda_{jk}|^2 = 1,
\end{equation}
where we use $\{ \ket{\Phi_j}_B \}_{j \in \{1,2,3,4\}}$ as a basis for the encoded qubits in $\mathcal{H}_B$ and an arbitrary orthonormal basis $\{ \ket{R_k} \}_{k \in \{1,\dots,d\}}$ for the remaining quantum state in $\mathcal{H}_R$. We can use $\{ \ket{\Phi_j}_B \}_{j \in \{1,2,3,4\}}$ as a basis for the quantum state in $\mathcal{H}_B$ since the lemma demands the qubits on $\mathcal{H}_B$ to be entirely within the code space.

Following from the conditions the probabilities for the projections $\{ \ket{\Phi_r} \bra{\Phi_r}\}_{r \in \{1,2,3,4\} }$ onto the four simultaneous eigenspaces of the commuting set of observables $\{ XX, ZZ \}$ are uniformly distributed. We use this to obtain a condition on the coefficients $\lambda_{jk}$ of the quantum state,
\begin{equation}
\begin{aligned}
	\forall r \in \{1,2,3,4\}: \ \frac{1}{4} = & \braket{\Psi | \Phi_r} \braket{\Phi_r | \Psi} \\
	= & \sum_{j,i=1}^4 \sum_{k,l=1}^d \lambda_{jk}^* \lambda_{il} \braket{\Phi_j | \Phi_r} \braket{\Phi_r | \Phi_i} \braket{R_k | R_l} \\
	= & \sum_{j,i=1}^4 \sum_{k,l=1}^d \lambda_{jk}^* \lambda_{il} \delta_{jr} \delta_{ri} \delta_{kl} \\
	= & \sum_{k=1}^4 \lambda_{rk}^* \lambda_{rk} \\
	= & \sum_{k=1}^4 |\lambda_{rk}|^2,
\end{aligned}
	\label{eq:lambda}
\end{equation}
where we used the decomposition in Eq.~\eqref{eq:psi-comp-b}. As a next step we use Eq.~\eqref{eq:lambda} to calculate the probability of an arbitrary unambiguous BM result $s \in \{s\}$ as defined in Eq.~\eqref{eq:s}, 
\begin{widetext}
\begin{equation}
\begin{aligned}
	p_{\ket{\psi} \bra{\psi}} (s) & = \bra{A} \bra{\psi} U^\dagger P_s U \ket{\psi} \ket{A} \\
	& = \sum_{j,i=1}^4 \sum_{k,l=1}^d \lambda_{j,k}^* \lambda_{il} \bra{A} \bra{\Phi_j} U^\dagger P_s P_s U \ket{\Phi_i} \ket{A} \braket{R_k | R_l} \\
	& = \sum_{j,i=1}^4 \sum_{k,l=1}^d \lambda_{j,k}^* \lambda_{il} \bra{A} \bra{\Phi_j} U^\dagger P_s U \ket{\Phi_i} \ket{A} \delta_{kl} \\
	& = \sum_{j,i=1}^4 \bra{A} \bra{\Phi_j} U^\dagger P_s U \ket{\Phi_i} \ket{A} \sum_{k=1}^d \lambda_{j,k}^* \lambda_{ik} \\
	& = \bra{A} \bra{\Phi_\sigma} U^\dagger P_s U \ket{\Phi_\sigma} \ket{A} \sum_{k=1}^d |\lambda_{\sigma k}|^2\\
	& = \frac{1}{4} \bra{A} \bra{\Phi_\sigma} U^\dagger P_s U \ket{\Phi_\sigma} \ket{A} \\
	& = p_{\varrho_B} (s),
\end{aligned}
	\label{eq:ps}
\end{equation}
\end{widetext}
where we inserted Eqs.~\eqref{eq:lambda} and~\eqref{eq:def-pb-a} in the final steps. Finally, we insert Eqs.~\eqref{eq:ps} into~\eqref{eq:def-pb-b} to obtain:
\begin{equation}
	\mathbb{P}_{\ket{\psi} \bra{\psi}} (\{s\}) = \mathbb{P}_B.
\end{equation}
\end{proof}

We conclude this section by combining the results of Cor.~\ref{cor:pm-success} and Lem.~\ref{lem:prob-success}, which together directly prove Lem.~\ref{lem:successful-bell-measurment}, which was presented in Sec.~\ref{sec:physical-bell-measurements-on-entangled-quantum-states}.

\section{Observables that commute with $S^{(\mathbb{M})}$}
\label{app:observables-that-commute-with-s}
In this appendix, we provide a complete formal treatment, along with an illustrative discussion of measurements where the observable commutes with the current stabilizer. Before we explore these cases in detail and present a formal analysis in Lem.~\ref{lem:commuting-observables}, we want to provide a simple example of how logical measurements are performed and how they fit into our formalism. Let us consider one logical qubit encoded in QPC($2,2$). Note that for simplicity, we only consider one logical qubit in this example. The stabilizer group of this four-qubit code is generated by:
\begin{equation}
	G_c = \{ XXXX, ZZII, IIZZ \}.
\end{equation}
The relevant logical operators of this code are:
\begin{equation}
	\{ XXII, IIXX \} \subset [\overline{X}],
\end{equation}
\begin{equation}
	\{ ZIZI, ZIIZ, IZZI, IZIZ \} \subset [\overline{Z}].
\end{equation}
Therefore we can express the stabilizer group of a logical qubit in an $\overline{X}$ eigenstate as follows:
\begin{equation}
	\begin{aligned}
	S 	& = \langle G_c \cup \{ l_x XXII \} \rangle\\
		& = \langle XXXX, ZZII, IIZZ, l_x XXII \rangle,
	\end{aligned}
\end{equation}
where $l_x$ is a random variable and its value is the logical $\overline{X}$ information of the logical qubit. The operator $XXII$ is an arbitrarily chosen representative of $[\overline{X}]$. The first step to perform a logical $\overline{X}$ measurement is to measure the observable $M_1 = XIII$. The observable $M_1$ anticommutes only with the second stabilizer generator $ZZII$ and we follow that the measurement result $m_1$ has equal probability for both outcomes. We obtain the post-measurement state:
\begin{equation}
\begin{aligned}
	S^{(m_1 M_1)} 	& = \langle G_c^{(m_1 M_1)} \cup \{m_1 M_1 \} \cup \{ l_x XXII \} \rangle\\
					& = \langle XXXX, IIZZ, m_1 XIII ,l_x XXII \rangle,
	\end{aligned}
\end{equation}
where the second code stabilizer generator was replaced by the measured observable. Now we complete the logical $\overline{X}$ measurement by measuring the observable $M_2 = IXII$. We note that the observable $M_2$ commutes with the current stabilizer $S^{(m_1 M_1)}$. Therefore, $M_2$ is an element of the stabilizer up to the sign $m_2$. We deduce the sign $m_2$ by finding the unique decomposition of $m_2 M_2$ in terms of the current stabilizer generators:
\begin{equation}
	m_2 M_2 = m_1 M_1 l_x XXII,
	\label{eq:example-composition}
\end{equation}
where we used that $M_1 M_2 = XXII \in [\overline{X}]$. We rearrange Eq.~\eqref{eq:example-composition} to
\begin{equation}
	m_1 m_2 M_1 M_2 = l_x XXII,
\end{equation}
obtaining our logical measurement outcome $m_1 m_2 = l_x$. In conclusion, we performed a logical $\overline{X}$ measurement by decomposing a logical $\overline{X}$ operator into a set of measurements: $M_1 M_2 = XXII = \overline{X}$, and then multiplying their outcomes $l_x = m_1 m_2$. In essence, we measured the logical operator $XXII \in [\overline{X}]$ by decomposing it into two single-qubit Pauli measurements $XIII$ and $IXII$ using the fact that $XIII \times IXII = XXII \in [\overline{X}]$. Expanding on the intuition build in this example we now present a general and formal treatment in Lem.~\ref{lem:commuting-observables}.

\begin{lemma}
\label{lem:commuting-observables}
(Observables that commute with $S^{(\mathbb{M})}$) Let us assume a current stabilizer state as defined in Eq.~\eqref{eq:general-state}:
\begin{equation}
	S^{(\mathbb{M})} = \langle G_c^{(\mathbb{M})} \cup \mathbb{M} \cup L^{(\mathbb{M})}_{x,z} \rangle.
	\tag{\ref{eq:general-state}}
\end{equation}
If an observable which commutes with the current stabilizer $S^{(\mathbb{M})}$ is measured, it will have no effect on the global quantum state and the outcome is predetermined by the stabilizer $S^{(\mathbb{M})}$ with unit probability. This implies two possibilities: either the observable is uncorrelated with the logical information or it is correlated. In the case where the observable is uncorrelated with the logical information, we know the outcome beforehand and we obtain no additional knowledge from the measurement; we merely completed measuring a code stabilizer in $\langle G_c \rangle$. In the case where the observable is correlated with the logical information the outcome is predetermined in one-to-one correspondence by one of the logical variables $\{ l_x, l_y, l_z \}$. Therefore, assuming no prior knowledge of the logical information, both outcomes are equally likely, since the logical variables are uniformly distributed. Consequently, the measurement will yield the value of one of the three logical variables. In other words, a logical measurement of one of the logical operators $\{ \overline{XX}, \overline{YY}, \overline{ZZ} \}$ is performed by measuring the observable.

\end{lemma}
\begin{proof}
Per assumption in the lemma we measure an observable $M$ with result $m$ which commutes with $S^{(\mathbb{M})}$. Thus we know that $mM \in S^{(\mathbb{M})}$. As we discussed in Sec.~\ref{sec:pauli-measurements-on-encoded-uniform-mixtures-of-bell-states}, we can use the reduced set of measurements $\mathbb{M}_r$ to ensure that we have a minimal generating set for the current stabilizer group $G_c^{{\mathbb{M}}} \cup \mathbb{M}_r \cup L_{x,z}^{(\mathbb{M})}$. Therefore, $mM$ has a unique decomposition in terms of this minimal generating set. Using the unique way $mM$ can decomposed using elements exclusively from the three subsets, we denote its decomposition as:
\begin{equation}
	mM = \gamma \mu \nu,
	\label{eq:measurement-decomposition}
\end{equation}
where $\gamma \in \langle G^{(\mathbb{M})}_c \rangle $, $\mu \in \langle \mathbb{M}_r \rangle$ and $\nu \in \langle L^{(\mathbb{M})}_{x,z}\rangle $. Recall that $\langle G^{(\mathbb{M})}_c \rangle \subseteq \langle G_c \rangle $ are code stabilizers. We recall from Lem.~\ref{lem:standard-form}, that the sign of every element in $\langle L^{(\mathbb{M})}_{x,z}\rangle$ is determined by one of the three logical variables $l_x$, $l_y$, or $l_z$. Additionally, Lem.~\ref{lem:standard-form} established that every element in $\langle G^{(\mathbb{M})}_c \rangle$ and $\langle \mathbb{M}_r \rangle$ is uncorrelated with the logical variables.

Now we will consider the two cases $\nu = I$ and $\nu \neq I$ separately. In the first case, $\nu = I$, Eq.~\eqref{eq:measurement-decomposition} simplifies to $mM = \gamma \mu$. The operator $\gamma \mu$ is a product of code stabilizers and observables of prior measurements, and is uncorrelated to the logical variables $l_x$ and $l_z$. Recall, that the logical variables are the only unknowns of the quantum state. Thus, in the case where $\nu = I$ we know the measurement result $m$ prior to the measurement and we learn nothing from the measurement. In the second case, where $\nu \neq I$, the decomposition of $mM$ is given by $mM = \gamma \mu \nu$, and we observe that the measurement result $m$ is uniquely determined by $l_x$, $l_y$, or $l_z$, depending on the factors of the logical stabilizers $L_{x,z}^{(\mathbb{M})}$ in the decomposition of $\nu$.

Therefore, from the measurement result $m$ we learn the value of the respective random variable. In other words, we conclude a logical measurement of either $\overline{XX}$, $\overline{YY}$, or $\overline{ZZ}$. To understand this logical measurement better let us rearrange Eq.~\eqref{eq:measurement-decomposition}:
\begin{equation}
	m M \mu = \gamma \nu.
	\label{eq:logical-measurement}
\end{equation}
The rhs of Eq.~\eqref{eq:logical-measurement} consists of code stabilizers and logical stabilizers. Therefore the rhs is itself a logical operator:
\begin{equation}
	\gamma \nu \in [\overline{XX}] \cup [\overline{YY}] \cup [\overline{ZZ}],
\end{equation}
with its sign determined by one of the logical variables, $l_x$, $l_y$ or $l_z$. The lhs of Eq.~\eqref{eq:logical-measurement} is the product of measured observables. Thus, upon examination of Eq.~\eqref{eq:logical-measurement}, it is apparent that the performed measurements $m M \mu$ constitute the logical measurement $\gamma \nu \in [\overline{XX}] \cup [\overline{YY}] \cup [\overline{ZZ}]$.

\end{proof}

The results of Lem.~\ref{lem:commuting-observables} should be unsurprising. Essentially, any observable $O$ can be factored into a set of observables $\{o_i\}$ via $O = \prod_i o_i$. This allows us to measure the set of observables ${o_i}$, with the product of their eigenvalues determining the eigenvalue of $O$.  Thus, in principle, any logical measurement can be decomposed into a set of obsevables which are simpler to measure, e.g. a decomposition into single-qubit Pauli measurements.

\section{Proof: Lem.~\ref{lem:acomm-logicals} of Sec.~\ref{sec:logical-operators-on-encoded-uniform-mixtures-of-bell-states}}
\label{app:proof-lem-acomm-logicals}

In this appendix, we provide the formal proof of Lem.~\ref{lem:acomm-logicals}.
\acommlogicalsrestatable*
\begin{proof}
Let us choose the two logical operators $\overline{XX} \in [\overline{XX}]$ and $\overline{ZZ} \in [\overline{ZZ}]$, which together form a logical BM. Without loss of generality, we choose these operators, but any pair from the set $\{ \overline{XX}, \overline{YY}, \overline{ZZ} \}$ would equally constitute a logical BM and could have been chosen for the proof of this lemma. Recall from Sec.~\ref{sec:encoded-uniform-mixture-of-bell-states} that these three possible pairs comprise all possible Pauli measurements that, when measured jointly, project onto the Bell basis.

Since the two logical qubits are encoded in independent stabilizer codes on disjoint sets of qubits, we may use their unique factorization into single-code logical operators:
\begin{equation}
	\begin{aligned}
	\overline{XX} 	& = \overline{X}_1 \otimes \overline{X}_2 \\
					& = \left( \overline{X}_1 \otimes I^{\otimes n_2} \right) \left( I^{\otimes n_1} \otimes \overline{X}_2 \right),
	\end{aligned}
\end{equation}
and similarly,
\begin{equation}
	\begin{aligned}
	\overline{ZZ} 	& = \overline{Z}_1 \otimes \overline{Z}_2 \\
					& = \left( \overline{Z}_1 \otimes I^{\otimes n_2} \right) \left( I^{\otimes n_1} \otimes \overline{Z}_2 \right).
	\end{aligned}
\end{equation}
Because in each factorization the first factor has support exclusively on the first code and the second factor exclusively on the second code, we make the following observations. Each of the pairs
\begin{equation}
	\left( \overline{X}_1 \otimes I^{\otimes n_2} \right), \ \left( I^{\otimes n_1} \otimes \overline{Z}_2 \right)
\end{equation}
and
\begin{equation}
	\left( I^{\otimes n_1} \otimes \overline{X}_2 \right), \ \left( \overline{Z}_1 \otimes I^{\otimes n_2} \right),
\end{equation}
anticommute in zero qubits. However, using the canonical single-code Pauli commutation relation
\begin{equation}
	\acomm{\overline{X}}{\overline{Z}}=0 \quad \text{for all } \overline{X} \in [\overline{X}], \ \overline{Z} \in [\overline{Z}],
\end{equation}
we deduce that each of the pairs
\begin{equation}
	\left( \overline{X}_1 \otimes I^{\otimes n_2} \right), \ \left( \overline{Z}_1 \otimes I^{\otimes n_2} \right),
\end{equation}
and
\begin{equation}
	\left( I^{\otimes n_1} \otimes \overline{X}_2 \right), \ \left( I^{\otimes n_1} \otimes \overline{Z}_2 \right)	,
\end{equation}
anticommute. Therefore, they anticommute in an odd number of qubits in their respective codes. Thus, combining the two observations we conclude that any two logical operators $\overline{XX}$ and $\overline{ZZ}$ anticommute in an odd number of qubits in each code.
\end{proof}

\section{Derivation of the single-code reduction}
\label{app:single-code-reduction}

In Sec.~\ref{sec:single-code-reduction}, we introduced the symmetry in the stabilizer generators during the first part of our measurement schemes and showed how this symmetry can be exploited to reduce our schemes to a single-code picture. In this appendix, we provide a more detailed technical derivation of this symmetry and the single-code reduction. If two operators have support exclusively on the first or the second code, respectively, and are the same operators including the sign, we refer to them as being the same or identical in their respective codes. In the following we prove the symmetry in the stabilizer generators during the first part of the scheme.

\begin{lemma}
\label{lem:symmetry}
(Symmetry of the stabilizer generators) We assume two logical qubits encoded in identical stabilizer codes and the exclusive use of transversal BMs and that we never measure an observable which anticommutes with at least one element of $L^{(\mathbb{M})}_{x,z}$ and commutes with the rest of the generators of the current stabilizer group $G_c^{(\mathbb{M})} \cup \mathbb{M}$. Under these assumptions, the current stabilizer generators, up to and including the first successful BM, can always be chosen to exhibit the following symmetry.

Every element of the stabilizer generators that is not transversal has support on only one of the two codes, and there exists another element in the generators which is the same operator on the other code up to a sign. The signs of these two operators may only differ for elements of the measurements $\mathbb{M}$, but not for the elements of the code generators $G_c^{\mathbb{M}}$. Furthermore, the logical generators $L_{x,z}^{(\mathbb{M})}$ can be chosen to be transversal at all times.
\end{lemma}
\begin{proof}
The initial stabilizer generators consist of the sets $G_c = G_1 \cup G_2$ and $L_{x,z}$. Assuming identical codes, the code stabilizer generators $G_1$ and $G_2$ are identical on both codes. Since the two logical qubits are encoded in independent stabilizer codes on disjoint sets of qubits, we may use the unique factorization of $l_x \overline{XX}$ in $L_{x,z}$ into single-code logical operators:
\begin{equation}
	\overline{XX} = \overline{X}_1 \otimes \overline{X}_2.
\end{equation}
Since the codes are identical, we may choose $\overline{X}_1$ and $\overline{X}_2$ to be identical on both codes, which makes $\overline{XX}$ a transversal operator. The argument applies analogously to $\overline{ZZ}$.

Having established that the symmetry is fulfilled for the initial state, we will prove by induction that it is preserved under transversal BMs up to and including the first successful BM. We denote the current stabilizer as:
\begin{equation}
	S = \langle G_1^{(\mathbb{M})} \cup G_2^{(\mathbb{M})} \cup \mathbb{M} \cup L_{x,z}^{(\mathbb{M})} \rangle,
\end{equation}
where $G_1^{(\mathbb{M})} \subseteq G_1$ and $G_2^{(\mathbb{M})} \subseteq G_2$. We start by examining how the stabilizer generators transform under a partial BM. Recall from Lem.~\ref{lem:forbidden-measurements} that, similarly to Lem.~\ref{lem:standard-form}, we exclude measurements which anticommute with at least one element of $L^{(\mathbb{M})}_{x,z}$ and commute with the rest of the current stabilizer generators $G_c^{(\mathbb{M})} \cup \mathbb{M}$ since they irreversibly destroy logical information. Let us denote the two single-qubit observables of the partial BMs as $M_ 1 = b_j \otimes I^{\otimes n_2}$ and $M_2 = I^{\otimes n_1} \otimes b_j$, which are identical in their respective codes, with $b_j \in \{X_j, Y_j, Z_j\}$ and $j$ labeling the qubit pair on which the BM is performed.

Since measuring an operator which commutes with the current stabilizer does not change the quantum state we only need to consider the case where at least one of $M_1$ and $M_2$ anticommute with an element of the current stabilizer generators. We note, that if $M_1$ anticommutes with an element $g_1 = g \otimes I^{\otimes n_2}$ of $G_1^{(\mathbb{M})}$ then $M_2$ anticommutes with an element $g_2 = I^{\otimes n_1} \otimes g$ of $G_2^{(\mathbb{M})}$ where $g_1$ and $g_2$ are the same operators on the first and second code, respectively. Therefore, measuring $M_1$ and $M_2$ replaces $g_1$ and $g_2$ with $m_1 M_1$ and $m_2 M_2$ in the stabilizer generator, where $m_1$ and $m_2$ are the measurement outcomes of $M_1$ and $M_2$, respectively.

Now, we consider the case where $M_1$ also anticommutes with an element $L$ in $L_{x,z}^{(\mathbb{M})}$. Since $M_2$ is identical to $M_1$ on the other code and $L$ is transversal, it follows that $M_2$ also anticommutes with $L$. Consequently, we replace $L$ with $g_1 g_2 L$ in the stabilizer generators. Note that the new logical operator $g_1 g_2 L$ lies in the same logical coset as $L$, since they differ only by code stabilizers. The new generator $g_1 g_2 L$ is transversal, as $g_1$ and $g_2$ are identically on their respective codes and $L$ is transversal. Lastly, $g_1 g_2 L$ commutes with $M_1$ and $M_2$. Thus, we can always transform the elements of $L_{x_z}^{(\mathbb{M})}$ to commute with the observables of the partial BMs while keeping them transversal. We complete the argument by noting that the sets $G_1^{(\mathbb{M})}$ and $G_2^{(\mathbb{M})}$ transform identically under the partial BM, preserving their symmetry.

As the final step, we examine how the stabilizer generators transform under a successful BM. The transversal observables measured by the BM commute with $L_{x,z}^{(\mathbb{M})}$, as all its elements are also transversal. If an element $g_1 = g \otimes I^{\otimes n_2}$ of $G_1^{(\mathbb{M})}$ anticommutes with a transversal observable $B$ than the same element $g_2 = I^{\otimes n_1} \otimes g$ of $G_2$ on the other code anticommutes with $B$ as well. Thus, the measurement replaces the code stabilizers $g_1$ and $g_2$ in the stabilizer generators with $g_1 g_2$ which is a transversal operator, since $g_1$ and $g_2$ are the same operators in their respective codes.
\end{proof}

The following Lem.~\ref{lem:transversal-bm} from the main text is a direct consequence of Lem.~\ref{lem:symmetry}. It states that partial BMs act as single-qubit measurements in the single-code picture.

\transversalbmrestatable*
\begin{proof}
In the proof of Lem.~\ref{lem:symmetry} we already established, that the current code stabilizers $G_1^{\mathbb{M}}$ and $G_2^{\mathbb{M}}$ have support exclusively on their respective codes and are identical on their respective codes. We again define the two single-qubit observables of the partial BMs as $M_ 1 = b_j \otimes I^{\otimes n_2}$ and $M_2 = I^{\otimes n_1} \otimes b_j$, with $b_j \in \{X_j, Y_j, Z_j\}$ and $j$ labeling the qubit pair on which the BM is performed. Thus, these two measurements, $M_1$ and $M_2$, act identically and independently on the two current code stabilizer generator sets, $G_1^{\mathbb{M}}$ and $G_2^{\mathbb{M}}$, as single-qubit observables $b_j$.
\end{proof}

\section{Proof: Thm.~\ref{thm:sufficient} (Sufficient conditions for an optimal logical Bell measurement)}
\label{app:proof-thm:sufficient}

\sufficientrestatable*
\begin{proof}
To make the derivation more explicit, we give the proof in the two-code picture. We split this proof into two parts. First, we demonstrate that the success probability for each transversal BM up to the first successful one is $\mathbb{P}_B$. Second, we prove that once a success occurs, the observables completing the logical BM can be obtained with probability one.

Recall, that we defined the logical two-qubit code using the trivial extension of the generators of the single-qubit code $G_c = \{g_s\}_{s \in \{ 1, \dots, n-1 \}}$:
\begin{equation}
	G_1 = \{ g_{1,s} \}_{s \in \{ 1, \dots, n-1\}} = \{ g_{s} \otimes I^{\otimes n} \}_{s \in \{1, \dots, n-1\}},
\end{equation}
\begin{equation}
	G_2 = \{ g_{2,s} \}_{s \in \{ 1, \dots, n-1 \}} = \{ I^{\otimes n} \otimes g_{s} \}_{s \in \{ 1, \dots, n-1 \}}.
\end{equation}

We start by analyzing the scheme before the first success occurred. The scheme proceeds in steps, with one transversal BM performed at each step. We define the set of measurements up to step $j$ as $\mathbb{M}_j$. For brevity, we denote $G^{(\mathbb{M}_j)}$ as $G^{(j)}$, $G_i^{(\mathbb{M}_j)}$ as $G_i^{(j)}$, $S^{(\mathbb{M}_j)}$ as $S^{(j)}$, $\overline{XX}^{(\mathbb{M}_j)}$ as $\overline{XX}^{(j)}$ and $\overline{ZZ}^{(\mathbb{M}_j)}$ as $\overline{ZZ}^{(j)}$. The stabilizer group of the initial logical uniform mixture of Bell states is given by:
\begin{equation}
\begin{aligned}
	S = S^{(0)} = \langle 	& \{ g_{1,s} \}_{s \in \{ 1, \dots, n-1 \}} \\
							& \cup \{ g_{2,s} \}_{s \in \{ 1, \dots, n-1 \}} \\
							& \cup \{ l_x \overline{XX}, l_z \overline{ZZ} \} \rangle.
\end{aligned}
	\label{eq:initial-scheme}
\end{equation}
In the following we will show by induction that, as long as no success occurred, the current stabilizer $S^{(j)} = \langle G^{(j)} \rangle$ at step $j$ is generated by
\begin{equation}
\begin{aligned}
	G^{(j)} = 		& \{ g_{1,s} \}_{s \in \{ j+1, \dots, n-1 \} } \\
					& \cup \{ g_{2,s} \}_{s \in \{ j+1, \dots, n-1 \} } \\
					& \cup \mathbb{M}_j \cup \{ l_x \overline{XX}^{(j)}, l_z \overline{ZZ}^{(j)} \}.
\end{aligned}
	\label{eq:current-scheme}
\end{equation}
We identify the following sets of the current code stabilizers for each code:
\begin{equation}
	G_1^{(j)} = \{ g_{1,s} \}_{s \in \{ j+1, \dots, n-1 \} }
\end{equation}
and
\begin{equation}
	G_2^{(j)} = \{ g_{2,s} \}_{s \in \{ j+1, \dots, n-1 \} }.
\end{equation}
The base case of Eq.~\eqref{eq:current-scheme} is already proven in Eq.~\eqref{eq:initial-scheme}. For the induction step, we assume that Eq.~\eqref{eq:current-scheme} holds for $j-1$. A partial transversal BM at step $j$ measures the single-qubit observables
\begin{equation}
	b_j \otimes I^{\otimes n}
\end{equation}
and
\begin{equation}
	I^{\otimes n} \otimes b_j.
\end{equation}
Using Lem.~\ref{lem:symmetry} we choose $\overline{XX}^{(j-1)}$ to be transversal at all times. Then, if the logical operator $\overline{XX}^{(j-1)}$ anticommutes with either $b_j \otimes I^{\otimes n}$ or $I^{\otimes n} \otimes b_j$ it anticommutes with both. In this case, we multiply the logical operator with the observables:
\begin{equation}
	\overline{XX}^{(j)} = b_j \otimes b_j \times \overline{XX}^{(j-1)},
\end{equation}
so that the new logical operator $\overline{XX}^{(j)}$ remains transversal and commutes with $b_j \otimes I^{\otimes n}$ and $I^{\otimes n} \otimes b_j$. We proceed analogously with $\overline{ZZ}^{(j-1)}$. This transformation of the logical operators does not change the current stabilizer group
\begin{equation}
	S^{(j-1)} = \langle \tilde{G}^{(j-1)} \rangle,
\end{equation}
where
\begin{equation}
\begin{aligned}
	\tilde{G}^{(j-1)} =	& \{ g_{1,s} \}_{s \in \{ j, \dots, n-1 \} } \\
						& \cup \{ g_{2,s} \}_{s \in  \{ j, \dots, n-1 \} } \\
						& \cup \mathbb{M}_{j-1} \cup \{ l_x \overline{XX}^{(j)}, l_z \overline{ZZ}^{(j)} \}.
\end{aligned}
\end{equation}
Recall that per assumptions of our schemes all elements of the measured observables $\mathbb{M}_j$ commute with each other. From conditions~1 and~2, it follows that $g_{1,j}$ is the only element of $\tilde{G}^{(j-1)}$ that anticommutes with $b_j \otimes I^{\otimes n}$. Thus, $b_j \otimes I^{\otimes n}$ replaces $g_{1,j}$ in $G^{(j)}$. Using the analogous argument for the second observable, we reach the final state defined in Eq.~\eqref{eq:current-scheme}. The set $\mathbb{M}_j$ contains all partial BM results up to the $j$-th step:
\begin{equation}
\begin{aligned}
	\mathbb{M}_j = \{ 	& r_{1,1} b_1 \otimes I^{\otimes n}, r_{2,1} I^{\otimes n} \otimes b_1, \dots, \\
						& r_{1,j} b_j \otimes I^{\otimes n}, r_{2,j} I^{\otimes n} \otimes b_j \},
\end{aligned}	
\end{equation}
where $r_{i,k}$ is the $k$th measurement result on code $i$.

Having established the current state at every step before the first success, we now proceed to discuss the transversal operators on the $j$-th qubit pair. Any measurement that reveals logical information commutes with the current stabilizer as shown in Lem.~\ref{lem:commuting-observables}. From condition~1, we deduce that all observables measured from partial BMs $b_j \otimes I^{\otimes n}$ and $I^{\otimes n} \otimes b_j$ anticommute with the current stabilizer. Consequently, without a successful transversal BM, we never measure a logical operator.

The three transversal operators are $b_j \otimes b_j$ together with $\tilde{b}_j \otimes \tilde{b}_j$ for $\tilde{b}_j \in \{X_j, Y_j, Z_j\} \setminus \{b_j\}$. Thus, the set of these three operators is always $\{ X_j \otimes X_j, Y_j \otimes Y_j, Z_j \otimes Z_j \}$. For the next argument, we use the ordering of the sequence $\mathbb{C}$ to refer to the elements of the sets $G_j^{(j)}$, i.e. $g_{1,j} = c_{j} \otimes I^{\otimes n}$ and $g_{2,j} = I^{\otimes n} \otimes c_{j}$. Let us first examine the case $j<n$. Condition~1 states that $b_j \otimes b_j$ always anticommutes with $g_{1,j}$ and $g_{2,j}$ and therefore does never commute with the current stabilizer $S^{(j)}$. It follows from condition~3 that $\tilde{b}_j \otimes \tilde{b}_j \in\{ X_j, Y_j, Z_j \} \setminus \{ b_j \}$ either anticommutes with elements of the current stabilizer $c_k \otimes I^{\otimes n} \in \{G_1^{(j)}\}$ and $I^{\otimes n} \otimes c_k\in \{G_2^{(j)}\}$ or completes a transversal logical operator. The first case follows directly from Eqs.~\eqref{eq:con-3-a} and~\eqref{eq:current-scheme}. Eq.~\eqref{eq:con-3-a} shows that $b_j$ anticommutes with $c_k$ for $k \geq j$, and Eq.~\eqref{eq:current-scheme} confirms that this $c_k$ belongs to the current stabilizer in step $j$. The second case becomes clear upon translating Eq.~\eqref{eq:con-3-b} into the two-code picture:
\begin{equation}
\begin{aligned}
	& \mu \in \langle b_1, \dots, b_{j-1} \rangle \wedge \mu \tilde{b}_j \in \left[ \overline{X} \right] \cup \left[ \overline{Y} \right] \cup \left[ \overline{Z} \right] \\
	& \Rightarrow \mu \otimes \mu \in \langle \mathbb{M}_j \rangle \\
	& \qquad \wedge  \mu \otimes \mu \times \tilde{b}_j \otimes \tilde{b}_j \in \left[ \overline{X} \overline{X} \right] \cup \left[ \overline{Y} \overline{Y} \right] \cup \left[ \overline{Z} \overline{Z} \right].
\end{aligned}
\end{equation}
For the case $j=n$ we infer from Eq.~\eqref{eq:current-scheme} that the current stabilizer generator after the $(n-1)$-th transversal BM is given by: $G^{(n-1)} = \mathbb{M}_{n-1} \cup \{ l_x \overline{XX}^{(n-1)}, l_z \overline{ZZ}^{(n-1)} \}$. Furthermore, recall that any operator commuting with the current stabilizer generators can be decomposed in terms of the same generators. Therefore, if $(n-1)$ qubit pairs where measured with partial outcomes, a transversal operator on the $n$-th qubit pair which commutes with the current stabilizer completes a logical measurement in $[\overline{XX}] \cup [\overline{YY}] \cup  [\overline{ZZ}]$. We recall that the outcomes of the logical operators $\left[ \overline{X} \overline{X} \right]$, $\left[ \overline{Y} \overline{Y} \right]$ and $\left[ \overline{Z} \overline{Z} \right]$ are uniformly distributed. In conclusion, for all $j$, each of the three transversal operators either anticommutes with the current stabilizer or completes a logical operator, and in both cases their outcomes are equally likely. Thus, by Lem.~\ref{lem:successful-bell-measurment}, the success probability of all transversal BMs up to the first success is $P_B$.

In the second part of the proof, we demonstrate that once a transversal BM is successful, the logical information can be obtained with probability one. Let us assume that a success occurred at index $s$. We will now discuss how we can measure $\overline{X}_s \otimes \overline{X}_s$ and $\overline{Z}_s \otimes \overline{Z}_s$ with probability one.

We refer back to the decompositions defined in Eqs.~\eqref{eq:logical-decomp-x} and~\eqref{eq:logical-decomp-z}:
\begin{equation}
	\overline{X}_s = \bigotimes_{t=1}^{n} u_{s,t}, \quad \text{where } u_{s,t} \ \in \{ I, X, Y, Z \},
\end{equation}
\begin{equation}
	\overline{Z}_s = \bigotimes_{t=1}^{n} v_{s,t}, \quad \text{where } v_{s,t} \in \{ I, X, Y, Z \}.
\end{equation}
Therefore, the two-code decompositions are given by:
\begin{equation}
\begin{aligned}
	\overline{X}_s \otimes \overline{X}_s = \bigotimes_{t=1}^{n} \left( u_{s,t} \otimes u_{s,t}, \right), \\
	\text{where } u_{s,t} \in \{ I, X, Y, Z \},
\end{aligned}
\end{equation}
\begin{equation}
\begin{aligned}
	\overline{Z}_s \otimes \overline{Z}_s = \bigotimes_{t=1}^{n} \left( v_{s,t} \otimes v_{s,t}, \right), \\
	\text{where } v_{s,t} \in \{ I, X, Y, Z \}.
\end{aligned}
\end{equation}
We separate the decompositions into three terms as follows:
\begin{equation}
\begin{aligned}
	\overline{X}_s \otimes \overline{X}_s = & \bigotimes_{t=1}^{s-1} \left( u_{s,t} \otimes u_{s,t} \right) \\
											& \otimes \left( u_{s,s} \otimes u_{s,s} \right) \\
											& \bigotimes_{t=s+1}^{n} \left( u_{s,t} \otimes u_{s,t} \right),
\end{aligned}
\end{equation}
\begin{equation}
\begin{aligned}
	\overline{Z}_s \otimes \overline{Z}_s = & \bigotimes_{t=1}^{s-1} \left( v_{s,t} \otimes v_{s,t} \right) \\
											& \otimes \left( v_{s,s} \otimes v_{s,s} \right) \\
											& \bigotimes_{t=s+1}^{n} \left( v_{s,t} \otimes v_{s,t} \right).
\end{aligned}
\end{equation}
For each of these decompositions, the first term has support exclusively on qubits that have already been measured, the second term has support on the qubit pair where a successful transversal BM occurred, and the third term has support exclusively on qubits that have not been measured yet. Condition~4 ensures that the logical operators do not conflict with previous measurements. Condition~5 guarantees that the $s$-th qubit pair is the sole pair requiring double information, and that the logical operators do not conflict with any unmeasured qubits. We will now verify this in a technical proof. We discuss each term individually.

It follows from condition~4 that the single code logical operator $\overline{X}_s$ commutes with all measured operators:
\begin{equation}
	\forall k<s : \quad \comm{\overline{X}_s}{b_k} = 0.
\end{equation}
Since $b_k$ are single-qubit operators i.e. they have support only on $k$, we deduce that the measurements commute with the decomposition in $k$:
\begin{equation}
	\forall k<s : \quad \comm{u_{s,k}}{b_k} = 0.
\end{equation}
We recall from Eq.~\eqref{eq:B-def}, that per definition $b_j$ cannot be an identity, since $\forall j : b_j \in \{ X_j, Y_j, Z_j \}$. Therefore, for every $k<s$, $u_{s,k}$ is either the identity $u_{s,k} = I$ or identical to $u_{s,k} = b_k$. Thus, the first term of $\overline{X}_s \otimes \overline{X}_s$ can be decomposed from prior measurements:
\begin{equation}
	\bigotimes_{t=1}^{s-1} \left( u_{s,t} \otimes u_{s,t} \right) \in \langle \mathbb{M}_{s-1} \rangle.
\end{equation}
The discussion for the first term $\bigotimes_{t=1}^{s-1} \left( v_{s,t} \otimes v_{s,t} \right)$ of the logical operator $\overline{Z}_s \otimes \overline{Z}_s$ proceeds analogously and leads to the same conclusion. The argument for the second terms is straightforward. Recall that two operators are said to conflict on a qubit if their decompositions into single-qubit Pauli operators require different Pauli information in that qubit. The decompositions of the two logical operators conflict in the $s$-th qubit pair, which is stated in Eq.~\eqref{eq:double-info}. However, since a successful transversal Bell measured occurred at index $s$, the operators $u_{s,t} \otimes u_{s,t}$ and $v_{s,s} \otimes v_{s,s}$ were both measured. Finally, it follwos from condition~5 that the logical operators $\overline{X}_s \otimes \overline{X}_s$ and $\overline{Z}_s \otimes \overline{Z}_s$ do not conflict on any unmeasured qubit. Since the logical operators $\overline{X}_s \otimes \overline{X}_s$ and $\overline{Z}_s \otimes \overline{Z}_s$ are transversal by definition, the unmeasured portions of these operators can be obtained either through transversal BMs, which always yield the partial logical information, or through single-qubit Pauli measurements.

In conclusion, we have shown that each transversal BM up to the first success has a success probability of $\mathbb{P}_B$, and that if any of them succeeds, the logical BM can be completed with probability one. For two identical $n$-qubit codes, we achieve the upper-bound success probability given in Thm.~\ref{thm:bound}: $1-(1-\mathbb{P}_B)^{n}$.
\end{proof}

\section{Rectangular rotated planar surface code}
\label{app:rectangular-rotated-planar-surface-code}
Our scheme for the rotated planar surface code extends naturally to the rectangular code with $r \neq m$. A small adaptation for the additional diagonals is sufficient to achieve an optimal logical Bell measurement. Without loss of generality we assume $r<m$. For the rectangular code, two cases arise depending on the parity of $r+m$. If $r+m$ is even, the number of diagonals is odd, and a middle diagonal exists, just as in the quadratic case. If $r+m$ is odd, the number of diagonals is even, so no single middle diagonal exists, and the bottom-right plaquette is a $Z$-plaquette instead of an $X$-plaquette. Our scheme generalizes seamlessly to both cases.

In the first part of the scheme, the top-left triangle, i.e., the diagonals up to the one with index sum $r$, and its bottom-right counterpart are measured, as illustrated in Figs.~\ref{fig:rotated-surface-rect-solution-A} and~\ref{fig:rotated-surface-rect-solution-B}. This part of the scheme is identical to the quadratic case.

\begin{figure*}
		\begin{subfigure}[c]{0.45\textwidth}
			\def\svgwidth{\textwidth}
			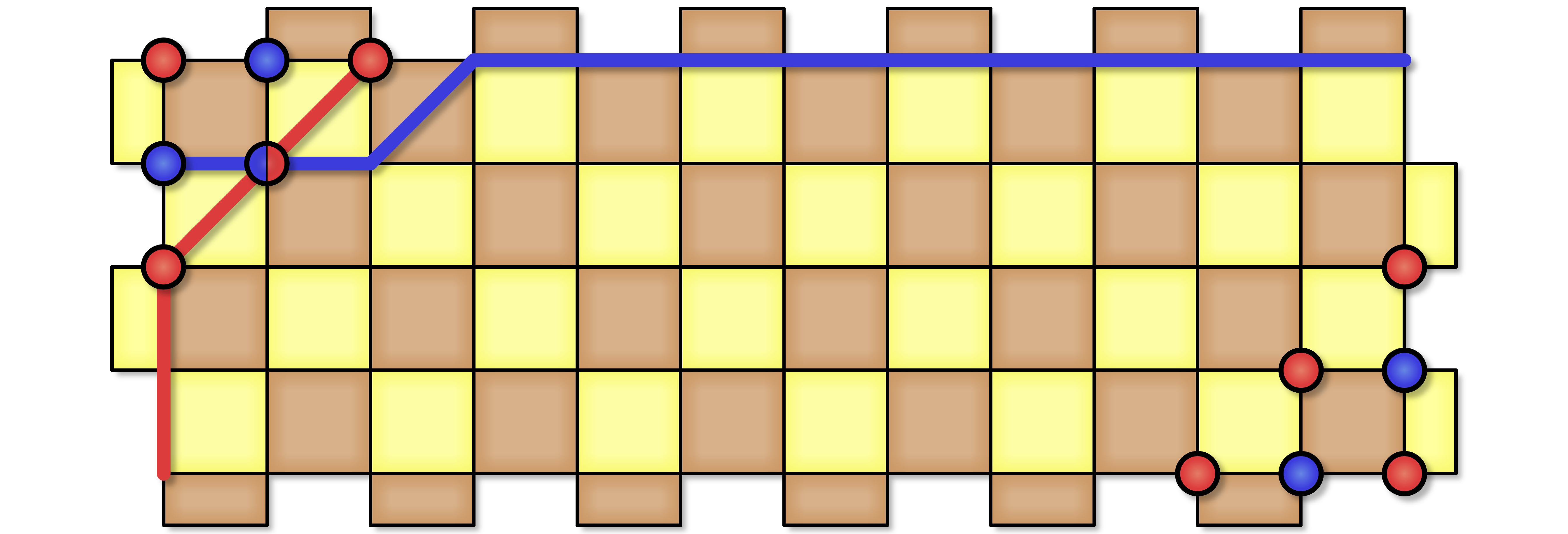
			\caption{\label{fig:rotated-surface-rect-solution-A}}
		\end{subfigure}									
		\begin{subfigure}[c]{0.45\textwidth}
			\def\svgwidth{\textwidth}
			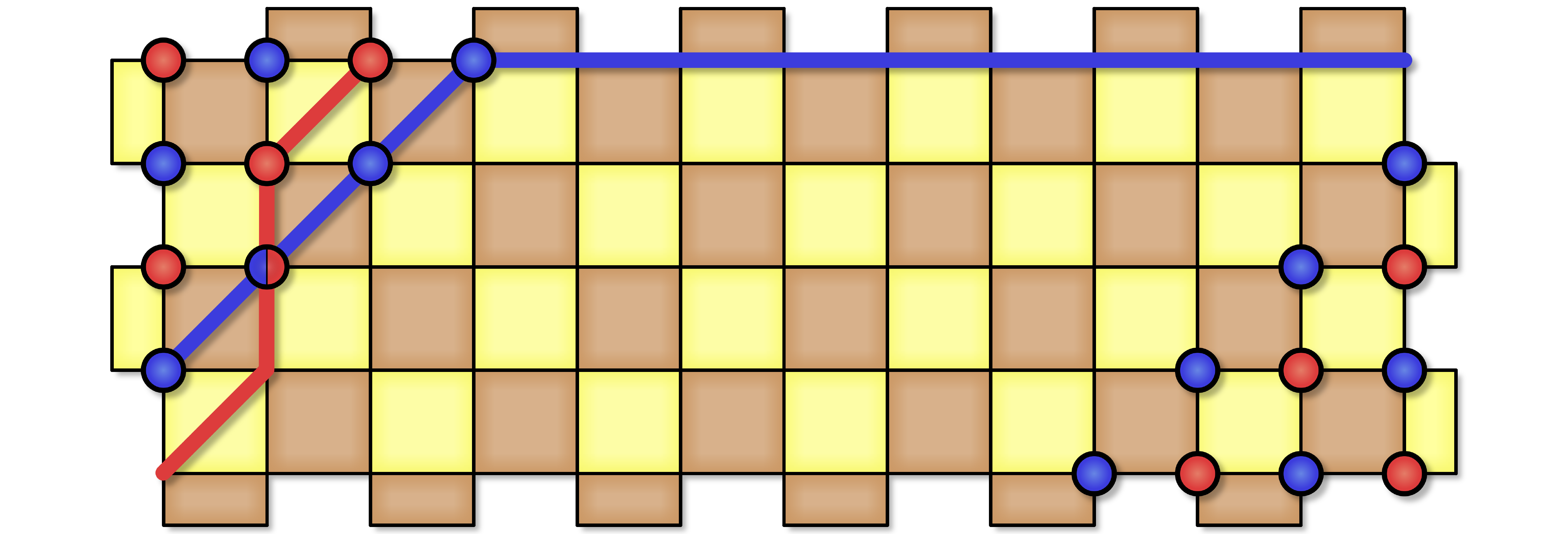
			\caption{\label{fig:rotated-surface-rect-solution-B}}
		\end{subfigure}
		\begin{subfigure}[c]{0.45\textwidth}
			\def\svgwidth{\textwidth}
			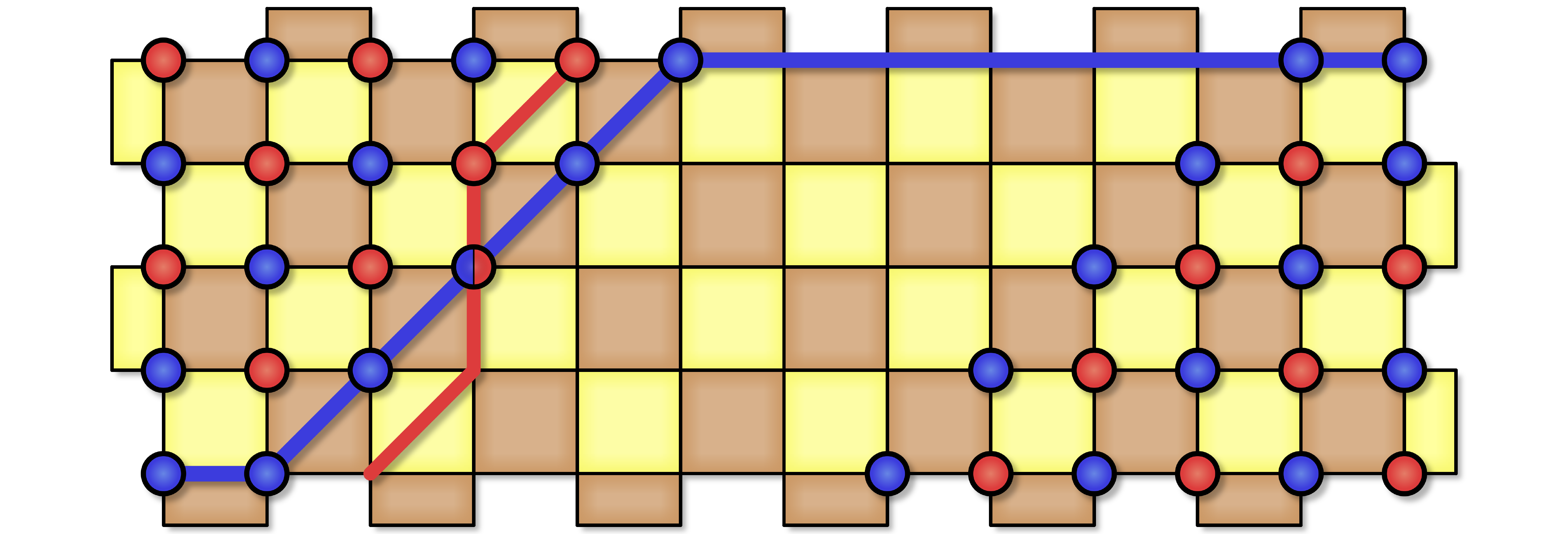
			\caption{\label{fig:rotated-surface-rect-solution-C}}
		\end{subfigure}									
		\begin{subfigure}[c]{0.45\textwidth}
			\def\svgwidth{\textwidth}
			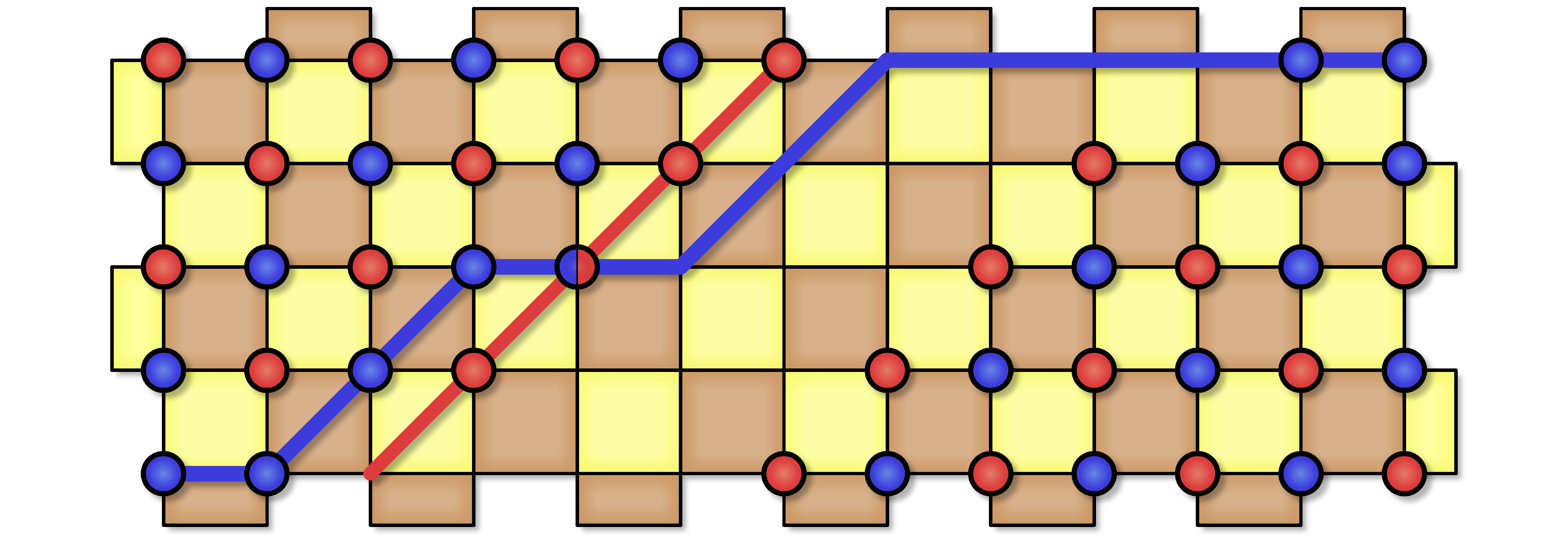
			\caption{\label{fig:rotated-surface-rect-solution-D}}
		\end{subfigure}		
		\begin{subfigure}[c]{0.45\textwidth}
			\def\svgwidth{\textwidth}
			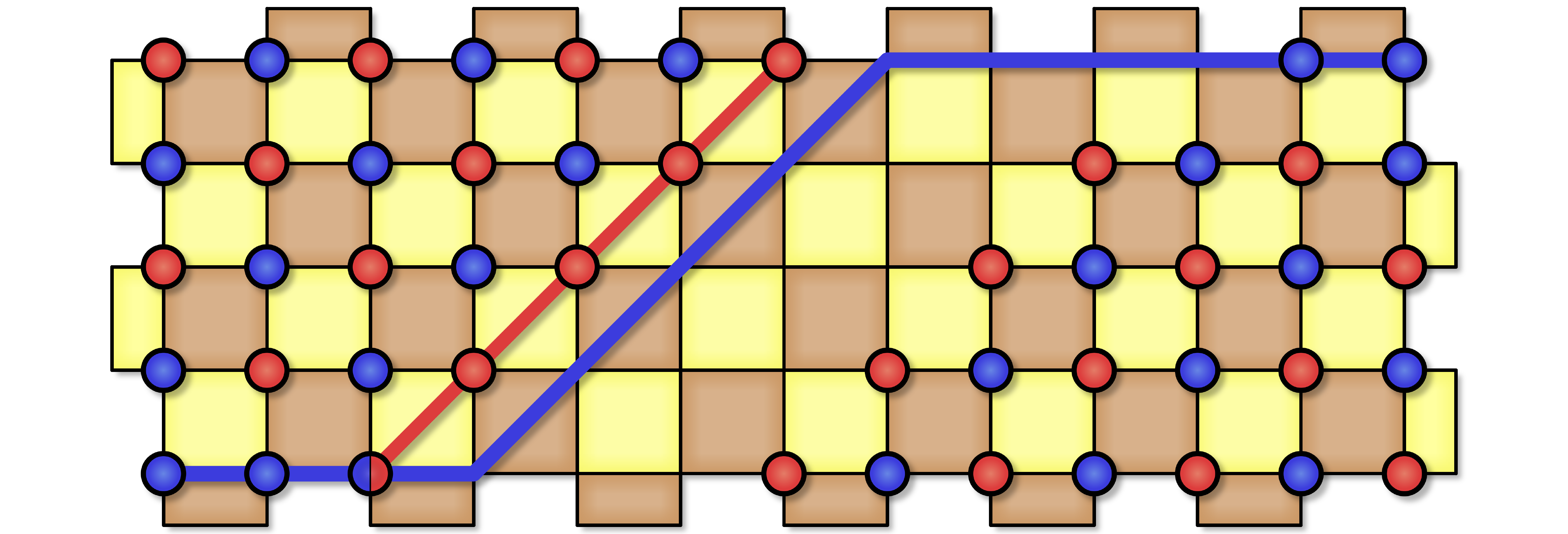
			\caption{\label{fig:rotated-surface-rect-solution-E}}
		\end{subfigure}									
		\begin{subfigure}[c]{0.45\textwidth}
			\def\svgwidth{\textwidth}
			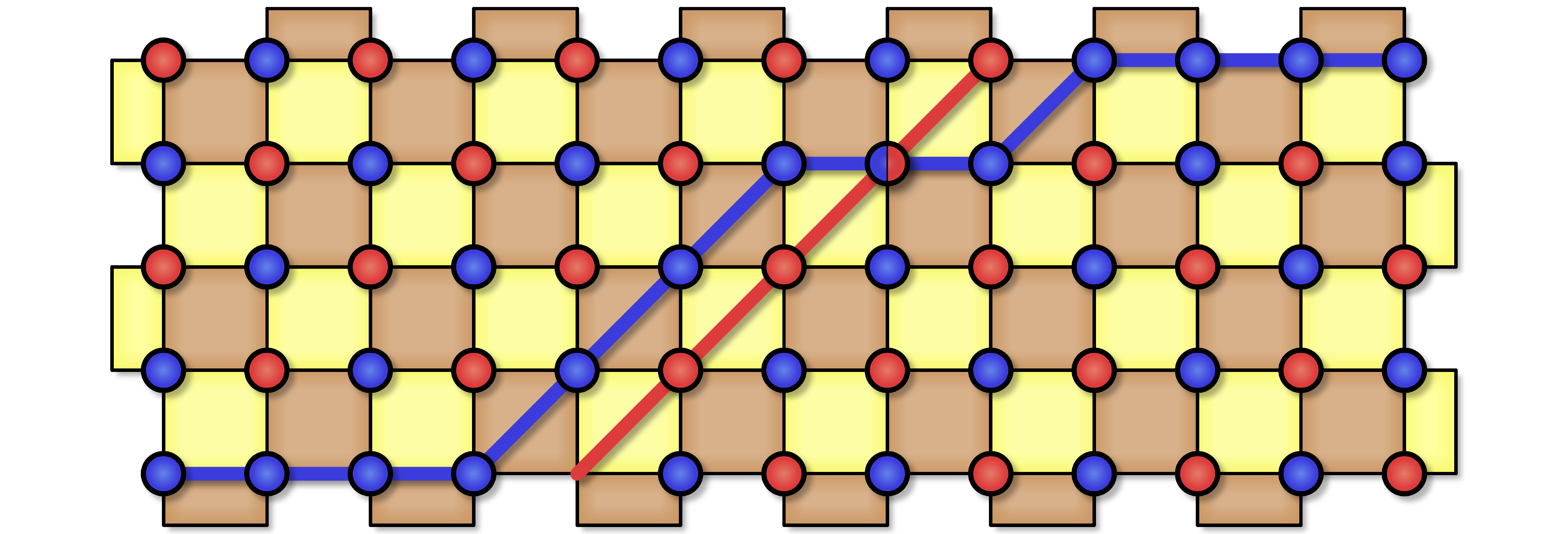
			\caption{\label{fig:rotated-surface-rect-solution-F}}
		\end{subfigure}
		\begin{subfigure}[c]{0.45\textwidth}
			\def\svgwidth{\textwidth}
			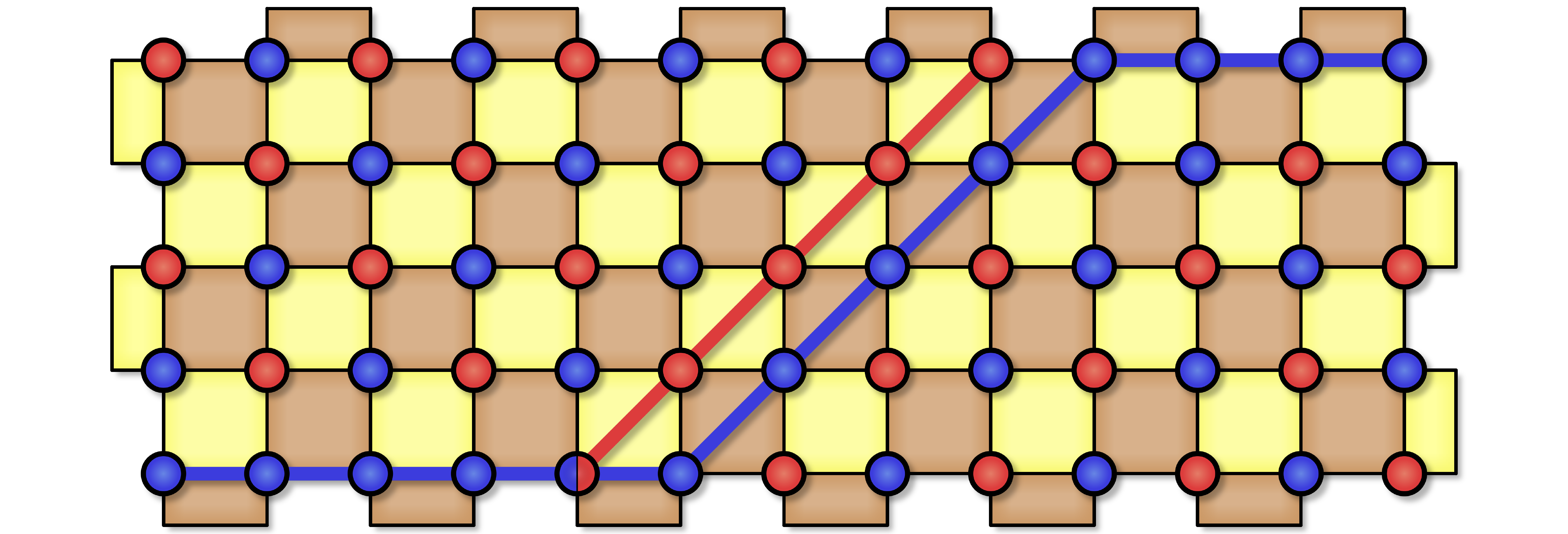
			\caption{\label{fig:rotated-surface-rect-solution-G}}
		\end{subfigure}									
	\caption{\label{fig:rotated-surface-rect-solution}Examples illustrating how the scheme measures the logical operators in the event of a successful BM for the rectangular rotated planar surface code. Red and blue qubits indicate $X$-BMs and $Z$-BMs, respectively. Qubits filled with both red and blue indicate a successful physical BM. Red and blue strings indicate the measured $\overline{X}$ and $\overline{Z}$ string, respectively. The following cases for the success vertex in the first part of the scheme are shown: (a) $X$-diagonal; (b) $Z$-diagonal. The following cases for the success vertex in the second part of the scheme are shown: (c) $Z$-diagonal; (d) $X$-diagonal, any qubit but the last; (e) $X$-diagonal, last qubit of the diagonal; (f) middle $X$-diagonal, any qubit but the last; (g) middle diagonal, last qubit of the diagonal. The solutions for the right side mirror those of the left side.}
\end{figure*}

The second part of the scheme addresses the remaining diagonals. As before, the scheme iterates over the diagonals in increasing order of index sum, i.e., from left to right, while a mirrored process runs from right to left. The order in which the vertices are addressed within each diagonal is identical to the quadratic case, as illustrated in Fig.~\ref{fig:rotated-surface-rect-scheme}.

\begin{figure}
	\def\svgwidth{0.45\textwidth}
	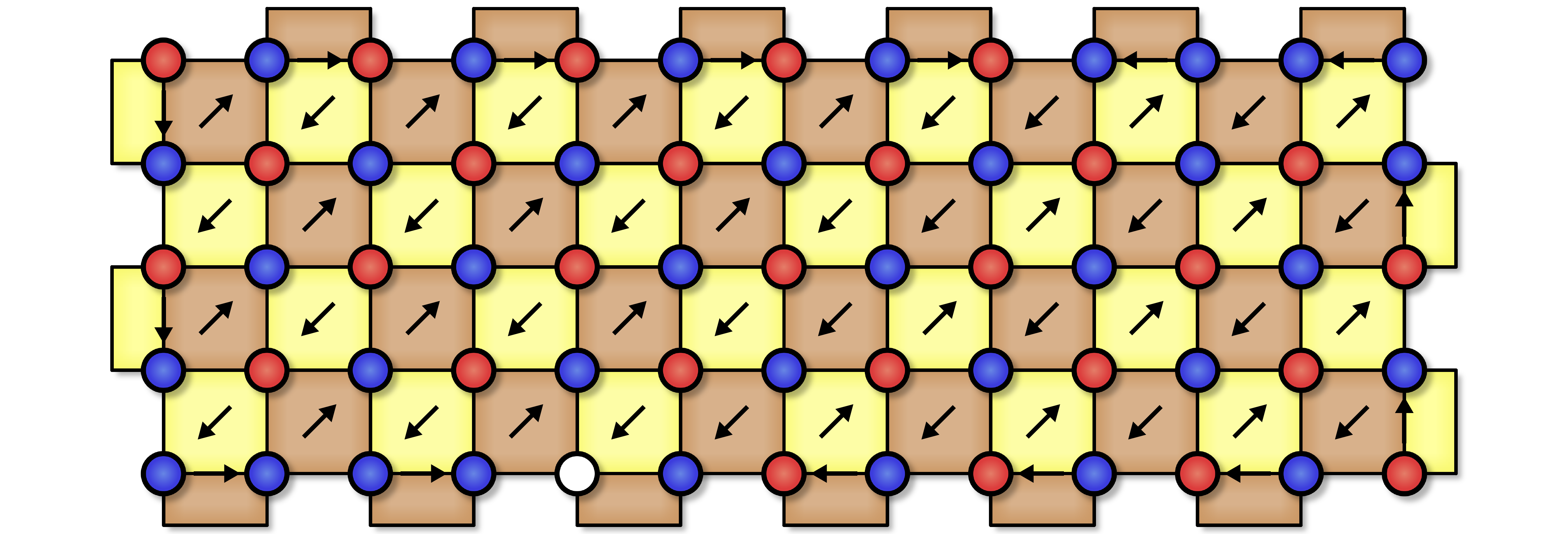
	\caption{\label{fig:rotated-surface-rect-scheme}Measurement scheme for the rectangular rotated planar surface code. The scheme starts simultaneously at the top-left and the bottom-right vertex. The qubits at the vertices are measured following the black arrows. Red and blue vertices are measured with $X$- and $Z$-BMs, respectively. The measurement type on the last qubit is inconsequential and thus remains uncolored.}
\end{figure}

The adaptation of the scheme is a small modification in the way $X$-diagonals are measured, which is illustrated in Figs.~\ref{fig:rotated-surface-rect-solution} and~\ref{fig:rotated-surface-rect-scheme}. In the following, we discuss the part of the scheme starting from the left; the mirrored part starting from the right is analogous. Instead of measuring the entire $X$-diagonal with $X$-BMs, we leave the vertex at the bottom boundary unmeasured. Only if no successful BM occurs in the rest of the diagonal will this last vertex be measured with a $Z$-BM. Thus, as long as no success occurs, the vertices from the bottom-left corner along the bottom boundary to the current diagonal are never measured with an $X$-BM. This ensures that a $Z$-diagonal can always be connected to the left boundary via this path.

Again, the strings used to complete the logical BM may decompose into several parts, some of which can vanish for qubits near the boundaries. However, the strings remain valid in these cases.

If one of the $X$-BMs on the $X$-diagonal is successful, the logical operators can be measured in a manner similar to the quadratic case, as illustrated in Figs.~\ref{fig:rotated-surface-rect-solution-D} and~\ref{fig:rotated-surface-rect-solution-F}. In this case, $\overline{X}$ is completed by performing an $X$-BM on the last vertex of the diagonal. This provides $X$ information for all qubits along the diagonal, which connects the top and bottom boundaries, thus forming an $\overline{X}$ operator. Let the index sum of the success vertex be $k$. The $\overline{Z}$ string starts at the bottom-left corner and follows the bottom boundary until it reaches the diagonal with index sum $k-1$. From there, it moves diagonally upwards to the row of the success vertex, then two steps horizontally to the right, crossing the $\overline{X}$ string at the success vertex. The $\overline{Z}$ string continues along the diagonal with index sum $k+1$ to the top boundary, then extends rightward along the top boundary to the top-right corner.

If the $Z$-BM on the last vertex touching the bottom boundary succeeds, the strings are essentially the same, as illustrated in Figs.~\ref{fig:rotated-surface-rect-solution-E} and~\ref{fig:rotated-surface-rect-solution-G}. The $\overline{X}$ operator is already measured, as $X$ information has been obtained on every vertex in the current diagonal. The $\overline{Z}$ string starts at the bottom-left corner, traverses the bottom boundary passing through the success vertex, and then follows the next diagonal with index sum $k+1$ upwards to the top boundary. From there, it extends rightward along the top boundary to the top-right corner.

Measuring the $Z$-diagonals is straightforward. The entire diagonal is measured with $Z$-BMs. The logical operators are measured similarly to the previous cases, as illustrated in Fig.~\ref{fig:rotated-surface-rect-solution-C}. $\overline{X}$ is completed by connecting the two adjacent diagonals, $k-1$, which touches the top boundary, and $k+1$, which touches the bottom boundary, vertically through the success vertex. The $\overline{Z}$ string starts at the bottom-left corner, traverses the bottom boundary to the current diagonal $k$, follows this upwards to the top boundary, and then extends rightward along the top boundary to the top-right corner.

The same argument for the $X$-diagonals and $Z$-diagonals also applies to a potential middle diagonal, and it does not matter whether we assign the middle diagonal to the part of the scheme coming from the right or from the left.

We now turn to the transformation of stabilizer generators through the measurement scheme, as illustrated in Fig.~\ref{fig:rotated-planar-surface-59-generators}. The argument is very similar to the quadratic case, but there is one difference. For such $X$-diagonals, where the last measurement is now a $Z$-BM, this final vertex now touches one $X$-plaquette from the prior diagonal and an $X$-boundary plaquette to the right. The stabilizer generator associated with the $X$-plaquette from the prior diagonal has already been replaced by a $Z$-BM of the prior diagonal. The only exception is the diagonal with index sum $r+1$, which touches only the $X$-boundary plaquette to the right. Therefore, the transformation of the stabilizer generators is straightforward to track, as each measurement successively replaces the next plaquette along the measured path. Recall from Sec.~\ref{sec:sufficient-conditions-for-an-optimal-logical-Bell-measurement-with-feedforward-based-linear-optics} that this argument does not need to apply to the very last qubit of the code.

For the quadratic case, we argued that an entire diagonal, along with its mirrored counterpart, can be measured simultaneously. This argument still holds, with the one exception that the $Z$-BM on the last vertex of an $X$-diagonal can only be performed if no successful BM has occurred along the rest of the diagonal, as an $X$-BM would otherwise be required on this qubit. However, since this qubit will always be measured with a $Z$-BM if no success has occurred, we can simply perform its measurement in the step for the next $Z$-diagonal. That means, an entire $Z$-diagonal together with the final BM in the previous $X$-diagonal can be measured simultaneously.

\begin{figure}
	\def\svgwidth{0.45\textwidth}
	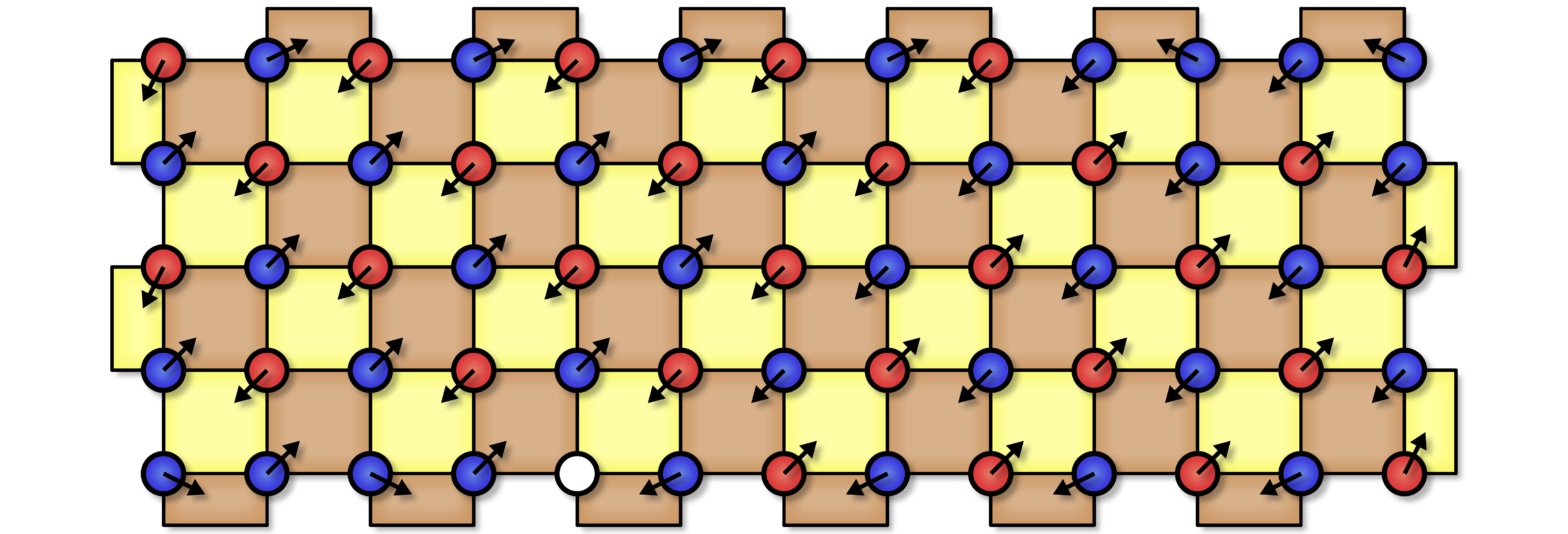
	\caption{\label{fig:rotated-planar-surface-59-generators}Schematic representation of the transformation of stabilizer generators for the rectangular rotated planar surface code. Red and blue qubits indicate $X$-BMs and $Z$-BMs, respectively. For each vertex, the black arrow points to the plaquette whose stabilizer generator is replaced by the BM on that vertex. Note that each vertex touches only one plaquette of the opposite type that has not been replaced by a previous measurement.}
\end{figure}

The argument to demonstrate the success probability for every physical BM is essentially identical to the rectangular case. In the first part of the scheme, which addresses the top-left triangle and its mirrored counterpart, the argument is identical to the quadratic case. For the second part of the scheme, let us, for now, exclude the middle diagonal and analyze an arbitrary diagonal. For each measurement, the $Y$ operator at any vertex on the diagonal anticommutes with the associated plaquette that also anticommutes with the type of the BM. For $Z$-diagonals, any single-qubit $X$ operator anticommutes with the adjacent $Z$-plaquette of the neighboring diagonal towards the middle of the lattice, which is unmeasured at this point. For $X$-diagonals, any single-qubit $Z$ operator anticommutes with the adjacent $X$-plaquette of the neighboring diagonal towards the middle of the lattice, which is unmeasured at this point. Finally, the single-qubit $X$ operator on the last vertex of an $X$-diagonal, which is measured with a $Z$-BM, completes an $\overline{X}$ string by obtaining $X$ information from this final vertex in the diagonal. If the sum $r+m$ is odd, so that no single middle diagonal exists, we select one of the two diagonals closest to the center to serve as the middle diagonal, which is measured last. The arguments for non-middle diagonals then apply to the diagonal not chosen, while the upcoming arguments for the middle diagonal apply to the selected one.

Again, for the middle diagonal, every argument from the quadratic case applies, with the caveat that if the middle diagonal is an $X$-diagonal, the two adjacent $Z$-diagonals do not touch the left and right boundaries. However, this issue can be easily resolved, as these diagonals can be connected to the left and right boundaries by traversing along the top and bottom boundaries, respectively, similar to the previous discussion on the logical operators.

With these observations and by applying Lem.~\ref{lem:successful-bell-measurment}, we conclude that the success probability for each BM is given by $\mathbb{P}_B$.

This concludes the generalization of our scheme to the rectangular code.

\section{Proofs: Optimal logical Bell measurements}
\label{app:proofs-optimal-schemes}
In this appendix, we provide the algebraic proofs of the optimality of our schemes based on Thm.~\ref{thm:sufficient}.

\subsection{Proof: Optimal logical Bell measurement for the quantum parity code}
\label{app:proof-qpc}
Recall from Sec.~\ref{sec:quantum-parity-code} that each qubit is indexed by a pair $(i,j)$, with $i \in \{1, \dots, r\}$ denoting the row and $j \in \{1, \dots, m\}$ enumerating the qubits within each row. To formally apply the conditions of Thm.~\ref{thm:sufficient}, we impose a lexicographic order on these index pairs. Furthermore, for convenience, we define $I$ as the set of all index pairs including the last qubit,
\begin{equation}
	\begin{aligned}
	I 	& = \{ (1,1), \dots, (r,m) \} \\
		& = \{ (i,j) \mid 1 \leq i \leq r, 1 \leq j \leq m \},
	\end{aligned}
\end{equation}
and $I'$ as the set of all index pairs excluding the last qubit,
\begin{equation}
	\begin{aligned}
	I' 	& = \{ (1,1), \dots, (r,m) \} \setminus \{ \left( r,m \right) \} \\
		& = \{ (i,j) \mid 1 \leq i \leq r, 1 \leq j \leq m, (i,j) \neq (r,m) \}.
	\end{aligned}
\end{equation}

First, we define the sequences, as discussed in Sec.~\ref{sec:quantum-parity-code}, of the code stabilizer generators
\begin{equation}
	\mathbb{C} = \left( c_{i,j} \right)_{(i,j) \in I' },
\end{equation}
where
\begin{equation}
	c_{i,j} =
    \begin{dcases}
      Z_{i,j} Z_{i,j+1} & \quad \text{if } j<m\\
      \prod_{t=1}^{m} X_{i,t} X_{i+1,t} & \quad \text{if } j=m \wedge i<r, \\
    \end{dcases}
    \label{eq:qpc-sequence-c}
\end{equation}
the measurement sequence
\begin{equation}
	\mathbb{B} = \left( b_{i,j} \right)_{(i,j) \in I' },
\end{equation}
where
\begin{equation}
	b_{i,j} =
	\begin{dcases}
		X_{i,j} & \quad \text{if } j<m \\
		Z_{i,j} & \quad \text{if } j=m, \\
	\end{dcases}
	\label{eq:qpc-sequence-b}
\end{equation}
and the sequence of pairs of logical operators
\begin{equation}
	\mathbb{L} = \left( \left( \overline{X}_{i,j} , \overline{Z}_{i,j} \right) \right)_{(i,j) \in I },
\end{equation}
where
\begin{equation}
	\overline{X}_{i,j} = \prod_{t=1}^{m} X_{i,t},
	\label{eq:qpc-sequence-x-logical}
\end{equation}
\begin{equation}
	\overline{Z}_{i,j} = Z_{i,j} Z_{i,m} \prod_{t=1}^{r} Z_{t,m}.
	\label{eq:qpc-sequence-z-logical}
\end{equation}
The $\overline{Z}$ operators in Eq.~\eqref{eq:qpc-sequence-z-logical} are chosen as in Sec.~\ref{sec:quantum-parity-code}, where the last qubit of each remaining row is used to complete the logical operator. Note that the factor $Z_{i,m}$ appears twice, once before the product and once within it, and thus cancels out. We now verify the conditions of Thm.~\ref{thm:sufficient} individually. Note that the conditions are stated as in Thm.~\ref{thm:sufficient}, using a single index. In the proof of each condition, the single index is replaced with the index pair.

\textit{Condition~1}: Each operator $b_j$ anticommutes with $c_j$:
\begin{equation}
	\forall j \in \{ 1 , \dots, n-1 \} : \quad \acomm{b_j}{c_j} = 0.
	\tag{\ref{eq:con-1}}
\end{equation}

\begin{proof}
We verify the condition	$\forall (i,j) \in I'$:
\begin{equation}
\begin{aligned}
	& \acomm{b_{i,j}}{c_{i,j}} = \\
	&	
	\begin{dcases}
		\acomm{X_{i,j}}{Z_{i,j}Z_{i,j+1}} & \quad \text{if } j<m\\
		\acomm{Z_{i,j}}{ \prod_{t=1}^{m} X_{i,t} X_{i+1,t}} & \quad \text{if } j=m \wedge i<r\\
	\end{dcases}
	\\
	& = 0. \\
\end{aligned}
\end{equation}
\end{proof}

\vspace{7\baselineskip}

\textit{Condition~2}: Each operator $b_j$ commutes with every later stabilizer generator:
\begin{equation}
	\forall k>j: \quad \comm{b_j}{c_k} = 0.
	\tag{\ref{eq:con-2}}
\end{equation}

\begin{proof}
In the equation below, we include relations between indices in brackets when they follow directly from the preceding index relation, as indicated by the symbol ``$\rightarrow$''. We verify the condition $\forall \left( k , l \right) > \left( i,j \right)$:
\begin{widetext}
\begin{equation}
\begin{aligned}
	\comm{b_{i,j}}{c_{k,l}} & =
	\begin{dcases}
		\comm{X_{i,j}}{Z_{k,l}Z_{k,l+1}} & \quad \text{if } j<m \wedge l<m \wedge k = i (\rightarrow l>j)\\
		\comm{X_{i,j}}{Z_{k,l}Z_{k,l+1}} & \quad \text{if } j<m \wedge l<m \wedge k > i \\
		\\
		\comm{X_{i,j}}{\prod_{t=1}^{m} X_{k,t} X_{k+1,t}} & \quad \text{if } j<m \wedge l=m (\rightarrow k<r)\\
		\\
		\comm{Z_{i,j}}{Z_{k,l}Z_{k,l+1}} & \quad \text{if } j=m (\rightarrow i<r) \wedge l<m\\
		\\
		\comm{Z_{i,j}}{\prod_{t=1}^{m} X_{k,t} X_{k+1,t}} & \quad \text{if } j=m \wedge l=m (\rightarrow i,k<r) (\rightarrow k>i)\\
	\end{dcases}
	\\
	& = 0.
\end{aligned}
\end{equation}
\end{widetext}
Both sequences $\mathbb{C}$ and $\mathbb{B}$, define two cases in their respective Eq.~\eqref{eq:qpc-sequence-c} and~\eqref{eq:qpc-sequence-b}, depending on whether the second index is strictly less than $m$ or equal to it, i.e, whether the index pair corresponds to the last qubit of a row. This leads to an initial total of four cases for combinations of one operator from each sequence. For simplicity, we further split the case where both indices $j$ and $l$ are less than $m$ into two subcases: one where the row indices are equal, and one where they differ. We now examine these cases one by one.

In the first subcase of the first case, note that $k = i$ implies $l > j$, since the equation holds for $(k, l) > (i, j)$ in lexicographic order. Therefore, the operators commute, since their support lies on different rows, $j$ and $l$. In the second subcase, where $k > i$, the supports are again on different rows, so the operators commute as well. In the second and third cases the operators commute because they consist of the same single-qubit Pauli operators, $X$ and $Z$, respectively. In the fourth case, the sequences are not defined on the last qubit $(r, m)$, which implies $i, k < r$. Moreover, since $(k, l) > (i, j)$, it follows that $k > i$, so the operators have support on different rows and thus commute.

\end{proof}

\textit{Condition~3}: For all $j \in \{ 1, \dots, n-1 \}$ each operator in the set $\tilde{b}_j \in \{ X_j, Y_j, Z_j \} \setminus \{ b_j \}$ either anticommutes with at least one non-prior stabilizer generator,
\begin{equation}
	\exists k \geq j: \quad \acomm{\tilde{b}_j}{c_k} = 0,
	\tag{\ref{eq:con-3-a}}
\end{equation}
or completes a logical measurement,
\begin{equation}
	\exists \mu \in \langle b_1, \dots, b_{j-1} \rangle: \quad \mu \tilde{b}_j \in \left[ \overline{X} \right] \cup \left[ \overline{Y} \right] \cup \left[ \overline{Z} \right].
	\tag{\ref{eq:con-3-b}}
\end{equation}

\begin{proof}
For every $Y_{i,j}$, the proof that it anticommutes with at least one non-prior stabilizer generator proceeds analogously to the calculation of condition~2, except for the final qubit $Y_{r,m}$, which completes a logical operator:
\begin{equation}
	\mu Y_{r,m} = Y_{r,m} \prod_{t=1}^{m} X_{r,t} \prod_{t=1}^{r} Z_{t,m} \in [\overline{Y}].
\end{equation}

For $j<m$, we show that all operators $\tilde{b}_{i,j} \in \{ X_{i,j} , Y_{i,j} , Z_{i,j} \} \setminus \{ b_{i,j} \}$ anticommute with at least one non-prior stabilizer generator $c_{k,l}$. Since $b_{i,j}=X_{i,j}$ for all $j<m$, the remaining operator is $Z_{i,j}$. We then show that:
\begin{equation}
	\exists \left( k , l \right) \geq \left( i,j \right): \quad \acomm{Z_{i,j}}{c_{k,l}} = 0.
\end{equation}
Specifically, the operator $Z_{i,j}$ anticommutes with $c_{i,m}$:
\begin{equation}
	\acomm{Z_{i,j}}{c_{i,m}} = \acomm{Z_{i,j}}{\prod_{t=1}^{m} X_{i,t} X_{i+1,t}} = 0,
\end{equation}
where we note that $j < m$ implies $(i,m) > (i,j)$ in lexicographic order.

Since $b_{i,j} = Z_{i,j}$ for all $j=m$, we have $\{ X_{i,j}, Y_{i,j}, Z_{i,j} \} \setminus \{ b_{i,j} \} = \{ X_{i,j}, Y_{i,j} \},$ and thus $X_{i,j} = \tilde{b}_{i,j} \in \{ X_{i,j}, Y_{i,j} \}$. We now show that $X_{i,j} = \tilde{b}_{i,j}$ completes a logical operator:
\begin{equation}
	\mu X_{i,j} = X_{i,j} \prod_{t=1}^{m-1} b_{i,t} = \prod_{t=1}^m X_{i,t} \in \left[ \overline{X} \right],
\end{equation}
where $\mu = \prod_{t=1}^{m-1} b_{i,t} \in \langle b_1, \dots, b_{i,j-1} \rangle$, and $j-1>0$ holds since $j=m\geq2$.
\end{proof}

\textit{Condition~4}: The logical operators $\overline{X}_j$ and $\overline{Z}_j$ commute with every prior element of $\mathbb{B}$ for all $j \in \{ 1, \dots, n \}$:
\begin{equation}
	\forall k<j : \quad \comm{\overline{X}_j}{b_k} = 0,
	\tag{\ref{eq:con-4-a}}
\end{equation}
\begin{equation}
	\forall k<j : \quad \comm{\overline{Z}_j}{b_k} = 0.
	\tag{\ref{eq:con-4-b}}
\end{equation}

\begin{proof}
We verify the condition	$\forall \left( k , l \right) < \left( i,j \right)$:
\begin{equation}
\begin{aligned}
	\comm{\overline{X}_{i,j}}{b_{k,l}} & =
	\begin{dcases}
		\comm{\prod_{t=1}^{m} X_{i,t}}{X_{k,l}} & \quad \text{if } l<m \\
		\comm{\prod_{t=1}^{m} X_{i,t}}{Z_{k,l}} & \quad \text{if } l=m \\
	\end{dcases}
	\\
	& = 0,  \\
\end{aligned}
\end{equation}
\begin{equation}
\begin{aligned}
	\comm{\overline{Z}_{i,j}}{b_{k,l}} & =
	\begin{dcases}
		\comm{Z_{i,j} Z_{i,m} \prod_{t=1}^{r} Z_{t,m}}{X_{k,l}} & \quad \text{if } l<m \\
		\comm{Z_{i,j} Z_{i,m} \prod_{t=1}^{r} Z_{t,m}}{Z_{k,l}} & \quad \text{if } l=m \\
	\end{dcases}
	\\
	& = 0. \\
\end{aligned}
\end{equation}

\end{proof}

\textit{Condition~5}: We decompose the logical operators into single-qubit Pauli operators to formulate the last condition:
\begin{equation}
	\overline{X}_j = \bigotimes_{t=1}^{n} u_{j,t}, \quad \text{where } u_{j,t} \in \{ I, X, Y, Z \},
	\tag{\ref{eq:logical-decomp-x}}
\end{equation}
\begin{equation}
	\overline{Z}_j = \bigotimes_{t=1}^{n} v_{j,t}, \quad \text{where } v_{j,t} \in \{ I, X, Y, Z \}.
	\tag{\ref{eq:logical-decomp-z}}
\end{equation}
The logical operators $\overline{X}_j$ and $\overline{Z}_j$ anticommute only in $j$:
\begin{equation}
	\forall j \in \{ 1, \dots, n \} : \acomm{u_{j,j}}{v_{j,j}} = 0,
	\tag{\ref{eq:double-info}}
\end{equation}
\begin{equation}
	\forall k,j \in \{1, \dots, n \} \wedge k \neq j : \comm{u_{j,k}}{v_{j,k}} = 0.
	\tag{\ref{eq:double-info-b}}
\end{equation}

\begin{proof}

The decomposition of the logical operators into single-qubit Pauli operators reads
\begin{equation}
	\overline{X}_{i,j} = \bigotimes_{t=1}^r \bigotimes_{s=1}^m u_{(i,j),(t,s)},
\end{equation}
\begin{equation}
	\overline{Z}_{i,j} = \bigotimes_{t=1}^r \bigotimes_{s=1}^m v_{(i,j),(t,s)},
\end{equation}
where $u_{(i,j),(t,s)},v_{(i,j),(t,s)} \in \{ I, X, Y, Z \}$. We take the decompositions from Eqs.~\eqref{eq:qpc-sequence-x-logical} and~\eqref{eq:qpc-sequence-z-logical}:
\begin{equation}
\begin{aligned}
	u_{(i,j),(t,s)} & = 
	\begin{dcases}
		X_{t,s} & \quad \text{if } i=t \\
		I_{t,s} & \quad \text{if } i \neq t, \\
	\end{dcases} 
\end{aligned}
\label{eq:qpc-logical-x-decomposition}
\end{equation}
\begin{equation}
\begin{aligned}
	v_{(i,j),(t,s)} & = 
	\begin{dcases}
		Z_{t,s} & \quad \text{if } (i,j)=(t,s) \vee (i \neq t \wedge s = m )\\
		I_{t,s} & \quad \text{else}.\\
	\end{dcases} 
\end{aligned}
\label{eq:qpc-logical-z-decomposition}
\end{equation}
It is straightforward to see that the logical operators $\overline{X}_{i,j}$ and $\overline{Z}_{i,j}$ anticommute in $(i,j)$, $\forall (i,j) \in \{ (1,1), \dots, (r,m) \}$:
\begin{equation}
\begin{aligned}
	\acomm{u_{(i,j),(i,j)}}{v_{(i,j),(i,j)}} = \acomm{X_{i,j}}{Z_{i,j}} = 0,
\end{aligned}
\end{equation}
and commute in every other qubit $\forall (i,j),(k,l) \in \{ (1,1), \dots, (r,m) \} \wedge (k,l) \neq (i,j)$:
\begin{equation}
\begin{aligned}
	& \quad \comm{u_{(i,j),(k,l)}}{v_{(i,j),(k,l)}} \\
	&	\quad =
		\begin{dcases}
			\comm{X_{k,l}}{Z_{k,l}} & \quad \text{if } i=k \wedge ( 	(i,j)=(k,l) \\
									&							\qquad \qquad \quad \vee \left( i \neq k \wedge l = m \right) ) \\
			0 						& \quad \text{else} 				\\
		\end{dcases} \\
	&	\quad =
		\begin{dcases}
			\comm{X_{k,l}}{Z_{k,l}} & \quad \text{if } i=k \wedge ( \left( i \neq k \wedge l = m \right) ) \\
			0 						& \quad \text{else} \\
		\end{dcases} \\
	& \quad = 0.
\end{aligned}
\label{eq:qpc-con-5-b-calc}
\end{equation}
In Eq. \eqref{eq:qpc-con-5-b-calc} we begin by separating the commutator into two cases.  The first case includes all terms from Eqs.~\eqref{eq:qpc-logical-x-decomposition} and~\eqref{eq:qpc-logical-z-decomposition} where the two operators could potentially anticommute. In the next step, we eliminated the term $(i,j) = (k,l)$, as this contradicts the assumption $(k,l) \neq (i,j)$. Finally, we observed that the condition $i = k \wedge i \neq k$ is inherently contradictory and therefore always evaluates to false.

\end{proof}

Having verified all conditions of Thm.~\ref{thm:sufficient} we conclude that our scheme is optimal.

\subsection{Proof: Optimal logical Bell measurement for the five-qubit code}
\label{app:proof-five-qubit}
First, we define the sequences of the code stabilizer generators
\begin{equation}
\begin{aligned}
	\mathbb{C} 	& = ( c_j )_{ j \in \{ 1, \dots, n-1 \} } \\
		& = ( XXYIY, YXXYI, IYXXY, YIYXX ),
\end{aligned}
\end{equation}
the measurement sequence
\begin{equation}
	\mathbb{B} = ( b_{j} = Y_j )_{j \in \{ 1, \dots, n-1 \} },
\end{equation}
and the sequence of pairs of logical operators
\begin{equation}
	\mathbb{L} = \left( \left( \overline{X}_j , \overline{Z}_j \right) \right)_{j \in \{ 1, \dots, n \} },
\end{equation}
where
\begin{equation}
\begin{aligned}
	( \overline{X}_{j} )_{j \in \{ 1, \dots, n \} } = ( 	& XIYYI, IXIYY, \\
															& YIXIY, YYIXI, \\
															& IYYIX ),
\end{aligned}
	\label{eq:five-qubit-sequence-x-logical}
\end{equation}
\begin{equation}
\begin{aligned}
	( \overline{Z}_{j} )_{j \in \{ 1, \dots, n \} } = (	& ZYIIY, YZYII, \\
															& IYZYI, IIYZY, \\
															& YIIYZ ).
\end{aligned}
	\label{eq:five-qubit-sequence-z-logical}
\end{equation}
The logical operators in Eqs.~\eqref{eq:five-qubit-sequence-x-logical} and~\eqref{eq:five-qubit-sequence-z-logical} can be obtained by multiplying $XXXXX \in [\overline{X}]$ and $ZZZZZ \in [\overline{Z}]$ with all elements of $\mathbb{C}$ and with $ZZXIX = XZZXI \times IXZZX \times XIXZZ \times ZXIXZ \in S$, respectively. We now verify the conditions of Thm.~\ref{thm:sufficient} individually.

\textit{Condition~1}: Each operator $b_j$ anticommutes with $c_j$:
\begin{equation}
	\forall j \in \{ 1 , \dots, n-1 \} : \quad \acomm{b_j}{c_j} = 0.
	\tag{\ref{eq:con-1}}
\end{equation}

\begin{proof}

We verify the condition for each index individually:
\begin{equation}
\begin{aligned}
	\acomm{b_1}{c_1} & = \acomm{Y_1}{XXYIY} = 0, \\
	\acomm{b_2}{c_2} & = \acomm{Y_2}{YXXYI} = 0, \\
	\acomm{b_3}{c_3} & = \acomm{Y_3}{IYXXY} = 0, \\
	\acomm{b_4}{c_4} & = \acomm{Y_4}{YIYXX} = 0.
\end{aligned}
\end{equation}

\end{proof}

\textit{Condition~2}: Each operator $b_j$ commutes with every later stabilizer generator:
\begin{equation}
	\forall k>j: \quad \comm{b_j}{c_k} = 0.
	\tag{\ref{eq:con-2}}
\end{equation}

\begin{proof}

Again, we verify the condition for each index individually:
\begin{equation}
\begin{aligned}
	\comm{b_1}{c_2} & = \comm{Y_1}{YXXYI} & = 0, \\
	\comm{b_1}{c_3} & = \comm{Y_1}{IYXXY} & = 0, \\
	\comm{b_1}{c_4} & = \comm{Y_1}{YIYXX} & = 0, \\
	\comm{b_2}{c_3} & = \comm{Y_2}{IYXXY} & = 0, \\
	\comm{b_2}{c_4} & = \comm{Y_2}{YIYXX} & = 0, \\
	\comm{b_3}{c_4} & = \comm{Y_3}{YIYXX} & = 0.
\end{aligned}
\end{equation}

\end{proof}

\textit{Condition~3}: For all $j \in \{ 1, \dots, n-1 \}$ each operator in the set $\tilde{b}_j \in \{ X_j, Y_j, Z_j \} \setminus \{ b_j \}$ either anticommutes with at least one non-prior stabilizer generator,
\begin{equation}
	\exists k \geq j: \quad \acomm{\tilde{b}_j}{c_k} = 0,
	\tag{\ref{eq:con-3-a}}
\end{equation}
or completes a logical measurement,
\begin{equation}
	\exists \mu \in \langle b_1, \dots, b_{j-1} \rangle: \quad \mu \tilde{b}_j \in \left[ \overline{X} \right] \cup \left[ \overline{Y} \right] \cup \left[ \overline{Z} \right].
	\tag{\ref{eq:con-3-b}}
\end{equation}

\begin{proof}

Since $b_j=Y_j$ for all $j$ the set $\tilde{b}_j \in \{ X_j , Y_j , Z_j \} \setminus \{ b_j \}$ becomes $\{ X_j , Z_j \}$. For $j<4$, we show that all operators $\tilde{b}_j \in \{ X_j , Z_j \}$ anticommute with at least one non-prior stabilizer generator $c_k$:
\begin{equation}
	\exists k \geq j : \acomm{\tilde{b}_j}{c_k} = 0,
\end{equation}
which we verify for each index individually:
\begin{equation}
\begin{aligned}
	\acomm{X_1}{c_2} & = \acomm{X_1}{YXXYI} & = 0, \\
	\acomm{Z_1}{c_2} & = \acomm{Z_1}{YXXYI} & = 0, \\
	\acomm{X_2}{c_3} & = \acomm{X_2}{IYXXY} & = 0, \\
	\acomm{Z_2}{c_3} & = \acomm{Z_2}{IYXXY} & = 0, \\
	\acomm{X_3}{c_4} & = \acomm{X_3}{YIYXX} & = 0, \\
	\acomm{Z_3}{c_4} & = \acomm{Z_3}{YIYXX} & = 0.
\end{aligned}
\end{equation}
For the final index $j=4$, we show that the operator $Z_4 = \tilde{b}_4 \in \{ X_4, Z_4 \}$ anticommutes with the non-prior stabilizer generator $c_4$:
\begin{equation}
	\acomm{Z_4}{c_4} = \acomm{Z_4}{YIYXX} = 0,
\end{equation}
and the operator $X_4 = \tilde{b}_4 \in \{ X_4, Z_4 \}$ completes a logical measurement:
\begin{equation}
	\mu X_4 = X_4 b_1 b_2 = YYIKI \in [\overline{X}],
\end{equation}
where $\mu = b_1 b_2 \in \langle b_1, \dots, b_4 \rangle$.

\end{proof}

\textit{Condition~4}: The logical operators $\overline{X}_j$ and $\overline{Z}_j$ commute with every prior element of $\mathbb{B}$ for all $j \in \{ 1, \dots, n \}$:
\begin{equation}
	\forall k<j : \quad \comm{\overline{X}_j}{b_k} = 0,
	\tag{\ref{eq:con-4-a}}
\end{equation}
\begin{equation}
	\forall k<j : \quad \comm{\overline{Z}_j}{b_k} = 0.
	\tag{\ref{eq:con-4-b}}
\end{equation}

\begin{proof}

We verify the condition for each index pair individually:
\begin{equation}
\begin{aligned}
	\comm{\overline{X}_2}{b_1} & = \comm{IXIYY}{Y_1} & = 0, \\
\end{aligned}
\end{equation}

\begin{equation}
\begin{aligned}
	\comm{\overline{X}_3}{b_1} & = \comm{YIXIY}{Y_1} & = 0, \\
	\comm{\overline{X}_3}{b_2} & = \comm{YIXIY}{Y_2} & = 0, \\
\end{aligned}
\end{equation}

\begin{equation}
\begin{aligned}
	\comm{\overline{X}_4}{b_1} & = \comm{YYIXI}{Y_1} & = 0, \\
	\comm{\overline{X}_4}{b_2} & = \comm{YYIXI}{Y_2} & = 0, \\
	\comm{\overline{X}_4}{b_3} & = \comm{YYIXI}{Y_3} & = 0, \\
\end{aligned}
\end{equation}

\begin{equation}
\begin{aligned}
	\comm{\overline{X}_5}{b_1} & = \comm{IYYIX}{Y_1} & = 0, \\
	\comm{\overline{X}_5}{b_2} & = \comm{IYYIX}{Y_2} & = 0, \\
	\comm{\overline{X}_5}{b_3} & = \comm{IYYIX}{Y_3} & = 0, \\
	\comm{\overline{X}_5}{b_4} & = \comm{IYYIX}{Y_4} & = 0, \\
\end{aligned}
\end{equation}

\begin{equation}
\begin{aligned}
	\comm{\overline{Z}_2}{b_1} & = \comm{YZYII}{Y_1} & = 0, \\
\end{aligned}
\end{equation}

\begin{equation}
\begin{aligned}
	\comm{\overline{Z}_3}{b_1} & = \comm{IYZYI}{Y_1} & = 0, \\
	\comm{\overline{Z}_3}{b_2} & = \comm{IYZYI}{Y_2} & = 0, \\
\end{aligned}
\end{equation}

\begin{equation}
\begin{aligned}
	\comm{\overline{Z}_4}{b_1} & = \comm{IIYZY}{Y_1} & = 0, \\
	\comm{\overline{Z}_4}{b_2} & = \comm{IIYZY}{Y_2} & = 0, \\
	\comm{\overline{Z}_4}{b_3} & = \comm{IIYZY}{Y_3} & = 0, \\
\end{aligned}
\end{equation}

\begin{equation}
\begin{aligned}
	\comm{\overline{Z}_5}{b_1} & = \comm{YIIYZ}{Y_1} & = 0, \\
	\comm{\overline{Z}_5}{b_2} & = \comm{YIIYZ}{Y_2} & = 0, \\
	\comm{\overline{Z}_5}{b_3} & = \comm{YIIYZ}{Y_3} & = 0, \\
	\comm{\overline{Z}_5}{b_4} & = \comm{YIIYZ}{Y_4} & = 0. \\
\end{aligned}
\end{equation}

\end{proof}

\textit{Condition~5}: We decompose the logical operators into single-qubit Pauli operators to formulate the last condition:
\begin{equation}
	\overline{X}_j = \bigotimes_{t=1}^{n} u_{j,t}, \quad \text{where } u_{j,t} \in \{ I, X, Y, Z \},
	\tag{\ref{eq:logical-decomp-x}}
\end{equation}
\begin{equation}
	\overline{Z}_j = \bigotimes_{t=1}^{n} v_{j,t}, \quad \text{where } v_{j,t} \in \{ I, X, Y, Z \}.
	\tag{\ref{eq:logical-decomp-z}}
\end{equation}
The logical operators $\overline{X}_j$ and $\overline{Z}_j$ anticommute only in $j$:
\begin{equation}
	\forall j \in \{ 1, \dots, n \} : \acomm{u_{j,j}}{v_{j,j}} = 0,
	\tag{\ref{eq:double-info}}
\end{equation}
\begin{equation}
	\forall k,j \in \{1, \dots, n \} \wedge k \neq j : \comm{u_{j,k}}{v_{j,k}} = 0.
	\tag{\ref{eq:double-info-b}}
\end{equation}

\begin{proof}

In principle the decompositions can be taken from Eqs.~\eqref{eq:five-qubit-sequence-x-logical} and~\eqref{eq:five-qubit-sequence-z-logical}, which can be used to calculate that the logical operators $\overline{X}_j$ and $\overline{Z}_j$ anticommute only in $j$. Since the result can be readily seen from Eqs.~\eqref{eq:five-qubit-sequence-x-logical} and~\eqref{eq:five-qubit-sequence-z-logical}, we omit the straightforward computation of these 25 equations.

\end{proof}

Having verified all conditions of Thm.~\ref{thm:sufficient} we conclude that our scheme is optimal.

\vspace{1\baselineskip}
\subsection{Proof: Optimal logical Bell measurement for the standard planar surface code}
\label{app:proof-standard-planar-surface-code}

Recall from Sec.~\ref{sec:standard-planar-surface-code} that we defined a coordinate system $(l,c)$, with $l \in \{1, \dots, 2r-1\}$ denoting the layer and $c \in \{1, \dots, 2m-1\}$ the column. Notably, positions where the sum $l + c$ is even correspond to qubits (edges), whereas positions where $l + c$ is odd correspond to vertices and faces associated with stabilizer generators. For clarity, we denote the indices of qubits by $(i,j)$, with $(i,j) \in I = \{ (1,1), \dots, (r,m) \} \cap \{ (i,j) \mid i+j \text{ even} \} $. Furthermore, for convenience, we define $I'$ as the set of all index pairs excluding the last qubit:
\begin{equation}
	\begin{aligned}
	I' = 	& \left( \{ 1, \dots, r \} \times \{ 1, \dots, m \} \setminus \{ (r,m) \} \right) \\
			& \cap \{ (i,j) \mid i+j \text{ even} \}.
	\end{aligned}
\end{equation}

In Sec.~\ref{sec:standard-planar-surface-code}, the ordering of the qubits was chosen to simplify the discussion of the stabilizer transformations. Here, we instead adopt a lexicographic ordering of the qubit indices $(i,j)$, which is more convenient for the algebraic derivation. This means that in this proof the layers are always measured left to right. Note that the order in which the qubits are addressed within a layer, including the $Z$-BM of the previous layer for $Z$-layers, is inconsequential in our scheme. We adapt the stabilizer generators to this ordering by replacing, in each even layer $i$, the face operator $Z_{i,3}$ with the product $Z_{i,1}Z_{i,3}$. As a result, each $X$-BM $X_{i,j}$ anticommutes with exactly one of the current stabilizer generators. This is especially the case for $X_{i,2}$, which anticommutes solely with $Z_{i,1}$ in every even layer $i$.

First, we define the sequences, of the code stabilizer generators
\begin{equation}
	\mathbb{C} = \left( c_{i,j} =
    \begin{dcases}
		X_{i,j+1} 		& \quad \text{if } i \text{ odd } \wedge j<m\\
		Z_{i+1,j} 		& \quad \text{if } i \text{ odd } \wedge j=m\\
		Z_{i,j-1} 		& \quad \text{if } i \text{ even}\wedge j\neq4\\
		Z_{i,1}Z_{i,3}	& \quad \text{if } i \text{ even}\wedge j=4\\
    \end{dcases}
    \right)_{(i,j) \in I' },
\end{equation}

the measurement sequence
\begin{equation}
	\mathbb{B} = \left( b_{i,j} =
	\begin{dcases}
		Z_{i,j} & \quad \text{if } i \text{ odd} \wedge j<m\\
		X_{i,j} & \quad \text{if } i \text{ odd} \wedge j=m\\
		X_{i,j} & \quad \text{if } i \text{ even}\\
	\end{dcases}
	\right)_{ (i,j) \in I' },
\end{equation}

and the sequence of pairs of logical operators

\begin{equation}
	\mathbb{L} = \left( \left( \overline{X}_{i,j} , \overline{Z}_{i,j} \right) \right)_{(i,j) \in I },
\end{equation}
where
\begin{widetext}
\begin{equation}
	\overline{X}_{i,j} =
	\begin{dcases}
		 \left( \prod_{t=1}^{\frac{i-1}{2}} X_{2t-1,2m-1} \right) \left( \prod_{t=1}^{\frac{2m-1-j}{2}} X_{i-1,j+2t-1} \right) \left( \prod_{t=0}^{\frac{2r-1-i}{2}} X_{i+2t,j} \right) & \quad \text{if } i \text{ odd}\\
		 \left( \prod_{t=1}^{\frac{i}{2}-1} X_{2t-1,2m-1} \right) \left( \prod_{t=0}^{\frac{2m-j}{2}-1} X_{i,j+2t} \right) \left( \prod_{t=0}^{\frac{2r-i}{2}-1} X_{i+2t+1,j-1} \right) & \quad \text{if } i \text{ even},\\
	\end{dcases}
\end{equation}
\begin{equation}
	\overline{Z}_{i,j} =
	\begin{dcases}
		 \prod_{t=0}^{\frac{2m}{2}} Z_{i,2t+1} & \quad \text{if } i \text{ odd}\\
		 \left( \prod_{t=0}^{\frac{j}{2}-1} Z_{i-1,2t+1} \right) Z_{i,j} \left( \prod_{t=0}^{\frac{2m-j}{2}-1} Z_{i+1,j+2t+1} \right) & \quad \text{if } i \text{ even}.\\
	\end{dcases}
\end{equation}
\end{widetext}
For edge cases, some of these parts may vanish. To account for this, we define an empty product $\prod_{t=a}^b$ with $b<a$ to be the identity in the equations above.

We now verify the conditions of Thm.~\ref{thm:sufficient} individually. Note that the conditions are stated as in Thm.~\ref{thm:sufficient}, i.e., using a single index. In the proof of each condition, the single index is replaced with the index pair.

\textit{Condition~1}: Each operator $b_j$ anticommutes with $c_j$:
\begin{equation}
	\forall j \in \{ 1 , \dots, n-1 \} : \quad \acomm{b_j}{c_j} = 0.
	\tag{\ref{eq:con-1}}
\end{equation}

\begin{proof}
We verify the condition $\forall (i,j) \in I' \}$:
\begin{equation}	
	\begin{aligned}
	\acomm{b_{i,j}}{c_{i,j}} & =
	\begin{dcases}
		\acomm{Z_{i,j}}{X_{i,j+1}} 		& \quad \text{if } i \text{ odd } \wedge j<m\\
		\acomm{X_{i,j}}{Z_{i+1,j}} 		& \quad \text{if } i \text{ odd } \wedge j=m\\
		\acomm{X_{i,j}}{Z_{i,j-1}}		& \quad \text{if } i \text{ even} \wedge j\neq4\\
		\acomm{X_{i,4}}{Z_{i,1}Z_{i,3}}	& \quad \text{if } i \text{ even} \wedge j=4\\
	\end{dcases}
	\\
	& = 0. \\
	\end{aligned}
\end{equation}
\end{proof}

\textit{Condition~2}: Each operator $b_j$ commutes with every later stabilizer generator:
\begin{equation}
	\forall k>j: \quad \comm{b_j}{c_k} = 0.
	\tag{\ref{eq:con-2}}
\end{equation}

\begin{proof}
We write the condition explicitly $\forall \left( k , l \right) > \left( i,j \right) \wedge k+l \text{ even} \wedge i+j \text{ even}$:
\begin{widetext}
\begin{equation}	
	\begin{aligned}
	\comm{b_{i,j}}{c_{k,l}} & =
	\begin{dcases}
		\comm{Z_{i,j}}{X_{k,l+1}} 		& \quad \text{if } i \text{ odd } \wedge j<m 	\wedge k \text{ odd } \wedge l<m\\
		\comm{Z_{i,j}}{Z_{k+1,l}}		& \quad \text{if } i \text{ odd } \wedge j<m	\wedge k \text{ odd } \wedge l=m\\
		\comm{Z_{i,j}}{Z_{k,l-1}}		& \quad \text{if } i \text{ odd } \wedge j<m	\wedge k \text{ even} \wedge l\neq4\\
		\comm{Z_{i,j}}{Z_{k,1}Z_{k,3}}	& \quad \text{if } i \text{ odd } \wedge j<m	\wedge k \text{ even} \wedge l=4\\
		\\
		\comm{X_{i,j}}{X_{k,l+1}}		& \quad \text{if } i \text{ odd } \wedge j=m	\wedge k \text{ odd } \wedge l<m\\
		\comm{X_{i,j}}{Z_{k+1,l}} 		& \quad \text{if } i \text{ odd } \wedge j=m 	\wedge k \text{ odd } \wedge l=m\\
		\comm{X_{i,j}}{Z_{k,l-1}}		& \quad \text{if } i \text{ odd } \wedge j=m	\wedge k \text{ even} \wedge l\neq4\\
		\comm{X_{i,j}}{Z_{k,1}Z_{k,3}}	& \quad \text{if } i \text{ odd } \wedge j=m	\wedge k \text{ even} \wedge l=4\\
		\\
		\comm{X_{i,j}}{X_{k,l+1}}		& \quad \text{if } i \text{ even}				\wedge k \text{ odd } \wedge l<m\\
		\comm{X_{i,j}}{Z_{k+1,l}}		& \quad \text{if } i \text{ even}				\wedge k \text{ odd } \wedge l=m\\
		\comm{X_{i,j}}{Z_{k,l-1}}		& \quad \text{if } i \text{ even} 				\wedge k \text{ even} \wedge l\neq4\\
		\comm{X_{i,j}}{Z_{k,1}Z_{k,3}}	& \quad \text{if } i \text{ even}				\wedge k \text{ even} \wedge l=4\\
	\end{dcases}
	\\
	& = 0. \\
	\end{aligned}
\end{equation}
\end{widetext}
While straightforward to verify by hand, we omit the explicit calculations here for brevity.

\end{proof}

\textit{Condition~3}: For all $j \in \{ 1, \dots, n-1 \}$ each operator in the set $\tilde{b}_j \in \{ X_j, Y_j, Z_j \} \setminus \{ b_j \}$ either anticommutes with at least one non-prior stabilizer generator,
\begin{equation}
	\exists k \geq j: \quad \acomm{\tilde{b}_j}{c_k} = 0,
	\tag{\ref{eq:con-3-a}}
\end{equation}
or completes a logical measurement,
\begin{equation}
	\exists \mu \in \langle b_1, \dots, b_{j-1} \rangle: \quad \mu \tilde{b}_j \in \left[ \overline{X} \right] \cup \left[ \overline{Y} \right] \cup \left[ \overline{Z} \right].
	\tag{\ref{eq:con-3-b}}
\end{equation}

\begin{proof}
For every $Y_{i,j}$, the proof that it anticommutes with at least one non-prior stabilizer generator proceeds analogously to the calculation of condition~2, except for the final qubit $Y_{2r-1,2m-1}$, which completes a logical operator:
\begin{equation}
	\begin{aligned}
	& \mu Y_{2r-1,2m-1} \\
	& = Y_{2r-1,2m-1} \prod_{{t=1}}^{r-1} b_{2t-1,2m-1} \prod_{{t=1}}^{m-1} b_{2r-1,2t-1}\\
	& = Y_{2r-1,2m-1} \prod_{{t=1}}^{r-1} X_{2t-1,2m-1} \prod_{{t=1}}^{m-1} Z_{2r-1,2t-1}\\
	& \in [\overline{Y}].
	\end{aligned}
\end{equation}
We go through the rest of the cases individually:

for $i \text{ odd} \wedge j<m \wedge j\neq 3$:

\begin{equation}
	\acomm{X_{i,j}}{c_{i+1,j+1}} = \acomm{X_{i,j}}{Z_{i+1,j}} = 0,
\end{equation}

for $i \text{ odd} \wedge j<m \wedge j=3$:

\begin{equation}
	\acomm{X_{i,j}}{c_{i+1,j+1}} = \acomm{X_{i,j}}{Z_{i+1,1}Z_{i+1,3}} = 0,
\end{equation}

for $i \text{ odd} \wedge j=m$:

\begin{equation}
	\mu Z_{i,m} = Z_{i,m} \prod_{t=1}^{m-1} b_{i,2t-1} = Z_{i,m} \prod_{t=1}^{m-1} Z_{i,2t-1} \in [\overline{Z}],
\end{equation}

and for $i \text{ even}$:
\begin{equation}
	\acomm{Z_{i,j}}{X_{i+1,j}} = 0.
\end{equation}

\end{proof}

The final two conditions are beyond the scope of a manual calculation. In principle, it could guide the development of an algorithmic approach to complete the algebraic proof. However, given the sheer size of the resulting expressions, such a computation is not expected to offer additional insight beyond the topological argument presented in the main text, and we therefore omit it here. For condition~4 we write the condition out explicitly to illustrate the complexity of the expression.

\textit{Condition~4}: The logical operators $\overline{X}_j$ and $\overline{Z}_j$ commute with every prior element of $\mathbb{B}$ for all $j \in \{ 1, \dots, n \}$:
\begin{equation}
	\forall k<j : \quad \comm{\overline{X}_j}{b_k} = 0,
	\tag{\ref{eq:con-4-a}}
\end{equation}
\begin{equation}
	\forall k<j : \quad \comm{\overline{Z}_j}{b_k} = 0.
	\tag{\ref{eq:con-4-b}}
\end{equation}

\begin{proof}
We write the condition explicitly $\forall \left( k , l \right) < \left( i,j \right)  \wedge i+j \text{ even} \wedge k+l \text{ even}$:
\begin{widetext}
\begin{equation}
\begin{aligned}
	\comm{\overline{X}_{i,j}}{b_{k,l}} & =
	\begin{dcases}
		\comm{\left( \prod_{t=1}^{\frac{i-1}{2}} X_{2t-1,2m-1} \right) \left( \prod_{t=1}^{\frac{2m-1-j}{2}} X_{i-1,j+2t-1} \right) \left( \prod_{t=0}^{\frac{2r-1-i}{2}} X_{i+2t,j} \right)}{Z_{k,l}} & \quad \text{if } i \text{ odd } \wedge k \text{ odd} \wedge l<m \\
		\comm{\overline{X}_{i,j}}{X_{k,l}} & \quad \text{if } i \text{ odd } \wedge k \text{ odd} \wedge l=m \\
		\comm{\overline{X}_{i,j}}{X_{k,l}} & \quad \text{if } i \text{ odd } \wedge k \text{ even} \\
		\\
		\comm{\left( \prod_{t=1}^{\frac{i}{2}-1} X_{2t-1,2m-1} \right) \left( \prod_{t=0}^{\frac{2m-j}{2}-1} X_{i,j+2t} \right) \left( \prod_{t=0}^{\frac{2r-i}{2}-1} X_{i+2t+1,j-1} \right)}{Z_{k,l}} & \quad \text{if } i \text{ even} \wedge k \text{ odd} \wedge l<m \\
		\comm{\overline{X}_{i,j}}{X_{k,l}} & \quad \text{if } i \text{ even} \wedge k \text{ odd} \wedge l=m \\
		\comm{\overline{X}_{i,j}}{X_{k,l}} & \quad \text{if } i \text{ even} \wedge k \text{ even} \\
	\end{dcases}
	\\
	& = 0,  \\
\end{aligned}
\end{equation}
\begin{equation}
\begin{aligned}
	\comm{\overline{Z}_{i,j}}{b_{k,l}} & =
	\begin{dcases}
		\comm{\overline{Z}_{i,j}}{Z_{k,l}} & \quad \text{if } i \text{ odd } \wedge k \text{ odd} \wedge l<m \\
		\comm{\prod_{t=0}^{\frac{2m}{2}} Z_{i,2t+1}}{X_{k,l}} & \quad \text{if } i \text{ odd } \wedge k \text{ odd} \wedge l=m \\
		\comm{\prod_{t=0}^{\frac{2m}{2}} Z_{i,2t+1}}{X_{k,l}} & \quad \text{if } i \text{ odd } \wedge k \text{ even} \\
		\\		
		\comm{\overline{Z}_{i,j}}{Z_{k,l}} & \quad \text{if } i \text{ even} \wedge k \text{ odd} \wedge l<m \\
		\comm{\left( \prod_{t=0}^{\frac{j}{2}-1} Z_{i-1,2t+1} \right) Z_{i,j} \left( \prod_{t=0}^{\frac{2m-j}{2}-1} Z_{i+1,j+2t+1} \right)}{X_{k,l}} & \quad \text{if } i \text{ even} \wedge k \text{ odd} \wedge l=m \\
		\comm{\left( \prod_{t=0}^{\frac{j}{2}-1} Z_{i-1,2t+1} \right) Z_{i,j} \left( \prod_{t=0}^{\frac{2m-j}{2}-1} Z_{i+1,j+2t+1} \right)}{X_{k,l}} & \quad \text{if } i \text{ even} \wedge k \text{ even} \\
	\end{dcases}
	\\
	& = 0. \\
\end{aligned}
\end{equation}
\end{widetext}
\end{proof}

\subsection{Proof: Optimal logical Bell measurement for the tree code}
\label{app:proof-tree}

As a first step, we need to index the qubits to impose an ordering on them. A natural approach would be to assign an index pair and impose a lexicographic order. However, similar to our approach in Sec.~\ref{sec:tree-code}, we refrain from explicitly assigning an index pair to each qubit. Instead, we implicitly assume that an arbitrary index~$v$ refers to a single qubit and impose the following order for any two indices~$v_1$ and~$v_2$:
\begin{equation}
    \depth(v_1) > \depth(v_2) \;\; \leftrightarrow \;\; v_1 < v_2.
\end{equation}

Recall that the graph state stabilizer generators are defined as $K_v = X_v \prod_{w \in \N(v)} Z_w$. Furthermore, we denoted by~$V$ the set of all indices, and by~$l$ the final index, such that $\forall v \in V \setminus \{l\} :\quad l > v$, where~$l$ is an arbitrary index satisfying $\depth(l) = 1$.

We define the sequences, as discussed in Sec.~\ref{sec:tree-code}, of the code stabilizer generators
\begin{equation}
	\begin{aligned}
		\mathbb{C} 	= ( c_j = K_j )_{j \in V \setminus \{l\}},
	\end{aligned}
\end{equation}
the measurement sequence
\begin{equation}
\begin{aligned}
	\mathbb{B} = ( b_j = Z_j )_{j \in V \setminus \{l\}},
\end{aligned}
\end{equation}
and the sequence of pairs of logical operators
\begin{equation}
	\mathbb{L} = \left( \left( \overline{X}_j , \overline{Z}_j \right) \right)_{j \in V},
\end{equation}
where
\begin{equation}
	\overline{X}_j = 
	\begin{dcases}
		Z_r \prod_{i=0}^{\frac{\depth(j)-2}{2}} K_{\anc(j,2i+1)}	& \quad \text{if } \depth(j) \ \text{even} \\ \\
		Z_r \prod_{i=0}^{\frac{\depth(j)-1}{2}} K_{\anc(j,2i)} 	& \quad \text{if } \depth(j) \ \text{odd},
	\end{dcases}
\label{eq:tree-sequence-x-logical}	
\end{equation}
\begin{equation}
	\overline{Z}_j = 
	\begin{dcases}
		K_r \prod_{i=0}^{\frac{\depth(j)-2}{2}} K_{\anc(j,2i)}		& \quad \text{if } \depth(j) \ \text{even} \\ \\
		K_r \prod_{i=0}^{\frac{\depth(j)-3}{2}} K_{\anc(j,2i+1)} 	& \quad \text{if } \depth(j) \ \text{odd}.
	\end{dcases}
\label{eq:tree-sequence-z-logical}
\end{equation}
We now verify the conditions of Thm.~\ref{thm:sufficient} individually.

\textit{Condition~1}: Each operator $b_j$ anticommutes with $c_j$:
\begin{equation}
	\forall j \in \{ 1 , \dots, n-1 \} : \quad \acomm{b_j}{c_j} = 0.
	\tag{\ref{eq:con-1}}
\end{equation}

\begin{proof}
We verify the condition	$\forall (i,j) \in V \setminus \{l\}$:
\begin{equation}
\begin{aligned}
	\acomm{b_j}{c_j}	& = \acomm{Z_j}{K_j} \\
						& = \acomm{Z_j}{X_j \prod_{w \in \N(j)} Z_w} \\
						& = 0.
\end{aligned}
\end{equation}

\end{proof}

\textit{Condition~2}: Each operator $b_j$ commutes with every later stabilizer generator:
\begin{equation}
	\forall k>j: \quad \comm{b_j}{c_k} = 0.
	\tag{\ref{eq:con-2}}
\end{equation}

\begin{proof}
We verify the condition	$\forall k > j$:
\begin{equation}
\begin{aligned}
	\comm{b_j}{c_k} 	& = \comm{Z_j}{K_k} \\
						& = \comm{Z_j}{X_k \prod_{w \in \N(k)} Z_w} \\
						& = 0.
\end{aligned}
\end{equation}

\end{proof}

\textit{Condition~3}: For all $j \in \{ 1, \dots, n-1 \}$ each operator in the set $\tilde{b}_j \in \{ X_j, Y_j, Z_j \} \setminus \{ b_j \}$ either anticommutes with at least one non-prior stabilizer generator,
\begin{equation}
	\exists k \geq j: \quad \acomm{\tilde{b}_j}{c_k} = 0,
	\tag{\ref{eq:con-3-a}}
\end{equation}
or completes a logical measurement,
\begin{equation}
	\exists \mu \in \langle b_1, \dots, b_{j-1} \rangle: \quad \mu \tilde{b}_j \in \left[ \overline{X} \right] \cup \left[ \overline{Y} \right] \cup \left[ \overline{Z} \right].
	\tag{\ref{eq:con-3-b}}
\end{equation}

\begin{proof}

Since $b_j=Z_j$ for all $j$, the set $\tilde{b}_j \in \{ X_j , Y_j , Z_j \} \setminus \{ b_j \}$ becomes $\{ X_j , Y_j \}$. For $\tilde{b}_j = Y_j$ it is straightforward to show that
\begin{equation}
	\forall j \in V \setminus \{l\}: \exists k \geq j : \acomm{Y_j}{c_k} = 0.
\end{equation}
Specifically, the element $c_j$ anticommutes with $Y_j$:
\begin{equation}
	\acomm{Y_j}{c_j} = \acomm{Y_j}{K_j} = 0.
\end{equation}
For all but the first level a similar argument applies to $\tilde{b} = X_j$ which anticommutes with the stabilizer generator associated with the parent $k = \anc(j,1)$, i.e., $\forall j \in \{ j \mid \depth(j) > 1 \}$ and $k = \anc(j,1)$:
\begin{equation}	
\begin{aligned}	
	\acomm{X_j}{c_{\anc(j,1)}} 	& = \acomm{X_j}{c_{\anc(j,1)}} \\
								& = \acomm{X_j}{X_k \prod_{w \in \N(k)} Z_w} \\
								& = 0.
\end{aligned}
\end{equation}
For the first level $\tilde{b} = X_j$ completes a logical $\overline{X}$ operator, i.e., $\forall j \in \{ j \mid \depth(j) = 1 \}$:
\begin{equation}
\begin{aligned}
	\mu X_j & = X_j \prod_{w \in \C(j)} b_w \\
			& = X_j \prod_{w \in \C(j)} Z_w \in [\overline{X}],
\end{aligned}
\end{equation}
where $\mu = \prod_{w \in \C(j)} b_w \in \langle b_1, \dots, b_l \rangle$.
\end{proof}

\textit{Condition~4}: The logical operators $\overline{X}_j$ and $\overline{Z}_j$ commute with every prior element of $\mathbb{B}$ for all $j \in \{ 1, \dots, n \}$:
\begin{equation}
	\forall k<j : \quad \comm{\overline{X}_j}{b_k} = 0,
	\tag{\ref{eq:con-4-a}}
\end{equation}
\begin{equation}
	\forall k<j : \quad \comm{\overline{Z}_j}{b_k} = 0.
	\tag{\ref{eq:con-4-b}}
\end{equation}

\begin{proof}
As a preliminary observation we note that all elements $b_k$ anticommute with any $K_{\anc(k,a)}$,
\begin{equation}
\begin{aligned}
	& \forall a>0 : \\
	& \comm{K_{\anc(k,a)}}{b_k} \\
	& \quad = \comm{X_{\anc(k,a)} \prod_{w \in \N(\anc(k,a))} Z_w}{Z_k} \\
	& \quad = 0.
\end{aligned}
	\label{eq:anc-comm}
\end{equation}

We now compute both commutators for the two cases of the logical operators, corresponding to odd and even $\depth(j)$. The first two steps of each calculation are a standard commutator identity and inserting $b_k = Z_k$. The third step is observing that the second commutator vanishes and applying Eq.~\eqref{eq:anc-comm}. We verify the condition for the first two cases, $\forall k<j \wedge \depth(j) \text{ even}$:
\begin{equation}
\begin{aligned}
	\comm{\overline{X}_j}{b_k} 	& = \comm{Z_r \prod_{i=0}^{\frac{\depth(j)-2}{2}} K_{\anc(j,2i+1)}}{b_k} \\
								& =	Z_r \comm{\prod_{i=0}^{\frac{\depth(j)-2}{2}} K_{\anc(j,2i+1)}}{b_k} \\
								& \phantom{= Z_r} + \comm{Z_r}{b_k} \prod_{i=0}^{\frac{\depth(j)-2}{2}} K_{\anc(j,2i+1)} \\
								& =	Z_r \comm{\prod_{i=0}^{\frac{\depth(j)-2}{2}} K_{\anc(j,2i+1)}}{Z_k} \\
								& \phantom{= Z_r} + \comm{Z_r}{Z_k} \prod_{i=0}^{\frac{\depth(j)-2}{2}} K_{\anc(j,2i+1)} \\
								& =	0,
\end{aligned}
\end{equation}
\newpage
and $\forall k<j \wedge \depth(j) \text{ odd}$:
\begin{equation}
\begin{aligned}
	\comm{\overline{X}_j}{b_k} 	& = \comm{Z_r \prod_{i=0}^{\frac{\depth(j)-1}{2}} K_{\anc(j,2i)}}{b_k} \\
								& = Z_r \comm{\prod_{i=0}^{\frac{\depth(j)-1}{2}} K_{\anc(j,2i)}}{b_k} \\
								& \phantom{= Z_r} + \comm{Z_r}{b_k} \prod_{i=0}^{\frac{\depth(j)-1}{2}} K_{\anc(j,2i)} \\
								& = Z_r \comm{\prod_{i=0}^{\frac{\depth(j)-1}{2}} K_{\anc(j,2i)}}{Z_k} \\
								& \phantom{= Z_r}+ \comm{Z_r}{Z_k} \prod_{i=0}^{\frac{\depth(j)-1}{2}} K_{\anc(j,2i)} \\
								& = Z_r \comm{K_{\anc(j,0)}}{Z_k} \\
								& = Z_r \comm{K_{j}}{Z_k} \\
								& = Z_r \comm{X_j \prod_{w \in \N(j)} Z_w}{Z_k} \\
								& = 0.
\end{aligned}
\end{equation}
In the final step, we used that the commutator vanishes because $k < j$, and therefore $X_j$ and $Z_k$ act on different qubits. We verify the condition for the third case, $\forall k<j \wedge \depth(j) \text{ even}$:
\begin{equation}
\begin{aligned}
	\comm{\overline{Z}_j}{b_k} 	& = \comm{K_r \prod_{i=0}^{\frac{\depth(j)-2}{2}} K_{\anc(j,2i)}}{b_k} \\
								& = K_r \comm{\prod_{i=0}^{\frac{\depth(j)-2}{2}} K_{\anc(j,2i)}}{b_k} \\
								& \phantom{= K_r}+ \comm{K_r}{b_k} \prod_{i=0}^{\frac{\depth(j)-2}{2}} K_{\anc(j,2i)} \\
								& = K_r \comm{\prod_{i=0}^{\frac{\depth(j)-2}{2}} K_{\anc(j,2i)}}{Z_k} \\
								& \phantom{= K_r}+ \comm{K_r}{Z_k} \prod_{i=0}^{\frac{\depth(j)-2}{2}} K_{\anc(j,2i)} \\
								& = K_r \comm{K_{\anc(j,0)}}{Z_k} \\
								& = K_r \comm{K_{j}}{Z_k} \\
								& = K_r \comm{X_j \prod_{w \in \N(j)} Z_w}{Z_k} \\
								& = 0.
\end{aligned}
\end{equation}
Again, in the final step, we used that the commutator vanishes because $k < j$, and therefore $X_j$ and $Z_k$ act on different qubits. We verify the condition for the last case, $\forall k<j \wedge \depth(j) \text{ odd}$:
\begin{equation}
\begin{aligned}
	\comm{\overline{Z}_j}{b_k} 	& = \comm{K_r \prod_{i=0}^{\frac{\depth(j)-3}{2}} K_{\anc(j,2i+1)}}{b_k} \\
								& = K_r \comm{\prod_{i=0}^{\frac{\depth(j)-3}{2}} K_{\anc(j,2i+1)}}{b_k} \\
								& \phantom{= K_r} + \comm{K_r}{b_k} \prod_{i=0}^{\frac{\depth(j)-3}{2}} K_{\anc(j,2i+1)} \\
								& = K_r \comm{\prod_{i=0}^{\frac{\depth(j)-3}{2}} K_{\anc(j,2i+1)}}{Z_k} \\
								& \phantom{= Z_r} + \comm{K_r}{Z_k} \prod_{i=0}^{\frac{\depth(j)-3}{2}} K_{\anc(j,2i+1)} \\
								& = 0.
\end{aligned}
\end{equation}

\end{proof}

\textit{Condition~5}: We decompose the logical operators into single-qubit Pauli operators to formulate the last condition:
\begin{equation}
	\overline{X}_j = \bigotimes_{t=1}^{n} u_{j,t}, \quad \text{where } u_{j,t} \in \{ I, X, Y, Z \},
	\tag{\ref{eq:logical-decomp-x}}
\end{equation}
\begin{equation}
	\overline{Z}_j = \bigotimes_{t=1}^{n} v_{j,t}, \quad \text{where } v_{j,t} \in \{ I, X, Y, Z \}.
	\tag{\ref{eq:logical-decomp-z}}
\end{equation}
The logical operators $\overline{X}_j$ and $\overline{Z}_j$ anticommute only in $j$:
\begin{equation}
	\forall j \in \{ 1, \dots, n \} : \acomm{u_{j,j}}{v_{j,j}} = 0,
	\tag{\ref{eq:double-info}}
\end{equation}
\begin{equation}
	\forall k,j \in \{1, \dots, n \} \wedge k \neq j : \comm{u_{j,k}}{v_{j,k}} = 0.
	\tag{\ref{eq:double-info-b}}
\end{equation}

\begin{proof}

For each vertex $v \in V$, we define the path set $P(v) \subseteq V$ to be the unique set of vertices forming the simple path from root $r$ to $v$. We take the decomposition from Eqs.~\eqref{eq:tree-sequence-x-logical} and~\eqref{eq:tree-sequence-z-logical}:

\begin{equation}
	u_{j,t} = 
	\begin{dcases}
		X_t		& \quad \text{if } t \in P(j) \wedge \depth(t) \text{ odd} \\ \\
		Z_t 	& \quad \text{if } t \notin P(j) \setminus \{j\} \\
				& \quad \quad \wedge \ \anc(t,1) \in P(j) \\
				& \quad \quad \wedge \ \depth(t) \text{ even} \\ \\
		I_t		& \quad \text{else, } \\
	\end{dcases}
	\label{eq:tree-logical-x-decomposition}
\end{equation}

\begin{equation}
	v_{j,t} = 
	\begin{dcases}
		X_t		& \quad \text{if } t \in P(j) \wedge \depth(t) \text{ even} \\ \\
		Z_t 	& \quad \text{if } t \notin P(j) \setminus \{j\} \\
				& \quad \quad \wedge \ \anc(t,1) \in P(j) \\
				& \quad \quad \wedge \ \depth(t) \text{ odd} \\ \\
		I_t		& \quad \text{else. } \\
	\end{dcases}
	\label{eq:tree-logical-z-decomposition}
\end{equation}
It is straightforward to verify that the logical operators anticommute for each $j$. We first consider the case of even $j$,
\begin{equation}
\begin{aligned}
	& \forall j \in V \wedge \depth(j) \text{ even} : \\
	& \acomm{u_{j,j}}{v_{j,j}} = \acomm{Z_j}{X_j} = 0,
\end{aligned}
\end{equation}
and then the case of odd $j$,
\begin{equation}
\begin{aligned}
	& \forall j \in V \wedge \depth(j) \text{ odd} : \\
	& \acomm{u_{j,j}}{v_{j,j}} = \acomm{X_j}{Z_j} = 0.
\end{aligned}
\end{equation}
It is also straightforward to argue by contradiction that the logical operators do not anticommute in any qubit other than~$j$. For the commutator $\comm{u_{j,k}}{v_{j,k}}$ to be non-zero, the operators $u_{j,k}$ and $v_{j,k}$ must be different single-qubit Pauli operators, e.g., $u_{j,k} = X_k$ and $v_{j,k} = Z_k$. In this case, Eq.~\eqref{eq:tree-logical-x-decomposition} implies that $k \in P(j)$, while Eq.~\eqref{eq:tree-logical-z-decomposition} implies that $k \in P(j) \setminus \{j\}$. The only index satisfying both conditions is $k = j$, completing the argument.

Having verified all conditions of Thm.~\ref{thm:sufficient} we conclude that our scheme is optimal.

\end{proof}

\subsection{Proof: Optimal logical Bell measurement for the Steane code}
\label{app:proof-steane}

First, we define the sequences, as discussed in Sec.~\ref{sec:steane-code}, of the code stabilizer generators
\begin{equation}
	\begin{aligned}
		\mathbb{C} 	= ( & ZZZZIII, IXIXXXI, IIXXIXX, \\
				& XIIXXIX, IZIZZZI, IIZZIZZ ),
	\end{aligned}
\end{equation}
the measurement sequence
\begin{equation}
\begin{aligned}
	\mathbb{B} 	& = ( b_{j} \}_{j \in \{ 1, \dots, n-1 ) } \\
				& = ( X_1, Z_2, Z_3, Z_4, X_5, X_6 ),
\end{aligned}
\end{equation}
and the sequence of pairs of logical operators
\begin{equation}
	\mathbb{L} = \left( \left( \overline{X}_{i,j} , \overline{Z}_{i,j} \right) \right)_{j \in \{ 1, \dots, n \} },
\end{equation}
where
\begin{equation}
\begin{aligned}
	( \overline{X}_{j} )_{j \in \{ 1, \dots, n \} } = ( 	& XXIIXII, XXIIXII, \\
															& XIXIIIX, XIIXIXI, \\
															& IIIIXXX, IIIIXXX, \\
															& IIIIXXX ),
\end{aligned}
	\label{eq:steane-sequence-x-logical}
\end{equation}
\begin{equation}
\begin{aligned}
	( \overline{Z}_{j} )_{j \in \{ 1, \dots, n \} } = (	& ZIZIIIZ, IZIZIIZ, \\
															& IIZZZII, IZIZIIZ, \\
															& IIZZZII, IZZIIZI, \\
															& IZIZIIZ ).
\end{aligned}
	\label{eq:steane-sequence-z-logical}
\end{equation}
Note that the sequences in Eqs.~\eqref{eq:steane-sequence-x-logical} and~\eqref{eq:steane-sequence-z-logical} are just one of the possible solutions presented in Sec.~\ref{sec:steane-code}. We now verify the conditions of Thm.~\ref{thm:sufficient} individually.

\textit{Condition~1}: Each operator $b_j$ anticommutes with $c_j$:
\begin{equation}
	\forall j \in \{ 1 , \dots, n-1 \} : \quad \acomm{b_j}{c_j} = 0.
	\tag{\ref{eq:con-1}}
\end{equation}

\begin{proof}

We verify the condition for each index individually:
\begin{equation}
\begin{aligned}
	\acomm{b_1}{c_1} & = \acomm{X_1}{ZZZZIII} & = 0, \\
	\acomm{b_2}{c_2} & = \acomm{Z_2}{IXIXXXI} & = 0, \\
	\acomm{b_3}{c_3} & = \acomm{Z_3}{IIXXIXX} & = 0, \\
	\acomm{b_4}{c_4} & = \acomm{Z_4}{XIIXXIX} & = 0, \\
	\acomm{b_5}{c_5} & = \acomm{X_5}{IZIZZZI} & = 0, \\
	\acomm{b_6}{c_6} & = \acomm{X_6}{IIZZIZZ} & = 0.
	\end{aligned}
\end{equation}

\end{proof}

\textit{Condition~2}: Each operator $b_j$ commutes with every later stabilizer generator:
\begin{equation}
	\forall k>j: \quad \comm{b_j}{c_k} = 0.
	\tag{\ref{eq:con-2}}
\end{equation}

\begin{proof}

Again, we verify the condition for each index individually:
\begin{equation}
\begin{aligned}
	\comm{b_1}{c_2} & = \comm{X_1}{IXIXXXI} & = 0, \\
	\comm{b_1}{c_3} & = \comm{X_1}{IIXXIXX} & = 0, \\
	\comm{b_1}{c_4} & = \comm{X_1}{XIIXXIX} & = 0, \\
	\comm{b_1}{c_5} & = \comm{X_1}{IZIZZZI} & = 0, \\
	\comm{b_1}{c_6} & = \comm{X_1}{IIZZIZZ} & = 0, \\
	\comm{b_2}{c_3} & = \comm{Z_2}{IIXXIXX} & = 0, \\
	\comm{b_2}{c_4} & = \comm{Z_2}{XIIXXIX} & = 0, \\
	\comm{b_2}{c_5} & = \comm{Z_2}{IZIZZZI} & = 0, \\
	\comm{b_2}{c_6} & = \comm{Z_2}{IIZZIZZ} & = 0, \\
	\comm{b_3}{c_4} & = \comm{Z_3}{XIIXXIX} & = 0, \\
	\comm{b_3}{c_5} & = \comm{Z_3}{IZIZZZI} & = 0, \\
	\comm{b_3}{c_6} & = \comm{Z_3}{IIZZIZZ} & = 0, \\
	\comm{b_4}{c_5} & = \comm{Z_4}{IZIZZZI} & = 0, \\
	\comm{b_4}{c_6} & = \comm{Z_4}{IIZZIZZ} & = 0, \\
	\comm{b_5}{c_6} & = \comm{Z_5}{IIZZIZZ} & = 0.
\end{aligned}
\end{equation}

\end{proof}

\textit{Condition~3}: For all $j \in \{ 1, \dots, n-1 \}$ each operator in the set $\tilde{b}_j \in \{ X_j, Y_j, Z_j \} \setminus \{ b_j \}$ either anticommutes with at least one non-prior stabilizer generator,
\begin{equation}
	\exists k \geq j: \quad \acomm{\tilde{b}_j}{c_k} = 0,
	\tag{\ref{eq:con-3-a}}
\end{equation}
or completes a logical measurement,
\begin{equation}
	\exists \mu \in \langle b_1, \dots, b_{j-1} \rangle: \quad \mu \tilde{b}_j \in \left[ \overline{X} \right] \cup \left[ \overline{Y} \right] \cup \left[ \overline{Z} \right].
	\tag{\ref{eq:con-3-b}}
\end{equation}

\begin{proof}

For $j=1$, it is trivial that the operators $\tilde{b}_1 \in \{ X_1 , Y_1 , Z_1 \} \setminus \{ b_1 \} = \{ X_1 , Z_1 \}$, anticommute with the non-prior stabilizer generators $XIIXXIX$ and $ZZZZIII$. Since $b_j = Z_j$ for $j \in \{ 2,3,4 \}$ the set $\tilde{b}_j \in \{ X_j , Y_j , Z_j \} \setminus \{ b_j \}$ becomes $\{ X_j , Y_j \}$. In Sec.~\ref{sec:steane-code}, we identified, for each of these indices, the non-prior stabilizer generators that anticommute with the corresponding elements of~$\mathbb{B}$, which we now formally list, $\forall j \in \{ 2,3,4 \} \wedge \tilde{b}_j \in \{ X_j , Y_j \}$:
\begin{equation}
\begin{aligned}
	\acomm{\tilde{b}_2}{c_5} & = \acomm{\tilde{b}_2}{IZIZZZI} & = 0, \\
	\acomm{\tilde{b}_3}{c_6} & = \acomm{\tilde{b}_3}{IIZZIZZ} & = 0, \\
	\acomm{\tilde{b}_4}{c_5} & = \acomm{\tilde{b}_4}{IZIZZZI} & = 0.
\end{aligned}
\end{equation}
Since $b_j = X_j$ for $j \in \{ 5,6 \}$ the set $\tilde{b}_j \in \{ X_j , Y_j , Z_j \} \setminus \{ b_j \}$ becomes $\{ Y_j , Z_j \}$. Again, in Sec.~\ref{sec:steane-code}, we identified, for each of these indices, the non-prior stabilizer generators that anticommute with the corresponding elements of~$\mathbb{B}$, which we now formally list, $\forall j \in \{ 5,6 \} \wedge \tilde{b}_j \in \{ Y_j , Z_j \}$:
\begin{equation}
\begin{aligned}
	\acomm{\tilde{b}_5}{c_5} & = \acomm{\tilde{b}_}{IZIZZZI} & = 0, \\
	\acomm{\tilde{b}_6}{c_6} & = \acomm{\tilde{b}_}{IIZZIZZ} & = 0.
\end{aligned}
\end{equation}
Finally, again as described Sec.~\ref{sec:steane-code}, the single-qubit operators on the final qubit all complete a logical operator:
\begin{equation}
	\mu_x X_7 = X_7 b_5 b_6 = IIIIXXX \in [\overline{X}],
\end{equation}
\begin{equation}
	\mu_y Y_7 = Y_7 b_2 b_4 b_5 b_6 = IZIZXXY \in [\overline{Y}],
\end{equation}
\begin{equation}
	\mu_z Z_7 = Z_7 b_2 b_4 = IZIZIIZ \in [\overline{Z}],
\end{equation}
where $\mu_x, \mu_y, \mu_z \in \langle b_1, \dots, b_6 \rangle$.

\end{proof}

\textit{Condition~4}: The logical operators $\overline{X}_j$ and $\overline{Z}_j$ commute with every prior element of $\mathbb{B}$ for all $j \in \{ 1, \dots, n \}$:
\begin{equation}
	\forall k<j : \quad \comm{\overline{X}_j}{b_k} = 0,
	\tag{\ref{eq:con-4-a}}
\end{equation}
\begin{equation}
	\forall k<j : \quad \comm{\overline{Z}_j}{b_k} = 0.
	\tag{\ref{eq:con-4-b}}
\end{equation}

\begin{proof}

We verify the condition for each index pair individually:
\begin{equation}
\begin{aligned}
	\comm{\overline{X}_2}{b_1} & = \comm{XXIIXII}{X_1} & = 0, \\
\end{aligned}
\end{equation}

\begin{equation}
\begin{aligned}
	\comm{\overline{X}_3}{b_1} & = \comm{XIXIIIX}{X_1} & = 0, \\
	\comm{\overline{X}_3}{b_2} & = \comm{XIXIIIX}{Z_2} & = 0, \\
\end{aligned}
\end{equation}

\begin{equation}
\begin{aligned}
	\comm{\overline{X}_4}{b_1} & = \comm{XIIXIXI}{X_1} & = 0, \\
	\comm{\overline{X}_4}{b_2} & = \comm{XIIXIXI}{Z_2} & = 0, \\
	\comm{\overline{X}_4}{b_3} & = \comm{XIIXIXI}{Z_3} & = 0, \\
\end{aligned}
\end{equation}

\begin{equation}
\begin{aligned}
	\comm{\overline{X}_5}{b_1} & = \comm{IIIIXXX}{X_1} & = 0, \\
	\comm{\overline{X}_5}{b_2} & = \comm{IIIIXXX}{Z_2} & = 0, \\
	\comm{\overline{X}_5}{b_3} & = \comm{IIIIXXX}{Z_3} & = 0, \\
	\comm{\overline{X}_5}{b_4} & = \comm{IIIIXXX}{Z_4} & = 0, \\
\end{aligned}
\end{equation}

\begin{equation}
\begin{aligned}
	\comm{\overline{X}_6}{b_1} & = \comm{IIIIXXX}{X_1} & = 0, \\
	\comm{\overline{X}_6}{b_2} & = \comm{IIIIXXX}{Z_2} & = 0, \\
	\comm{\overline{X}_6}{b_3} & = \comm{IIIIXXX}{Z_3} & = 0, \\
	\comm{\overline{X}_6}{b_4} & = \comm{IIIIXXX}{Z_4} & = 0, \\
	\comm{\overline{X}_6}{b_5} & = \comm{IIIIXXX}{X_5} & = 0, \\
\end{aligned}
\end{equation}

\begin{equation}
\begin{aligned}
	\comm{\overline{X}_7}{b_1} & = \comm{IIIIXXX}{X_1} & = 0, \\
	\comm{\overline{X}_7}{b_2} & = \comm{IIIIXXX}{Z_2} & = 0, \\
	\comm{\overline{X}_7}{b_3} & = \comm{IIIIXXX}{Z_3} & = 0, \\
	\comm{\overline{X}_7}{b_4} & = \comm{IIIIXXX}{Z_4} & = 0, \\
	\comm{\overline{X}_7}{b_5} & = \comm{IIIIXXX}{X_5} & = 0, \\
	\comm{\overline{X}_7}{b_6} & = \comm{IIIIXXX}{X_6} & = 0. \\
\end{aligned}
\end{equation}

\begin{equation}
\begin{aligned}
	\comm{\overline{Z}_2}{b_1} & = \comm{IZIZIIZ}{X_1} & = 0, \\
\end{aligned}
\end{equation}

\begin{equation}
\begin{aligned}
	\comm{\overline{Z}_3}{b_1} & = \comm{IIZZZII}{X_1} & = 0, \\
	\comm{\overline{Z}_3}{b_2} & = \comm{IIZZZII}{Z_2} & = 0, \\
\end{aligned}
\end{equation}

\begin{equation}
\begin{aligned}
	\comm{\overline{Z}_4}{b_1} & = \comm{IZIZIIZ}{X_1} & = 0, \\
	\comm{\overline{Z}_4}{b_2} & = \comm{IZIZIIZ}{Z_2} & = 0, \\
	\comm{\overline{Z}_4}{b_3} & = \comm{IZIZIIZ}{Z_3} & = 0, \\
\end{aligned}
\end{equation}

\begin{equation}
\begin{aligned}
	\comm{\overline{Z}_5}{b_1} & = \comm{IIZZZII}{X_1} & = 0, \\
	\comm{\overline{Z}_5}{b_2} & = \comm{IIZZZII}{Z_2} & = 0, \\
	\comm{\overline{Z}_5}{b_3} & = \comm{IIZZZII}{Z_3} & = 0, \\
	\comm{\overline{Z}_5}{b_4} & = \comm{IIZZZII}{Z_4} & = 0, \\
\end{aligned}
\end{equation}

\begin{equation}
\begin{aligned}
	\comm{\overline{Z}_6}{b_1} & = \comm{IZZIIZI}{X_1} & = 0, \\
	\comm{\overline{Z}_6}{b_2} & = \comm{IZZIIZI}{Z_2} & = 0, \\
	\comm{\overline{Z}_6}{b_3} & = \comm{IZZIIZI}{Z_3} & = 0, \\
	\comm{\overline{Z}_6}{b_4} & = \comm{IZZIIZI}{Z_4} & = 0, \\
	\comm{\overline{Z}_6}{b_5} & = \comm{IZZIIZI}{X_5} & = 0, \\
\end{aligned}
\end{equation}

\begin{equation}
\begin{aligned}
	\comm{\overline{Z}_7}{b_1} & = \comm{IZIZIIZ}{X_1} & = 0, \\
	\comm{\overline{Z}_7}{b_2} & = \comm{IZIZIIZ}{Z_2} & = 0, \\
	\comm{\overline{Z}_7}{b_3} & = \comm{IZIZIIZ}{Z_3} & = 0, \\
	\comm{\overline{Z}_7}{b_4} & = \comm{IZIZIIZ}{Z_4} & = 0, \\
	\comm{\overline{Z}_7}{b_5} & = \comm{IZIZIIZ}{X_5} & = 0, \\
	\comm{\overline{Z}_7}{b_6} & = \comm{IZIZIIZ}{X_6} & = 0. \\
\end{aligned}
\end{equation}

\end{proof}

\textit{Condition~5}: We decompose the logical operators into single-qubit Pauli operators to formulate the last condition:
\begin{equation}
	\overline{X}_j = \bigotimes_{t=1}^{n} u_{j,t}, \quad \text{where } u_{j,t} \in \{ I, X, Y, Z \},
	\tag{\ref{eq:logical-decomp-x}}
\end{equation}
\begin{equation}
	\overline{Z}_j = \bigotimes_{t=1}^{n} v_{j,t}, \quad \text{where } v_{j,t} \in \{ I, X, Y, Z \}.
	\tag{\ref{eq:logical-decomp-z}}
\end{equation}
The logical operators $\overline{X}_j$ and $\overline{Z}_j$ anticommute only in $j$:
\begin{equation}
	\forall j \in \{ 1, \dots, n \} : \acomm{u_{j,j}}{v_{j,j}} = 0,
	\tag{\ref{eq:double-info}}
\end{equation}
\begin{equation}
	\forall k,j \in \{1, \dots, n \} \wedge k \neq j : \comm{u_{j,k}}{v_{j,k}} = 0.
	\tag{\ref{eq:double-info-b}}
\end{equation}

\begin{proof}
In principle the decomposition can be taken from Eqs.~\eqref{eq:steane-sequence-x-logical} and~\eqref{eq:steane-sequence-z-logical}, which can be used to calculate that the logical operators $\overline{X}_j$ and $\overline{Z}_j$ anticommute only in $j$. Since the result can be readily seen from Eqs.~\eqref{eq:steane-sequence-x-logical} and~\eqref{eq:steane-sequence-z-logical}, we omit the straightforward computation of these 49 equations.

\end{proof}

Having verified all conditions of Thm.~\ref{thm:sufficient} we conclude that our scheme is optimal.

\bibliographystyle{apsrev4-2} 
\bibliography{bib-refs} 

\end{document}